\documentclass[11pt]{article}
\usepackage{amsfonts}
\usepackage{amssymb,amsmath}
\usepackage{amsthm,amsgen}
\usepackage{graphicx,epsfig}
\usepackage[mathscr]{eucal}
\usepackage{rawfonts}
\usepackage{pictex}
\usepackage{color}
\newtheorem{thm}{THEOREM}[section]
\newtheorem*{thm*}{THEOREM}
\newtheorem{cor}[thm]{COROLLARY}
\newtheorem{lem}[thm]{LEMMA}
\newtheorem{prop}[thm]{PROPOSITION}
\theoremstyle{definition}
\newtheorem{defn}[thm]{Definition}
\newtheorem{rem}[thm]{Remark}

\newcommand{\subsubsubsection}[1]{\textbf{#1}}

\newcounter{subequation}
	\newenvironment{subequation}%
	{\addtocounter{equation}{-1}%
	\stepcounter{subequation}%
	\begin{equation}}%
	{\end{equation}%
}

\newcommand{\sgn}{\mbox{sgn}}

\newcommand{\bA}{\mathbf{A}}

\newcommand{\ba}{\mathbf{a}}

\newcommand{\be}{\mathbf{e}}
\newcommand{\bF}{\mathbf{F}}

\newcommand{\bg}{\mathbf{g}}

\newcommand{\beq}{\begin{equation}}
\newcommand{\eeq}{\end{equation}}
\newcommand{\bseq}{\begin{subequation}}
\newcommand{\eseq}{\end{subequation}}

\newcommand{\refeq}[1]{(\ref{#1})}

\newcommand{\sign}{{\mathrm{sign\,}}}

\newcommand{\fa}{\mathfrak{a}}

\newcommand{\fh}{\mathfrak{h}}
\newcommand{\fl}{\mathfrak{l}}
\newcommand{\fm}{\mathfrak{m}}

\newcommand{\ft}{\mathfrak{t}}

\newcommand{\fB}{\mathfrak{B}}

\newcommand{\fD}{\mathfrak{D}}

\newcommand{\fH}{\mathfrak{H}}

\newcommand{\fM}{\mathfrak{M}}
\newcommand{\fN}{\mathfrak{N}}
\newcommand{\fO}{\mathfrak{O}}
\newcommand{\fP}{\mathfrak{P}}
\newcommand{\fQ}{\mathfrak{Q}}
\newcommand{\fR}{\mathfrak{R}}
\newcommand{\fS}{\mathfrak{S}}
\newcommand{\fT}{\mathfrak{T}}
\newcommand{\fV}{\mathfrak{V}}

\newcommand{\p}{\partial}

\newcommand{\cF}{\Phi}

\newcommand{\cA}{\mathcal{A}}

\newcommand{\cR}{\mathcal{R}}
\newcommand{\cC}{\mathcal{C}}

\newcommand{\cE}{{\mathcal E}}

\newcommand{\cH}{{\mathcal H}}
\newcommand{\cI}{{\mathcal I}}

\newcommand{\cK}{{\mathcal K}}

\newcommand{\cM}{{\mathcal M}}
\newcommand{\cN}{{\mathcal N}}
\newcommand{\cO}{\mathcal{O}}
\newcommand{\cP}{{\mathcal P}}

\newcommand{\cS}{{\mathcal S}}

\newcommand{\cW}{\mathcal{W}}

\newcommand{\sS}{\mathcal{S}}
\newcommand{\cZ}{\mathcal{Z}}
\newcommand{\Cset}{{\mathbb C}}

\newcommand{\Rset}{{\mathbb R}}
\newcommand{\Sset}{{\mathbb S}}
\newcommand{\Zset}{{\mathbb Z}}

\newcommand{\la}{\lambda}
\newcommand{\La}{\Lambda}
\newcommand{\de}{\delta}
\newcommand{\De}{\varpi}
\newcommand{\al}{\alpha}

\newcommand{\ga}{\gamma}
\newcommand{\Ga}{\Gamma}
\newcommand{\ep}{\epsilon}

\newcommand{\ka}{\kappa}
\newcommand{\om}{\omega}
\newcommand{\Om}{\Omega}
\newcommand{\si}{\sigma}
\newcommand{\Si}{\Sigma}
\newcommand{\nab}{\nabla}
\newcommand{\half}{\frac{1}{2}}

\newcommand{\diag}{\mbox{diag}}

\newcommand{\bna}{\begin{eqnarray}}
\newcommand{\ena}{\end{eqnarray}}
\newcommand{\bea}{\begin{eqnarray*}}
\newcommand{\eea}{\end{eqnarray*}}
\newcommand{\ben}{\begin{enumerate}}
\newcommand{\een}{\end{enumerate}}
\newcommand{\bi}{\begin{itemize}}
\newcommand{\ei}{\end{itemize}}

\newcommand{\RR}{{\mathbb R}}

\newcommand{\rad}{\mathrm{rad}}

\newcommand{\sH}{{\sf H}}

\newcommand{\rcyl}{{\mathcal{Z}}}
\begin{document}
%
%
\title{The Dirac point electron in zero-gravity Kerr--Newman spacetime}
\author{\normalsize\sc{M. K.-H. Kiessling and A. S. Tahvildar-Zadeh}\\
	{$\phantom{nix}$}\\[-0.1cm]
        \normalsize Department of Mathematics\\[-0.1cm]
	Rutgers, The State University of New Jersey\\[-0.1cm]
	110 Frelinghuysen Rd., Piscataway, NJ 08854}
\vspace{-0.3cm}
\date{Orig. submission: Nov. 13, 2014; Revised version: March 21, 2015} 
\maketitle
\begin{abstract} 
\noindent
	Dirac's wave equation for a point electron in the topologically nontrivial maximal analytically extended
electromagnetic Kerr--Newman spacetime is studied in a limit $G\to 0$, where $G$ is Newton's constant of universal 
gravitation.
	The following results are obtained:
the formal Dirac Hamiltonian on the static spacelike slices is essentially self-adjoint;
the spectrum of the self-adjoint extension is symmetric about zero, featuring a 
continuum with a gap about zero that, under two smallness conditions, contains a point spectrum.
       The symmetry result extends to the Dirac operator on a generalization of the zero-$G$ Kerr--Newman 
spacetime with different electric-monopole / magnetic-dipole-moment ratio.
\end{abstract}


\vfill
\hrule
\smallskip
\copyright{2014. The authors.}
\newpage

\section{Introduction}
	In 1976 Chandrasekhar \cite{ChandraDIRACinKERRgeom,ChandraDIRACinKERRgeomERR},
Page \cite{Page76}, and Toop \cite{Too76}		
showed that Dirac's equation for a point electron in the Kerr--Newman spacetime \cite{NCCEPT65} separates essentially 
completely \footnote{In contrast to the familiar separation-of-variables results for, say,
 the Laplacian in a rectangular box or a cylinder or a sphere, Chandrasekhar, Page, and Toop 
 obtained a system of ODEs for functions of only one variable each which is \emph{not} of 
 triangular structure, and so cannot be  solved one equation at a time.}
in oblate spheroidal coordinates.
	Although this remarkable discovery enabled detailed mathematical studies of the behavior of a Dirac electron in a charged, 
rotating black hole spacetime
\cite{KalMil1992,BelMar99,FinsterETalDperDNEerr,FinsterETalDperDNE,FinsterETalDKNa,FinsterETalDKNb,WINKLMEIERa,WINKLMEIERb,WINKLMEIERc,BelCac2010}
(see also \cite{ChandraBOOKonBH} for neutral rotating black holes), there are perplexing conceptual issues which await clarification.
	Beside those that hark back to the enigmatic quantum-mechanical meaning of Dirac's equation in Minkowski spacetime, see 
\cite{Thaller} and \cite{KieTah14b}, 
serious new issues arise because of the physically somewhat questionable character of the Kerr--Newman solution, 
unveiled by Carter \cite{Car68}; see also \cite{ONeillBOOK,HehlREVIEW}.

        Namely, the maximal analytical extension of the stationary axisymmetric Kerr--Newman spacetime has a very strong
curvature singularity on a timelike\footnote{In a 
  limiting sense of course, since the metric is singular on this surface.}
cylindrical surface whose cross-section with constant-$t$ hypersurfaces is a circle; here, $t$ is a
coordinate pertinent to the \emph{asymptotically (at spacelike $\infty$) timelike} Killing field that encodes the stationarity of the
``outer regions'' of the Kerr--Newman spacetime.
	This circle is commonly referred to as \emph{the ``ring'' singularity}.
	The region near the ring is especially pathological since it harbors closed timelike 
loops.\footnote{The timelike ring singularity of the Kerr--Newman manifold is itself the limit 
 of closed timelike loops, for which reason it is not possible to interpret this singular source 
 of the stationary and axisymmetric Kerr--Newman electromagnetic fields
 outside of the outer ergosphere horizon as a ``rotating charged ring.'' 
 For a careful analysis of the ring sources of the electromagnetic z$\,G$KN fields, see \cite{zGKN}.}
	Carter \cite{Car68} also showed that  the maximal analytically extended Kerr--Newman manifold
is ``cross-linked through the ring.''\footnote{The complement of a wedding ring in ordinary three-dimensional Euclidean space
		is topologically non-trivial, too, but ``looping through the ring once brings you back to where you began;''
		in a spacelike slice of the maximal analytically extended Kerr--Newman spacetime ``you need to
		loop through the ring twice to get back to square one.''}
 Interestingly, this non-trivial topology was discovered a few years 
earlier in a family of static vacuum spacetimes by Zipoy \cite{Zip66}, who completely described their maximal analytical 
extension.  
  Since Zipoy seems to have been the first to discover this non-trivial topology in exact spacetime solutions to Einstein's 
vacuum equations, we henceforth will refer to it as the \emph{Zipoy topology}.
        Carter showed that the Zipoy topology survives the vanishing-charge
limit of the Kerr--Newman manifold, which yields the maximal analytic extension \cite{BoyLin67} of Kerr's solution \cite{Ker63}
to Einstein's vacuum equations (${R}_{\mu\nu}=0$), cf.~\cite{HawEll73}.
	He furthermore showed that this topology also survives the vanishing-mass
 limit of the Kerr manifold, which yields an
otherwise flat vacuum spacetime consisting of two static spacetime ends which are cross-linked through the ring.
	This vanishing-mass limit of the Kerr manifold coincides with the vanishing-mass limit of Zipoy's oblate spheroidal family
of static vacuum spacetimes\footnote{In the same paper Zipoy also described another, 
 prolate spheroidal family of static vacuum spacetimes, whose metric is nowadays known as ``Zipoy--Voorhees metric."}
	
	In the black-hole sector of their parameter space the Kerr--Newman spacetimes also have a Cauchy horizon, an
event horizon, and an ergosphere horizon; see \cite{HawEll73,ONeillBOOK,HehlREVIEW}.
	From the ``safe perspective of an observer at spatial infinity'' the ring singularity, the acausal region, and the 
Cauchy horizon are invisible, being ``hidden'' behind the event horizon, and no exotic or even objectionable physics would 
ever seem to happen: 
a Dirac spinor wavefunction initially supported outside the event horizon will either keep spreading within the outer region 
or eventually (as $t\to\infty$) accumulate (in parts or wholly) at the event horizon, see \cite{FinsterETalDKNa,FinsterETalDKNb}.
	Yet, an inquisitive physicist may also want to study 
the spinor wave function in other coordinates designed to ``follow it across the event horizon,''\footnote{When the analogous
  study was carried out by Oppenheimer and Snyder \cite{OppSny} for classical gaseous matter undergoing gravitational collapse it
  leveled the ground for building our modern understanding of the physics of gravitational collapse, involving the formation
  of black holes and their singularities.
Poetically speaking their work revealed that there is more physics in general relativity than meets the (distant observer's) eye.}
however, it is neither clear how to continue in a ``physically correct'' manner beyond the
Cauchy horizon\footnote{If instead of the Cauchy problem one studies $t$-periodic solutions, then one can continue across the
		Cauchy and the event horizons using a weak matching procedure \cite{FinsterETalDperDNEerr,FinsterETalDperDNE}.
		However, Finster et al. \cite{FinsterETalDperDNEerr,FinsterETalDperDNE} found that no $t$-periodic
		solutions exist which are normalized over a constant-$t$ slice of the ``physical black hole spacetime;'' see
		main text.}
nor what to make of the acausal region of closed timelike loops, nor how to correctly handle the timelike singularity.
	Moreover, if one inquires into the ``physics beyond the event horizon,'' one has the option of allowing the support of
the initial spinor wave function to be spread over both asymptotically flat ends, or some other parts of the maximal
analytically extended Kerr--Newman spacetime.
	It is not so clear which options  (if any) are physically reasonable and which ones are science fiction,
although astrophysicists can argue for the ``physical black hole spacetime,'' i.e. the part of the
maximal analytically extended Kerr--Newman spacetime which is the asymptotic limit of the topologically
simple spacetime of a charged, rotating star collapsing into a black hole \cite{StraumannBOOK}.

        The horizons are absent in the hyper-extremal parameter regime.
        Even though the absence of the Cauchy horizon is a welcome simplification,
this regime is rarely studied because the absence of the event horizon renders the singularity ``naked,'' and the (weak) 
cosmic censorship hypothesis, according to which ``nature abhors naked singularities'' \cite{PenroseNAKEDsing,PenroseCOSMICcensor},
has (unfortunately) discouraged  physicists from investigating spacetimes with naked singularities.
        Yet once a piece of a Dirac spinor wave function has crossed the event horizon of a Kerr--Newman black hole 
it is no longer shielded from possible harm done by the spacetime singularity --- viz., inside the event
horizon the singularity \emph{is} naked ---, and so one may as well study the effects of naked singularities directly.
        Be that as it may, the hyper-extremal regime retains the closed timelike loops which 
according to the standard interpretation of general relativity turn the entire manifold into a \emph{causally vicious set}, 
something that many physicists (including the authors) would regard as physically suspicious.

        A strategy to rid the Kerr--Newman manifold from its Cauchy horizon, and all its other acausal aspects, is
to take a \emph{zero-gravity limit} $G\to 0$, where $G$ is Newton's constant of universal gravitation.
	This would be quite uninteresting if the zero-$G$ limit of the Kerr--Newman manifold would simply yield a
Minkowski spacetime decorated with the electric field of a point charge and the magnetic field of a point dipole,
as one might be tempted to guess from the asymptotically flat ends of the Kerr--Newman spacetime.
        However, as shown by one of us in the accompanying paper \cite{zGKN},
the zero-gravity limit of the maximal analytically extended Kerr--Newman spacetime yields a static, flat,
yet \emph{two-leafed, cross-linked} spacetime which is decorated with
Appell--Sommerfeld \cite{Appell,Som97} electromagnetic fields\footnote{The zero-$G$ limit of the electromagnetic Kerr--Newman
	fields yields fields originally discovered by Appell \cite{Appell}, who obtained them from 
        the Coulomb potential of a point charge by a complex translation of the charge's position.
	Appell noticed that the fields change sign when looping once through the ring, but did not conclude
        --- apparently --- that they live naturally on a topologically nontrivial space.
	Sommerfeld \cite{Som97} seems to have been first to introduce  ``branched Riemann spaces,'' 
	three-dimensional analogues of topologically non-trivial Riemann surfaces, to which we will refer as 
        \emph{Sommerfeld spaces}, and to construct electromagnetic fields (harmonic functions) 
        on them, which in general we will call \emph{Sommerfeld fields}.
	Eventually Evans \cite{Evans51} and his students \cite{Neu51}, \cite{Alz04} laid their rigorous foundations.}
$\bF = d\bA$ with $\bA$ as in \refeq{def:AKN},
whose sources are certain ``finite charge and current distributions''
supported by the one-dimensional ring.\footnote{Strictly speaking, the ring singularity
        is not part of the manifold; it's rather a ring ``defect.''}
        Although  the gravitational (viz.: curvature) aspects of the Kerr--Newman manifold, its event horizon included, 
vanish in this limit too, one does retain the topological, the singular, and all the electromagnetic aspects of the 
spacetime.
        Studies of the Dirac equation for a point electron in this zero-$G$ Kerr--Newman (z$G$KN) spacetime will therefore 
illuminate  the role played by the topological and electromagnetic aspects of the Kerr--Newman manifold in the
relativistic quantum mechanics of the electron.

        In this paper we study the Dirac equation for a point electron in static, electromagnetic, flat spacetimes with 
Zipoy topology which include the z$G$KN spacetimes as special case, but which in general can sport 
Sommerfeld fields $\bF = d\bA$ with $\bA$ as in \refeq{def:AQI}, which differ from the Appell--Sommerfeld fields only in a single number,
the ratio $I\pi a^2/Qa$ of their magnetic dipole moment to the magnetic dipole moment of the 
Appell--Sommerfeld fields of same charge $Q$; here, $|a|$ is the radius of the ring singularity.
	By constructing an operator that anti-commutes with the pertinent Dirac Hamiltonian we 
show that the spectrum of any of its self-adjoint extensions is symmetric about zero; this result holds for 
arbitrary $(Q,I)$. 
   All other results are obtained for Dirac's electron in the z$G$KN spacetime ($I\pi a^2/Qa=1$): by
adapting an argument of Winklmeier--Yamada for the Dirac equation of a point electron in the outer region
of the Kerr--Newman black hole spacetime, we show that our formal Dirac Hamiltonian is essentially self-adjoint 
on a spacelike slice of the maximal analytically extended, static z$G$KN spacetime.
	Then we exploit the Chandrasekhar--Page--Toop separation-of-variables method for Dirac's equation on a
general Kerr--Newman spacetime, and the Pr\"ufer transform, to show that the self-adjoint Dirac operator 
has a continuous spectrum with a gap about zero which, under two smallness conditions, contains a pure point spectrum
associated with time-periodic $L^2$ spinor fields, representing
bound states of Dirac's point electron in the electromagnetic field of the ring singularity of the z$G$KN spacetime.

        In the next section we formulate our main results about the Dirac equation for a point electron in the z$G$KN spacetime;
one result is valid also for a Dirac electron in static, flat spacetimes having Zipoy topology featuring  electromagnetic Sommerfeld 
fields of arbitrary $I\pi a/Q$-ratio. 
	In sections 3, 4, 5, 6, and 7 we prove our main theorems about the spectrum.
        In section 8 we conclude with a list of interesting questions left unanswered by this work.
\vskip-0.9truecm $\phantom{nix}$

\section{Formulation of the main results}

     We begin by formulating the Dirac equation for a point electron in electromagnetic, static, flat spacetimes with 
Zipoy topology which generalize zero-$G$ Kerr--Newman spacetimes to zero-$G$ Kerr spacetimes equipped with Sommerfeld 
fields of arbitrary $I\pi a/Q$-ratio. 
     We then state our main theorems about the spectrum of the pertinent Dirac operators.

\vskip-20pt$\phantom{nix}$
\subsection{Dirac's equation for a point electron on zero-$G$ Kerr spacetimes equipped with electromagnetic 
  Sommerfeld fields of arbitrary $I\pi a/Q$-ratio}

\subsubsection{Zero-$G$ Kerr spacetimes}\label{sec:zGK}
	Our limit $G\to 0$ of the maximal analytic extension of the well-known {\em Kerr} family of stationary, axisymmetric
spacetime solutions of Einstein's vacuum equations yields a one-parameter  family of static, flat, but topologically
nontrivial spacetimes\footnote{ We emphasize that Zipoy \cite{Zip66} found a large class of static, axisymmetric, 
  flat but topologically nontrivial solutions to the Einstein vacuum equations which in their zero-$G$ limit coincide with 
  the members of the zero-$G$ Kerr spacetime family.}
$(\cM,\bg)$ which consist of two ``cross-linked leafs.''
      Explicitly, let $\cC\equiv\{(t,r,\theta,\varphi): t\in\Rset, r\in\Rset,\theta\in[0,\pi], \varphi\in[0,2\pi)\}$
denote a rectangular ``four-dimensional cylinder,'' and let
$\cS\equiv \{(t,r,\theta,\varphi): t\in\Rset, r=0, \theta =\pi/2,\varphi\in[0,2\pi)\}\subset\cC$
denote a rectangular ``two-dimensional slab'' in $\cC$.
        Then $\cC\setminus \cS$ is a covering chart of oblate spheroidal (Boyer--Lindquist, or BL) coordinates\footnote{The
         notation $(t,r,\theta,\varphi)$ for the BL coordinates is standard in the relativity literature, and should not be
         confused for instance with the Schwarzschild coordinates on the outer region of that spacetime, or with just 
         standard spherical coordinates of Minkowski spacetime.
	 All standard non-flat $(t,r,\theta,\varphi)$ coordinate systems reduce to standard spherical coordinates
	of the flat Minkowski spacetime near ``$r = \infty$.''
		Note that in BL coordinates $r$ takes any real value.}
for this spacetime, with line element
\beq\label{metricOS}
ds_{\mathbf{g}}^2
=
dt^2 - \left(r^2+a^2\right) \sin^2\theta\, d\varphi^2 - \frac{r^2+a^2 \cos^2\theta}{r^2+a^2} \left(dr^2 + \left(r^2+a^2\right) d\theta^2\right);
\eeq
here, $a^2>0$ is the only parameter of these spacetimes, and we have set the speed of light $c=1$.
 Our sign convention of $(+,-,-,-)$ for the metric follows \cite{FinsterETalDperDNE}.

        The static, axisymmetric character of these zero-$G$ Kerr spacetimes is manifest in (\ref{metricOS}).
        Also, since $r\in\Rset$  occurs strictly quadratically in (\ref{metricOS}), it is clear that the manifold consists of two
``conjoined identical twins.''
	To exhibit their flatness, and in the process also the topological nontrivial juncture, we introduce
cylindrical coordinates $(t,\varrho,z,\varphi)$ on Minkowski spacetime $\Rset^{1,3}$, with the same $(t,\varphi)$ as in
Boyer--Lindquist coordinates, and with the cylindrical coordinates $(\varrho,z)$ related to  the elliptical coordinates
$(r,\theta)$ by
\beq\label{CYLtoELL}
\varrho = \sqrt{r^2+a^2}\sin\theta,\qquad z = r \cos \theta.
\eeq
       In cylindrical  $(t,\varrho,z,\varphi)$ coordinates the metric takes the familiar form for flat Minkowski spacetime 
\beq
ds_{\mathbf{g}}^2 = dt^2 - d\varrho^2 - \varrho^2d\varphi^2 - dz^2,
\eeq
except that the map (\ref{CYLtoELL}) makes it plain that the chart $\cC\setminus\cS$ will be mapped into \emph{two} copies of Minkowski
spacetime which are ``doubly conjoined,'' in a smooth yet crossing manner, at the set  spanned by $\cS$.
  The metric ${\mathbf{g}}$ given by the line element (\ref{metricOS}) has a singularity at $\cS$, which is
the singularity \emph{of the spacetime}, not \emph{in the spacetime}. 
  The set $\cS$ is the boundary of a timelike open solid cylinder in the z$G$KN manifold;
the cross section at any instant $t$ of this cylindrical surface is a translate of the {\em ring}
$\cR_{0}\equiv\{(t,r,\theta,\varphi):t=0,r=0, \theta=\pi/2,\varphi\in[0,2\pi)\}$ of Euclidean radius $\sqrt{a^2}=|a|>0$, 
for which reason one speaks of a ``ring singularity;'' we write $|a|$ because we need to allow $a\in\Rset\backslash\{0\}$
for reasons that become clear in the next subsection.
 The points on the ring are {\em conical singularities} for the metric, meaning that the limit as the radius goes to zero
of the ratio of the circumference to radius of a small circle centered at a point of the ring and lying in a meridional
plane $\varphi=$const. is not $2\pi$; instead, here it is $4\pi$.
  See \cite{zGKN} for details.

        The key topological features of this manifold can easily be visualized.
	Namely, although the fixed timelike planes $\{(t,r,\theta,\varphi): t=t_0,\varphi=\varphi_0\}$ 
cannot be embedded into $\Rset^3$, each such plane can be immersed in it, the immersion consisting of
two Euclidean half planes, stacked up upon each other, then cut along a line segment of length $|a|$ orthogonal to the planes'
boundaries ``with scissors,'' then smoothly ``cross-glued'' at the cut such that the ``upper'' and ``lower'' sheets are cross-linked
like an $\times$ along the cut, while remaining like $\|$ beyond the cut; the singular endpoint of the $\times$-line is not
part of the two-sheeted manifold.
 Shown in Fig.1  are the ring singularity and the part $\{r\in(-1,1),\theta\in(0,\pi)\}$ of a constant-azimuth section 
(slightly curved, for the purpose of visualization) of the two-sheeted static spacelike slice of the zero-$G$ Kerr spacetime.
 The coordinate grid shows a few of the curves $\theta = \mbox{const.}$ (hyperbolas, transiting  through the ring from one sheet 
to the other) and $r = \mbox{const.}$ (oblate semi-ellipses, remaining outside the ring on a single sheet).
\vspace{-.1cm}$\phantom{nix}$
\begin{figure}[ht]
\begin{center}
\includegraphics[width=9cm,height=7cm]{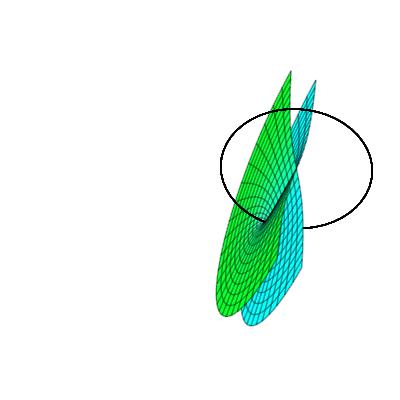}
\end{center}
\vspace{-2cm}
\caption{An illustration of the Zipoy topology.\label{fig:ZipTop}}
\end{figure}

        We can view $\cM$ as a {\em branched covering} over the base manifold $\Rset^{1,3}\setminus\cS$ (mildly abusing notation), 
with the projection map $\Pi: \cM \to \Rset^{1,3}\setminus\cS$ being $\Pi(t,r,\theta,\varphi) = (t,\varrho,z,\varphi)$.
        The pre-image of a point in the base consists of two points, degenerating into one point at each ``ring''
\beq
\cR_{t_0} = \{(t,r,\theta,\varphi)\ |\ t=t_0,\ r=0, \ \theta = \pi/2, 0\leq \varphi\ \leq 2\pi \}.
\eeq
         The pullback of the Minkowski metric $\boldsymbol{\eta}$ under $\Pi$ endows $\cM$ with a flat Lorentzian metric 
$\mathbf{g} = \Pi^*\boldsymbol{\eta}$, whose line element is given in (\ref{metricOS}).

\subsubsection{Zero-G Kerr-Newman spacetimes}\label{sec:zGKN}
         The  spacetime $(\cM,\bg)$ introduced above can be decorated with any static electromagnetic Sommerfeld field $\bF=d\bA$, 
satisfying the flat space Maxwell equations locally but respecting the topologically nontrivial character of the spacetime.
         The zero-$G$ limit of the maximal analytical extension of the {\em Kerr--Newman} family\footnote{We recall that 
for fixed $G>0$ (and speed of light $c$) the Kerr--Newman family is a three-parameter family of electrovac spacetimes, 
the parameters being ADM mass (energy) $M(>0)$, ADM angular momentum $J=Ma\in\Rset$, and total charge $Q\in\Rset$, 
all defined in a single asymptotic end.
  Note that in units where $c=1$ the angular momentum per unit mass $a$ has physical dimension of length; equivalently, of time.}
of stationary axisymmetric solutions to the Einstein--Maxwell equations, written in BL coordinates,
yields precisely the zero-$G$ Kerr (z$G$K) spacetime decorated with a particular electromagnetic Sommerfeld field,
the \emph{Appell--Sommerfeld field}, whose four-potential one-form reads
\beq\label{def:AKN}
\bA = - \frac{r}{r^2+a^2\cos^2\theta} (Qdt - Qa \sin^2\theta\, d\varphi).
\eeq
  Here, $Q$ is the total charge ``seen from infinity'' in the $r>0$ sheet, defined by computing the electric outward flux 
through a spherical surface surrounding the ring singularity in that sheet, and we need to allow $a\in\RR\backslash\{0\}$ to
accomodate a magnetic dipole moment vector pointing either ``up'' or ``down'' w.r.t. the $z$ axis defined in (\ref{CYLtoELL}).

  The field $\bF$ is singular on the same ring $\cR_t$ as is the metric, while for $r$ very large positive its electric and
magnetic components  approach, respectively, the asymptotics of an ``electric monopole field in $\Rset^3$ of a charge $Q$'' and a 
``magnetic dipole field in $\Rset^3$ of dipole moment $Qa$,'' aligned parallel or anti-parallel to the ``$z$ direction;''
for $-r$ very large positive, i.e. in the other sheet,
these electric monopole and magnetic dipole fields correspond to a charge $-Q$ and magnetic dipole moment $-Qa$.

\begin{rem}\textit{
 Note that in contrast to the metric (\ref{metricOS}), the electromagnetic potential (\ref{def:AKN}) is not invariant 
under a change of sign of $a\in \RR\backslash\{0\}$; however, the difference is merely in the direction of the magnetic 
moment vector that corresponds to the above potential, which either points along or opposite to the ``$z$ direction'' 
defined by (\ref{CYLtoELL}).
 Since a 180-degree rotation around a diametrical axis through the ring transforms its magnetic moment into the negative thereof,
and since all objectively physical quantities are invariant under such a rotation, the choices $a>0$ and $a<0$ are physically 
equivalent, given $Q$ and $a^2>0$.
 Therefore, without loss of generality, we can choose $a>0$ in all ensuing calculations involving (\ref{def:AKN}).}
\end{rem}

\begin{rem}\textit{
  It may be tempting to speculate whether the magnetic dipole moment $Qa$ can be interpreted as due to a ``gyrating charged ring,'' 
with ``$a$'' the angular momentum per unit mass ``of the singularity,'' as has been attempted for Kerr--Newman spacetimes.
 Moreover, since Kerr--Newman spacetimes have a {\em gyromagnetic ratio} $Q/M =:g_{\mbox{\tiny{KN}}}^{}Q/2M$ 
(in units with $c=1$), amounting to a $g$-factor $g_{\mbox{\tiny{KN}}}^{}=2$ (see \cite{Car68}), and
since the KN parameters $(M,Q,a)$ are independent of $G$, and so is the KN gyromagnetic ratio, one could be
tempted to assign the z$\,G$KN spacetime the same gyromagnetic ratio of $Q/M$ and $g$-factor of $2$.
 However, since $M$ does not show in the z$\,G$KN metric, such an assignment would be reasonable only if there were no other
way to construct z$\,G$KN than taking the zero-$G$ limit of KN.  
 Yet this is not the case: as already pointed out in footnote 10, the underlying spacetime manifold of z$\,G$KN can be obtained
as zero-$G$ limit of either, the \emph{stationary} family of Kerr spacetimes --- having both an ADM mass $M_{\mbox{\tiny{K}}}^{}$ 
and ADM angular momentum $J=M_{\mbox{\tiny{K}}}^{}a$ ---, or a \emph{static} family of Zipoy spacetimes --- having an ADM mass 
$M_{\mbox{\tiny{Z}}}^{}$ but zero ADM angular momentum; note also that $M_{\mbox{\tiny{Z}}}^{}\neq M_{\mbox{\tiny{K}}}^{}$ in 
general.
 This (say) z$\,G$Z spacetime can now be equipped with an arbitrary Sommerfeld field, in particular: the Appell--Sommerfeld
field of z$\,G$KN, without being logical compelled to interpret its magnetic moment $Qa$ as being due to a ``gyrating ring of
charge $Q$'' with angular momentum per unit mass ``$a$,'' although this is logically possible; 
in any event, it's better to refrain from assigning z$\,G$KN any {\em spacetime $g$-factor}.
}
\end{rem}
\subsubsection{Generalizations of z$G$KN spacetimes to arbitrary charge and current.}

      Since the electric and magnetic components of Maxwell's vacuum equations decouple in the zero-$G$ limit, 
to decorate the zero-$G$ Kerr spacetime with a generalization of the electromagnetic Appell--Sommerfeld field $\bF = d\bA$ 
having electric charge $Q\in\RR$ and current $I\in\RR$, all that needs to be done is to replace the magnetic dipole moment $Qa$ in 
formula (\ref{def:AKN}) by $I\pi a^2$, thus 
\beq\label{def:AQI}
\bA = - \frac{r}{r^2+a^2\cos^2\theta} (Q dt - I\pi a^2 \sin^2\theta\, d\varphi).
\eeq
   Again, electric charge $Q$ and magnetic dipole moment $I\pi a^2$ 
are as ``seen'' from spacelike infinity in the $r>0$ sheet; 
viewing from spacelike infinity in the other sheet one ``sees'' $-Q$ and $-I\pi a^2$.

\begin{rem}\textit{
  The Sommerfeld field (\ref{def:AQI}) makes it plain that the z$\,G$KN spacetimes are but a special one-parameter
subfamily in a two-parameter family of qualitatively similar electromagnetic spacetimes with arbitrary charge $Q$ and magnetic 
moment $I\pi a^2$ (given $a^2$). 
 The ease with which this result was accomplished stands in stark contrast to the difficulties in generalizing
the Kerr--Newman family to electromagnetic spacetimes with a magnetic moment different from $Qa$. }
\end{rem}

\begin{rem}\textit{ 
  As does the metric (\ref{metricOS}), the generalized Sommerfeld field (\ref{def:AQI}) depends on ``$a$'' only through
$a^2$, hence it is invariant under a sign change of $a\in \RR\backslash\{0\}$. 
 So when (\ref{metricOS}) is combined with (\ref{def:AQI}) we can choose $a>0$ without any further ado. 
 Of course, in (\ref{def:AQI}) there is now the new parameter $I\in\RR$, but by the same rotation argument as
given for the z$\,G$KN fields the choices $I>0$ and $I<0$ are physically equivalent, given $Q$ and $a^2>0$.
 Therefore, without loss of generality, from here on we choose $\sign(I)=\sign(Q)$, and $a>0$.}
\end{rem}

\subsubsection{The Dirac equation on electromagnetic spacetimes: Cartan's frame method}

In arbitrary cordinates $(x^\mu)$ (with $c=1$ and $\hbar=1$), 
the Dirac equation for a spin-$1/2$ electron of empirical rest mass $m$ and charge $-e<0$ interacting (through minimal coupling) 
with an electromagnetic field $\bF = d\bA$ in a spacetime $(\cM,\bg)$ reads
\beq\label{eq:DirEqA}
{\tilde\ga}^\mu (-i \nab_\mu  + e A_\mu) \Psi + m \Psi = 0;
\eeq
here $\nab$ is the covariant derivative (on bi-spinors) associated to the spacetime metric $\bg$, and the $({\tilde\ga^\mu})_{\mu=0}^3$ are 
Dirac matrices associated to this metric, i.e. satisfying
\beq
\tilde{\ga}^\mu \tilde{\ga}^\nu + \tilde{\ga}^\nu \tilde{\ga}^\mu = 2 g^{\mu\nu}{\boldsymbol{1}}_{4\times4},
\eeq
while the $A_\mu$ are the pertinent components of 
the electromagnetic potential, $\bA = A_\mu dx^\mu$. 

Using Cartan's frame method (see \cite{BrillCohen66} and refs. therein) one can express the above covariant derivative on spinors 
in terms of standard derivatives:
\beq
\tilde{\ga}^\mu \nab_\mu    = \ga^\mu\be_\mu  +\frac{1}{4} \Om_{\mu\nu\la} \ga^\la\ga^\mu\ga^\nu,
\eeq
where the $\ga^\nu$ are Dirac gamma matrices for the Minkowski spacetime, satisfying 
\beq\label{MinkGamma}
\ga^\nu\ga^\mu+\ga^\mu\ga^\nu = 2{{\eta}}^{\mu\nu}{\boldsymbol{1}}_{4\times4},
\eeq
with
\beq
({\boldsymbol{\eta}}) = \diag(1,-1,-1,-1)
\eeq 
being the matrix of the Minkowski metric in rectangular coordinates;
and $\{\be_\mu\}_{\mu=0}^3$ is a {\em Cartan frame}, i.e. an orthonormal frame of vectors spanning the tangent space at each point
of the spacetime manifold.
We thus have
\beq (\be_\mu)^\nu (\be_\la)^\ka {g}_{\nu\ka} = {{\eta}}_{\mu\la}.
\eeq

On the one hand, it follows that
\beq
\tilde{\ga}^\mu = (\be_\nu)^\mu \ga^\nu.
\eeq
On the other hand, let $\{\boldsymbol{\om}^\mu\}_{\mu=0}^3$ denote the {\em dual} frame to $\{\be_\mu\}$,
 i.e. the orthonormal basis for the cotangent space at each point of the manifold that is dual to the basis for the tangent space:
\beq 
\boldsymbol{\om}_\mu \big(\be^\nu\big)  = \be^\nu \big( \boldsymbol{\om}_\mu\big) = \de_\mu^\nu.
\eeq
Then the $\Om_{\mu\nu\la}$ are by definition the {\em Ricci rotation coefficients} of the frame $\{\boldsymbol{\om}^\mu\}_{\mu=0}^3$, 
defined in the following way:  Let the one-forms $\boldsymbol{\Om}^\mu_\nu$ satisfy
\beq
d \boldsymbol{\om}^\mu + \boldsymbol{\Om}^\mu_\nu \wedge \boldsymbol{\om}^\nu = 0.
\eeq
This does not uniquely define the $\boldsymbol{\Om}^\mu_\nu$.
  However, there exists a unique set of such 1-forms satisfying the extra condition
\beq
\boldsymbol{\Om}_{\mu\nu} = - \boldsymbol{\Om}_{\nu\mu},
\eeq 
where the first index is lowered by the Minkowski metric:   $\boldsymbol{\Om}_{\mu\nu} := {{\eta}}_{\mu\la}\boldsymbol{\Om}^{\la}_\nu$.
Since $\left\{\boldsymbol{\om}^\mu\right\}$ forms a basis for the space of 1-forms, we then have  
$\boldsymbol{\Om}_{\mu\nu} = \Om_{\mu\nu\la}\boldsymbol{\om}^\la$, which defines the rotation coefficients $\Om_{\mu\nu\la}$.

The Dirac equation (\ref{eq:DirEqA})
on a spacetime $(\cM,\bg)$ with an electromagnetic 4-potential $\bA$ can thus be written in the following form:
\beq\label{eq:DirEqAstandard}
\ga^\mu \left(\be_\mu + \Ga_\mu + i e  \tilde{A}_\mu\right)\Psi  + im\Psi = 0;
\eeq
here, the $\Ga_\mu$ are connection coefficients,
\beq
\Ga_\mu := \frac{1}{4} \Om_{\nu\la\mu} \ga^\nu\ga^\la = \frac{1}{8} \Om_{\nu\la\mu}[\ga^\nu,\ga^\la],
\eeq
and the $\tilde{A}_\mu$ are the components of the potential $\bA$ in the ${\boldsymbol{\om}^\mu}$ basis, 
i.e. $\bA =\tilde{A}_\mu \boldsymbol{\om}^\mu$, or,
\beq
\tilde{A}_\mu  := (\be_\mu)^\nu A_\nu.
\eeq
\subsubsection{Frame formulation of the Dirac equation  on z$G$K spacetimes featuring generalized electromagnetic 
              Sommerfeld fields with arbitrary $I\pi a/Q$-ratio}

As explained in section \ref{sec:zGK}, the single chart $\cC\setminus\cS$ 
of oblate spheroidal coordinates $(t,r,\theta,\varphi)$ covers the whole zero-$G$ Kerr spacetime $(\cM,\bg)$,
and in section \ref{sec:zGKN} we saw that in these coordinates the generalized electromagnetic Sommerfeld one-form
$\bA$ is everywhere on  $(\cM,\bg)$ given by the simple formula (\ref{def:AQI}). 
It is therefore natural that one would like to write Dirac's equation \refeq{eq:DirEqA} in these coordinates as well, 
in the hope of achieving at least some partial separation of variables.\footnote{The idea of using special frames adapted 
to a coordinate system in order to separate spinorial wave equations in those coordinates goes back to Kinnersley \cite{Kin69} 
and Teukolsky \cite{Teu72}.}

 However, unlike Cartesian coordinates $(x^\mu)$ in Minkowski spacetime, oblate spheroidal coordinate derivatives do not give 
rise to an orthonormal basis for the tangent space at each point of a zero-$G$ Kerr spacetime.
 Thus, to bring \refeq{eq:DirEqA} into the Cartan form \refeq{eq:DirEqAstandard} using oblate spheroidal coordinates, one
also needs to construct a suitable Cartan frame. 
  Following Chandrasekhar \cite{ChandraDIRACinKERRgeom,ChandraDIRACinKERRgeomERR}, Page \cite{Page76}, Toop \cite{Too76} (see also
Carter-McLenaghan \cite{CarMcL82}), we introduce a special orthonormal frame $\{\be_\mu\}_{\mu=0}^3$ on the 
tangent bundle $T\cM$ which is adapted to the oblate spheroidal coordinates, such that the Dirac equation 
takes a comparatively simple form.

 We begin by introducing a Cartan (co-)frame $\{\boldsymbol{\om}^\mu\}_{\mu=0}^3$ for the cotangent bundle\footnote{This 
   particular frame is called a {\em canonical symmetric tetrad} in \cite{CarMcL82}.}
\beq\label{CcoF}
\boldsymbol{\om}^0 := \frac{\De}{|\rho|} (dt - a \sin^2\theta\, d\varphi),\quad
\boldsymbol{\om}^1 := |\rho|d\theta,\quad
\boldsymbol{\om}^2 := \frac{\sin\theta}{|\rho|} (-a dt + \De^2 d\varphi),\quad
\boldsymbol{\om}^3 := \frac{|\rho|}{\De}dr,
\eeq
with the abbreviations
\beq\label{vpirho}
\De := \sqrt{r^2 + a^2}, \quad \rho:= r + i a \cos\theta.
\eeq
Let us denote the oblate spheroidal coordinates $(t,r,\theta,\varphi)$ collectively by $(y^\nu)$.
  Let $g_{\mu\nu}$ denote the coefficients of the spacetime metric \refeq{metricOS} in oblate spheroidal coordinates, 
i.e. $g_{\mu\nu} = \bg\Big(\frac{\p}{\p y^\mu},\frac{\p}{\p y^\nu}\Big)$.
     One easily checks that written in the $\{\boldsymbol{\om}^{\mu}\}$ frame, the spacetime line element is
\beq 
ds_{\bg}^2 = g_{\mu\nu}dy^\mu dy^\nu = {{\eta}}_{\alpha\beta} \boldsymbol{\om}^{\alpha}\boldsymbol{\om}^{\beta}.
\eeq
 This shows that the frame $\{\boldsymbol{\om}^\mu\}_{\mu=0}^3$ is indeed orthonormal.
 With respect to this frame the electromagnetic Sommerfeld  potential \refeq{def:AQI} becomes
$\bA = \tilde{A}_\mu \boldsymbol{\om}^\mu$, with
\beq\label{def:Atilde}
\tilde{A}_0 =  -Q\frac{r}{|\rho|\De} -  \left(Q-{I\pi a}\right)\frac{a^2r\sin^2\theta}{\De |\rho|^3},\quad
\tilde{A}_1 = 0, \quad
\tilde{A}_2 = - \left(Q-I\pi a\right)\frac{ar\sin\theta}{|\rho|^3},\quad 
\tilde{A}_3 = 0.
\eeq

\begin{rem}\textit{
 We observe that for $Q =I\pi a$, all but one of the quantities $\tilde{A}_\mu$ vanish, and the non-vanishing one,
$\tilde{A}_0$, reduces to $-{Qr}/{|\rho|\De}$.}
\end{rem}
\begin{rem}\textit{
 Clearly, the Cartan co-frame (\ref{CcoF}) is not invariant under the replacement $a\to-a$; it is adapted to (\ref{def:AKN}).
 If we replace $a\to-a$ in (\ref{def:AKN}), we also need to replace $a\to-a$ in (\ref{CcoF}); in the same vein, replacing
$a\to-a$ in (\ref{CcoF}) corresponds to replacing $I\to-I$ in (\ref{def:AQI}).
 Therefore, in the following, a change $a\to-a$ needs to be simultaneously accompanied by the change $I\to-I$.}
\end{rem}

  Next, let the frame of vector fields $\{{\be}_{\mu}\}$ be the {\em dual} frame to $\{\boldsymbol{\om}^{\mu}\}$.
 Thus $\{\be_\mu\}$ yields an orthonormal basis for the tangent space at each point in the manifold:
\beq
\be_0 = \frac{\De}{|\rho|} \p^{}_t + \frac{a}{\De|\rho|} \p^{}_\varphi,\quad
\be_1 = \frac{1}{|\rho|}\p^{}_\theta,\quad
\be_2 = \frac{a\sin\theta}{|\rho|} \p^{}_t + \frac{1}{|\rho|\sin\theta} \p^{}_\varphi,\quad
\be_3 = \frac{\De}{|\rho|} \p^{}_r\;.
\eeq

  Next, the anti-symmetric matrix $\big(\boldsymbol{\Om}_{\mu\nu}\big) = \big({{\eta}}_{\mu\la}\boldsymbol{\Om}^\la_\nu\big)$ is computed to be
\beq
(\boldsymbol{\Om}_{\mu\nu}) = \left(\begin{array}{cccc}
0&-C\boldsymbol{\om}^0 - D\boldsymbol{\om}^2 & D\boldsymbol{\om}^1 - B\boldsymbol{\om}^3 & -A\boldsymbol{\om}^0- B \boldsymbol{\om}^2\\
& 0 & D\boldsymbol{\om}^0 + F\boldsymbol{\om}^2 &-E \boldsymbol{\om}^1 - C\boldsymbol{\om}^3  \\
&\mbox{(anti-sym)}& 0 & -B \boldsymbol{\om}^0 - E\boldsymbol{\om}^2 \\
& & & 0
\end{array}\right),
\eeq
with
\beq
A := \frac{a^2 r \sin^2\theta}{\De |\rho|^3},\,
B := \frac{a r \sin\theta}{|\rho|^3},\,
C := \frac{a^2 \sin\theta\cos\theta}{|\rho|^3},\,
D := \frac{a\cos\theta\De}{|\rho|^3},\,
E := \frac{r\De}{|\rho|^3},\,
F := \frac{\De^2\cos\theta}{|\rho|^3\sin\theta}.
\eeq

 With respect to this frame on a zero-$G$ Kerr spacetime, and picking the {\em Weyl} 
representation\footnote{Here and throughout this paper we use the {\em Weyl} (spinor) 
  representation for the gamma matrices, see \cite{ThallerBOOK} for details.} 
for the Dirac matrices $\ga^\mu$,  the covariant derivative part of 
the Dirac operator (\ref{eq:DirEqA}) can be expressed with the help of the operator
\beq
\fO := \tilde{\ga}^\mu\nab_\mu = \left(\begin{array}{cc} 0 & \fl'+\fm'\\ \fl+\fm &0 \end{array}\right),
\eeq
where
\beq
\fl := 
\frac{1}{|\rho|} \left(\begin{array}{cc} D_+ & L_- \\ L_+ & D_-\end{array}\right)
\eeq
and
\beq
\fl' := 
\frac{1}{|\rho|} \left(\begin{array}{cc} D_- & -L_- \\ -L_+ & D_+\end{array}\right),
\eeq
with
\beq\label{eq:DpmLpm}
D_\pm := \pm \De \p^{}_r + \left( \De \p^{}_t + \frac{a}{\De} \p^{}_\varphi\right),
\qquad 
L_\pm :=\p^{}_\theta  \pm i \left(a \sin\theta\,\p^{}_t+\csc\theta \p^{}_\varphi\right),
\eeq
while
\beq
\begin{aligned}
\fm &:=  \half\bigl[ (-2C+F+iB)\si^{}_1+(-A+2E+iD)\si^{}_3\bigr] \\
&\ = \frac{1}{2|\rho|} \left(\begin{array}{cc} \frac{r}{\De}+ \frac{\De}{\bar{\rho}} &\cot\theta +  \frac{ia\sin\theta}{\bar{\rho}}\\
\cot\theta + \frac{ia\sin\theta}{\bar{\rho}} &  -\frac{r}{\De} - \frac{\De}{\bar{\rho}}\end{array}\right)
\end{aligned}
\eeq
and 
\beq
\fm': = \half\bigl[ (2C-F+iB)\si^{}_1+(A-2E+iD)\si^{}_3\bigr] = -\fm^*,
\eeq
where the $\si^{}_k$ are Pauli matrices:
\beq
 \si^{}_1 = \left(\begin{array}{cc} 0 & 1\\ 1 & 0\end{array}\right),\quad
 \si^{}_2 = \left(\begin{array}{cc} 0 & -i\\ i &\ 0\end{array}\right),\quad
 \si^{}_3 = \left(\begin{array}{cc} 1 & \ 0\\ 0 & -1\end{array}\right).
\eeq

  We note that the principal part of $|\rho|\fO$ has an additive separation property:
\beq\label{eq:DDprincipal}
\begin{aligned}
|\rho|\left(\begin{array}{cc} 0 & \fl'\\ \fl & 0\end{array}\right) 
=
\left[
 \ga^3 \De \p^{}_r + \ga^0\left((\De \p^{}_t + \frac{a}{\De} \p^{}_\varphi\right)\right]
 + 
\Bigl[ \ga^1 \p^{}_\theta + \ga^2(a\sin\theta \p^{}_t + \csc\theta\,\p^{}_\varphi)\Bigr],
\end{aligned}
\eeq
where the coefficients of the two square-bracketed operators are functions of only $r$, respectively only $\theta$.
 Moreover, it is possible to transform away the lower order term in $\fO$, so that exact separation can be achieved for $|\rho|\fO$.
 Namely, let 
\beq
\chi(r,\theta) := \half \log( \De \bar{\rho}\sin\theta).
\eeq  
 It is easy to see that
\beq
\fm = \fl\chi,\qquad \fm' = \fl'\bar{\chi}.
\eeq
 Let us therefore define the diagonal matrix
\beq\label{def:D}
\fD := \diag( e^{-\chi},e^{-\chi}, e^{-\bar{\chi}}, e^{-\bar{\chi}})
\eeq
and a new bispinor $\hat{\Psi}$ related to the original $\Psi$ by
\beq
\Psi = \fD \hat{\Psi}.
\eeq
 Denoting the upper and lower components of a bispinor $\Psi$ by $\psi_1$ and $\psi_2$ respectively, it then follows that
\beq
(\fl + \fm)\psi_1 = 
(\fl + \fm)(e^{-\chi}\hat{\psi}_1) =
 e^{-\chi} \left[ \fl - \fl\chi + \fm\right]\hat{\psi}_1 =
 e^{-\chi} \fl\hat{\psi}_1,
\eeq
and similarly
\beq
(\fl'+ \fm')\psi_2 = e^{-\bar{\chi}} \fl'\hat{\psi}_2.
\eeq

 We now put it all together.
 We set 
\beq
\fR := \diag(\rho,\rho,\bar{\rho},\bar{\rho})
\eeq 
and note that $|\rho|\fD^{-*}\fD = \fR$ while $\fD^{-*}\ga^\mu\fD = \ga^\mu$.
 Thus, setting $\Psi = \fD \hat{\Psi}$ in \refeq{eq:DirEqA} and left-multiplying the equation by the diagonal matrix 
$\fD' := |\rho|\fD^{-*}$  we  conclude that $\hat{\Psi}$ solves a new Dirac equation
\beq\label{eq:newDir}
\left(|\rho|\ga^\mu (\be_\mu + ie\tilde{A}_\mu) + im\fR\right) \hat{\Psi} = 0.
\eeq

  Finally, let us compute the Hamiltonian form of \refeq{eq:newDir}. 
  Let matrices $M^\mu$ be defined by
\beq
|\rho|\ga^\mu\be_\mu = M^\mu\p^{}_\mu.
\eeq
  Thus in particular
\beq
M^0 = \De\ga^0 + a \sin\theta\, \ga^2.
\eeq
We  may thus rewrite \refeq{eq:newDir} as
\beq
M^0 \p^{}_t\hat{\Psi} = - \left( M^k\p^{}_k + ie|\rho|\ga^\mu\tilde{A}_\mu + im\fR\right)\hat{\Psi},
\eeq
so that, defining
\beq\label{def:Hhat}
\hat{H} := -i (M^0)^{-1} \left( M^k\p^{}_k + ie|\rho|\ga^\mu\tilde{A}_\mu + im\fR\right),
\eeq
we can now rewrite the Dirac equation \refeq{eq:newDir} in Hamiltonian form:
\beq\label{eq:DIRACeqHAMformat}
i \p^{}_t \hat{\Psi} = \hat{H}\hat{\Psi}.
\eeq

\begin{rem}\textit{
We note that for $Q \ne I\pi a$ the quantity $|\rho|\ga^\mu\tilde{A}_\mu$ in \refeq{def:Hhat} is a function of both $r$ and $\theta$, 
and unlike the other terms in the Dirac equation \refeq{eq:newDir}  it does {\em not} separate into a sum of two 
terms each depending only on one of these variables. 
It follows that the Dirac equation will not be exactly separable in its four spacetime variables on zGK spacetimes decorated with a 
generalized Sommerfeld field having a magnetic moment different from $Qa$.}

\textit{Even when $Q =I\pi a$, so that $|\rho|\ga^\mu\tilde{A}_\mu$ reduces to  $|\rho|\ga^0\tilde{A}_0 = -({Qr}/{\De})\ga^0$, which is 
a function of only $r$, the separation of variables Ansatz does not yield a system of ordinary differential equations 
which can be solved one at a time, unlike the situation for the familiar Dirac equation for the spectrum of Hydrogen 
in Minkowski spacetime.}
\end{rem}

\subsubsection{A Hilbert space for $\hat{H}$}
In order to decide what is the correct inner product to use for the space of bispinor fields defined on the z$G$KN spacetime, 
we pause to consider the action for the original Dirac equation \refeq{eq:DirEqA}, which should be obtainable from this equation 
upon left-multiplying it by the {\em conjugate bispinor} $\overline{\Psi}$, defined as
\beq
\overline{\Psi} := \Psi^\dag \ga^0,
\eeq
and integrating the result on the spacetime. Note that the dagger in the above formula is the usual notation for ``conjugate-transpose'',
 i.e. $\Psi^\dag = \Psi^{*t}$, and that $\ga^0$ is the zero-th Dirac gamma matrix for the Minkowski space, defined by 
\refeq{MinkGamma}.\footnote{On the Minkowski space, the $\gamma^0$ matrix plays a double role:  In addition to being one 
of the four Dirac gamma matrices, it is also the matrix of the Hermitian quadratic form of signature (2,2) defined on the 
bispinor space, which is one of the key geometric structures needed in order to define the Dirac operator on a general four-dimensional 
Lorentzian manifold, see \cite{CanJad98} for details.
  In our context, the first role is played by $\tilde{\ga}^0$, and the second one by $\ga^0$.}
 Thus, using oblate spheroidal coordinates,
\beq
\cS[\Psi] = \int dt \int_{\Si_t} \Psi^\dag \ga^0 \left[ \tilde{\ga}^\mu \nab_\mu \Psi + \dots \right] d\mu^{}_{\Si_t},
\eeq
where 
\beq\label{vol-elem}
d\mu^{}_{\Si_t} = |\rho|^2\sin\theta  d\theta d\varphi dr
\eeq
 is the volume element of $\Si_t$, the spacelike 
$t=$ constant slice of z$G$KN.
It follows that the natural inner product for bispinors on $\Si_t$ needs to be
\beq
\langle \Psi,\Phi\rangle = \int_\Si \Psi^\dag\ga^0\tilde{\ga}^0 \Phi d\mu^{}_\Si 
= \int_0^{2\pi}\int_0^\pi \int_{-\infty}^\infty \Psi^\dag M \Phi |\rho|^2 \sin\theta d\theta d\varphi dr,
\eeq
with
\beq
M := \ga^0 \tilde{\ga}^0 = \ga^0 \be_\nu^0 \ga^\nu = \frac{\De}{|\rho|} \al^0 + \frac{a\sin\theta}{|\rho|} \al^2.
\eeq
Here, $\alpha^2$ is the second one of the three Dirac alpha matrices in the Weyl (spinor) represenation, 
viz.
\beq
\al^k = \ga^0 \ga^k = \left(\begin{array}{cc} \si^{}_k & \ 0\\ 0 & -\si^{}_k\end{array}\right),\qquad k=1,2,3;
\eeq
for notational convenience, we have also set
\beq
\al^0 =  \left(\begin{array}{cc} \boldsymbol{1}_{2\times2} & 0 \\ 0 & \boldsymbol{1}_{2\times2}\end{array}\right)
\eeq
for the $4\times4$ identity matrix.

Now, let $\Psi = \fD \hat{\Psi}$ and $\Phi = \fD \hat{\Phi}$, with $\fD$ as in \refeq{def:D}.
  Then we have
\beq
\langle \Psi,\Phi \rangle = 
\int_{-\infty}^\infty 
\int_0^{2\pi}
\int_0^\pi 
\hat{\Psi}^\dag \hat{M} \hat{\Phi} d\theta d\varphi dr,
\eeq
where
\beq
\hat{M} := \al^0 + \frac{a\sin\theta}{\De} \al^2.
\eeq
The eigenvalues of $\hat{M}$ are $\la_\pm = 1 \pm \frac{a\sin\theta}{\De}$, both of which are positive everywhere on this space with
Zipoy topology.
	(Note that $\la_- \to 0$ on the ring, which is not part of the space time but at its boundary.)
	We may thus take the above as the definition of a positive definite inner product
 given by the matrix $\hat{M}$ for bispinors defined on the rectangular cylinder
$\rcyl :=\RR\times [0,\pi]\times  [0,2\pi]$  (which is the $t=const.$ section of $\cC$) with its natural measure:
\beq\label{def:innerPROD}
\langle \hat{\Psi},\hat{\Phi}\rangle_{\hat{M}} := \int_{\rcyl} \hat{\Psi}^\dag\hat{M} \hat{\Phi} d\theta d\varphi dr.
\eeq
 An alternative way of arriving at this inner product is to define the conserved Dirac current
\beq\label{dir-curr}
j^\mu = \overline{\Psi}\tilde{\ga}^\mu \Psi = \Psi^\dag \ga^0 \tilde{\ga}^\mu \Psi,
\eeq
and consider the integral of its time component $j^0$ on the Cauchy hypersurface $\Si_t$ with its induced measure \refeq{vol-elem}:
\beq
\int_{\Si_t} j^0 d\mu_\Si  = \int_{-\infty}^\infty 
\int_0^{2\pi}
\int_0^\pi 
\Psi^\dag \ga^0\tilde{\ga}^0 \Psi |\rho|^2 \sin\theta d\theta d\varphi dr 
= \int_{-\infty}^\infty 
\int_0^{2\pi}
\int_0^\pi 
\hat{\Psi}^\dag \hat{M}\hat{\Psi} d\theta d\varphi dr.
\eeq
 The corresponding Hilbert space is denoted by ${\sf H}$, thus
\beq
{\sf H} 
:= \left\{ \hat\Psi:\rcyl \to \Cset^4\ | \ \|\hat\Psi\|_{\hat{M}}^2 := \langle \hat\Psi,\hat\Psi\rangle_{\hat{M}} < \infty \right\}.
\eeq
Note  that ${\sf H}$ is \emph{not equivalent} to standard $L^2(\rcyl)$ whose inner product has the identity matrix in place of $\hat{M}$.

 After these preparations we are now ready to state our main results.

\subsection{Statement of the Main Theorems}

Our results about  the symmetry of the spectrum are valid
for the Dirac Hamiltonian on a static spacelike slice of the zero-$G$ Kerr spacetime
decorated with Sommerfeld fields of arbitrary charge $Q$ and current $I$.
The essential self-adjointness, and location of essential and point spectra, are stated only
for the Dirac Hamiltonian on a static spacelike slice of the z$G$KN spacetime; however,
we conjecture that these results also hold for the more general Hamiltonian as long as the coupling
constant  $(Q -I\pi a)e$ is sufficiently small.

In the ensuing four sections we will prove the following Theorems about  $\hat{H}$.

\subsubsection{Symmetry of the spectrum of the Dirac Hamiltonians}

We shall find an operator which anti-commutes with any self-adjoint extension of
the formal Dirac operator $\hat{H}$ on $\sf H$, with the help of which we prove:

\begin{thm}\label{thm:sym}
  Let any self-adjoint extension of the formal Dirac operator $\hat{H}$ on $\sf H$ be denoted by the same letter.
  Suppose $E\in\mathrm{spec}\,\hat{H}$. 
  Then $-E\in\mathrm{spec}\,\hat{H}$. 
\end{thm}

Note that the above result holds for \emph{any} self-adjoint extension of $\hat{H}$, whatever $Q$ and $I$ are.

\subsubsection{Essential self-adjointness of the Dirac Hamiltonian on z$G$KN}

 Let $\rcyl^*$ denote $\rcyl$ with the ring singularity removed.
 By adapting an argument of Winklmeier--Yamada \cite{WINKLMEIERc}, we shall prove:
\begin{thm}\label{thm:esa}
For $Q =I\pi a$, i.e. for z$\,G$KN,  the operator $\hat{H}$ with domain $C^\infty_c(\rcyl^*,\Cset^4)$ is e.s.a. in~$\sf H$.  
\end{thm}

\subsubsection{The continuous spectrum of the Dirac Hamiltonians on z$G$KN}

By adapting an argument of Weidmann \cite{Wei82}, we shall prove:

\begin{thm}\label{thm:essspec}
For $Q=I\pi a$  the continuous spectrum of $\hat{H}$ on $\sf H$ is $\Rset\setminus(-m,m)$.
\end{thm}

\subsubsection{The point spectrum of the Dirac Hamiltonian on z$G$KN}

With the help of the Chandrasekhar--Page--Toop formalism to separate variables, and the Pr\"ufer transform,
we will be able to control the point spectrum for the z$G$KN Dirac Hamiltonian:

\begin{thm}\label{thm:ptspec}
Suppose $Q =I\pi a$.
Then, if $2m|a|<1$ and $|eQ|<\sqrt{2m|a|(1-2m|a|)}$, the point spectrum of $\hat{H}$ on $\sf H$ is nonempty and located in 
$(-m,m)$; the end points are not included.
\end{thm}

This completes the formulation of our main results. 
We next turn to their proofs.
The proofs of our main theorems are distributed over four sections corresponding to the various aspects of the spectrum, 
i.e. symmetry, essential self-adjointness, continuous spectrum, and point spectrum.

\section{Proof of  Theorem \ref{thm:sym} (Symmetry of the energy spectrum)}\label{sec:proofofsymmetry}

Suppose $E\in\RR$ is an eigenvalue of $\hat{H}$.
  Then there exists $\hat\Psi\in \sH$ such that
\beq
\hat{H} \hat\Psi = E \hat\Psi.
\eeq
Suppose one can find a bounded linear, or conjugate-linear, operator $\hat{C}:\sH \to \sH$ that anti-commutes with $\hat{H}$, i.e.
\beq
\bigl[ \hat{C}, \hat{H}\bigr]_+^{} = \hat{C} \hat{H} + \hat{H}\hat{C} = 0.
\eeq
It is then easy to see that $-E$ must also be an eigenvalue of $\hat{H}$, since
\beq
\hat{H} \hat{C} \hat\Psi = - \hat{C} \hat{H} \hat\Psi = -\hat{C} E \hat\Psi = -E \hat{C}\hat\Psi.
\eeq
This argument can be extended to show the symmetry of other parts of the spectrum. (See e.g. Glazman \cite{Gla65}, p. 205.)

Let $\hat{K}:\sH \to \sH$ denote the complex conjugation opertor $\hat{K} \hat\Psi (x) = \hat\Psi^*(x)$, 
and let $\hat{S}:\sH\to \sH$ denote the  operator
$(\hat{S} \hat\Psi) (x) = \hat\Psi (\varsigma(x))$ where $\varsigma:\cZ \to \cZ$ is the sheet swapping map, 
\beq
\varsigma(r,\theta,\varphi) = (-r,\pi-\theta,\varphi).
\eeq
We claim that the operator $\hat{C}:\sH \to \sH$ given (in Weyl representation) by $\hat{C} := \ga^0 \hat{K} \hat{S}$, viz.
\beq
(\hat{C} \hat\Psi) (x) = \ga^0 \hat\Psi^*(\varsigma(x)),
\eeq
anti-commutes with $\hat{H}$.
 Note that the double-sheetedness of the underlying space plays an essential role in the definition of this operator.

\begin{rem}\emph{
 The operator $\hat{C}$ should not be confused with the operator $\tilde{C}$ given in Weyl representation by
$$
\tilde{C} := i\ga^2 \hat{K}.
$$
 One easily checks that if $\hat\Psi$ solves $i\hbar\p_t\hat\Psi = (\hat{H}_0 + e\cA)\hat\Psi$, then $\tilde{C}\hat\Psi$ solves
$i\hbar\p_t(\tilde{C}\hat\Psi) = (\hat{H}_0 - e\cA)(\tilde{C}\hat\Psi)$.
 In particular, if $\hat\Psi$ is an eigen-bi-spinor of $\hat{H}_0 + e\cA$ with eigenvalue $E$, then $\tilde{C}\hat\Psi$ is an eigen-bi-spinor of 
$\hat{H}_0 - e\cA$ with eigenvalue $-E$ (note that the two Hamiltonians here are different!).
  For this reason $\tilde{C}$ is called the {\em charge conjugation} operator.  }
\end{rem}
 To prove the claim, first note that $\ga^0 = \beta$ anti-commutes with all three $\al^k$ matrices.
  Recall that
\beq\label{def:hamhat}
\hat{H}(x) = \hat{M}^{-1}\fH
\eeq
and
\beq\label{def:hamfrak}
\fH := 
  -i\al^3 \p^{}_r + \frac{1}{\De} \Bigl( - i \al^1 \p^{}_\theta - i\al^2 \csc\theta\,\p^{}_\varphi\Bigr) -
 \frac{ia}{\De^2} \al^0 \p^{}_\varphi + \frac{m}{\De} \ga^0 \fR + \frac{e|\rho|}{\De} \left(\tilde{A}^0(x) \al^0 + 
\tilde{A}^2(x)\al^2\right).
\eeq
Now
\beq\label{minv}
\hat{M}^{-1} = \frac{\De^2}{|\rho|^2} \left(\al^0 - \frac{a\sin\theta}{\De} \al^2\right).
\eeq
Thus, keeping in mind that $\overline{\al^2} = -\al^2$, we find that
\beq
\hat{C} \hat{M}^{-1} =
\ga^0 \overline{\hat{M}^{-1}\circ\varsigma} \hat{K}\hat{S}=
  \frac{\De^2}{|\rho|^2} \ga^0\left(\al^0 +\frac{a\sin\theta}{\De}\al^2\right) \hat{K}\hat{S}=
  \frac{\De^2}{|\rho|^2}\left(\al^0 - \frac{a\sin\theta}{\De} \al^2\right)\ga^0\hat{K}\hat{S}=
 \hat{M}^{-1}\hat{C}.
\eeq
So we only need to check that $\hat{C}$ anti-commutes with $\fH$.

  It is enough to check that each term in $\fH$ goes through an odd number of sign changes (either one or three) as the 
three operators $\hat{K}$, $\hat{S}$, and multiplication by $\ga^0$, filter through that term. 
 Recalling that the potential $\bA$ is anti-symmetric with respect to sheet swap:  $\tilde{A}_\mu\circ\varsigma = - \tilde{A}_\mu$,
this becomes obvious for most terms in $\fH$. 
 Only the term involving $\fR$ requires some care. 
 We first check that $\fR\circ\varsigma = -\fR$ and that $\ga^0 \overline{\fR} = \fR \ga^0$. 
 Then
\beq
\hat{C} \ga^0 \fR = \ga^0 \ga^0 \overline{\fR\circ\varsigma} \hat{K} \hat{S} = 
-  \ga^0 \ga^0 \overline{\fR} \hat{K} \hat{S}=  
- \ga^0 \fR \ga^0 \hat{K} \hat{S} 
= - \ga^0\fR\hat{C},
\eeq
establishing the anti-commutation property.
The proof of Theorem \ref{thm:sym} is complete.

 Before moving on to the proof of the next theorem on the list, we pause briefly to recall our earlier discussion that
showed that the physics does not change if in the z$G$KN electromagnetic spacetime solution one changes $a\to-a$, 
respectively changes $I\to-I$ or $a\to-a$ in its generalization involving $Qa\to I\pi a^2$.
 This suggests that the spectrum of our Dirac Hamiltonian must be invariant under these transformations. 
 However, recall that our co-frame is adapted to the electromagnetic fields written as in (\ref{def:AKN}), respectively (\ref{def:AQI}),
so that a sign change $a\to-a$ needs to be accompanied by a sign change $I\to-I$ in order for the physics (here: the spectrum of the
Hamiltonian) to remain unchanged.
 We now use a variant of the strategy of proof of Theorem \ref{thm:sym} to prove exactly this.
 Thus we write the Hamiltonian defined by \refeq{def:hamhat} as $\hat{H} = \hat{H}_{a,I}$ to emphasize the dependence 
on the two parameters $a$ and $I$. 

\begin{prop}\label{prop:apos}
 There exists an involutive isometry $C: \sH \to \sH$ such that 
\beq\label{intertwine}
C \hat{H}_{a,I} = \hat{H}_{-a,-I} C.
\eeq
 Thus, the spectral properties of the two Hamiltonians are identical.
\end{prop}
\begin{proof}
Let $T: \sH \to \sH$ be defined by
$$
T \hat{\Psi} (r,\theta,\varphi) = \hat{\Psi}(r,\pi - \theta,-\varphi),
$$
and let $\vec{S}$ denote the ``classical spin operator,'' in Weyl representation given by:
$$
S_k := \left( \begin{array}{cc} \si_k & \\ & \si_k \end{array}\right),\qquad k=1,2,3.
$$
It can then be readily checked from \refeq{def:hamhat}, \refeq{def:hamfrak}, \refeq{minv}, and \refeq{def:Atilde} that 
$$
C := S_3 T
$$
will do the job.  
\end{proof}

\section{Proof of Theorem \ref{thm:esa} (Essential self-adjointness ($Q=I\pi a$))}

 We now show that the Dirac Hamiltonian $\hat{H}$ is essentially self-adjoint on $\sf H$;
recall that $\sf H$ is equipped with the inner product\footnote{As pointed out by the anonymous referee,
 an alternate proof of essential self-adjointness might be possible following the strategy of  Chernoff \cite{Che73}, 
who proved essential self-adjointness of the Dirac operator on certain {\em complete} spacetimes. 
 Such a proof would be highly welcome, indeed, for it would avoid a partial wave decomposition and be more direct. 
However, due to the presence of the naked ring singularity in the z$G$KN spacetime, the underlying manifold is not
 complete and thus it is far from obvious how to generalize Chernoff's result.}
 \refeq{def:innerPROD}.

 We observe that $M^0 = \De \ga^0 \hat{M}$, so we may rewrite \refeq{def:Hhat} as
\beq
\hat{H} = 
\hat{M}^{-1}\ga^0\left(\frac{-i}{\De} M^k\p^{}_k + e\frac{|\rho|}{\De}\ga^\mu\tilde{A}_\mu  + \frac{m}{\De} \fR\right) 
=  \hat{M}^{-1} \fH,
\eeq
where
\beq
\fH := \fM + \frac{1}{\De} \fN + \fP +\fQ,
\eeq
with
\bna
\fM & := & -i \al^3 \p^{}_r  \\
\fN & :=& -i \al^1 \p^{}_\theta  -i \al^2\csc\theta \p^{}_\varphi\\
\fP & := & -i \frac{a}{\De^2}\al^0 \p^{}_\varphi + \frac{m}{\De} \ga^0\fR \\
\fQ & :=  & e\frac{|\rho|}{\De}\ga^0\ga^\mu\tilde{A}_\mu\,.
\ena
Thus,
\beq
\langle \hat{\Psi} , \hat{H} \hat{\Phi}\rangle_{\hat{M}} = \int_{\rcyl}
\hat{\Psi} \fH \hat{ \Phi}  d\theta d\varphi dr.
\eeq
 Evidently, $\fH$ is Hermitian symmetric on the Hilbert space $L^2(\rcyl;\Cset^4)$ with its natural inner product  
\beq \label{def:IP}
(\hat\Phi,\hat\Psi) = \int_{\rcyl} \hat\Phi^\dag\hat\Psi d\theta d\varphi dr.
\eeq 
 It is furthermore easy to see that $\hat{H}$ is e.s.a. on ${\sf H}$ if and only if $\fH$ is e.s.a. on 
$L^2(\rcyl;\Cset^4)$.

We shall  prove that $\fH$ is e.s.a. on $L^2(\rcyl;\Cset^4)$ when $Q=I\pi a$, i.e. for a Dirac point electron in z$G$KN.

\begin{thm}\label{thm:esaGOTIC}
 For $Q=I\pi a$ the operator $\fH$ with domain $C^\infty_c(\rcyl^*,\Cset^4)$ is e.s.a. in $L^2(\rcyl,\Cset^4)$. 
\end{thm}

\begin{proof}
Let us write
\beq
\fH = \fH^0 + \fQ.
\eeq
Here, $\fH^0$ is the free Hamiltonian.
  We will first show that $\fH^0$ is essentially self-adjoint; this proof 
is an easy adaptation of the method first employed by Winklmeier and Yamada \cite{WINKLMEIERc}.
	We will then conclude essential self-adjointness of $\fH$ by using a perturbation argument.

        To this end, let us consider the decomposition with respect to the azimuthal angle $\varphi$ of the Hilbert space 
$L^2(\rcyl;\Cset^4)$ into partial wave subspaces $L^2([0,\pi]\times\RR,d\theta dr)$:
\beq\label{pwd}
L^2(\rcyl;\Cset^4) = \oplus_{\ka\in\Zset+\half} \left(L^2_\ka([0,\pi]\times\RR,d\theta dr)\right)^4
\eeq
corresponding to the expansion of a bispinor field $\hat\Psi \in L^2(\rcyl;\Cset^4)$ given by 
\beq\label{decomp}
\hat\Psi (r,\theta,\varphi) = \sum_{\ka\in\Zset+\half} e^{i\ka\varphi}\hat\Psi_\ka(r,\theta).
\eeq
For a discussion of why $\kappa$ needs to be a half-integer, see \cite{FinsterETalDperDNE,FinsterETalDperDNEerr}.

Let $\fH^0_\ka := \left.\fH^0\right|_{L^2_\ka}$. 
 Then $\fH^0_\ka = \fS_\ka + \De^{-1}\fT_\ka +\fB_\ka$, with
\beq
\fS_\ka = -i \al^3 \p^{}_r ,
\qquad
\fT_\ka =-i \al^1 \p^{}_\theta  + \al^2\ka \csc\theta 
= 
\left(\begin{array}{cc} \ft_\ka& 0\\ 0 & -\ft_\ka\end{array}\right),
\qquad
\fB_\ka 
= 
\frac{a\ka}{\De^2}\al^0 + \frac{m}{\De} \ga^0\fR.
\eeq
 We note that $\fB_\ka$ is a symmetric {\em bounded} multiplication operator on $L^2_\ka$; in fact,
\beq
\|\fB_\ka \|_{L^\infty} \leq |{\ka}/{a}| +m,
\eeq
so that the task of showing e.s.a.-ness of $\fH^0_\ka$ reduces to showing e.s.a.-ness of $\fH'_\ka  := \fS_\ka+\De^{-1}\fT_\ka$.

Now $\fH'_\ka$ is block-diagonal:
\beq\label{fHprimeDecomp}
\fH'_\ka = \left(\begin{array}{cc}\fh'_\ka & 0 \\  0 & -\fh'_\ka \end{array}\right),\qquad \fh'_\ka := -i\si_3 \p_r + \De^{-1}\ft_\ka,\qquad
 \ft_\ka := -i \si_1 \p_\theta  + \si_2\ka \csc\theta.
\eeq
Thus it is enough to show $\fh'_\ka$ is e.s.a.
 We do so by showing that $\ker(\fh'_\ka \pm i) = \{0\}$: 
 Suppose $\hat\psi_\ka \in \left(L^2_\ka([0,\pi]\times\RR,d\theta dr)\right)^2$ satisfies
\beq\label{eq:ker}
\fh'_\ka \hat\psi_\ka = \pm i \hat\psi_\ka.
\eeq
 As observed in \cite{WINKLMEIERc}, it is possible to decompose \refeq{eq:ker} with respect to the eigenspaces of the operator
\beq
\fa_\ka:= W\ft_\ka W^{-1} = -i\si^{}_2 \p^{}_\theta + \ka \csc\theta \si^{}_1,
\eeq
where
\beq\label{def:W}
W := \left(\begin{array}{cc} 0 & 1\\ i & 0\end{array}\right).
\eeq
 The operator $\fa_\ka$ has pure point spectrum and a complete set of eigenfunctions.  
 More precisely, one has the following result \cite{WinklmeierPHD} (here quoted from \cite{WINKLMEIERc}):

\begin{thm} (Winklmeier, 2006) 
For all $\ka \in \Zset + \half$ the operator $\fa_\ka$ with domain $(C^\infty_c((0,\pi)))^2$ is essentially self-adjoint 
in $(L^2((0,\pi),d\theta))^2$.
  Its closure (denoted again by $\fa_\ka$) is compactly invertible and its spectrum consists of simple eigenvalues only, given by
\beq\label{def:lank}
\la_n^\ka := \mbox{sgn}(n)\left(|\ka|-{\textstyle\half} + |n|\right),\qquad n \in \Zset^* = \Zset\setminus\{0\},
\eeq
 with corresponding normalized eigenfunctions $\{g_n^\ka\}_{n\in\Zset^*}$ forming a complete orthonormal set in 
$(L^2((0,\pi),d\theta))^2$. 
 Moreover,
\beq
\la_{-n}^\ka = -\la_n^\ka,\qquad g_{-n}^\ka = -\si^{}_3 g_n^\ka.
\eeq
\end{thm}
We can therefore write
\beq\label{winkdec}
\hat\psi_\ka(r,\theta)= \sum_{n\in\Zset^*} \xi_n(r) g_n^\ka(\theta),
\eeq
with functions $\xi_n\in L^2(\RR,dr)$.
  Hence, performing a similarity transform on \refeq{eq:ker} with $W$ and projecting 
on the spans of $g_n^\ka$ and $g_{-n}^\ka$ we obtain the following system (see \cite{WINKLMEIERc} for details):
\beq
\left(\begin{array}{cc} \frac{\la_n^\ka}{\De} & i \p^{}_r \\
i\p^{}_r  & -\frac{\la_n^\ka}{\De}\end{array}\right)
\left(\begin{array}{c} \xi_n \\
\xi_{-n}\end{array}\right) = \pm i \left(\begin{array}{c} \xi_n \\
\xi_{-n}\end{array}\right)
\eeq
However the operator in the above eigenvalue problem $\cC_\ka = i\si^{}_1\p^{}_r + \frac{\la_n^\ka}{\De}\si^{}_3$ is clearly e.s.a., 
since $\frac{\la_n^\ka}{\De}$ is bounded, hence $\xi_n = \xi_{-n} = 0$.

  This completes the proof of essential self-adjointness of $\fH^0$.

  Consider now the term $\fQ= e\frac{|\rho|}{\De}\ga^0\ga^\mu\tilde{A}_\mu$ coming from the electromagnetic potential.
  It can be rewritten as  $\fQ = -eQ\fV_1 - e(Q-I\pi a)\fV_2$, where
\beq\label{eq:fQdecomp}
\fV_1 := \frac{r}{\De^2}\al^0,\qquad \fV_2 := \frac{ar\sin\theta}{\De|\rho|^2}\hat{M}.
\eeq
The first term, $\fV_1$, is clearly  bounded, whereas the second one, $\fV_2$, blows up on the ring.
However, since by hypothesis we restrict ourselves to the case $Q=I\pi a$, the $\fV_2$ term is absent from
$\fQ$, and   essential self-adjointness of $\fH$ follows easily from that of $\fH^0$ and the boundedness of $eQ\fV_1$. 

The proof of Theorem \ref{thm:esaGOTIC} is complete.
\end{proof}

 For the proof the remaining statements in this paper we rely on the fact that the Dirac equation of a point electron in z$G$KN separates into 
four (coupled) ordinary differential equations, each of which depends on only one of the four oblate spheroidal coordinates, with the coupling 
being effected through shared parameters in the equations.  
This is carried out in the next section before we resume with proving our claims.

\section{Chandrasekhar--Page--Toop separation-of-variables ($Q=I\pi a$)}

 When $Q=I\pi a$ the Dirac equation \refeq{eq:DIRACeqHAMformat} for the bispinor $\hat{\Psi}$ allows 
a clear separation also for the remaining $r$ and $\theta$ derivatives (commonly referred to in 
the literature as ``radial'' and ``angular'' derivatives, even though $r$ is 
not a radial distance and $\theta$ is not an angle, except at infinity).
 Thus, when $Q=I\pi a$ the Dirac equation \refeq{eq:DIRACeqHAMformat} becomes
\beq\label{eq:DirSep}
(\hat{R} +\hat{A}) \hat{\Psi} = 0,
\eeq
where 
\bna
\hat{R}& := & \left(\begin{array}{cccc} imr & 0 &D_-+ieQ\frac{r}{\varpi} & 0 \\
0 & imr & 0 & D_++ieQ\frac{r}{\varpi}\\
D_++ieQ\frac{r}{\varpi} & 0 & imr & 0 \\
0 & D_-+ieQ\frac{r}{\varpi} & 0 & imr \end{array} \right),\\
\hat{A} &:= & \left(\begin{array}{cccc} -m a \cos\theta & 0 & 0 & -L_- \\
0 & -ma\cos\theta & -L_+ &0 \\
0 & L_- & ma\cos\theta & 0 \\
L_+ & 0 & 0 & ma\cos\theta
\end{array}\right),
\ena
where $D_\pm$ and $L_\pm$ have been given in (\ref{eq:DpmLpm}).
 Once a solution $\hat{\Psi}$ to \refeq{eq:DirSep} is found, the bispinor $\Psi := \fD\hat{\Psi}$
solves the original Dirac equation \refeq{eq:DirEqA}.

\subsubsection{The Chandrasekhar Ansatz}
  Assume now that a solution $\hat{\Psi}$ of \refeq{eq:DirSep} is of the form
\beq\label{chandra-ansatz} 
\hat{\Psi} = e^{-i(Et-\kappa \varphi)} \left( \begin{array}{c}R_1S_1\\ R_2 S_2\\ R_2 S_1\\ R_1 S_2 \end{array}\right),
\eeq
with $R_k$ being complex-valued functions of $r$ alone, and $S_k$ real-valued functions of $\theta$ alone.  
Let
\beq
\vec{R} := \left(\begin{array}{c} R_1\\ R_2\end{array}\right),\qquad \vec{S} := \left(\begin{array}{c} S_1\\ S_2\end{array}\right).
\eeq
Plugging the Chandrasekhar Ansatz \eqref{chandra-ansatz} into \eqref{eq:DirSep} one easily finds that there must be 
$\la\in\Cset$ such that 
\beq\label{eq:rad} 
T_{rad}\vec{R} =  E\vec{R},
\eeq
\beq\label{eq:ang}
T_{ang}\vec{S} = \la \vec{S},
\eeq
where
\bna
T_{rad} & :=  \label{eq:Trad} 
& \left(\begin{array}{cc} d_- 
&-m\frac{r}{\De} - i\frac{\la}{\De} \\ -m\frac{r}{\De}+i\frac{\la}{\De} 
& -d_+ \end{array}\right)
\\
T_{ang}& := \label{eq:Srad} 
& \left(\begin{array}{cc}  -ma\cos\theta & -l_- \\
 l_+ &ma\cos\theta  \end{array}\right)
\ena
The operators $d_\pm$ and $l_\pm$ are now ordinary differential operators in $r$ and $\theta$ respectively, 
with coefficients that depend on the unknown $E$, and parameters $a$, $\kappa$, and $eQ$:
\bna\label{opdefs}
d_\pm & := & i \frac{d}{dr} \pm \frac{-a\kappa + eQ r}{\De^2}\\
l_\pm & := & \frac{d}{d\theta} \mp \left( aE\sin\theta - \kappa \csc\theta\right)
\ena
 The angular operator $T_{ang}$ in \refeq{eq:ang} is easily seen to be essentially self-adjoint on
 $(C^\infty_c((0,\pi),\sin\theta d\theta))^2$ and in fact is self-adjoint on its domain inside $(L^2((0,\pi),\sin\theta d\theta))^2$
(e.g. \cite{SufFacCos83,WINKLMEIERa}) with purely point spectrum $\la=\la_n(am,aE,\kappa)$, $n\in \Zset\setminus 0$.
  Thus in particular  $\la \in \RR$.
It then follows that the radial operator $T_{rad}$ is also essentially self-adjoint on $(C^\infty_c(\RR, dr))^2$ and in 
fact self-adjoint on its domain inside $(L^2(\RR,dr))^2.$

Suppose $\vec{R} = (R_1,R_2)^T \in (L^2(\RR))^2$ is a nontrivial solution to $T_{rad} \vec{R} = E\vec{R}$, with $E\in \RR$.
Then
\bea
\frac{dR_1}{dr} - i\left(E -  \frac{ a \kappa-eQr}{\De^2}\right)R_1 +\frac{1}{\De} (imr - \la) R_2 & = & 0\\
-\frac{d R_2}{dr} -i \left(E- \frac{ a \kappa-eQr}{\De^2}\right) R_2  +\frac{1}{\De} (imr + \la)R_1 & = & 0.
\eea
Multiply the first equation by $\bar{R}_1$ and the second equation by $\bar{R}_2$,  add them and take the real part, to obtain
\beq
\frac{d}{dr} \left(|R_1|^2 - |R_2|^2\right) = 0.
\eeq
Thus the difference of the moduli squared of $R_1$ and $R_2$ is constant, hence zero since they need to be integrable at infinity. 
I.e.,
\beq
|R_1| = |R_2| := R.
\eeq
Let $R_j = R e^{i\Phi_j}$ for $j=1,2$.
  Multiply the first equation by $\bar{R}_2$, multiply the complex conjugate of the second equation  by $R_1$, and add them to obtain
\beq
\frac{d}{dr} \left(\frac{R_1}{ \bar{R}_2}\right) = 0.
\eeq
Thus the ratio $R_1 / \bar{R}_2$, and hence the sum of the arguments $\Phi_1+ \Phi_2$ must be a constant, say $\de$.
 Thus $R_1 = \bar{R}_2e^{i\de}$.
Since multiplication by a constant phase factor is a gauge transformation for Dirac bispinors, we can replace $\hat{\Psi}$ 
with $\hat{\Psi}' = e^{-i\de/2}\hat{\Psi}$ without changing anything.
  The spinor thus obtained has the same form as \refeq{chandra-ansatz}, now with $R'_1 = \bar{R}'_2$.
  Thus without loss of generality we can assume $\de = 0$ and $R_1 = \bar{R}_2$.

This motivates us to set
\beq
R_1 =\frac{1}{\sqrt{2}}( v-iu),\qquad R_2  =\frac{1}{\sqrt{2}}( v + iu)
\eeq
 for real funcions $u$ and $v$.
   Consider the unitary matrix
\beq
U := \frac{1}{\sqrt{2}}\left(\begin{array}{cc} -i & 1 \\ \ i & 1 \end{array}\right).
\eeq
A change of basis using $U$ brings the radial system \refeq{eq:rad} into the following standard (Hamiltonian) form
\beq\label{eq:hamil}
(H_{rad} -E)\left(\begin{array}{c} u \\ v \end{array}\right) = \left(\begin{array}{c}0 \\ 0 \end{array}\right),
\eeq
where
\beq\label{eq:Hrad}
H_{rad} := \left(\begin{array}{cc} m \frac{r}{\De} + \frac{\ga r+a\kappa}{\De^2} & -\p^{}_r + \frac{\la}{\De} \\[20pt]
 \p^{}_r +\frac{\la}{\De} & -m\frac{r}{\De} + \frac{\ga r+a\kappa}{\De^2}  \end{array}\right),
\eeq
(cf. \cite{ThallerBOOK}, eq (7.105)) with
\beq
\ga := -eQ <0.
\eeq  
\section{Proof of Theorem \ref{thm:essspec}  (Continuous spectrum of $\hat{H}$ on z$G$KN)}\label{sec:contspec}

 Following Weidmann \cite{Wei82} we now prove the theorem about the continuous spectrum of $\hat{H}$.
  Recall the partial wave decomposition \eqref{pwd}. 
 Let $\hat{H}_\ka$ denote the restriction of $\hat{H}$ to $L^2_\ka$.  
 The Chandrasekhar separation \refeq{chandra-ansatz} and equation \refeq{eq:rad} yield that 
the spectrum of $\hat{H}_\ka$ coincides with that of $T_{rad}$, which coincides with that of $H_{rad}$ 
since these last two are unitarily equivalent.
  Furthermore, the spectrum of $\hat{H}$ equals the union of the spectra of $\hat{H}_\ka$. 
 Thus in order to prove the claim about the essential spectrum, it suffices to show that it 
holds for $H_{rad}$ regardless of the values of $\ka$ and $\lambda$.  

Since $H_{rad}$ is a radial Dirac operator, one can then use results that are particular to one dimension. 
 One such result is due to Weidmann \cite{Wei82}: 
\begin{thm*} Let $P$ and $J$ be matrices such that $H_{rad} = J\p_r + P$.
 Suppose $P$ can be written as $P_1+P_2$ in such a way that each component 
of $P_1$ is integrable in $[R,\infty)$ for some $R>0$, $P_2$ is of bounded variation on $[R,\infty)$ and 
$$
\lim_{r\to \infty} P(r) = \left(\begin{array}{cc} a & 0 \\ 0  & b \end{array} \right),\qquad a> b.
$$
Then {\em each} self-adjoint extension of $h$ has a purely absolutely continuous spectrum in $(-\infty,b]\cup[a,\infty)$.  
\end{thm*}
Using this result, our claim follows by noticing that the hypotheses on $P$ are satisfied, and 
$$
\lim_{r\to \infty} P(r) = \left(\begin{array}{cc} m & 0 \\ 0 & -m \end{array}\right).
$$
Proof of Theorem \ref{thm:essspec} is  complete.

\section{Proof of Theorem \ref{thm:ptspec} (Point spectrum of $\hat{H}$ on z$G$KN)}

By the remarks at the beginning of Section~\ref{sec:contspec}, we are interested in the eigenvalues $E$ and square-integrable eigenfunctions in 
$L^2(\Rset,dr)^2$ of the operator $H_\rad$.
  One complication is that in our case the radial Hamiltonian $H_{rad}$ depends on the unknown eigenvalues $\la$ 
of the angular operator $T_{ang}$ in \refeq{eq:ang}, which in turn depend on the energy $E$.
 Since the angular operator is the same as the one on Kerr and Kerr-Newman spacetime studied in \cite{SufFacCos83,  WINKLMEIERa},
and since it is known that for a given value of $E$ there is a largest negative eigenvalue $\la = \La(E)$, our strategy is to show the 
existence, for a given value of $\lambda<0$, of a smallest positive eigenvalue $E = \cE(\la)$ for $H_{rad}$, and then set up an 
iteration that 
converges to a  pair $(E,\la)$ for which the radial \refeq{eq:rad} and the angular \refeq{eq:ang} equations jointly have 
$L^2$ solutions, thereby establishing the existence of a ``positive-energy eigenstate'' for the full Dirac Hamiltonian; note that
by the symmetry of the spectrum there also exists a ``negative-energy eigenstate.''

\subsection{The Pr\"ufer transform}
Consider the equations \refeq{eq:hamil} and \refeq{eq:ang} for unknowns $(u,v)$ and $(S_1,S_2)$.
  Let us define new unknowns $(R,\Om)$ and $(S,\Theta)$ via the Pr\"ufer transform \cite{Pru26}
\beq\label{eq:prufer}
u =\sqrt{2} R \cos\frac{\Om}{2},\quad v = \sqrt{2} R \sin\frac{\Om}{2},\quad S_1 = S \cos\frac{\Theta}{2},\quad S_2 = S \sin\frac{\Theta}{2}.
\eeq
Thus
\beq
R =\half\sqrt{u^2+v^2},\quad\Om = 2\tan^{-1}\frac{v}{u},\quad S = \sqrt{S_1^2+S_2^2},\quad \Theta = 2\tan^{-1}\frac{S_2}{S_1}.
\eeq
As a result, $R_1 = -iRe^{i\Om/2}$ and $R_2 = iRe^{-i\Om/2}$.  Hence $\hat{\Psi}$ can be re-expressed in terms of the Pr\"ufer variables,
thus
\beq\label{ontology}
\hat{\Psi}(t,r,\theta,\varphi) = R(r)S(\theta)e^{-i(Et-\ka \varphi)} \left(\begin{array}{l}
 -i\cos(\Theta(\theta)/2)e^{+i\Om(r)/2}\\
+i\sin(\Theta(\theta)/2) e^{-i\Om(r)/2}\\
+i\cos(\Theta(\theta)/2)e^{-i\Om(r)/2}\\
-i\sin(\Theta(\theta)/2)e^{+i\Om(r)/2}\end{array}\right),
\eeq
and we obtain the following equations for the new unknowns
\bna
\frac{d}{dr}\Om    &=& 2 \frac{mr}{\De} \cos\Om + 2\frac{\la}{\De} \sin\Om +2\frac{a\kappa + \gamma r}{\De^2} - 2E ,\label{eq:Om}\\
\frac{d}{dr} \ln R &=& \frac{mr}{\De}\sin\Om - \frac{\la}{\De} \cos\Om .\label{eq:R}
\ena
Similarly,
\bna
\frac{d}{d\theta}\Theta &=& -2ma\cos\theta\cos\Theta + 2\left(aE \sin\theta - \frac{\kappa}{\sin\theta}\right)\sin\Theta + 2\la,\label{eq:Theta}\\
\frac{d}{d\theta} \ln S &=& -ma \cos\theta\sin\Theta - \left(aE\sin\theta - \frac{\kappa}{\sin\theta}\right)\cos\Theta. \label{eq:S}
\ena
To simplify the analysis of these systems and reduce the number of parameters involved, we will henceforth set $m=1$.
  Note that this is always possible by defining the constants $a'=ma$, $E'=E/m$, and a change of variable $r'=mr$.

\begin{rem}
\textit{Equations \eqref{eq:Om}--\eqref{eq:S} have certain symmetries that are connected with the action of operators $\hat{C}$ and $\tilde{C}$ introduced in Section~\ref{sec:proofofsymmetry}:  It is easy to see that these four equations are preserved under each of the following two transformations:
\beq\label{trans1}
r \to -r,\quad \theta \to \pi - \theta,\quad \la \to -\la, \quad \ka \to -\ka, \quad E \to -E,
\eeq
and 
\beq\label{trans2}
\Om \to \pi - \Om,\quad \Theta \to \pi -\Theta, \quad \la \to -\la,  \quad \ka \to -\ka, \quad E \to -E, \quad \ga \to -\ga.
\eeq
The map \refeq{trans1} corresponds to the action of $\hat{C}$ and the map \refeq{trans2} to the action of $\tilde{C}$.  Note that the latter does not preserve the  Hamiltonian, since the sign of $\ga = -eQ$ is changed.  }
\end{rem}
\subsection{The realm of $L^2$ solutions}
We note that in both of the above systems (\ref{eq:Om},\ref{eq:R}) and (\ref{eq:Theta},\ref{eq:S}), when a solution to the first equation is known, the second equation 
in the system can be solved by quadrature.
  Moreover, the requirement that $R$ and $S$ be $L^2$ functions of their argument determines what boundary values 
the solutions to the $\Omega$ and $\Theta$ equations should have.
  More precisely,
\begin{prop}\label{prop:bndryvals} 
Any bispinor $\hat{\Psi}$ of the form \refeq{ontology} constructed from solutions of (\ref{eq:Om}), (\ref{eq:R}), (\ref{eq:Theta}), (\ref{eq:S}),
 with $|\kappa|\geq \half$, $E>0$ and $\la<0$, belongs to the Hilbert space $\sH$ provided
 \beq\label{asympOm}
\lim_{r\to -\infty} \Omega(r) = -\pi + \cos^{-1}(E),\qquad \lim_{r\to \infty} \Omega(r) = - \cos^{-1}(E),
\eeq
and
\beq\label{asympTh} \Theta(0) = 0,\qquad \Theta(\pi) =-\pi.
\eeq
\end{prop}
\begin{proof}
It is straightforward to compute that for a $\hat{\Psi}$ of the form \refeq{ontology},
\bea
\|\hat{\Psi}\|_{\hat{M}}^2 
&=& 
2\int_0^{2\pi}\int_0^\pi\int_{-\infty}^\infty 
R^2(r)S^2(\theta)\left( 1+ \frac{a\sin\theta}{\De} \sin\Theta(\theta) \sin\Omega(r)\right) dr d\theta d\varphi \\
&=& 
4\pi\left[ \int_{-\infty}^\infty R^2 dr \int_0^\pi 
S^2 d\theta + a \int_{-\infty}^\infty R^2 \sin\Om \frac{dr}{\De} \int_0^\pi S^2 \sin\Theta \sin\theta d\theta\right] \\
&\leq& 
8 \pi \|R\|^2_{L^2} \|S\|_{L^2}^2,
\eea
and thus $\hat\Psi \in \sH$ provided $R \in L^2(\RR,dr)$ and $S \in L^2((0,\pi),d\theta)$.  

Now \refeq{eq:Theta} can be written as a smooth dynamical system in the $(\theta,\Theta)$ plane by introducing a new independent 
variable $\tau$ such that $\frac{d\theta}{d\tau} = \sin\theta$.
Then, with dot representing differentiation in $\tau$,  we have, 
\beq\label{dynsysTh}
\left\{\begin{array}{rcl}
\dot{\theta} & = & \sin\theta\\
       \dot{\Theta} & = & -2a\sin\theta\cos\theta\cos\Theta+2aE\sin^2\theta\sin\Theta - 2\ka\sin\Theta + 2\la\sin\theta
       \end{array}\right.
\eeq
Identifying the line $\Theta=\pi$ with $\Theta=-\pi$, this becomes a dynamical system on a closed finite cylinder 
$\cC_1=[0,\pi]\times\Sset^1$. 
 The only equilibrium points of the flow are on the two circular boundaries: 
Two on the left boundary: $S_- = (0,0)$, $N_- = (0,\pi)$; two on the right: $S_+ = (\pi,-\pi)$ and $N_+ = (\pi,0)$.  

For $\ka>0$, the linearization of the flow at the equilibrium points reveals that $S_-$ and $S_+$ are hyperbolic saddle points (with 
eigenvalues $\{1,-2\ka\}$ and $\{-1,2\ka\}$ respectively), while $N_-$ is a  source node (with  eigenvalues 1 and $2\ka$) 
and $N_+$ is a sink node (with eigenvalues $-1$ and $-2\ka$).  Note that the situation with $\ka<0$ is entirely analogous, with the 
critical points switching their roles.  For the remainder of this section therefore, we will assume $\ka>0$.  

  The $\al$-limit set of the orbit of any point in the interior of the cylinder must necessarily be either $S_-$ or $N_-$,
 and likewise its $\om$-limit set can only be either $N_+$ or $S_+$. 
 The only possible boundary values for $\Theta(\theta)$ are therefore $0$ and $\pm\pi$ at each endpoint of the 
interval $0\leq \theta\leq \pi$.  

The boundary values (\ref{asympTh}) correspond to a heteroclinic orbit connecting the two saddles $S_-$ and $S_+$.
  Suppose such a saddles connection exists.
  Since the eigendirection corresponding to the unstable manifold of $S_-$ is $v_1 = (\ka+\half,\la-a)^T$ and the 
stable manifold of $S_+$ has the same eigendirection, it follows that for the said saddles connection, we have
\beq
\frac{d\Theta }{d\theta}_{|_{\theta=0}} = \frac{d\Theta}{d\theta}_{|_{\theta = \pi}} = \frac{\la -a}{\ka+\half} =:\de <0.
\eeq
Thus $\Theta = \de\theta + o(1)$ as $\theta\to 0$ and $\Theta = -\pi+\de(\theta-\pi) +o(1)$ as $\theta \to \pi$. 
Consider now the $S$ equation \refeq{eq:S}.  
By the above, 
\beq
\frac{d}{d\theta}\ln S =\left\{\begin{array}{ll}  \frac{\ka}{\theta} +o(1) & \mbox{ as }\theta \to 0\\
\frac{\ka}{\theta-\pi} + o(1) & \mbox{ as }\theta \to \pi
\end{array}\right.
\eeq
Integrating in $\theta$ we thereby conclude that $S \sim |\theta|^\ka$ for $\theta$ small and $S \sim |\pi-\theta|^\ka$ 
for $\theta$ near $\pi$.  
Therefore $S$ is integrable on $(0,\pi)$ and indeed it belongs to $L^p((0,\pi),d\theta)$ for any $p\geq 1$.  

 In an analogous manner one shows that $S\not\in L^2$ when $\Theta$ takes boundary values corresponding to either of 
the nodes $N_-$ or $N_+$.

  Consider next the $\Om$ equation \refeq{eq:Om}.
  It can also be rewritten as a smooth dynamical system on a cylinder, in this case by setting 
$\tau := \frac{r}{a}$ as new independent variable, as well as introducing a new dependent variable
\beq
\xi := \tan^{-1}\frac{r}{a} = \tan^{-1}\tau 
\eeq
Then, with dot again representing differentiation in $\tau$, \refeq{eq:Om} is equivalent to
\beq\label{dynsysOm}
\left\{\begin{array}{rcl}
\dot{\xi} & = & \cos^2\xi \\
\dot{\Om} & = & 2a\sin\xi\cos\Om+2\la\cos\xi\sin\Om + 2\ga\sin\xi\cos\xi+2\ka\cos^2\xi - 2aE
\end{array}\right.
\eeq
Once again, identifying $\Om=-\pi$ with $\Om  = \pi$ turns this into a smooth flow on the closed finite cylinder 
$\cC_2 := [-\frac{\pi}{2},\frac{\pi}{2}]\times \Sset^1$.  The only equilibrium points of the flow are on the two circular boundaries. 
For $E\in[0,1)$ there are two equilibria on each: $S_- = (-\frac{\pi}{2},-\pi+\cos^{-1}E)$ and $N_-=(-\frac{\pi}{2},\pi - \cos^{-1}E)$ 
on the left boundary, and $S_+ = (\frac{\pi}{2}, -\cos^{-1}E)$ and $N_+ = (\frac{\pi}{2},\cos^{-1}E)$ on the right boundary.
 $S_\mp$ are non-hyperbolic (degenerate) saddle-nodes, with eigenvalues $0$ and $\pm 2a\sqrt{1-E^2}$, while $N_-$ is a degenerate 
source-node and $N_+$ a degenerate sink-node.\footnote{For $E=1$ each $S,N$ pair coalesces into one  degenerate equilibrium: 
$N^1_- = (-\frac{\pi}{2},\pm\pi)$ and $N^1_+ = (\frac{\pi}{2},0)$ with both eigenvalues being zero.}

 The boundary values \refeq{asympOm} correspond to a heteroclinic orbit connecting $S_-$ and $S_+$. 
 Suppose such a saddles connection exists, and consider the $R$ equation \refeq{eq:R}.
   As $r \to \pm\infty$, we will then have
\beq
\frac{d}{dr}\ln R \sim - \sgn(r)\sqrt{1-E^2}
\eeq
so that integrating in $r$ we will obtain
\beq
R(r) \sim e^{-|r|\sqrt{1-E^2}}\qquad \mbox{ as } r\to \pm\infty
\eeq
which ensures that $R$ is integrable at infinity.
  Since the right-hand-side of the $R$ equation is smooth in $\Om$ and $r$,
 and $\Om$ itself is smooth, it follows that $R \in L^p(\RR,dr)$ for all $p\geq 1$.

 In an analogous manner one shows that $R\not\in L^2$ when $\Omega$ takes boundary values corresponding to either of 
the nodes $N_-$ or $N_+$.
\end{proof}

\subsection{Existence of heteroclinic orbits connecting the two saddles}
{F}rom the proof of Proposition~\ref{prop:bndryvals} it is evident that in order to establish the existence of an eigenfunction 
for the Dirac Hamiltonian of a point electron in the z$G$KN spacetime, we need to show that there exists a pair $(E,\la)$ such 
that both dynamical systems \refeq{dynsysTh} and \refeq{dynsysOm} have a {\em saddle-saddle connecting orbit} for those values 
of $E$ and $\la$.  
We call this type of  orbit a {\em saddles connector} for the corresponding flow.  
We pave the road for our proof by recalling some general facts of flow on a cylinder.

\subsubsection{Flow on a finite cylinder}
Let $\cC := [x_-^{},x_+^{}]\times \Sset^1$ be a finite cylinder. 
We denote its universal cover by $\bar{\cC} := [x_-^{},x_+^{}]\times\RR$, with coordinates $(x,y)$, and fix a fundamental 
domain $\widetilde{\cC} := [x_-^{},x_+^{}]\times[-\pi,\pi)$ in $\bar{\cC}$.
 Consider the flow $\Phi_t$ on $\bar{\cC}$ given by the dynamical system
\beq\label{eq:flow}
\left\{\begin{array}{rcl}
          \dot{x} & = & f(x)\\ \dot{y} & = & g(x,y)
         \end{array}\right.
\eeq
where the dot  represents differentiation with respect to a formal ``time'' parameter $\tau$, 
the functions $f$ and $g$ are smooth, and $g$ is $2\pi$-periodic in $y$: $g(x,y) = g(x,y+2\pi)$.
  Let us moreover assume that $f$ satisfies
\beq f(x_-^{}) = f(x_+^{}) = 0,\qquad f(x)>0\ \forall\ x\in(x_-^{},x_+^{})
\eeq
while $g$ satisfies
\beq
g(x_-^{},y) = 0 \implies y\in\{n_-^{}, s_-^{}\},\qquad g(x_+^{},y) = 0 \implies y\in\{n_+^{},s_+^{}\}
\eeq
where $-\pi\leq s_-^{} < n_-^{} \leq\pi$ and $-\pi\leq s_+^{} < n_+^{} \leq\pi$.
  These assumptions imply that the following four distinct points in $\cC$ are equilibrium points for the flow:
\beq
N_\pm := (x_\pm,n_\pm),\qquad S_\pm := (x_\pm, s_\pm).
\eeq
We shall further assume that the flow does not have any non-wandering points other than the above four equilibria.  

The following assumptions fix the character of the four equilibrium points:
\beq
f'(x_-^{})\geq 0,\quad f'(x_+^{})\leq 0,\quad f''(x_\pm) \ne 0
\eeq
(where by $f'(x_\pm)$ we mean the left derivate at $x_+^{}$ and the right derivative at $x_-^{}$),
and 
\beq\label{cond:hyp}
D_yg(x_-^{},n_-^{})>0,\quad D_yg(x_-^{},s_-^{})<0,\quad D_yg(x_+^{},n_+^{})<0,\quad D_yg(x_+^{},s_+^{})>0,
\eeq
where $D_yg$ is the $y$-derivative of $g(x,y)$.
Thus $N_-$ is a (source) node, $N_+$ a (sink) node, and $S_\pm$ are saddle points.
 These will be hyperbolic if $f'(x_\pm) \ne 0$, and non-hyperbolic (degenerate) otherwise.  

Later on, in order to have a well-defined notion of index for certain distinguished orbits on $\cC$, 
we will also assume that the locations of the equilibria on the boundary of the cylinder are not arbitrary,
 but are subject to the single condition
\beq\label{assump}
 s_-^{} - n_-^{} = n_+^{} - s_+^{} \qquad(\mbox{mod}\ 2\pi)
\eeq
(This is a condition on $g(x,y)$.  
Although we will not pursue this approach here, under this condition \refeq{eq:flow} can be viewed as a flow 
with two equilibrium points on a {\em 2-torus}.)

For a point $p\in\cC$, let $\cO(p)$ denote the flow orbit through $p$.
  Since $\cC$ is compact, all orbits are 
complete, meaning they exist for all $s\in \RR$, and since the flow is autonomous, two orbits are either disjoint 
or they coincide.
  The orbit of an equilibrium point consists of only one point, namely the equilibrium itself.
  The $\omega$-limit of any other orbit in $\cC$ can be either $N_+$ or $S_+$, and the $\alpha$-limit likewise can
only be either $N_-$ or $S_-$.
  All these facts are easy consequences of the existence and uniqueness theorem for ODEs.  

\subsubsection{Connecting orbits and corridors}
Given a flow on $\cC$ as in the above, there are two distinguished orbits in the interior of the cylinder:
  Let ${\cW}^-$ denote the unique orbit of the flow whose $\al$-limit point is the saddle $S_-$, and let ${\cW}^+$ denote 
the unique orbit whose $\om$-limit point is the saddle $S_+$.
  In the hyperbolic case ($f'(x_\pm)\ne 0$) the uniqueness is immediate because ${\cW}^-$ is  the unstable manifold of 
$S_-$ and ${\cW}^+$ is the stable manifold of $S_+$.
  In the non-hyperbolic case $f'(x_\pm)= 0$ the orbits ${\cW}^\pm$ are {\em center} manifolds for the corresponding 
saddle-nodes $S_\pm$.  Recall that center manifolds may be non-unique, but in our case the uniqueness is assured 
because the equilibrium points are on the boundary of the domain, so the relevant part of the center manifolds are 
on the ``saddle side'' of the equilibrium, and not on the ``node side'' (see Figure~\ref{fig:sn}).
\begin{figure}[ht]
\begin{center}
\includegraphics[scale=0.3]{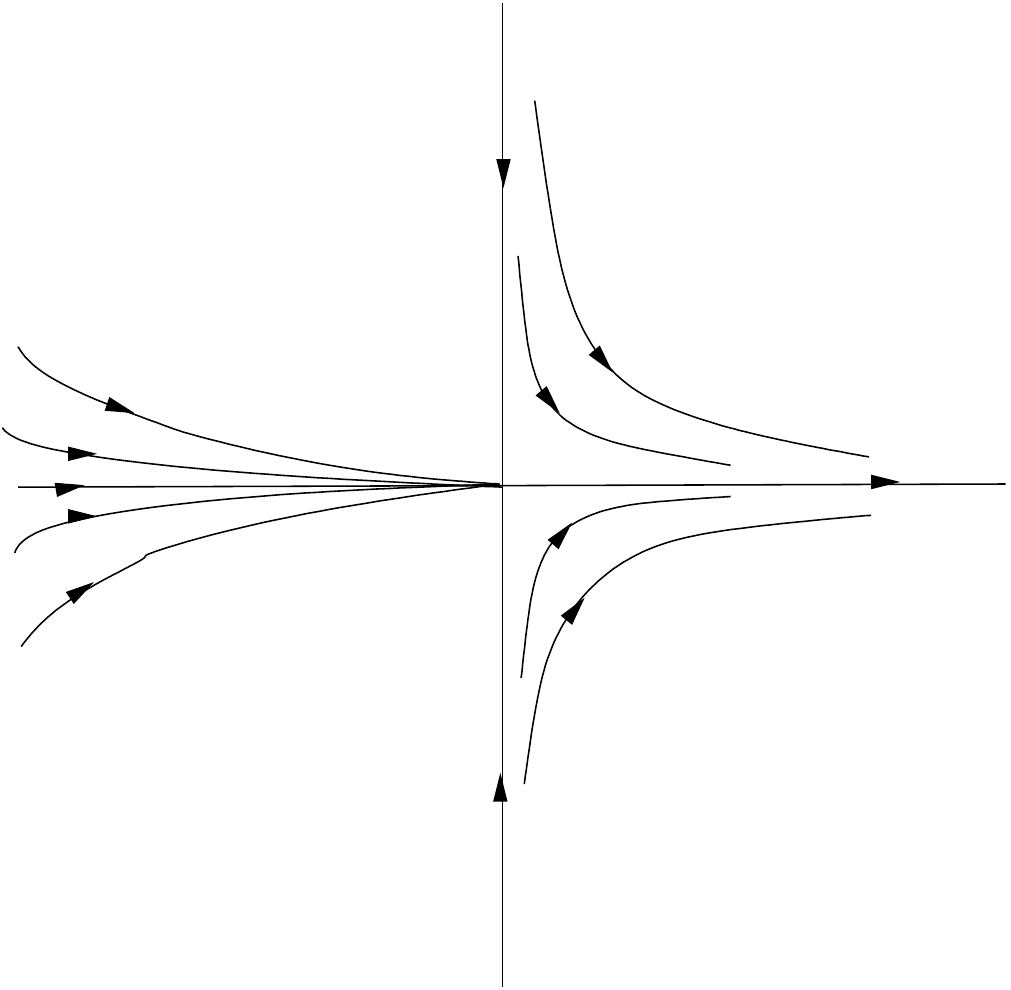}
\end{center}
\caption{\label{fig:sn} Flow near a saddle-node. The node part lies outside of
the domain of concern.}
\end{figure}
If ${\cW}^+$ and ${\cW}^-$ intersect, they must coincide, and the resulting orbit will connect the two saddle points, 
i.e. it will be the saddles connector we are after.
  Let us therefore assume that they are disjoint.
  The $\om$-limit of ${\cW}^-$ must then necessarily be $N_+$, and the $\al$-limit of ${\cW}^+$ must be $N_-$.  

On the other hand the assumptions we have made about the flow imply that there are also two orbits of the flow 
on the left boundary of the cylinder connecting $N_-$ with $S_-$, call them $(N_-S_-)_\pm$, with $+$ denoting 
the counterclockwise one (when viewed from a point on the cylinder's axis and to the left of the cylinder), and similarly two joining $S_+$ with $N_+$, called $(S_+N_+)_\pm$.
  Consider therefore the following collection of six heteroclinic orbits 
\beq
\cH := \{ (N_-S_-)_\pm, {\cW}^\pm, (S_+N_+)_\pm\}.
\eeq
The cylinder $\cC$ is divided into two invariant regions $\cK_1$ and $\cK_2$, called {\em corridors}, by these 
orbits:  $\cC = \cK_1 \cup \cH \cup \cK_2$.
  We would like to distinguish one of these two corridors.
  We do so as follows: Consider the lifting of the flow to the universal cover $\bar{\cC}$.
  Let $\widetilde{S}_-$ denote the unique copy of the node $S_-$ that lies in the fundamental 
domain $\widetilde{\cC}$, and let $\widetilde{\cW}^-$ be the unique orbit in $\bar{\cC}$ whose $\al$-limit point is $\widetilde{S}_-$.
  The $\om$-limit point of this orbit is thus some copy of the node $N_+$, call it$\bar{N}_+$, which has coordinates 
$(x_+^{},n_+^{}-2\pi k_+)$ for some $k_+\in \Zset$.
  Similarly, let $\widetilde{S}_+$ denote the unique point in the preimage of $S_+$ under the covering map that lies in the 
fundamental domain $\widetilde{\cC}$ and let $\widetilde{\cW}^+$ denote the unique orbit whose $\om$-limit point is this $\widetilde{S}_+$.
  Let $\bar{N}_- = (x_-^{},n_-^{} +2\pi k_-)$, $k_-\in \Zset$ be the $\al$-limit point of $\widetilde{\cW}^+$.
  By definition the corridor $\cK_1$ is the open domain in $\bar{\cC}$ whose boundary contains the two orbits $\widetilde{\cW}^-$ and $\widetilde{\cW}^+$.

We note that in $\bar{\cC}$ only one of the two corridors will have {\em both} of these orbits on its boundary, 
so this is the distinguishing feature of $\cK_1$.  

We orient the boundary of $\cK_1$ (which is a closed simple curve) in such a way that the orientation induced on $\widetilde{\cW}^-$ coincides with the direction of the flow on that orbit.
\begin{figure}[ht]
 \begin{center}
\font\thinlinefont=cmr5
\begingroup\makeatletter\ifx\SetFigFont\undefined%
\gdef\SetFigFont#1#2#3#4#5{%
  \reset@font\fontsize{#1}{#2pt}%
  \fontfamily{#3}\fontseries{#4}\fontshape{#5}%
  \selectfont}%
\fi\endgroup%
\mbox{\beginpicture
\setcoordinatesystem units <0.50000cm,0.50000cm>
\unitlength=0.50000cm
\linethickness=1pt
\setplotsymbol ({\makebox(0,0)[l]{\tencirc\symbol{'160}}})
\setshadesymbol ({\thinlinefont .})
\setlinear
\linethickness= 0.500pt
\setplotsymbol ({\thinlinefont .})
{\color[rgb]{0,0,0}\ellipticalarc axes ratio  0.953:2.540  360 degrees 
	from  6.032 16.510 center at  5.080 16.510
}%
\linethickness= 0.500pt
\setplotsymbol ({\thinlinefont .})
{\color[rgb]{0,0,0}\putrule from  5.080 13.970 to 12.700 13.970
}%
\linethickness= 0.500pt
\setplotsymbol ({\thinlinefont .})
{\color[rgb]{0,0,0}\putrule from  5.080 19.050 to 12.700 19.050
}%
\linethickness= 0.500pt
\setplotsymbol ({\thinlinefont .})
{\color[rgb]{0,0,0}\plot  5.397 14.287  5.715 13.970 /
}%
\linethickness= 0.500pt
\setplotsymbol ({\thinlinefont .})
{\color[rgb]{0,0,0}\plot  5.715 14.605  6.350 13.970 /
}%
\linethickness= 0.500pt
\setplotsymbol ({\thinlinefont .})
{\color[rgb]{0,0,0}\plot  7.010 13.966  5.804 15.045 /
}%
\linethickness= 0.500pt
\setplotsymbol ({\thinlinefont .})
{\color[rgb]{0,0,0}\plot  5.994 15.553  7.391 14.199 /
}%
\linethickness= 0.500pt
\setplotsymbol ({\thinlinefont .})
{\color[rgb]{0,0,0}\plot  6.058 16.104  7.942 14.389 /
}%
\linethickness= 0.500pt
\setplotsymbol ({\thinlinefont .})
{\color[rgb]{0,0,0}\plot  6.037 16.739  8.175 14.685 /
}%
\linethickness= 0.500pt
\setplotsymbol ({\thinlinefont .})
{\color[rgb]{0,0,0}\plot  7.645 15.828  8.492 15.024 /
}%
\linethickness= 0.500pt
\setplotsymbol ({\thinlinefont .})
{\color[rgb]{0,0,0}\plot  8.407 15.680  8.894 15.257 /
}%
\linethickness= 0.500pt
\setplotsymbol ({\thinlinefont .})
{\color[rgb]{0,0,0}\plot  8.894 15.828  9.318 15.405 /
}%
\linethickness= 0.500pt
\setplotsymbol ({\thinlinefont .})
{\color[rgb]{0,0,0}\plot  9.275 16.082  9.974 15.447 /
}%
\linethickness= 0.500pt
\setplotsymbol ({\thinlinefont .})
{\color[rgb]{0,0,0}\plot  9.529 16.463 12.133 13.923 /
}%
\linethickness= 0.500pt
\setplotsymbol ({\thinlinefont .})
{\color[rgb]{0,0,0}\plot 10.524 14.855 11.540 13.902 /
}%
\linethickness= 0.500pt
\setplotsymbol ({\thinlinefont .})
{\color[rgb]{0,0,0}\plot  9.847 16.739 12.725 13.987 /
}%
\linethickness= 0.500pt
\setplotsymbol ({\thinlinefont .})
{\color[rgb]{0,0,0}\plot 10.291 17.035 12.979 14.326 /
}%
\linethickness= 0.500pt
\setplotsymbol ({\thinlinefont .})
{\color[rgb]{0,0,0}\plot 10.736 17.204 13.233 14.707 /
}%
\linethickness= 0.500pt
\setplotsymbol ({\thinlinefont .})
{\color[rgb]{0,0,0}\plot 11.223 17.289 13.382 15.151 /
}%
\linethickness= 0.500pt
\setplotsymbol ({\thinlinefont .})
{\color[rgb]{0,0,0}\plot 12.048 17.120 13.318 15.828 /
}%
\linethickness= 0.500pt
\setplotsymbol ({\thinlinefont .})
{\color[rgb]{0,0,0}\plot  4.386 17.691  4.216 17.479 /
\putrule from  4.216 17.479 to  4.153 17.479
}%
\linethickness= 0.500pt
\setplotsymbol ({\thinlinefont .})
{\color[rgb]{0,0,0}\plot  4.534 17.141  4.153 16.760 /
}%
\linethickness= 0.500pt
\setplotsymbol ({\thinlinefont .})
{\color[rgb]{0,0,0}\plot  4.682 16.675  4.132 16.125 /
}%
\linethickness= 0.500pt
\setplotsymbol ({\thinlinefont .})
{\color[rgb]{0,0,0}\plot  4.873 16.188  4.153 15.617 /
}%
\linethickness= 0.500pt
\setplotsymbol ({\thinlinefont .})
{\color[rgb]{0,0,0}\plot  5.021 15.807  4.301 15.088 /
}%
\linethickness= 0.500pt
\setplotsymbol ({\thinlinefont .})
{\color[rgb]{0,0,0}\plot  5.254 15.405  4.492 14.601 /
}%
\linethickness= 0.500pt
\setplotsymbol ({\thinlinefont .})
{\color[rgb]{0,0,0}\plot  5.508 14.939  4.640 14.156 /
}%
\linethickness= 0.500pt
\setplotsymbol ({\thinlinefont .})
{\color[rgb]{0,0,0}\plot  5.719 14.601  5.063 13.966 /
}%
\linethickness= 0.500pt
\setplotsymbol ({\thinlinefont .})
{\color[rgb]{0,0,0}\plot  6.037 14.326  5.698 13.966 /
}%
\linethickness= 0.500pt
\setplotsymbol ({\thinlinefont .})
{\color[rgb]{0,0,0}\plot  6.460 14.114  6.354 13.966 /
}%
\linethickness= 0.500pt
\setplotsymbol ({\thinlinefont .})
{\color[rgb]{0,0,0}\plot 12.112 14.601 11.434 13.966 /
}%
\linethickness= 0.500pt
\setplotsymbol ({\thinlinefont .})
{\color[rgb]{0,0,0}\plot 12.344 14.220 12.112 13.987 /
}%
\linethickness= 0.500pt
\setplotsymbol ({\thinlinefont .})
{\color[rgb]{0,0,0}\plot 11.921 15.066 11.223 14.410 /
}%
\linethickness= 0.500pt
\setplotsymbol ({\thinlinefont .})
{\color[rgb]{0,0,0}\plot 11.836 15.596 11.434 15.257 /
}%
\linethickness= 0.500pt
\setplotsymbol ({\thinlinefont .})
{\color[rgb]{0,0,0}\plot 11.773 16.209 11.498 15.892 /
}%
\linethickness= 0.500pt
\setplotsymbol ({\thinlinefont .})
{\color[rgb]{0,0,0}\plot 11.773 16.781 11.604 16.590 /
}%
\linethickness= 0.500pt
\setplotsymbol ({\thinlinefont .})
{\color[rgb]{0,0,0}\plot 11.985 17.797 11.667 17.416 /
}%
\linethickness= 0.500pt
\setplotsymbol ({\thinlinefont .})
{\color[rgb]{0,0,0}\plot 12.090 18.411 11.709 17.987 /
}%
\linethickness= 0.500pt
\setplotsymbol ({\thinlinefont .})
{\color[rgb]{0,0,0}\plot 12.090 19.025 11.836 18.792 /
}%
\linethickness= 0.500pt
\setplotsymbol ({\thinlinefont .})
{\color[rgb]{0,0,0}\plot 12.493 19.046 13.022 18.601 /
}%
\linethickness= 0.500pt
\setplotsymbol ({\thinlinefont .})
{\color[rgb]{0,0,0}\plot 12.958 18.347 13.360 17.966 /
}%
\linethickness=1pt
\setplotsymbol ({\makebox(0,0)[l]{\tencirc\symbol{'160}}})
{\color[rgb]{0,0,0}\plot  5.715 18.415  5.717 18.411 /
\plot  5.717 18.411  5.721 18.400 /
\plot  5.721 18.400  5.730 18.379 /
\plot  5.730 18.379  5.743 18.349 /
\plot  5.743 18.349  5.762 18.309 /
\plot  5.762 18.309  5.783 18.260 /
\plot  5.783 18.260  5.808 18.201 /
\plot  5.808 18.201  5.836 18.138 /
\plot  5.836 18.138  5.865 18.070 /
\plot  5.865 18.070  5.897 18.000 /
\plot  5.897 18.000  5.927 17.930 /
\plot  5.927 17.930  5.956 17.863 /
\plot  5.956 17.863  5.986 17.799 /
\plot  5.986 17.799  6.013 17.738 /
\plot  6.013 17.738  6.039 17.678 /
\plot  6.039 17.678  6.064 17.623 /
\plot  6.064 17.623  6.088 17.573 /
\plot  6.088 17.573  6.111 17.524 /
\plot  6.111 17.524  6.134 17.477 /
\plot  6.134 17.477  6.155 17.433 /
\plot  6.155 17.433  6.176 17.388 /
\plot  6.176 17.388  6.198 17.346 /
\plot  6.198 17.346  6.219 17.304 /
\plot  6.219 17.304  6.240 17.261 /
\plot  6.240 17.261  6.261 17.219 /
\plot  6.261 17.219  6.284 17.175 /
\plot  6.284 17.175  6.308 17.130 /
\plot  6.308 17.130  6.331 17.086 /
\plot  6.331 17.086  6.356 17.041 /
\plot  6.356 17.041  6.382 16.995 /
\plot  6.382 16.995  6.409 16.948 /
\plot  6.409 16.948  6.439 16.902 /
\plot  6.439 16.902  6.469 16.853 /
\plot  6.469 16.853  6.498 16.806 /
\plot  6.498 16.806  6.530 16.760 /
\plot  6.530 16.760  6.564 16.713 /
\plot  6.564 16.713  6.598 16.667 /
\plot  6.598 16.667  6.632 16.620 /
\plot  6.632 16.620  6.665 16.578 /
\plot  6.665 16.578  6.699 16.533 /
\plot  6.699 16.533  6.735 16.493 /
\plot  6.735 16.493  6.771 16.453 /
\plot  6.771 16.453  6.807 16.413 /
\plot  6.807 16.413  6.845 16.377 /
\plot  6.845 16.377  6.881 16.341 /
\plot  6.881 16.341  6.919 16.305 /
\plot  6.919 16.305  6.960 16.271 /
\plot  6.960 16.271  6.998 16.239 /
\plot  6.998 16.239  7.040 16.207 /
\plot  7.040 16.207  7.082 16.173 /
\plot  7.082 16.173  7.129 16.142 /
\plot  7.129 16.142  7.176 16.112 /
\plot  7.176 16.112  7.224 16.080 /
\plot  7.224 16.080  7.275 16.051 /
\plot  7.275 16.051  7.326 16.019 /
\plot  7.326 16.019  7.379 15.989 /
\plot  7.379 15.989  7.434 15.962 /
\plot  7.434 15.962  7.489 15.934 /
\plot  7.489 15.934  7.544 15.909 /
\plot  7.544 15.909  7.601 15.883 /
\plot  7.601 15.883  7.656 15.860 /
\plot  7.656 15.860  7.711 15.837 /
\plot  7.711 15.837  7.766 15.818 /
\plot  7.766 15.818  7.819 15.799 /
\plot  7.819 15.799  7.870 15.782 /
\plot  7.870 15.782  7.921 15.767 /
\plot  7.921 15.767  7.969 15.752 /
\plot  7.969 15.752  8.016 15.742 /
\plot  8.016 15.742  8.062 15.731 /
\plot  8.062 15.731  8.107 15.723 /
\plot  8.107 15.723  8.149 15.716 /
\plot  8.149 15.716  8.200 15.710 /
\plot  8.200 15.710  8.249 15.706 /
\plot  8.249 15.706  8.297 15.704 /
\putrule from  8.297 15.704 to  8.348 15.704
\plot  8.348 15.704  8.397 15.708 /
\plot  8.397 15.708  8.446 15.712 /
\plot  8.446 15.712  8.494 15.718 /
\plot  8.494 15.718  8.543 15.729 /
\plot  8.543 15.729  8.592 15.742 /
\plot  8.592 15.742  8.640 15.756 /
\plot  8.640 15.756  8.689 15.773 /
\plot  8.689 15.773  8.735 15.792 /
\plot  8.735 15.792  8.782 15.814 /
\plot  8.782 15.814  8.827 15.837 /
\plot  8.827 15.837  8.871 15.862 /
\plot  8.871 15.862  8.913 15.888 /
\plot  8.913 15.888  8.954 15.915 /
\plot  8.954 15.915  8.996 15.945 /
\plot  8.996 15.945  9.036 15.974 /
\plot  9.036 15.974  9.076 16.006 /
\plot  9.076 16.006  9.108 16.036 /
\plot  9.108 16.036  9.142 16.066 /
\plot  9.142 16.066  9.178 16.095 /
\plot  9.178 16.095  9.214 16.127 /
\plot  9.214 16.127  9.250 16.161 /
\plot  9.250 16.161  9.288 16.197 /
\plot  9.288 16.197  9.326 16.235 /
\plot  9.326 16.235  9.366 16.273 /
\plot  9.366 16.273  9.409 16.311 /
\plot  9.409 16.311  9.451 16.351 /
\plot  9.451 16.351  9.493 16.394 /
\plot  9.493 16.394  9.536 16.434 /
\plot  9.536 16.434  9.580 16.476 /
\plot  9.580 16.476  9.624 16.516 /
\plot  9.624 16.516  9.669 16.559 /
\plot  9.669 16.559  9.711 16.599 /
\plot  9.711 16.599  9.756 16.639 /
\plot  9.756 16.639  9.800 16.677 /
\plot  9.800 16.677  9.842 16.713 /
\plot  9.842 16.713  9.885 16.749 /
\plot  9.885 16.749  9.927 16.785 /
\plot  9.927 16.785  9.970 16.817 /
\plot  9.970 16.817 10.012 16.849 /
\plot 10.012 16.849 10.054 16.880 /
\plot 10.054 16.880 10.097 16.910 /
\plot 10.097 16.910 10.141 16.940 /
\plot 10.141 16.940 10.185 16.967 /
\plot 10.185 16.967 10.230 16.995 /
\plot 10.230 16.995 10.279 17.022 /
\plot 10.279 17.022 10.327 17.050 /
\plot 10.327 17.050 10.376 17.075 /
\plot 10.376 17.075 10.427 17.101 /
\plot 10.427 17.101 10.480 17.124 /
\plot 10.480 17.124 10.533 17.147 /
\plot 10.533 17.147 10.585 17.170 /
\plot 10.585 17.170 10.640 17.192 /
\plot 10.640 17.192 10.693 17.211 /
\plot 10.693 17.211 10.748 17.230 /
\plot 10.748 17.230 10.801 17.247 /
\plot 10.801 17.247 10.854 17.261 /
\plot 10.854 17.261 10.907 17.276 /
\plot 10.907 17.276 10.958 17.287 /
\plot 10.958 17.287 11.009 17.300 /
\plot 11.009 17.300 11.057 17.308 /
\plot 11.057 17.308 11.106 17.314 /
\plot 11.106 17.314 11.153 17.321 /
\plot 11.153 17.321 11.199 17.327 /
\plot 11.199 17.327 11.246 17.329 /
\plot 11.246 17.329 11.295 17.333 /
\putrule from 11.295 17.333 to 11.343 17.333
\putrule from 11.343 17.333 to 11.394 17.333
\plot 11.394 17.333 11.445 17.331 /
\plot 11.445 17.331 11.496 17.329 /
\plot 11.496 17.329 11.549 17.323 /
\plot 11.549 17.323 11.601 17.316 /
\plot 11.601 17.316 11.654 17.308 /
\plot 11.654 17.308 11.709 17.300 /
\plot 11.709 17.300 11.762 17.287 /
\plot 11.762 17.287 11.815 17.274 /
\plot 11.815 17.274 11.868 17.259 /
\plot 11.868 17.259 11.921 17.244 /
\plot 11.921 17.244 11.972 17.228 /
\plot 11.972 17.228 12.023 17.211 /
\plot 12.023 17.211 12.071 17.192 /
\plot 12.071 17.192 12.118 17.173 /
\plot 12.118 17.173 12.162 17.153 /
\plot 12.162 17.153 12.207 17.132 /
\plot 12.207 17.132 12.249 17.111 /
\plot 12.249 17.111 12.289 17.088 /
\plot 12.289 17.088 12.330 17.065 /
\plot 12.330 17.065 12.370 17.041 /
\plot 12.370 17.041 12.408 17.018 /
\plot 12.408 17.018 12.448 16.990 /
\plot 12.448 16.990 12.486 16.965 /
\plot 12.486 16.965 12.526 16.935 /
\plot 12.526 16.935 12.565 16.906 /
\plot 12.565 16.906 12.605 16.874 /
\plot 12.605 16.874 12.645 16.842 /
\plot 12.645 16.842 12.683 16.808 /
\plot 12.683 16.808 12.723 16.772 /
\plot 12.723 16.772 12.761 16.736 /
\plot 12.761 16.736 12.797 16.701 /
\plot 12.797 16.701 12.835 16.665 /
\plot 12.835 16.665 12.869 16.629 /
\plot 12.869 16.629 12.903 16.593 /
\plot 12.903 16.593 12.935 16.557 /
\plot 12.935 16.557 12.967 16.521 /
\plot 12.967 16.521 12.994 16.487 /
\plot 12.994 16.487 13.022 16.451 /
\plot 13.022 16.451 13.049 16.417 /
\plot 13.049 16.417 13.073 16.385 /
\plot 13.073 16.385 13.098 16.351 /
\plot 13.098 16.351 13.121 16.315 /
\plot 13.121 16.315 13.147 16.277 /
\plot 13.147 16.277 13.170 16.241 /
\plot 13.170 16.241 13.193 16.201 /
\plot 13.193 16.201 13.216 16.159 /
\plot 13.216 16.159 13.240 16.114 /
\plot 13.240 16.114 13.263 16.068 /
\plot 13.263 16.068 13.288 16.017 /
\plot 13.288 16.017 13.316 15.962 /
\plot 13.316 15.962 13.341 15.903 /
\plot 13.341 15.903 13.369 15.843 /
\plot 13.369 15.843 13.394 15.784 /
\plot 13.394 15.784 13.420 15.729 /
\plot 13.420 15.729 13.443 15.676 /
\plot 13.443 15.676 13.462 15.634 /
\plot 13.462 15.634 13.477 15.600 /
\plot 13.477 15.600 13.485 15.577 /
\plot 13.485 15.577 13.492 15.564 /
\plot 13.492 15.564 13.494 15.558 /
}%
\linethickness=1pt
\setplotsymbol ({\makebox(0,0)[l]{\tencirc\symbol{'160}}})
{\color[rgb]{0,0,0}\plot  6.985 13.970  6.991 13.972 /
\plot  6.991 13.972  7.004 13.978 /
\plot  7.004 13.978  7.027 13.987 /
\plot  7.027 13.987  7.059 14.002 /
\plot  7.059 14.002  7.104 14.021 /
\plot  7.104 14.021  7.154 14.044 /
\plot  7.154 14.044  7.209 14.069 /
\plot  7.209 14.069  7.269 14.095 /
\plot  7.269 14.095  7.326 14.122 /
\plot  7.326 14.122  7.383 14.148 /
\plot  7.383 14.148  7.436 14.175 /
\plot  7.436 14.175  7.487 14.201 /
\plot  7.487 14.201  7.531 14.224 /
\plot  7.531 14.224  7.573 14.247 /
\plot  7.573 14.247  7.614 14.271 /
\plot  7.614 14.271  7.650 14.294 /
\plot  7.650 14.294  7.686 14.317 /
\plot  7.686 14.317  7.719 14.343 /
\plot  7.719 14.343  7.753 14.366 /
\plot  7.753 14.366  7.785 14.393 /
\plot  7.785 14.393  7.817 14.421 /
\plot  7.817 14.421  7.851 14.448 /
\plot  7.851 14.448  7.885 14.478 /
\plot  7.885 14.478  7.916 14.510 /
\plot  7.916 14.510  7.950 14.544 /
\plot  7.950 14.544  7.986 14.575 /
\plot  7.986 14.575  8.020 14.611 /
\plot  8.020 14.611  8.054 14.645 /
\plot  8.054 14.645  8.090 14.681 /
\plot  8.090 14.681  8.124 14.715 /
\plot  8.124 14.715  8.156 14.749 /
\plot  8.156 14.749  8.189 14.783 /
\plot  8.189 14.783  8.221 14.815 /
\plot  8.221 14.815  8.251 14.846 /
\plot  8.251 14.846  8.280 14.876 /
\plot  8.280 14.876  8.308 14.903 /
\plot  8.308 14.903  8.335 14.929 /
\plot  8.335 14.929  8.361 14.952 /
\plot  8.361 14.952  8.388 14.975 /
\plot  8.388 14.975  8.418 15.001 /
\plot  8.418 15.001  8.450 15.026 /
\plot  8.450 15.026  8.484 15.050 /
\plot  8.484 15.050  8.515 15.073 /
\plot  8.515 15.073  8.551 15.094 /
\plot  8.551 15.094  8.585 15.115 /
\plot  8.585 15.115  8.621 15.134 /
\plot  8.621 15.134  8.659 15.153 /
\plot  8.659 15.153  8.697 15.172 /
\plot  8.697 15.172  8.735 15.189 /
\plot  8.735 15.189  8.774 15.208 /
\plot  8.774 15.208  8.812 15.225 /
\plot  8.812 15.225  8.852 15.242 /
\plot  8.852 15.242  8.890 15.259 /
\plot  8.890 15.259  8.928 15.276 /
\plot  8.928 15.276  8.970 15.293 /
\plot  8.970 15.293  9.002 15.308 /
\plot  9.002 15.308  9.038 15.323 /
\plot  9.038 15.323  9.074 15.339 /
\plot  9.074 15.339  9.112 15.356 /
\plot  9.112 15.356  9.150 15.373 /
\plot  9.150 15.373  9.193 15.392 /
\plot  9.193 15.392  9.235 15.409 /
\plot  9.235 15.409  9.279 15.426 /
\plot  9.279 15.426  9.324 15.443 /
\plot  9.324 15.443  9.370 15.460 /
\plot  9.370 15.460  9.415 15.475 /
\plot  9.415 15.475  9.462 15.490 /
\plot  9.462 15.490  9.506 15.500 /
\plot  9.506 15.500  9.550 15.511 /
\plot  9.550 15.511  9.593 15.519 /
\plot  9.593 15.519  9.635 15.528 /
\plot  9.635 15.528  9.675 15.532 /
\plot  9.675 15.532  9.716 15.534 /
\putrule from  9.716 15.534 to  9.754 15.534
\plot  9.754 15.534  9.790 15.530 /
\plot  9.790 15.530  9.826 15.526 /
\plot  9.826 15.526  9.862 15.519 /
\plot  9.862 15.519  9.898 15.509 /
\plot  9.898 15.509  9.934 15.496 /
\plot  9.934 15.496  9.967 15.481 /
\plot  9.967 15.481 10.003 15.464 /
\plot 10.003 15.464 10.039 15.443 /
\plot 10.039 15.443 10.073 15.420 /
\plot 10.073 15.420 10.107 15.397 /
\plot 10.107 15.397 10.141 15.369 /
\plot 10.141 15.369 10.175 15.342 /
\plot 10.175 15.342 10.204 15.312 /
\plot 10.204 15.312 10.236 15.280 /
\plot 10.236 15.280 10.264 15.248 /
\plot 10.264 15.248 10.289 15.217 /
\plot 10.289 15.217 10.315 15.185 /
\plot 10.315 15.185 10.338 15.153 /
\plot 10.338 15.153 10.359 15.119 /
\plot 10.359 15.119 10.380 15.088 /
\plot 10.380 15.088 10.399 15.054 /
\plot 10.399 15.054 10.418 15.018 /
\plot 10.418 15.018 10.437 14.980 /
\plot 10.437 14.980 10.454 14.942 /
\plot 10.454 14.942 10.471 14.901 /
\plot 10.471 14.901 10.488 14.861 /
\plot 10.488 14.861 10.505 14.819 /
\plot 10.505 14.819 10.520 14.776 /
\plot 10.520 14.776 10.535 14.734 /
\plot 10.535 14.734 10.549 14.692 /
\plot 10.549 14.692 10.564 14.649 /
\plot 10.564 14.649 10.577 14.609 /
\plot 10.577 14.609 10.590 14.569 /
\plot 10.590 14.569 10.602 14.531 /
\plot 10.602 14.531 10.615 14.493 /
\plot 10.615 14.493 10.628 14.459 /
\plot 10.628 14.459 10.638 14.427 /
\plot 10.638 14.427 10.651 14.395 /
\plot 10.651 14.395 10.664 14.366 /
\plot 10.664 14.366 10.679 14.332 /
\plot 10.679 14.332 10.696 14.298 /
\plot 10.696 14.298 10.715 14.264 /
\plot 10.715 14.264 10.736 14.230 /
\plot 10.736 14.230 10.761 14.196 /
\plot 10.761 14.196 10.789 14.161 /
\plot 10.789 14.161 10.820 14.122 /
\plot 10.820 14.122 10.852 14.084 /
\plot 10.852 14.084 10.884 14.046 /
\plot 10.884 14.046 10.914 14.014 /
\plot 10.914 14.014 10.935 13.991 /
\plot 10.935 13.991 10.947 13.976 /
\plot 10.947 13.976 10.954 13.970 /
}%
\linethickness=1pt
\setplotsymbol ({\makebox(0,0)[l]{\tencirc\symbol{'160}}})
{\color[rgb]{0,0,0}\plot  4.286 18.098  4.288 18.093 /
\plot  4.288 18.093  4.290 18.083 /
\plot  4.290 18.083  4.295 18.066 /
\plot  4.295 18.066  4.301 18.036 /
\plot  4.301 18.036  4.312 17.998 /
\plot  4.312 17.998  4.326 17.947 /
\plot  4.326 17.947  4.341 17.886 /
\plot  4.341 17.886  4.360 17.814 /
\plot  4.360 17.814  4.381 17.733 /
\plot  4.381 17.733  4.405 17.649 /
\plot  4.405 17.649  4.430 17.558 /
\plot  4.430 17.558  4.456 17.465 /
\plot  4.456 17.465  4.481 17.369 /
\plot  4.481 17.369  4.508 17.276 /
\plot  4.508 17.276  4.534 17.185 /
\plot  4.534 17.185  4.559 17.096 /
\plot  4.559 17.096  4.583 17.012 /
\plot  4.583 17.012  4.606 16.931 /
\plot  4.606 16.931  4.629 16.855 /
\plot  4.629 16.855  4.652 16.781 /
\plot  4.652 16.781  4.674 16.711 /
\plot  4.674 16.711  4.695 16.645 /
\plot  4.695 16.645  4.716 16.582 /
\plot  4.716 16.582  4.737 16.523 /
\plot  4.737 16.523  4.758 16.466 /
\plot  4.758 16.466  4.777 16.408 /
\plot  4.777 16.408  4.798 16.353 /
\plot  4.798 16.353  4.820 16.298 /
\plot  4.820 16.298  4.843 16.245 /
\plot  4.843 16.245  4.866 16.188 /
\plot  4.866 16.188  4.889 16.131 /
\plot  4.889 16.131  4.915 16.074 /
\plot  4.915 16.074  4.942 16.017 /
\plot  4.942 16.017  4.970 15.958 /
\plot  4.970 15.958  5.000 15.896 /
\plot  5.000 15.896  5.031 15.833 /
\plot  5.031 15.833  5.063 15.767 /
\plot  5.063 15.767  5.099 15.699 /
\plot  5.099 15.699  5.137 15.627 /
\plot  5.137 15.627  5.177 15.551 /
\plot  5.177 15.551  5.222 15.473 /
\plot  5.222 15.473  5.266 15.390 /
\plot  5.266 15.390  5.313 15.306 /
\plot  5.313 15.306  5.362 15.219 /
\plot  5.362 15.219  5.410 15.132 /
\plot  5.410 15.132  5.459 15.047 /
\plot  5.459 15.047  5.508 14.965 /
\plot  5.508 14.965  5.550 14.889 /
\plot  5.550 14.889  5.592 14.819 /
\plot  5.592 14.819  5.626 14.757 /
\plot  5.626 14.757  5.656 14.707 /
\plot  5.656 14.707  5.679 14.666 /
\plot  5.679 14.666  5.696 14.639 /
\plot  5.696 14.639  5.707 14.620 /
\plot  5.707 14.620  5.713 14.609 /
\plot  5.713 14.609  5.715 14.605 /
}%
\linethickness= 0.500pt
\setplotsymbol ({\thinlinefont .})
{\color[rgb]{0,0,0}\plot  5.715 14.605  5.721 14.601 /
\plot  5.721 14.601  5.736 14.590 /
\plot  5.736 14.590  5.759 14.573 /
\plot  5.759 14.573  5.793 14.548 /
\plot  5.793 14.548  5.836 14.516 /
\plot  5.836 14.516  5.884 14.480 /
\plot  5.884 14.480  5.935 14.444 /
\plot  5.935 14.444  5.988 14.406 /
\plot  5.988 14.406  6.037 14.370 /
\plot  6.037 14.370  6.083 14.338 /
\plot  6.083 14.338  6.126 14.309 /
\plot  6.126 14.309  6.166 14.281 /
\plot  6.166 14.281  6.202 14.258 /
\plot  6.202 14.258  6.234 14.237 /
\plot  6.234 14.237  6.265 14.216 /
\plot  6.265 14.216  6.295 14.199 /
\plot  6.295 14.199  6.325 14.182 /
\plot  6.325 14.182  6.361 14.161 /
\plot  6.361 14.161  6.397 14.141 /
\plot  6.397 14.141  6.433 14.125 /
\plot  6.433 14.125  6.469 14.108 /
\plot  6.469 14.108  6.505 14.091 /
\plot  6.505 14.091  6.541 14.074 /
\plot  6.541 14.074  6.576 14.061 /
\plot  6.576 14.061  6.612 14.048 /
\plot  6.612 14.048  6.644 14.036 /
\plot  6.644 14.036  6.674 14.025 /
\plot  6.674 14.025  6.703 14.017 /
\plot  6.703 14.017  6.729 14.008 /
\plot  6.729 14.008  6.752 14.002 /
\plot  6.752 14.002  6.773 13.995 /
\plot  6.773 13.995  6.801 13.989 /
\plot  6.801 13.989  6.826 13.985 /
\plot  6.826 13.985  6.852 13.981 /
\plot  6.852 13.981  6.879 13.976 /
\plot  6.879 13.976  6.907 13.974 /
\plot  6.907 13.974  6.934 13.972 /
\putrule from  6.934 13.972 to  6.960 13.972
\plot  6.960 13.972  6.977 13.970 /
\putrule from  6.977 13.970 to  6.983 13.970
\putrule from  6.983 13.970 to  6.985 13.970
}%
%
%
\linethickness= 0.500pt
\setplotsymbol ({\thinlinefont .})
{\color[rgb]{0,0,0}\plot 11.113 13.970 11.115 13.976 /
\plot 11.115 13.976 11.119 13.989 /
\plot 11.119 13.989 11.125 14.010 /
\plot 11.125 14.010 11.136 14.044 /
\plot 11.136 14.044 11.148 14.089 /
\plot 11.148 14.089 11.165 14.144 /
\plot 11.165 14.144 11.184 14.207 /
\plot 11.184 14.207 11.206 14.275 /
\plot 11.206 14.275 11.225 14.345 /
\plot 11.225 14.345 11.246 14.417 /
\plot 11.246 14.417 11.267 14.489 /
\plot 11.267 14.489 11.286 14.556 /
\plot 11.286 14.556 11.303 14.624 /
\plot 11.303 14.624 11.318 14.688 /
\plot 11.318 14.688 11.335 14.749 /
\plot 11.335 14.749 11.347 14.810 /
\plot 11.347 14.810 11.360 14.867 /
\plot 11.360 14.867 11.373 14.927 /
\plot 11.373 14.927 11.383 14.986 /
\plot 11.383 14.986 11.394 15.045 /
\plot 11.394 15.045 11.405 15.107 /
\plot 11.405 15.107 11.411 15.157 /
\plot 11.411 15.157 11.419 15.208 /
\plot 11.419 15.208 11.426 15.263 /
\plot 11.426 15.263 11.434 15.318 /
\plot 11.434 15.318 11.441 15.375 /
\plot 11.441 15.375 11.449 15.433 /
\plot 11.449 15.433 11.455 15.494 /
\plot 11.455 15.494 11.464 15.558 /
\plot 11.464 15.558 11.470 15.623 /
\plot 11.470 15.623 11.477 15.689 /
\plot 11.477 15.689 11.485 15.756 /
\plot 11.485 15.756 11.491 15.824 /
\plot 11.491 15.824 11.498 15.894 /
\plot 11.498 15.894 11.506 15.966 /
\plot 11.506 15.966 11.513 16.036 /
\plot 11.513 16.036 11.519 16.108 /
\plot 11.519 16.108 11.525 16.178 /
\plot 11.525 16.178 11.532 16.248 /
\plot 11.532 16.248 11.538 16.317 /
\plot 11.538 16.317 11.544 16.385 /
\plot 11.544 16.385 11.551 16.451 /
\plot 11.551 16.451 11.557 16.516 /
\plot 11.557 16.516 11.561 16.580 /
\plot 11.561 16.580 11.568 16.643 /
\plot 11.568 16.643 11.574 16.705 /
\plot 11.574 16.705 11.578 16.764 /
\plot 11.578 16.764 11.585 16.823 /
\plot 11.585 16.823 11.589 16.880 /
\plot 11.589 16.880 11.595 16.946 /
\plot 11.595 16.946 11.601 17.012 /
\plot 11.601 17.012 11.608 17.077 /
\plot 11.608 17.077 11.614 17.143 /
\plot 11.614 17.143 11.620 17.209 /
\plot 11.620 17.209 11.627 17.274 /
\plot 11.627 17.274 11.633 17.342 /
\plot 11.633 17.342 11.640 17.407 /
\plot 11.640 17.407 11.646 17.473 /
\plot 11.646 17.473 11.654 17.537 /
\plot 11.654 17.537 11.661 17.600 /
\plot 11.661 17.600 11.669 17.664 /
\plot 11.669 17.664 11.676 17.725 /
\plot 11.676 17.725 11.682 17.784 /
\plot 11.682 17.784 11.690 17.844 /
\plot 11.690 17.844 11.697 17.899 /
\plot 11.697 17.899 11.703 17.951 /
\plot 11.703 17.951 11.709 18.002 /
\plot 11.709 18.002 11.716 18.051 /
\plot 11.716 18.051 11.722 18.095 /
\plot 11.722 18.095 11.728 18.140 /
\plot 11.728 18.140 11.735 18.180 /
\plot 11.735 18.180 11.741 18.218 /
\plot 11.741 18.218 11.748 18.256 /
\plot 11.748 18.256 11.756 18.309 /
\plot 11.756 18.309 11.767 18.360 /
\plot 11.767 18.360 11.775 18.409 /
\plot 11.775 18.409 11.786 18.455 /
\plot 11.786 18.455 11.794 18.500 /
\plot 11.794 18.500 11.805 18.544 /
\plot 11.805 18.544 11.813 18.584 /
\plot 11.813 18.584 11.824 18.625 /
\plot 11.824 18.625 11.832 18.661 /
\plot 11.832 18.661 11.841 18.694 /
\plot 11.841 18.694 11.849 18.724 /
\plot 11.849 18.724 11.858 18.752 /
\plot 11.858 18.752 11.864 18.777 /
\plot 11.864 18.777 11.870 18.798 /
\plot 11.870 18.798 11.874 18.819 /
\plot 11.874 18.819 11.881 18.838 /
\plot 11.881 18.838 11.887 18.866 /
\plot 11.887 18.866 11.891 18.891 /
\plot 11.891 18.891 11.896 18.917 /
\plot 11.896 18.917 11.900 18.944 /
\plot 11.900 18.944 11.902 18.972 /
\plot 11.902 18.972 11.904 18.999 /
\putrule from 11.904 18.999 to 11.904 19.025
\plot 11.904 19.025 11.906 19.042 /
\putrule from 11.906 19.042 to 11.906 19.048
\putrule from 11.906 19.048 to 11.906 19.050
}%
%
%
\linethickness= 0.500pt
\setplotsymbol ({\thinlinefont .})
{\color[rgb]{0,0,0}\putrule from 12.700 19.050 to 12.698 19.050
\putrule from 12.698 19.050 to 12.689 19.050
\plot 12.689 19.050 12.670 19.048 /
\plot 12.670 19.048 12.639 19.044 /
\plot 12.639 19.044 12.601 19.039 /
\plot 12.601 19.039 12.562 19.031 /
\plot 12.562 19.031 12.524 19.022 /
\plot 12.524 19.022 12.490 19.012 /
\plot 12.490 19.012 12.459 18.999 /
\plot 12.459 18.999 12.431 18.984 /
\plot 12.431 18.984 12.408 18.967 /
\plot 12.408 18.967 12.383 18.944 /
\plot 12.383 18.944 12.366 18.927 /
\plot 12.366 18.927 12.351 18.908 /
\plot 12.351 18.908 12.334 18.887 /
\plot 12.334 18.887 12.317 18.862 /
\plot 12.317 18.862 12.298 18.836 /
\plot 12.298 18.836 12.281 18.804 /
\plot 12.281 18.804 12.262 18.773 /
\plot 12.262 18.773 12.243 18.737 /
\plot 12.243 18.737 12.224 18.699 /
\plot 12.224 18.699 12.205 18.658 /
\plot 12.205 18.658 12.186 18.616 /
\plot 12.186 18.616 12.167 18.574 /
\plot 12.167 18.574 12.150 18.527 /
\plot 12.150 18.527 12.131 18.481 /
\plot 12.131 18.481 12.114 18.434 /
\plot 12.114 18.434 12.097 18.383 /
\plot 12.097 18.383 12.082 18.335 /
\plot 12.082 18.335 12.065 18.282 /
\plot 12.065 18.282 12.052 18.244 /
\plot 12.052 18.244 12.040 18.201 /
\plot 12.040 18.201 12.029 18.157 /
\plot 12.029 18.157 12.016 18.112 /
\plot 12.016 18.112 12.004 18.066 /
\plot 12.004 18.066 11.991 18.017 /
\plot 11.991 18.017 11.978 17.966 /
\plot 11.978 17.966 11.966 17.913 /
\plot 11.966 17.913 11.951 17.858 /
\plot 11.951 17.858 11.938 17.803 /
\plot 11.938 17.803 11.925 17.746 /
\plot 11.925 17.746 11.913 17.689 /
\plot 11.913 17.689 11.902 17.632 /
\plot 11.902 17.632 11.889 17.575 /
\plot 11.889 17.575 11.879 17.515 /
\plot 11.879 17.515 11.866 17.458 /
\plot 11.866 17.458 11.858 17.403 /
\plot 11.858 17.403 11.847 17.348 /
\plot 11.847 17.348 11.839 17.293 /
\plot 11.839 17.293 11.830 17.240 /
\plot 11.830 17.240 11.822 17.187 /
\plot 11.822 17.187 11.813 17.137 /
\plot 11.813 17.137 11.807 17.088 /
\plot 11.807 17.088 11.800 17.039 /
\plot 11.800 17.039 11.794 16.986 /
\plot 11.794 16.986 11.788 16.933 /
\plot 11.788 16.933 11.783 16.880 /
\plot 11.783 16.880 11.779 16.828 /
\plot 11.779 16.828 11.775 16.772 /
\plot 11.775 16.772 11.771 16.717 /
\plot 11.771 16.717 11.769 16.662 /
\plot 11.769 16.662 11.764 16.605 /
\putrule from 11.764 16.605 to 11.764 16.548
\plot 11.764 16.548 11.762 16.491 /
\putrule from 11.762 16.491 to 11.762 16.434
\putrule from 11.762 16.434 to 11.762 16.377
\plot 11.762 16.377 11.764 16.320 /
\putrule from 11.764 16.320 to 11.764 16.264
\plot 11.764 16.264 11.769 16.209 /
\plot 11.769 16.209 11.771 16.154 /
\plot 11.771 16.154 11.775 16.101 /
\plot 11.775 16.101 11.779 16.049 /
\plot 11.779 16.049 11.783 15.998 /
\plot 11.783 15.998 11.788 15.947 /
\plot 11.788 15.947 11.794 15.898 /
\plot 11.794 15.898 11.800 15.847 /
\plot 11.800 15.847 11.807 15.799 /
\plot 11.807 15.799 11.815 15.748 /
\plot 11.815 15.748 11.824 15.697 /
\plot 11.824 15.697 11.832 15.644 /
\plot 11.832 15.644 11.843 15.591 /
\plot 11.843 15.591 11.851 15.536 /
\plot 11.851 15.536 11.864 15.481 /
\plot 11.864 15.481 11.874 15.424 /
\plot 11.874 15.424 11.887 15.367 /
\plot 11.887 15.367 11.900 15.312 /
\plot 11.900 15.312 11.913 15.255 /
\plot 11.913 15.255 11.927 15.200 /
\plot 11.927 15.200 11.942 15.145 /
\plot 11.942 15.145 11.955 15.092 /
\plot 11.955 15.092 11.970 15.039 /
\plot 11.970 15.039 11.982 14.990 /
\plot 11.982 14.990 11.997 14.942 /
\plot 11.997 14.942 12.012 14.897 /
\plot 12.012 14.897 12.025 14.855 /
\plot 12.025 14.855 12.037 14.812 /
\plot 12.037 14.812 12.052 14.774 /
\plot 12.052 14.774 12.065 14.736 /
\plot 12.065 14.736 12.084 14.688 /
\plot 12.084 14.688 12.101 14.641 /
\plot 12.101 14.641 12.120 14.597 /
\plot 12.120 14.597 12.141 14.552 /
\plot 12.141 14.552 12.160 14.510 /
\plot 12.160 14.510 12.181 14.470 /
\plot 12.181 14.470 12.203 14.429 /
\plot 12.203 14.429 12.224 14.393 /
\plot 12.224 14.393 12.245 14.357 /
\plot 12.245 14.357 12.266 14.326 /
\plot 12.266 14.326 12.287 14.296 /
\plot 12.287 14.296 12.306 14.268 /
\plot 12.306 14.268 12.327 14.243 /
\plot 12.327 14.243 12.347 14.222 /
\plot 12.347 14.222 12.363 14.201 /
\plot 12.363 14.201 12.383 14.182 /
\plot 12.383 14.182 12.408 14.158 /
\plot 12.408 14.158 12.431 14.137 /
\plot 12.431 14.137 12.459 14.116 /
\plot 12.459 14.116 12.490 14.095 /
\plot 12.490 14.095 12.524 14.072 /
\plot 12.524 14.072 12.562 14.048 /
\plot 12.562 14.048 12.601 14.025 /
\plot 12.601 14.025 12.639 14.004 /
\plot 12.639 14.004 12.670 13.987 /
\plot 12.670 13.987 12.689 13.976 /
\plot 12.689 13.976 12.698 13.970 /
\putrule from 12.698 13.970 to 12.700 13.970
}%
%
%
\linethickness= 0.500pt
\setplotsymbol ({\thinlinefont .})
{\color[rgb]{0,0,0}\plot 12.700 19.050 12.704 19.044 /
\plot 12.704 19.044 12.711 19.033 /
\plot 12.711 19.033 12.725 19.012 /
\plot 12.725 19.012 12.747 18.980 /
\plot 12.747 18.980 12.774 18.938 /
\plot 12.774 18.938 12.806 18.887 /
\plot 12.806 18.887 12.844 18.830 /
\plot 12.844 18.830 12.884 18.766 /
\plot 12.884 18.766 12.926 18.703 /
\plot 12.926 18.703 12.967 18.639 /
\plot 12.967 18.639 13.007 18.578 /
\plot 13.007 18.578 13.045 18.519 /
\plot 13.045 18.519 13.081 18.464 /
\plot 13.081 18.464 13.113 18.411 /
\plot 13.113 18.411 13.142 18.364 /
\plot 13.142 18.364 13.170 18.318 /
\plot 13.170 18.318 13.195 18.275 /
\plot 13.195 18.275 13.219 18.235 /
\plot 13.219 18.235 13.242 18.197 /
\plot 13.242 18.197 13.261 18.161 /
\plot 13.261 18.161 13.282 18.123 /
\plot 13.282 18.123 13.303 18.085 /
\plot 13.303 18.085 13.324 18.045 /
\plot 13.324 18.045 13.346 18.004 /
\plot 13.346 18.004 13.365 17.962 /
\plot 13.365 17.962 13.386 17.920 /
\plot 13.386 17.920 13.405 17.875 /
\plot 13.405 17.875 13.424 17.831 /
\plot 13.424 17.831 13.443 17.786 /
\plot 13.443 17.786 13.462 17.740 /
\plot 13.462 17.740 13.479 17.693 /
\plot 13.479 17.693 13.496 17.647 /
\plot 13.496 17.647 13.511 17.598 /
\plot 13.511 17.598 13.526 17.551 /
\plot 13.526 17.551 13.540 17.505 /
\plot 13.540 17.505 13.553 17.456 /
\plot 13.553 17.456 13.564 17.410 /
\plot 13.564 17.410 13.574 17.363 /
\plot 13.574 17.363 13.583 17.319 /
\plot 13.583 17.319 13.591 17.272 /
\plot 13.591 17.272 13.600 17.223 /
\plot 13.600 17.223 13.606 17.181 /
\plot 13.606 17.181 13.612 17.137 /
\plot 13.612 17.137 13.617 17.090 /
\plot 13.617 17.090 13.621 17.041 /
\plot 13.621 17.041 13.625 16.993 /
\plot 13.625 16.993 13.629 16.940 /
\plot 13.629 16.940 13.631 16.885 /
\plot 13.631 16.885 13.636 16.830 /
\putrule from 13.636 16.830 to 13.636 16.772
\plot 13.636 16.772 13.638 16.713 /
\putrule from 13.638 16.713 to 13.638 16.654
\putrule from 13.638 16.654 to 13.638 16.595
\plot 13.638 16.595 13.636 16.533 /
\putrule from 13.636 16.533 to 13.636 16.472
\plot 13.636 16.472 13.631 16.413 /
\plot 13.631 16.413 13.629 16.351 /
\plot 13.629 16.351 13.625 16.292 /
\plot 13.625 16.292 13.621 16.235 /
\plot 13.621 16.235 13.617 16.178 /
\plot 13.617 16.178 13.612 16.121 /
\plot 13.612 16.121 13.606 16.063 /
\plot 13.606 16.063 13.600 16.006 /
\plot 13.600 16.006 13.593 15.955 /
\plot 13.593 15.955 13.587 15.903 /
\plot 13.587 15.903 13.578 15.850 /
\plot 13.578 15.850 13.570 15.797 /
\plot 13.570 15.797 13.561 15.742 /
\plot 13.561 15.742 13.553 15.684 /
\plot 13.553 15.684 13.542 15.627 /
\plot 13.542 15.627 13.530 15.570 /
\plot 13.530 15.570 13.519 15.511 /
\plot 13.519 15.511 13.506 15.452 /
\plot 13.506 15.452 13.494 15.392 /
\plot 13.494 15.392 13.479 15.333 /
\plot 13.479 15.333 13.464 15.276 /
\plot 13.464 15.276 13.449 15.219 /
\plot 13.449 15.219 13.434 15.162 /
\plot 13.434 15.162 13.418 15.109 /
\plot 13.418 15.109 13.403 15.056 /
\plot 13.403 15.056 13.386 15.003 /
\plot 13.386 15.003 13.369 14.954 /
\plot 13.369 14.954 13.352 14.908 /
\plot 13.352 14.908 13.335 14.863 /
\plot 13.335 14.863 13.318 14.819 /
\plot 13.318 14.819 13.299 14.776 /
\plot 13.299 14.776 13.282 14.736 /
\plot 13.282 14.736 13.259 14.692 /
\plot 13.259 14.692 13.238 14.645 /
\plot 13.238 14.645 13.212 14.601 /
\plot 13.212 14.601 13.185 14.558 /
\plot 13.185 14.558 13.157 14.514 /
\plot 13.157 14.514 13.125 14.470 /
\plot 13.125 14.470 13.089 14.421 /
\plot 13.089 14.421 13.051 14.372 /
\plot 13.051 14.372 13.011 14.321 /
\plot 13.011 14.321 12.967 14.271 /
\plot 12.967 14.271 12.922 14.216 /
\plot 12.922 14.216 12.876 14.165 /
\plot 12.876 14.165 12.833 14.116 /
\plot 12.833 14.116 12.793 14.072 /
\plot 12.793 14.072 12.759 14.034 /
\plot 12.759 14.034 12.734 14.006 /
\plot 12.734 14.006 12.715 13.985 /
\plot 12.715 13.985 12.704 13.974 /
\plot 12.704 13.974 12.700 13.970 /
}%
%
%
\linethickness=1pt
\setplotsymbol ({\makebox(0,0)[l]{\tencirc\symbol{'160}}})
{\color[rgb]{0,0,0}\plot 13.652 17.145 13.648 17.151 /
\plot 13.648 17.151 13.642 17.164 /
\plot 13.642 17.164 13.627 17.189 /
\plot 13.627 17.189 13.606 17.225 /
\plot 13.606 17.225 13.578 17.274 /
\plot 13.578 17.274 13.545 17.333 /
\plot 13.545 17.333 13.506 17.399 /
\plot 13.506 17.399 13.464 17.471 /
\plot 13.464 17.471 13.422 17.543 /
\plot 13.422 17.543 13.379 17.615 /
\plot 13.379 17.615 13.339 17.685 /
\plot 13.339 17.685 13.301 17.752 /
\plot 13.301 17.752 13.263 17.814 /
\plot 13.263 17.814 13.229 17.871 /
\plot 13.229 17.871 13.200 17.926 /
\plot 13.200 17.926 13.170 17.975 /
\plot 13.170 17.975 13.142 18.019 /
\plot 13.142 18.019 13.115 18.062 /
\plot 13.115 18.062 13.092 18.102 /
\plot 13.092 18.102 13.068 18.140 /
\plot 13.068 18.140 13.045 18.176 /
\plot 13.045 18.176 13.015 18.220 /
\plot 13.015 18.220 12.988 18.265 /
\plot 12.988 18.265 12.960 18.307 /
\plot 12.960 18.307 12.931 18.349 /
\plot 12.931 18.349 12.901 18.392 /
\plot 12.901 18.392 12.871 18.432 /
\plot 12.871 18.432 12.842 18.472 /
\plot 12.842 18.472 12.812 18.512 /
\plot 12.812 18.512 12.783 18.553 /
\plot 12.783 18.553 12.753 18.589 /
\plot 12.753 18.589 12.725 18.625 /
\plot 12.725 18.625 12.696 18.658 /
\plot 12.696 18.658 12.668 18.688 /
\plot 12.668 18.688 12.641 18.718 /
\plot 12.641 18.718 12.615 18.743 /
\plot 12.615 18.743 12.590 18.768 /
\plot 12.590 18.768 12.565 18.792 /
\plot 12.565 18.792 12.541 18.811 /
\plot 12.541 18.811 12.509 18.836 /
\plot 12.509 18.836 12.478 18.860 /
\plot 12.478 18.860 12.444 18.881 /
\plot 12.444 18.881 12.408 18.902 /
\plot 12.408 18.902 12.366 18.923 /
\plot 12.366 18.923 12.321 18.946 /
\plot 12.321 18.946 12.272 18.967 /
\plot 12.272 18.967 12.222 18.989 /
\plot 12.222 18.989 12.173 19.010 /
\plot 12.173 19.010 12.129 19.027 /
\plot 12.129 19.027 12.095 19.039 /
\plot 12.095 19.039 12.076 19.046 /
\plot 12.076 19.046 12.067 19.050 /
\putrule from 12.067 19.050 to 12.065 19.050
}%
%
%
\put{\SetFigFont{6}{7.2}{\rmdefault}{\mddefault}{\updefault}{\color[rgb]{0,0,0}$\cK_1$}%
} [lB] at 10.058 15.977
%
%
\put{\SetFigFont{6}{7.2}{\rmdefault}{\mddefault}{\updefault}{\color[rgb]{0,0,0}$S_+$}%
} [lB] at 13.678 15.426
%
%
\put{\SetFigFont{6}{7.2}{\rmdefault}{\mddefault}{\updefault}{\color[rgb]{0,0,0}$N_+$}%
} [lB] at 13.741 17.141
%
%
\put{\SetFigFont{6}{7.2}{\rmdefault}{\mddefault}{\updefault}{\color[rgb]{0,0,0}$S_-$}%
} [lB] at  3.243 18.051
\put{\SetFigFont{6}{7.2}{\rmdefault}{\mddefault}{\updefault}{\color[rgb]{0,0,0}$N_-$}%
} [lB] at  5.698 18.517
\linethickness=0pt
\putrectangle corners at  3.211 19.120 and 13.773 13.877
\endpicture}
\qquad
\font\thinlinefont=cmr5
\begingroup\makeatletter\ifx\SetFigFont\undefined%
\gdef\SetFigFont#1#2#3#4#5{%
  \reset@font\fontsize{#1}{#2pt}%
  \fontfamily{#3}\fontseries{#4}\fontshape{#5}%
  \selectfont}%
\fi\endgroup%
\mbox{\beginpicture
\setcoordinatesystem units <0.50000cm,0.50000cm>
\unitlength=0.50000cm
\linethickness=1pt
\setplotsymbol ({\makebox(0,0)[l]{\tencirc\symbol{'160}}})
\setshadesymbol ({\thinlinefont .})
\setlinear
%
%
\linethickness= 0.500pt
\setplotsymbol ({\thinlinefont .})
{\color[rgb]{0,0,0}\putrule from  5.080 16.510 to 15.240 16.510
}%
%
%
\linethickness= 0.500pt
\setplotsymbol ({\thinlinefont .})
{\color[rgb]{0,0,0}\putrule from  5.080 11.430 to 15.240 11.430
}%
%
%
\linethickness=1pt
\setplotsymbol ({\makebox(0,0)[l]{\tencirc\symbol{'160}}})
{\color[rgb]{0,0,0}\putrule from  5.080 21.590 to 15.240 21.590
}%
%
%
\linethickness=1pt
\setplotsymbol ({\makebox(0,0)[l]{\tencirc\symbol{'160}}})
{\color[rgb]{0,0,0}\putrule from 15.240 21.590 to 15.240  6.315
\putrule from 15.240  6.350 to  5.045  6.350
\putrule from  5.080  6.350 to  5.080 21.590
}%
%
%
\linethickness= 0.500pt
\setplotsymbol ({\thinlinefont .})
{\color[rgb]{0,0,0}\plot  6.985 15.081  7.938 15.716 /
%
%
\plot  7.761 15.523  7.938 15.716  7.691 15.628 /
}%
%
%
\linethickness= 0.500pt
\setplotsymbol ({\thinlinefont .})
{\color[rgb]{0,0,0}\plot 13.970 16.192 14.446 15.875 /
%
%
\plot 14.200 15.963 14.446 15.875 14.270 16.069 /
}%
%
%
\linethickness= 0.500pt
\setplotsymbol ({\thinlinefont .})
{\color[rgb]{0,0,0}\putrule from  5.080 17.462 to  5.080 16.034
%
%
\plot  5.017 16.288  5.080 16.034  5.143 16.288 /
}%
%
%
\linethickness= 0.500pt
\setplotsymbol ({\thinlinefont .})
{\color[rgb]{0,0,0}\putrule from 15.240 12.383 to 15.240 10.795
%
%
\plot 15.176 11.049 15.240 10.795 15.304 11.049 /
}%
%
%
\linethickness= 0.500pt
\setplotsymbol ({\thinlinefont .})
{\color[rgb]{0,0,0}\plot  5.080 12.700  5.082 12.706 /
\plot  5.082 12.706  5.088 12.717 /
\plot  5.088 12.717  5.101 12.740 /
\plot  5.101 12.740  5.118 12.774 /
\plot  5.118 12.774  5.141 12.819 /
\plot  5.141 12.819  5.171 12.876 /
\plot  5.171 12.876  5.207 12.941 /
\plot  5.207 12.941  5.245 13.013 /
\plot  5.245 13.013  5.285 13.089 /
\plot  5.285 13.089  5.328 13.168 /
\plot  5.328 13.168  5.370 13.244 /
\plot  5.370 13.244  5.412 13.320 /
\plot  5.412 13.320  5.453 13.392 /
\plot  5.453 13.392  5.491 13.460 /
\plot  5.491 13.460  5.529 13.523 /
\plot  5.529 13.523  5.565 13.583 /
\plot  5.565 13.583  5.599 13.638 /
\plot  5.599 13.638  5.632 13.691 /
\plot  5.632 13.691  5.664 13.739 /
\plot  5.664 13.739  5.696 13.786 /
\plot  5.696 13.786  5.730 13.830 /
\plot  5.730 13.830  5.762 13.875 /
\plot  5.762 13.875  5.795 13.917 /
\plot  5.795 13.917  5.825 13.955 /
\plot  5.825 13.955  5.857 13.995 /
\plot  5.857 13.995  5.891 14.034 /
\plot  5.891 14.034  5.925 14.074 /
\plot  5.925 14.074  5.961 14.112 /
\plot  5.961 14.112  5.997 14.152 /
\plot  5.997 14.152  6.035 14.192 /
\plot  6.035 14.192  6.073 14.232 /
\plot  6.073 14.232  6.113 14.275 /
\plot  6.113 14.275  6.155 14.315 /
\plot  6.155 14.315  6.198 14.357 /
\plot  6.198 14.357  6.242 14.398 /
\plot  6.242 14.398  6.284 14.440 /
\plot  6.284 14.440  6.329 14.480 /
\plot  6.329 14.480  6.375 14.520 /
\plot  6.375 14.520  6.420 14.561 /
\plot  6.420 14.561  6.464 14.601 /
\plot  6.464 14.601  6.509 14.639 /
\plot  6.509 14.639  6.553 14.677 /
\plot  6.553 14.677  6.598 14.713 /
\plot  6.598 14.713  6.642 14.751 /
\plot  6.642 14.751  6.684 14.785 /
\plot  6.684 14.785  6.727 14.821 /
\plot  6.727 14.821  6.769 14.855 /
\plot  6.769 14.855  6.811 14.889 /
\plot  6.811 14.889  6.854 14.922 /
\plot  6.854 14.922  6.892 14.954 /
\plot  6.892 14.954  6.932 14.986 /
\plot  6.932 14.986  6.972 15.018 /
\plot  6.972 15.018  7.013 15.050 /
\plot  7.013 15.050  7.055 15.081 /
\plot  7.055 15.081  7.097 15.115 /
\plot  7.097 15.115  7.142 15.149 /
\plot  7.142 15.149  7.188 15.183 /
\plot  7.188 15.183  7.235 15.219 /
\plot  7.235 15.219  7.283 15.255 /
\plot  7.283 15.255  7.332 15.291 /
\plot  7.332 15.291  7.383 15.329 /
\plot  7.383 15.329  7.434 15.365 /
\plot  7.434 15.365  7.487 15.403 /
\plot  7.487 15.403  7.540 15.439 /
\plot  7.540 15.439  7.595 15.477 /
\plot  7.595 15.477  7.648 15.515 /
\plot  7.648 15.515  7.703 15.551 /
\plot  7.703 15.551  7.758 15.589 /
\plot  7.758 15.589  7.813 15.625 /
\plot  7.813 15.625  7.868 15.661 /
\plot  7.868 15.661  7.923 15.697 /
\plot  7.923 15.697  7.978 15.731 /
\plot  7.978 15.731  8.033 15.765 /
\plot  8.033 15.765  8.088 15.801 /
\plot  8.088 15.801  8.143 15.835 /
\plot  8.143 15.835  8.198 15.869 /
\plot  8.198 15.869  8.255 15.900 /
\plot  8.255 15.900  8.306 15.932 /
\plot  8.306 15.932  8.357 15.962 /
\plot  8.357 15.962  8.410 15.991 /
\plot  8.410 15.991  8.465 16.021 /
\plot  8.465 16.021  8.520 16.053 /
\plot  8.520 16.053  8.577 16.085 /
\plot  8.577 16.085  8.636 16.116 /
\plot  8.636 16.116  8.697 16.148 /
\plot  8.697 16.148  8.759 16.180 /
\plot  8.759 16.180  8.822 16.212 /
\plot  8.822 16.212  8.888 16.245 /
\plot  8.888 16.245  8.956 16.277 /
\plot  8.956 16.277  9.023 16.309 /
\plot  9.023 16.309  9.091 16.341 /
\plot  9.091 16.341  9.163 16.375 /
\plot  9.163 16.375  9.233 16.404 /
\plot  9.233 16.404  9.305 16.436 /
\plot  9.305 16.436  9.377 16.466 /
\plot  9.377 16.466  9.449 16.495 /
\plot  9.449 16.495  9.521 16.525 /
\plot  9.521 16.525  9.593 16.552 /
\plot  9.593 16.552  9.665 16.578 /
\plot  9.665 16.578  9.735 16.603 /
\plot  9.735 16.603  9.807 16.626 /
\plot  9.807 16.626  9.874 16.650 /
\plot  9.874 16.650  9.944 16.671 /
\plot  9.944 16.671 10.012 16.692 /
\plot 10.012 16.692 10.080 16.711 /
\plot 10.080 16.711 10.147 16.728 /
\plot 10.147 16.728 10.213 16.745 /
\plot 10.213 16.745 10.279 16.760 /
\plot 10.279 16.760 10.346 16.775 /
\plot 10.346 16.775 10.412 16.787 /
\plot 10.412 16.787 10.478 16.800 /
\plot 10.478 16.800 10.545 16.811 /
\plot 10.545 16.811 10.613 16.821 /
\plot 10.613 16.821 10.681 16.830 /
\plot 10.681 16.830 10.751 16.836 /
\plot 10.751 16.836 10.823 16.842 /
\plot 10.823 16.842 10.894 16.849 /
\plot 10.894 16.849 10.966 16.851 /
\plot 10.966 16.851 11.041 16.855 /
\putrule from 11.041 16.855 to 11.115 16.855
\putrule from 11.115 16.855 to 11.191 16.855
\plot 11.191 16.855 11.267 16.853 /
\plot 11.267 16.853 11.343 16.849 /
\plot 11.343 16.849 11.419 16.844 /
\plot 11.419 16.844 11.496 16.838 /
\plot 11.496 16.838 11.570 16.830 /
\plot 11.570 16.830 11.646 16.819 /
\plot 11.646 16.819 11.720 16.808 /
\plot 11.720 16.808 11.794 16.796 /
\plot 11.794 16.796 11.866 16.781 /
\plot 11.866 16.781 11.938 16.766 /
\plot 11.938 16.766 12.008 16.749 /
\plot 12.008 16.749 12.076 16.730 /
\plot 12.076 16.730 12.141 16.711 /
\plot 12.141 16.711 12.207 16.690 /
\plot 12.207 16.690 12.270 16.667 /
\plot 12.270 16.667 12.332 16.643 /
\plot 12.332 16.643 12.393 16.618 /
\plot 12.393 16.618 12.452 16.593 /
\plot 12.452 16.593 12.509 16.565 /
\plot 12.509 16.565 12.569 16.535 /
\plot 12.569 16.535 12.624 16.506 /
\plot 12.624 16.506 12.681 16.474 /
\plot 12.681 16.474 12.736 16.440 /
\plot 12.736 16.440 12.791 16.404 /
\plot 12.791 16.404 12.848 16.366 /
\plot 12.848 16.366 12.903 16.328 /
\plot 12.903 16.328 12.958 16.286 /
\plot 12.958 16.286 13.013 16.241 /
\plot 13.013 16.241 13.070 16.195 /
\plot 13.070 16.195 13.125 16.148 /
\plot 13.125 16.148 13.180 16.097 /
\plot 13.180 16.097 13.236 16.044 /
\plot 13.236 16.044 13.291 15.991 /
\plot 13.291 15.991 13.343 15.936 /
\plot 13.343 15.936 13.396 15.879 /
\plot 13.396 15.879 13.449 15.820 /
\plot 13.449 15.820 13.502 15.761 /
\plot 13.502 15.761 13.551 15.699 /
\plot 13.551 15.699 13.602 15.638 /
\plot 13.602 15.638 13.648 15.574 /
\plot 13.648 15.574 13.695 15.513 /
\plot 13.695 15.513 13.739 15.450 /
\plot 13.739 15.450 13.784 15.386 /
\plot 13.784 15.386 13.824 15.323 /
\plot 13.824 15.323 13.864 15.259 /
\plot 13.864 15.259 13.902 15.196 /
\plot 13.902 15.196 13.940 15.132 /
\plot 13.940 15.132 13.974 15.069 /
\plot 13.974 15.069 14.008 15.007 /
\plot 14.008 15.007 14.042 14.944 /
\plot 14.042 14.944 14.072 14.880 /
\plot 14.072 14.880 14.103 14.817 /
\plot 14.103 14.817 14.129 14.760 /
\plot 14.129 14.760 14.154 14.702 /
\plot 14.154 14.702 14.177 14.643 /
\plot 14.177 14.643 14.203 14.584 /
\plot 14.203 14.584 14.226 14.522 /
\plot 14.226 14.522 14.249 14.461 /
\plot 14.249 14.461 14.273 14.398 /
\plot 14.273 14.398 14.294 14.332 /
\plot 14.294 14.332 14.317 14.264 /
\plot 14.317 14.264 14.338 14.196 /
\plot 14.338 14.196 14.359 14.127 /
\plot 14.359 14.127 14.381 14.055 /
\plot 14.381 14.055 14.402 13.981 /
\plot 14.402 13.981 14.423 13.904 /
\plot 14.423 13.904 14.442 13.826 /
\plot 14.442 13.826 14.461 13.748 /
\plot 14.461 13.748 14.482 13.667 /
\plot 14.482 13.667 14.499 13.587 /
\plot 14.499 13.587 14.518 13.504 /
\plot 14.518 13.504 14.535 13.420 /
\plot 14.535 13.420 14.554 13.335 /
\plot 14.554 13.335 14.571 13.250 /
\plot 14.571 13.250 14.586 13.164 /
\plot 14.586 13.164 14.603 13.079 /
\plot 14.603 13.079 14.618 12.992 /
\plot 14.618 12.992 14.633 12.903 /
\plot 14.633 12.903 14.647 12.816 /
\plot 14.647 12.816 14.662 12.730 /
\plot 14.662 12.730 14.675 12.641 /
\plot 14.675 12.641 14.690 12.552 /
\plot 14.690 12.552 14.702 12.463 /
\plot 14.702 12.463 14.715 12.374 /
\plot 14.715 12.374 14.728 12.285 /
\plot 14.728 12.285 14.738 12.196 /
\plot 14.738 12.196 14.751 12.105 /
\plot 14.751 12.105 14.764 12.012 /
\plot 14.764 12.012 14.774 11.936 /
\plot 14.774 11.936 14.783 11.858 /
\plot 14.783 11.858 14.793 11.779 /
\plot 14.793 11.779 14.804 11.697 /
\plot 14.804 11.697 14.815 11.614 /
\plot 14.815 11.614 14.825 11.529 /
\plot 14.825 11.529 14.834 11.441 /
\plot 14.834 11.441 14.846 11.350 /
\plot 14.846 11.350 14.857 11.256 /
\plot 14.857 11.256 14.867 11.157 /
\plot 14.867 11.157 14.878 11.055 /
\plot 14.878 11.055 14.891 10.950 /
\plot 14.891 10.950 14.903 10.839 /
\plot 14.903 10.839 14.916 10.723 /
\plot 14.916 10.723 14.929 10.604 /
\plot 14.929 10.604 14.942 10.480 /
\plot 14.942 10.480 14.956 10.348 /
\plot 14.956 10.348 14.971 10.215 /
\plot 14.971 10.215 14.986 10.075 /
\plot 14.986 10.075 15.001  9.931 /
\plot 15.001  9.931 15.016  9.783 /
\plot 15.016  9.783 15.033  9.631 /
\plot 15.033  9.631 15.050  9.478 /
\plot 15.050  9.478 15.064  9.322 /
\plot 15.064  9.322 15.081  9.165 /
\plot 15.081  9.165 15.098  9.011 /
\plot 15.098  9.011 15.113  8.856 /
\plot 15.113  8.856 15.128  8.706 /
\plot 15.128  8.706 15.143  8.560 /
\plot 15.143  8.560 15.157  8.422 /
\plot 15.157  8.422 15.172  8.293 /
\plot 15.172  8.293 15.183  8.172 /
\plot 15.183  8.172 15.196  8.064 /
\plot 15.196  8.064 15.204  7.967 /
\plot 15.204  7.967 15.212  7.882 /
\plot 15.212  7.882 15.221  7.811 /
\plot 15.221  7.811 15.227  7.751 /
\plot 15.227  7.751 15.232  7.705 /
\plot 15.232  7.705 15.236  7.671 /
\plot 15.236  7.671 15.238  7.648 /
\putrule from 15.238  7.648 to 15.238  7.631
\plot 15.238  7.631 15.240  7.624 /
\putrule from 15.240  7.624 to 15.240  7.620
}%
%
%
\linethickness= 0.500pt
\setplotsymbol ({\thinlinefont .})
{\color[rgb]{0,0,0}\plot 15.240 15.240 15.236 15.242 /
\plot 15.236 15.242 15.229 15.248 /
\plot 15.229 15.248 15.215 15.261 /
\plot 15.215 15.261 15.193 15.278 /
\plot 15.193 15.278 15.164 15.304 /
\plot 15.164 15.304 15.124 15.335 /
\plot 15.124 15.335 15.077 15.373 /
\plot 15.077 15.373 15.024 15.418 /
\plot 15.024 15.418 14.963 15.466 /
\plot 14.963 15.466 14.897 15.519 /
\plot 14.897 15.519 14.829 15.574 /
\plot 14.829 15.574 14.760 15.632 /
\plot 14.760 15.632 14.688 15.687 /
\plot 14.688 15.687 14.616 15.744 /
\plot 14.616 15.744 14.546 15.797 /
\plot 14.546 15.797 14.476 15.852 /
\plot 14.476 15.852 14.408 15.903 /
\plot 14.408 15.903 14.343 15.953 /
\plot 14.343 15.953 14.275 16.002 /
\plot 14.275 16.002 14.211 16.049 /
\plot 14.211 16.049 14.146 16.095 /
\plot 14.146 16.095 14.080 16.142 /
\plot 14.080 16.142 14.014 16.188 /
\plot 14.014 16.188 13.947 16.235 /
\plot 13.947 16.235 13.877 16.281 /
\plot 13.877 16.281 13.805 16.328 /
\plot 13.805 16.328 13.733 16.377 /
\plot 13.733 16.377 13.678 16.413 /
\plot 13.678 16.413 13.623 16.449 /
\plot 13.623 16.449 13.568 16.485 /
\plot 13.568 16.485 13.509 16.521 /
\plot 13.509 16.521 13.447 16.559 /
\plot 13.447 16.559 13.386 16.597 /
\plot 13.386 16.597 13.322 16.635 /
\plot 13.322 16.635 13.255 16.673 /
\plot 13.255 16.673 13.187 16.711 /
\plot 13.187 16.711 13.115 16.749 /
\plot 13.115 16.749 13.041 16.789 /
\plot 13.041 16.789 12.967 16.828 /
\plot 12.967 16.828 12.888 16.866 /
\plot 12.888 16.866 12.808 16.904 /
\plot 12.808 16.904 12.725 16.942 /
\plot 12.725 16.942 12.641 16.978 /
\plot 12.641 16.978 12.556 17.014 /
\plot 12.556 17.014 12.467 17.050 /
\plot 12.467 17.050 12.378 17.081 /
\plot 12.378 17.081 12.287 17.115 /
\plot 12.287 17.115 12.194 17.145 /
\plot 12.194 17.145 12.101 17.175 /
\plot 12.101 17.175 12.006 17.200 /
\plot 12.006 17.200 11.910 17.225 /
\plot 11.910 17.225 11.815 17.249 /
\plot 11.815 17.249 11.718 17.268 /
\plot 11.718 17.268 11.620 17.287 /
\plot 11.620 17.287 11.523 17.304 /
\plot 11.523 17.304 11.426 17.316 /
\plot 11.426 17.316 11.328 17.327 /
\plot 11.328 17.327 11.229 17.333 /
\plot 11.229 17.333 11.132 17.340 /
\plot 11.132 17.340 11.034 17.342 /
\putrule from 11.034 17.342 to 10.937 17.342
\plot 10.937 17.342 10.837 17.338 /
\plot 10.837 17.338 10.740 17.331 /
\plot 10.740 17.331 10.643 17.323 /
\plot 10.643 17.323 10.543 17.310 /
\plot 10.543 17.310 10.446 17.295 /
\plot 10.446 17.295 10.346 17.276 /
\plot 10.346 17.276 10.253 17.257 /
\plot 10.253 17.257 10.162 17.236 /
\plot 10.162 17.236 10.069 17.211 /
\plot 10.069 17.211  9.974 17.183 /
\plot  9.974 17.183  9.878 17.156 /
\plot  9.878 17.156  9.781 17.124 /
\plot  9.781 17.124  9.682 17.090 /
\plot  9.682 17.090  9.582 17.054 /
\plot  9.582 17.054  9.481 17.016 /
\plot  9.481 17.016  9.379 16.978 /
\plot  9.379 16.978  9.273 16.938 /
\plot  9.273 16.938  9.169 16.895 /
\plot  9.169 16.895  9.064 16.851 /
\plot  9.064 16.851  8.956 16.806 /
\plot  8.956 16.806  8.850 16.762 /
\plot  8.850 16.762  8.740 16.717 /
\plot  8.740 16.717  8.632 16.671 /
\plot  8.632 16.671  8.524 16.626 /
\plot  8.524 16.626  8.414 16.580 /
\plot  8.414 16.580  8.306 16.535 /
\plot  8.306 16.535  8.198 16.491 /
\plot  8.198 16.491  8.090 16.447 /
\plot  8.090 16.447  7.984 16.404 /
\plot  7.984 16.404  7.878 16.362 /
\plot  7.878 16.362  7.775 16.324 /
\plot  7.775 16.324  7.671 16.286 /
\plot  7.671 16.286  7.569 16.250 /
\plot  7.569 16.250  7.470 16.216 /
\plot  7.470 16.216  7.374 16.184 /
\plot  7.374 16.184  7.279 16.157 /
\plot  7.279 16.157  7.186 16.131 /
\plot  7.186 16.131  7.095 16.108 /
\plot  7.095 16.108  7.008 16.089 /
\plot  7.008 16.089  6.924 16.072 /
\plot  6.924 16.072  6.841 16.059 /
\plot  6.841 16.059  6.761 16.051 /
\plot  6.761 16.051  6.682 16.044 /
\plot  6.682 16.044  6.608 16.042 /
\plot  6.608 16.042  6.538 16.044 /
\plot  6.538 16.044  6.469 16.049 /
\plot  6.469 16.049  6.403 16.059 /
\plot  6.403 16.059  6.337 16.072 /
\plot  6.337 16.072  6.276 16.091 /
\plot  6.276 16.091  6.219 16.112 /
\plot  6.219 16.112  6.164 16.140 /
\plot  6.164 16.140  6.111 16.169 /
\plot  6.111 16.169  6.060 16.205 /
\plot  6.060 16.205  6.013 16.245 /
\plot  6.013 16.245  5.967 16.290 /
\plot  5.967 16.290  5.922 16.341 /
\plot  5.922 16.341  5.880 16.398 /
\plot  5.880 16.398  5.840 16.461 /
\plot  5.840 16.461  5.800 16.529 /
\plot  5.800 16.529  5.762 16.605 /
\plot  5.762 16.605  5.726 16.688 /
\plot  5.726 16.688  5.690 16.777 /
\plot  5.690 16.777  5.656 16.872 /
\plot  5.656 16.872  5.622 16.976 /
\plot  5.622 16.976  5.590 17.086 /
\plot  5.590 17.086  5.558 17.204 /
\plot  5.558 17.204  5.529 17.329 /
\plot  5.529 17.329  5.499 17.462 /
\plot  5.499 17.462  5.469 17.602 /
\plot  5.469 17.602  5.442 17.748 /
\plot  5.442 17.748  5.412 17.903 /
\plot  5.412 17.903  5.387 18.062 /
\plot  5.387 18.062  5.359 18.227 /
\plot  5.359 18.227  5.334 18.396 /
\plot  5.334 18.396  5.309 18.570 /
\plot  5.309 18.570  5.285 18.743 /
\plot  5.285 18.743  5.262 18.919 /
\plot  5.262 18.919  5.241 19.094 /
\plot  5.241 19.094  5.220 19.266 /
\plot  5.220 19.266  5.201 19.435 /
\plot  5.201 19.435  5.182 19.596 /
\plot  5.182 19.596  5.165 19.751 /
\plot  5.165 19.751  5.150 19.895 /
\plot  5.150 19.895  5.137 20.028 /
\plot  5.137 20.028  5.124 20.149 /
\plot  5.124 20.149  5.114 20.256 /
\plot  5.114 20.256  5.105 20.350 /
\plot  5.105 20.350  5.099 20.428 /
\plot  5.099 20.428  5.093 20.494 /
\plot  5.093 20.494  5.088 20.544 /
\plot  5.088 20.544  5.084 20.582 /
\plot  5.084 20.582  5.082 20.608 /
\plot  5.082 20.608  5.080 20.625 /
\putrule from  5.080 20.625 to  5.080 20.633
\putrule from  5.080 20.633 to  5.080 20.637
}%
%
%
\linethickness= 0.500pt
\setplotsymbol ({\thinlinefont .})
{\color[rgb]{0,0,0}\putrule from 15.240 17.780 to 15.240 17.784
\putrule from 15.240 17.784 to 15.240 17.795
\plot 15.240 17.795 15.238 17.814 /
\plot 15.238 17.814 15.236 17.844 /
\plot 15.236 17.844 15.234 17.884 /
\plot 15.234 17.884 15.229 17.937 /
\plot 15.229 17.937 15.225 18.002 /
\plot 15.225 18.002 15.219 18.078 /
\plot 15.219 18.078 15.215 18.165 /
\plot 15.215 18.165 15.208 18.260 /
\plot 15.208 18.260 15.200 18.362 /
\plot 15.200 18.362 15.193 18.468 /
\plot 15.193 18.468 15.185 18.576 /
\plot 15.185 18.576 15.179 18.684 /
\plot 15.179 18.684 15.170 18.790 /
\plot 15.170 18.790 15.164 18.895 /
\plot 15.164 18.895 15.157 18.995 /
\plot 15.157 18.995 15.149 19.092 /
\plot 15.149 19.092 15.143 19.185 /
\plot 15.143 19.185 15.136 19.274 /
\plot 15.136 19.274 15.130 19.357 /
\plot 15.130 19.357 15.126 19.435 /
\plot 15.126 19.435 15.119 19.509 /
\plot 15.119 19.509 15.115 19.581 /
\plot 15.115 19.581 15.109 19.649 /
\plot 15.109 19.649 15.105 19.713 /
\plot 15.105 19.713 15.100 19.774 /
\plot 15.100 19.774 15.094 19.833 /
\plot 15.094 19.833 15.090 19.890 /
\plot 15.090 19.890 15.085 19.947 /
\plot 15.085 19.947 15.081 20.003 /
\plot 15.081 20.003 15.075 20.068 /
\plot 15.075 20.068 15.071 20.134 /
\plot 15.071 20.134 15.064 20.199 /
\plot 15.064 20.199 15.058 20.265 /
\plot 15.058 20.265 15.052 20.331 /
\plot 15.052 20.331 15.045 20.398 /
\plot 15.045 20.398 15.039 20.466 /
\plot 15.039 20.466 15.033 20.538 /
\plot 15.033 20.538 15.026 20.612 /
\plot 15.026 20.612 15.018 20.690 /
\plot 15.018 20.690 15.009 20.771 /
\plot 15.009 20.771 15.001 20.853 /
\plot 15.001 20.853 14.992 20.940 /
\plot 14.992 20.940 14.982 21.027 /
\plot 14.982 21.027 14.973 21.114 /
\plot 14.973 21.114 14.965 21.201 /
\plot 14.965 21.201 14.956 21.281 /
\plot 14.956 21.281 14.948 21.355 /
\plot 14.948 21.355 14.942 21.421 /
\plot 14.942 21.421 14.935 21.476 /
\plot 14.935 21.476 14.931 21.520 /
\plot 14.931 21.520 14.927 21.552 /
\plot 14.927 21.552 14.925 21.573 /
\plot 14.925 21.573 14.922 21.586 /
\putrule from 14.922 21.586 to 14.922 21.590
}%
%
%
\linethickness= 0.500pt
\setplotsymbol ({\thinlinefont .})
{\color[rgb]{0,0,0}\putrule from  5.080 10.160 to  5.080 10.156
\plot  5.080 10.156  5.082 10.145 /
\plot  5.082 10.145  5.084 10.126 /
\plot  5.084 10.126  5.086 10.096 /
\plot  5.086 10.096  5.093 10.056 /
\plot  5.093 10.056  5.099 10.001 /
\plot  5.099 10.001  5.105  9.934 /
\plot  5.105  9.934  5.114  9.855 /
\plot  5.114  9.855  5.124  9.764 /
\plot  5.124  9.764  5.137  9.665 /
\plot  5.137  9.665  5.150  9.557 /
\plot  5.150  9.557  5.163  9.442 /
\plot  5.163  9.442  5.177  9.326 /
\plot  5.177  9.326  5.190  9.205 /
\plot  5.190  9.205  5.205  9.087 /
\plot  5.205  9.087  5.218  8.968 /
\plot  5.218  8.968  5.232  8.854 /
\plot  5.232  8.854  5.245  8.744 /
\plot  5.245  8.744  5.258  8.638 /
\plot  5.258  8.638  5.271  8.537 /
\plot  5.271  8.537  5.283  8.439 /
\plot  5.283  8.439  5.294  8.348 /
\plot  5.294  8.348  5.304  8.263 /
\plot  5.304  8.263  5.315  8.183 /
\plot  5.315  8.183  5.326  8.107 /
\plot  5.326  8.107  5.336  8.035 /
\plot  5.336  8.035  5.345  7.965 /
\plot  5.345  7.965  5.353  7.902 /
\plot  5.353  7.902  5.364  7.840 /
\plot  5.364  7.840  5.372  7.783 /
\plot  5.372  7.783  5.381  7.726 /
\plot  5.381  7.726  5.389  7.673 /
\plot  5.389  7.673  5.397  7.620 /
\plot  5.397  7.620  5.410  7.548 /
\plot  5.410  7.548  5.423  7.480 /
\plot  5.423  7.480  5.433  7.415 /
\plot  5.433  7.415  5.446  7.349 /
\plot  5.446  7.349  5.461  7.286 /
\plot  5.461  7.286  5.474  7.222 /
\plot  5.474  7.222  5.489  7.159 /
\plot  5.489  7.159  5.505  7.093 /
\plot  5.505  7.093  5.520  7.027 /
\plot  5.520  7.027  5.539  6.960 /
\plot  5.539  6.960  5.558  6.890 /
\plot  5.558  6.890  5.577  6.820 /
\plot  5.577  6.820  5.596  6.750 /
\plot  5.596  6.750  5.616  6.680 /
\plot  5.616  6.680  5.637  6.612 /
\plot  5.637  6.612  5.654  6.551 /
\plot  5.654  6.551  5.671  6.496 /
\plot  5.671  6.496  5.685  6.447 /
\plot  5.685  6.447  5.696  6.411 /
\plot  5.696  6.411  5.704  6.384 /
\plot  5.704  6.384  5.711  6.365 /
\plot  5.711  6.365  5.713  6.354 /
\plot  5.713  6.354  5.715  6.350 /
}%
%
%
\linethickness= 0.500pt
\setplotsymbol ({\thinlinefont .})
{\color[rgb]{0,0,0}\plot  8.255  6.350  8.259  6.352 /
\plot  8.259  6.352  8.270  6.354 /
\plot  8.270  6.354  8.289  6.361 /
\plot  8.289  6.361  8.316  6.369 /
\plot  8.316  6.369  8.357  6.380 /
\plot  8.357  6.380  8.407  6.394 /
\plot  8.407  6.394  8.467  6.411 /
\plot  8.467  6.411  8.537  6.433 /
\plot  8.537  6.433  8.613  6.456 /
\plot  8.613  6.456  8.695  6.479 /
\plot  8.695  6.479  8.780  6.505 /
\plot  8.780  6.505  8.867  6.528 /
\plot  8.867  6.528  8.951  6.553 /
\plot  8.951  6.553  9.036  6.576 /
\plot  9.036  6.576  9.116  6.600 /
\plot  9.116  6.600  9.195  6.623 /
\plot  9.195  6.623  9.269  6.642 /
\plot  9.269  6.642  9.339  6.663 /
\plot  9.339  6.663  9.406  6.680 /
\plot  9.406  6.680  9.470  6.699 /
\plot  9.470  6.699  9.533  6.714 /
\plot  9.533  6.714  9.591  6.731 /
\plot  9.591  6.731  9.650  6.746 /
\plot  9.650  6.746  9.705  6.761 /
\plot  9.705  6.761  9.760  6.773 /
\plot  9.760  6.773  9.815  6.786 /
\plot  9.815  6.786  9.870  6.799 /
\plot  9.870  6.799  9.923  6.814 /
\plot  9.923  6.814  9.978  6.826 /
\plot  9.978  6.826 10.035  6.839 /
\plot 10.035  6.839 10.092  6.852 /
\plot 10.092  6.852 10.149  6.862 /
\plot 10.149  6.862 10.209  6.875 /
\plot 10.209  6.875 10.268  6.888 /
\plot 10.268  6.888 10.329  6.900 /
\plot 10.329  6.900 10.393  6.911 /
\plot 10.393  6.911 10.456  6.924 /
\plot 10.456  6.924 10.520  6.934 /
\plot 10.520  6.934 10.585  6.947 /
\plot 10.585  6.947 10.651  6.957 /
\plot 10.651  6.957 10.717  6.966 /
\plot 10.717  6.966 10.780  6.977 /
\plot 10.780  6.977 10.846  6.985 /
\plot 10.846  6.985 10.911  6.993 /
\plot 10.911  6.993 10.975  7.002 /
\plot 10.975  7.002 11.038  7.008 /
\plot 11.038  7.008 11.102  7.015 /
\plot 11.102  7.015 11.163  7.021 /
\plot 11.163  7.021 11.223  7.025 /
\plot 11.223  7.025 11.282  7.029 /
\plot 11.282  7.029 11.339  7.032 /
\plot 11.339  7.032 11.396  7.034 /
\plot 11.396  7.034 11.453  7.036 /
\plot 11.453  7.036 11.508  7.038 /
\putrule from 11.508  7.038 to 11.563  7.038
\putrule from 11.563  7.038 to 11.620  7.038
\plot 11.620  7.038 11.680  7.036 /
\plot 11.680  7.036 11.739  7.034 /
\plot 11.739  7.034 11.798  7.032 /
\plot 11.798  7.032 11.860  7.027 /
\plot 11.860  7.027 11.921  7.023 /
\plot 11.921  7.023 11.982  7.017 /
\plot 11.982  7.017 12.046  7.010 /
\plot 12.046  7.010 12.107  7.004 /
\plot 12.107  7.004 12.171  6.996 /
\plot 12.171  6.996 12.234  6.987 /
\plot 12.234  6.987 12.298  6.977 /
\plot 12.298  6.977 12.361  6.966 /
\plot 12.361  6.966 12.423  6.955 /
\plot 12.423  6.955 12.484  6.945 /
\plot 12.484  6.945 12.543  6.932 /
\plot 12.543  6.932 12.601  6.919 /
\plot 12.601  6.919 12.658  6.907 /
\plot 12.658  6.907 12.713  6.894 /
\plot 12.713  6.894 12.766  6.881 /
\plot 12.766  6.881 12.816  6.869 /
\plot 12.816  6.869 12.865  6.854 /
\plot 12.865  6.854 12.912  6.841 /
\plot 12.912  6.841 12.958  6.828 /
\plot 12.958  6.828 13.001  6.814 /
\plot 13.001  6.814 13.045  6.799 /
\plot 13.045  6.799 13.094  6.784 /
\plot 13.094  6.784 13.140  6.765 /
\plot 13.140  6.765 13.187  6.748 /
\plot 13.187  6.748 13.236  6.729 /
\plot 13.236  6.729 13.282  6.710 /
\plot 13.282  6.710 13.331  6.687 /
\plot 13.331  6.687 13.379  6.665 /
\plot 13.379  6.665 13.432  6.640 /
\plot 13.432  6.640 13.487  6.612 /
\plot 13.487  6.612 13.542  6.583 /
\plot 13.542  6.583 13.602  6.553 /
\plot 13.602  6.553 13.661  6.521 /
\plot 13.661  6.521 13.718  6.490 /
\plot 13.718  6.490 13.775  6.460 /
\plot 13.775  6.460 13.826  6.430 /
\plot 13.826  6.430 13.873  6.405 /
\plot 13.873  6.405 13.909  6.384 /
\plot 13.909  6.384 13.936  6.369 /
\plot 13.936  6.369 13.955  6.358 /
\plot 13.955  6.358 13.966  6.352 /
\plot 13.966  6.352 13.970  6.350 /
}%
%
%
\linethickness= 0.500pt
\setplotsymbol ({\thinlinefont .})
{\color[rgb]{0,0,0}\plot  9.525  6.350  9.531  6.352 /
\plot  9.531  6.352  9.544  6.354 /
\plot  9.544  6.354  9.565  6.361 /
\plot  9.565  6.361  9.601  6.367 /
\plot  9.601  6.367  9.648  6.378 /
\plot  9.648  6.378  9.705  6.392 /
\plot  9.705  6.392  9.773  6.407 /
\plot  9.773  6.407  9.847  6.422 /
\plot  9.847  6.422  9.925  6.441 /
\plot  9.925  6.441 10.008  6.458 /
\plot 10.008  6.458 10.088  6.475 /
\plot 10.088  6.475 10.166  6.492 /
\plot 10.166  6.492 10.243  6.507 /
\plot 10.243  6.507 10.317  6.519 /
\plot 10.317  6.519 10.384  6.532 /
\plot 10.384  6.532 10.450  6.543 /
\plot 10.450  6.543 10.513  6.553 /
\plot 10.513  6.553 10.573  6.562 /
\plot 10.573  6.562 10.630  6.568 /
\plot 10.630  6.568 10.685  6.574 /
\plot 10.685  6.574 10.740  6.581 /
\plot 10.740  6.581 10.793  6.585 /
\plot 10.793  6.585 10.848  6.587 /
\plot 10.848  6.587 10.903  6.591 /
\plot 10.903  6.591 10.958  6.593 /
\putrule from 10.958  6.593 to 11.015  6.593
\plot 11.015  6.593 11.072  6.596 /
\putrule from 11.072  6.596 to 11.134  6.596
\plot 11.134  6.596 11.193  6.593 /
\putrule from 11.193  6.593 to 11.254  6.593
\plot 11.254  6.593 11.318  6.591 /
\plot 11.318  6.591 11.381  6.587 /
\plot 11.381  6.587 11.445  6.585 /
\plot 11.445  6.585 11.510  6.581 /
\plot 11.510  6.581 11.574  6.576 /
\plot 11.574  6.576 11.637  6.572 /
\plot 11.637  6.572 11.701  6.568 /
\plot 11.701  6.568 11.762  6.562 /
\plot 11.762  6.562 11.822  6.557 /
\plot 11.822  6.557 11.879  6.551 /
\plot 11.879  6.551 11.934  6.545 /
\plot 11.934  6.545 11.989  6.538 /
\plot 11.989  6.538 12.040  6.532 /
\plot 12.040  6.532 12.088  6.528 /
\plot 12.088  6.528 12.135  6.521 /
\plot 12.135  6.521 12.181  6.515 /
\plot 12.181  6.515 12.224  6.509 /
\plot 12.224  6.509 12.279  6.500 /
\plot 12.279  6.500 12.332  6.492 /
\plot 12.332  6.492 12.385  6.483 /
\plot 12.385  6.483 12.435  6.475 /
\plot 12.435  6.475 12.488  6.466 /
\plot 12.488  6.466 12.541  6.456 /
\plot 12.541  6.456 12.596  6.443 /
\plot 12.596  6.443 12.653  6.433 /
\plot 12.653  6.433 12.713  6.420 /
\plot 12.713  6.420 12.774  6.405 /
\plot 12.774  6.405 12.831  6.392 /
\plot 12.831  6.392 12.886  6.380 /
\plot 12.886  6.380 12.933  6.369 /
\plot 12.933  6.369 12.971  6.361 /
\plot 12.971  6.361 12.996  6.354 /
\plot 12.996  6.354 13.011  6.352 /
\plot 13.011  6.352 13.018  6.350 /
}%
%
%
\linethickness=0pt
\setplotsymbol ({\thinlinefont \ })
{\color[rgb]{0,0,0}\plot  5.397 17.939  5.408 17.867 /
\plot  5.408 17.867  5.421 17.795 /
\plot  5.421 17.795  5.431 17.725 /
\plot  5.431 17.725  5.444 17.653 /
\plot  5.444 17.653  5.455 17.583 /
\plot  5.455 17.583  5.467 17.513 /
\plot  5.467 17.513  5.480 17.443 /
\plot  5.480 17.443  5.493 17.374 /
\plot  5.493 17.374  5.508 17.306 /
\plot  5.508 17.306  5.520 17.238 /
\plot  5.520 17.238  5.535 17.170 /
\plot  5.535 17.170  5.550 17.105 /
\plot  5.550 17.105  5.563 17.041 /
\plot  5.563 17.041  5.580 16.980 /
\plot  5.580 16.980  5.594 16.919 /
\plot  5.594 16.919  5.609 16.861 /
\plot  5.609 16.861  5.624 16.806 /
\plot  5.624 16.806  5.641 16.753 /
\plot  5.641 16.753  5.656 16.705 /
\plot  5.656 16.705  5.671 16.656 /
\plot  5.671 16.656  5.687 16.612 /
\plot  5.687 16.612  5.702 16.571 /
\plot  5.702 16.571  5.719 16.533 /
\plot  5.719 16.533  5.734 16.497 /
\plot  5.734 16.497  5.751 16.463 /
\plot  5.751 16.463  5.768 16.430 /
\plot  5.768 16.430  5.791 16.391 /
\plot  5.791 16.391  5.812 16.355 /
\plot  5.812 16.355  5.838 16.324 /
\plot  5.838 16.324  5.863 16.292 /
\plot  5.863 16.292  5.891 16.262 /
\plot  5.891 16.262  5.920 16.237 /
\plot  5.920 16.237  5.950 16.212 /
\plot  5.950 16.212  5.982 16.188 /
\plot  5.982 16.188  6.016 16.169 /
\plot  6.016 16.169  6.052 16.150 /
\plot  6.052 16.150  6.088 16.133 /
\plot  6.088 16.133  6.124 16.118 /
\plot  6.124 16.118  6.159 16.108 /
\plot  6.159 16.108  6.198 16.097 /
\plot  6.198 16.097  6.236 16.087 /
\plot  6.236 16.087  6.274 16.080 /
\plot  6.274 16.080  6.312 16.074 /
\plot  6.312 16.074  6.350 16.068 /
\plot  6.350 16.068  6.390 16.063 /
\plot  6.390 16.063  6.430 16.059 /
\plot  6.430 16.059  6.462 16.057 /
\plot  6.462 16.057  6.498 16.055 /
\plot  6.498 16.055  6.534 16.053 /
\plot  6.534 16.053  6.572 16.051 /
\plot  6.572 16.051  6.612 16.049 /
\putrule from  6.612 16.049 to  6.655 16.049
\putrule from  6.655 16.049 to  6.699 16.049
\plot  6.699 16.049  6.746 16.051 /
\putrule from  6.746 16.051 to  6.794 16.051
\plot  6.794 16.051  6.845 16.055 /
\plot  6.845 16.055  6.898 16.059 /
\plot  6.898 16.059  6.953 16.063 /
\plot  6.953 16.063  7.010 16.070 /
\plot  7.010 16.070  7.070 16.076 /
\plot  7.070 16.076  7.129 16.085 /
\plot  7.129 16.085  7.192 16.095 /
\plot  7.192 16.095  7.254 16.106 /
\plot  7.254 16.106  7.317 16.118 /
\plot  7.317 16.118  7.383 16.131 /
\plot  7.383 16.131  7.449 16.146 /
\plot  7.449 16.146  7.516 16.163 /
\plot  7.516 16.163  7.584 16.180 /
\plot  7.584 16.180  7.654 16.199 /
\plot  7.654 16.199  7.726 16.218 /
\plot  7.726 16.218  7.781 16.235 /
\plot  7.781 16.235  7.838 16.252 /
\plot  7.838 16.252  7.897 16.271 /
\plot  7.897 16.271  7.957 16.290 /
\plot  7.957 16.290  8.020 16.311 /
\plot  8.020 16.311  8.084 16.332 /
\plot  8.084 16.332  8.149 16.355 /
\plot  8.149 16.355  8.217 16.379 /
\plot  8.217 16.379  8.287 16.404 /
\plot  8.287 16.404  8.359 16.430 /
\plot  8.359 16.430  8.431 16.455 /
\plot  8.431 16.455  8.505 16.480 /
\plot  8.505 16.480  8.579 16.508 /
\plot  8.579 16.508  8.655 16.538 /
\plot  8.655 16.538  8.731 16.565 /
\plot  8.731 16.565  8.807 16.593 /
\plot  8.807 16.593  8.884 16.622 /
\plot  8.884 16.622  8.960 16.650 /
\plot  8.960 16.650  9.034 16.677 /
\plot  9.034 16.677  9.108 16.707 /
\plot  9.108 16.707  9.182 16.734 /
\plot  9.182 16.734  9.254 16.760 /
\plot  9.254 16.760  9.324 16.787 /
\plot  9.324 16.787  9.392 16.813 /
\plot  9.392 16.813  9.459 16.838 /
\plot  9.459 16.838  9.523 16.861 /
\plot  9.523 16.861  9.586 16.885 /
\plot  9.586 16.885  9.646 16.906 /
\plot  9.646 16.906  9.705 16.927 /
\plot  9.705 16.927  9.762 16.948 /
\plot  9.762 16.948  9.815 16.967 /
\plot  9.815 16.967  9.870 16.986 /
\plot  9.870 16.986  9.938 17.010 /
\plot  9.938 17.010 10.003 17.033 /
\plot 10.003 17.033 10.069 17.054 /
\plot 10.069 17.054 10.132 17.073 /
\plot 10.132 17.073 10.194 17.092 /
\plot 10.194 17.092 10.255 17.109 /
\plot 10.255 17.109 10.317 17.126 /
\plot 10.317 17.126 10.376 17.141 /
\plot 10.376 17.141 10.435 17.156 /
\plot 10.435 17.156 10.494 17.168 /
\plot 10.494 17.168 10.552 17.181 /
\plot 10.552 17.181 10.607 17.192 /
\plot 10.607 17.192 10.662 17.200 /
\plot 10.662 17.200 10.715 17.209 /
\plot 10.715 17.209 10.767 17.215 /
\plot 10.767 17.215 10.818 17.221 /
\plot 10.818 17.221 10.869 17.225 /
\plot 10.869 17.225 10.918 17.228 /
\plot 10.918 17.228 10.964 17.230 /
\putrule from 10.964 17.230 to 11.011 17.230
\putrule from 11.011 17.230 to 11.057 17.230
\putrule from 11.057 17.230 to 11.102 17.230
\plot 11.102 17.230 11.146 17.228 /
\plot 11.146 17.228 11.193 17.223 /
\plot 11.193 17.223 11.235 17.221 /
\plot 11.235 17.221 11.278 17.217 /
\plot 11.278 17.217 11.322 17.213 /
\plot 11.322 17.213 11.369 17.209 /
\plot 11.369 17.209 11.415 17.202 /
\plot 11.415 17.202 11.464 17.194 /
\plot 11.464 17.194 11.515 17.185 /
\plot 11.515 17.185 11.568 17.177 /
\plot 11.568 17.177 11.620 17.166 /
\plot 11.620 17.166 11.676 17.156 /
\plot 11.676 17.156 11.733 17.143 /
\plot 11.733 17.143 11.790 17.130 /
\plot 11.790 17.130 11.849 17.115 /
\plot 11.849 17.115 11.908 17.101 /
\plot 11.908 17.101 11.968 17.084 /
\plot 11.968 17.084 12.029 17.067 /
\plot 12.029 17.067 12.088 17.050 /
\plot 12.088 17.050 12.150 17.031 /
\plot 12.150 17.031 12.209 17.010 /
\plot 12.209 17.010 12.268 16.990 /
\plot 12.268 16.990 12.327 16.969 /
\plot 12.327 16.969 12.387 16.948 /
\plot 12.387 16.948 12.446 16.925 /
\plot 12.446 16.925 12.503 16.902 /
\plot 12.503 16.902 12.562 16.878 /
\plot 12.562 16.878 12.622 16.853 /
\plot 12.622 16.853 12.675 16.830 /
\plot 12.675 16.830 12.732 16.804 /
\plot 12.732 16.804 12.789 16.779 /
\plot 12.789 16.779 12.846 16.751 /
\plot 12.846 16.751 12.905 16.724 /
\plot 12.905 16.724 12.965 16.694 /
\plot 12.965 16.694 13.026 16.662 /
\plot 13.026 16.662 13.089 16.631 /
\plot 13.089 16.631 13.153 16.599 /
\plot 13.153 16.599 13.216 16.565 /
\plot 13.216 16.565 13.282 16.529 /
\plot 13.282 16.529 13.348 16.493 /
\plot 13.348 16.493 13.413 16.457 /
\plot 13.413 16.457 13.479 16.419 /
\plot 13.479 16.419 13.545 16.383 /
\plot 13.545 16.383 13.608 16.345 /
\plot 13.608 16.345 13.672 16.307 /
\plot 13.672 16.307 13.735 16.271 /
\plot 13.735 16.271 13.794 16.233 /
\plot 13.794 16.233 13.854 16.197 /
\plot 13.854 16.197 13.911 16.161 /
\plot 13.911 16.161 13.968 16.125 /
\plot 13.968 16.125 14.021 16.091 /
\plot 14.021 16.091 14.072 16.057 /
\plot 14.072 16.057 14.122 16.023 /
\plot 14.122 16.023 14.169 15.991 /
\plot 14.169 15.991 14.216 15.960 /
\plot 14.216 15.960 14.262 15.928 /
\plot 14.262 15.928 14.302 15.898 /
\plot 14.302 15.898 14.343 15.871 /
\plot 14.343 15.871 14.381 15.841 /
\plot 14.381 15.841 14.419 15.811 /
\plot 14.419 15.811 14.457 15.780 /
\plot 14.457 15.780 14.493 15.748 /
\plot 14.493 15.748 14.531 15.714 /
\plot 14.531 15.714 14.567 15.678 /
\plot 14.567 15.678 14.601 15.642 /
\plot 14.601 15.642 14.637 15.602 /
\plot 14.637 15.602 14.671 15.562 /
\plot 14.671 15.562 14.702 15.517 /
\plot 14.702 15.517 14.734 15.473 /
\plot 14.734 15.473 14.766 15.424 /
\plot 14.766 15.424 14.796 15.371 /
\plot 14.796 15.371 14.823 15.318 /
\plot 14.823 15.318 14.851 15.261 /
\plot 14.851 15.261 14.878 15.202 /
\plot 14.878 15.202 14.901 15.138 /
\plot 14.901 15.138 14.925 15.075 /
\plot 14.925 15.075 14.946 15.007 /
\plot 14.946 15.007 14.967 14.935 /
\plot 14.967 14.935 14.986 14.863 /
\plot 14.986 14.863 15.003 14.787 /
\plot 15.003 14.787 15.020 14.707 /
\plot 15.020 14.707 15.035 14.626 /
\plot 15.035 14.626 15.047 14.539 /
\plot 15.047 14.539 15.060 14.450 /
\plot 15.060 14.450 15.071 14.357 /
\plot 15.071 14.357 15.081 14.260 /
\plot 15.081 14.260 15.088 14.186 /
\plot 15.088 14.186 15.094 14.105 /
\plot 15.094 14.105 15.100 14.025 /
\plot 15.100 14.025 15.107 13.940 /
\plot 15.107 13.940 15.111 13.854 /
\plot 15.111 13.854 15.115 13.763 /
\plot 15.115 13.763 15.119 13.672 /
\plot 15.119 13.672 15.124 13.576 /
\plot 15.124 13.576 15.128 13.477 /
\plot 15.128 13.477 15.130 13.377 /
\plot 15.130 13.377 15.132 13.276 /
\plot 15.132 13.276 15.134 13.172 /
\plot 15.134 13.172 15.136 13.066 /
\plot 15.136 13.066 15.138 12.960 /
\putrule from 15.138 12.960 to 15.138 12.852
\plot 15.138 12.852 15.141 12.744 /
\putrule from 15.141 12.744 to 15.141 12.636
\putrule from 15.141 12.636 to 15.141 12.529
\plot 15.141 12.529 15.138 12.423 /
\putrule from 15.138 12.423 to 15.138 12.315
\putrule from 15.138 12.315 to 15.138 12.211
\plot 15.138 12.211 15.136 12.107 /
\plot 15.136 12.107 15.134 12.006 /
\plot 15.134 12.006 15.132 11.906 /
\plot 15.132 11.906 15.130 11.811 /
\plot 15.130 11.811 15.128 11.716 /
\plot 15.128 11.716 15.126 11.627 /
\plot 15.126 11.627 15.121 11.540 /
\plot 15.121 11.540 15.119 11.458 /
\plot 15.119 11.458 15.117 11.379 /
\plot 15.117 11.379 15.113 11.305 /
\plot 15.113 11.305 15.109 11.235 /
\plot 15.109 11.235 15.107 11.172 /
\plot 15.107 11.172 15.102 11.110 /
\plot 15.102 11.110 15.100 11.055 /
\plot 15.100 11.055 15.096 11.005 /
\plot 15.096 11.005 15.092 10.958 /
\plot 15.092 10.958 15.088 10.916 /
\plot 15.088 10.916 15.085 10.880 /
\plot 15.085 10.880 15.081 10.848 /
\plot 15.081 10.848 15.075 10.810 /
\plot 15.075 10.810 15.071 10.780 /
\plot 15.071 10.780 15.064 10.763 /
\plot 15.064 10.763 15.058 10.759 /
\plot 15.058 10.759 15.052 10.763 /
\plot 15.052 10.763 15.045 10.778 /
\plot 15.045 10.778 15.039 10.801 /
\plot 15.039 10.801 15.033 10.837 /
\plot 15.033 10.837 15.026 10.880 /
\plot 15.026 10.880 15.018 10.933 /
\plot 15.018 10.933 15.011 10.994 /
\plot 15.011 10.994 15.003 11.064 /
\plot 15.003 11.064 14.994 11.140 /
\plot 14.994 11.140 14.986 11.223 /
\plot 14.986 11.223 14.978 11.311 /
\plot 14.978 11.311 14.969 11.405 /
\plot 14.969 11.405 14.961 11.502 /
\plot 14.961 11.502 14.950 11.599 /
\plot 14.950 11.599 14.942 11.701 /
\plot 14.942 11.701 14.931 11.803 /
\plot 14.931 11.803 14.922 11.906 /
\plot 14.922 11.906 14.912 12.008 /
\plot 14.912 12.008 14.901 12.109 /
\plot 14.901 12.109 14.891 12.211 /
\plot 14.891 12.211 14.880 12.311 /
\plot 14.880 12.311 14.870 12.408 /
\plot 14.870 12.408 14.861 12.488 /
\plot 14.861 12.488 14.851 12.567 /
\plot 14.851 12.567 14.840 12.645 /
\plot 14.840 12.645 14.829 12.725 /
\plot 14.829 12.725 14.819 12.804 /
\plot 14.819 12.804 14.808 12.884 /
\plot 14.808 12.884 14.796 12.962 /
\plot 14.796 12.962 14.783 13.043 /
\plot 14.783 13.043 14.770 13.123 /
\plot 14.770 13.123 14.757 13.204 /
\plot 14.757 13.204 14.743 13.284 /
\plot 14.743 13.284 14.728 13.367 /
\plot 14.728 13.367 14.713 13.447 /
\plot 14.713 13.447 14.696 13.526 /
\plot 14.696 13.526 14.679 13.606 /
\plot 14.679 13.606 14.662 13.684 /
\plot 14.662 13.684 14.645 13.763 /
\plot 14.645 13.763 14.626 13.839 /
\plot 14.626 13.839 14.607 13.913 /
\plot 14.607 13.913 14.590 13.985 /
\plot 14.590 13.985 14.569 14.057 /
\plot 14.569 14.057 14.550 14.127 /
\plot 14.550 14.127 14.531 14.194 /
\plot 14.531 14.194 14.510 14.260 /
\plot 14.510 14.260 14.491 14.321 /
\plot 14.491 14.321 14.470 14.383 /
\plot 14.470 14.383 14.448 14.442 /
\plot 14.448 14.442 14.427 14.499 /
\plot 14.427 14.499 14.406 14.554 /
\plot 14.406 14.554 14.385 14.609 /
\plot 14.385 14.609 14.364 14.660 /
\plot 14.364 14.660 14.340 14.711 /
\plot 14.340 14.711 14.313 14.772 /
\plot 14.313 14.772 14.283 14.831 /
\plot 14.283 14.831 14.252 14.891 /
\plot 14.252 14.891 14.220 14.948 /
\plot 14.220 14.948 14.186 15.005 /
\plot 14.186 15.005 14.152 15.062 /
\plot 14.152 15.062 14.116 15.117 /
\plot 14.116 15.117 14.080 15.174 /
\plot 14.080 15.174 14.040 15.229 /
\plot 14.040 15.229 14.002 15.282 /
\plot 14.002 15.282 13.959 15.337 /
\plot 13.959 15.337 13.919 15.390 /
\plot 13.919 15.390 13.877 15.441 /
\plot 13.877 15.441 13.832 15.492 /
\plot 13.832 15.492 13.790 15.541 /
\plot 13.790 15.541 13.748 15.589 /
\plot 13.748 15.589 13.703 15.634 /
\plot 13.703 15.634 13.661 15.678 /
\plot 13.661 15.678 13.619 15.723 /
\plot 13.619 15.723 13.576 15.763 /
\plot 13.576 15.763 13.534 15.803 /
\plot 13.534 15.803 13.494 15.841 /
\plot 13.494 15.841 13.454 15.877 /
\plot 13.454 15.877 13.413 15.913 /
\plot 13.413 15.913 13.375 15.947 /
\plot 13.375 15.947 13.335 15.981 /
\plot 13.335 15.981 13.293 16.017 /
\plot 13.293 16.017 13.250 16.053 /
\plot 13.250 16.053 13.206 16.089 /
\plot 13.206 16.089 13.161 16.125 /
\plot 13.161 16.125 13.117 16.161 /
\plot 13.117 16.161 13.070 16.197 /
\plot 13.070 16.197 13.022 16.233 /
\plot 13.022 16.233 12.975 16.269 /
\plot 12.975 16.269 12.924 16.305 /
\plot 12.924 16.305 12.874 16.339 /
\plot 12.874 16.339 12.823 16.375 /
\plot 12.823 16.375 12.772 16.408 /
\plot 12.772 16.408 12.721 16.442 /
\plot 12.721 16.442 12.668 16.474 /
\plot 12.668 16.474 12.617 16.506 /
\plot 12.617 16.506 12.565 16.535 /
\plot 12.565 16.535 12.514 16.563 /
\plot 12.514 16.563 12.463 16.590 /
\plot 12.463 16.590 12.414 16.616 /
\plot 12.414 16.616 12.366 16.639 /
\plot 12.366 16.639 12.315 16.662 /
\plot 12.315 16.662 12.268 16.684 /
\plot 12.268 16.684 12.220 16.703 /
\plot 12.220 16.703 12.171 16.722 /
\plot 12.171 16.722 12.122 16.739 /
\plot 12.122 16.739 12.071 16.756 /
\plot 12.071 16.756 12.021 16.772 /
\plot 12.021 16.772 11.970 16.787 /
\plot 11.970 16.787 11.915 16.802 /
\plot 11.915 16.802 11.860 16.815 /
\plot 11.860 16.815 11.805 16.828 /
\plot 11.805 16.828 11.745 16.840 /
\plot 11.745 16.840 11.686 16.851 /
\plot 11.686 16.851 11.625 16.861 /
\plot 11.625 16.861 11.563 16.870 /
\plot 11.563 16.870 11.502 16.878 /
\plot 11.502 16.878 11.441 16.885 /
\plot 11.441 16.885 11.377 16.889 /
\plot 11.377 16.889 11.316 16.893 /
\plot 11.316 16.893 11.254 16.897 /
\putrule from 11.254 16.897 to 11.193 16.897
\plot 11.193 16.897 11.134 16.899 /
\putrule from 11.134 16.899 to 11.074 16.899
\plot 11.074 16.899 11.017 16.897 /
\plot 11.017 16.897 10.960 16.895 /
\plot 10.960 16.895 10.905 16.891 /
\plot 10.905 16.891 10.850 16.887 /
\plot 10.850 16.887 10.795 16.880 /
\plot 10.795 16.880 10.744 16.874 /
\plot 10.744 16.874 10.693 16.868 /
\plot 10.693 16.868 10.643 16.859 /
\plot 10.643 16.859 10.590 16.851 /
\plot 10.590 16.851 10.537 16.840 /
\plot 10.537 16.840 10.484 16.830 /
\plot 10.484 16.830 10.429 16.819 /
\plot 10.429 16.819 10.372 16.806 /
\plot 10.372 16.806 10.315 16.794 /
\plot 10.315 16.794 10.257 16.779 /
\plot 10.257 16.779 10.198 16.762 /
\plot 10.198 16.762 10.139 16.747 /
\plot 10.139 16.747 10.080 16.730 /
\plot 10.080 16.730 10.022 16.713 /
\plot 10.022 16.713  9.963 16.694 /
\plot  9.963 16.694  9.904 16.675 /
\plot  9.904 16.675  9.847 16.658 /
\plot  9.847 16.658  9.790 16.639 /
\plot  9.790 16.639  9.732 16.618 /
\plot  9.732 16.618  9.677 16.599 /
\plot  9.677 16.599  9.624 16.580 /
\plot  9.624 16.580  9.572 16.561 /
\plot  9.572 16.561  9.519 16.542 /
\plot  9.519 16.542  9.468 16.523 /
\plot  9.468 16.523  9.417 16.504 /
\plot  9.417 16.504  9.366 16.482 /
\plot  9.366 16.482  9.315 16.463 /
\plot  9.315 16.463  9.265 16.444 /
\plot  9.265 16.444  9.214 16.423 /
\plot  9.214 16.423  9.161 16.402 /
\plot  9.161 16.402  9.108 16.381 /
\plot  9.108 16.381  9.053 16.358 /
\plot  9.053 16.358  8.998 16.334 /
\plot  8.998 16.334  8.943 16.309 /
\plot  8.943 16.309  8.886 16.284 /
\plot  8.886 16.284  8.827 16.258 /
\plot  8.827 16.258  8.767 16.231 /
\plot  8.767 16.231  8.708 16.203 /
\plot  8.708 16.203  8.649 16.173 /
\plot  8.649 16.173  8.587 16.146 /
\plot  8.587 16.146  8.528 16.116 /
\plot  8.528 16.116  8.469 16.085 /
\plot  8.469 16.085  8.410 16.055 /
\plot  8.410 16.055  8.350 16.023 /
\plot  8.350 16.023  8.293 15.994 /
\plot  8.293 15.994  8.236 15.962 /
\plot  8.236 15.962  8.181 15.930 /
\plot  8.181 15.930  8.126 15.898 /
\plot  8.126 15.898  8.071 15.867 /
\plot  8.071 15.867  8.018 15.835 /
\plot  8.018 15.835  7.965 15.803 /
\plot  7.965 15.803  7.912 15.769 /
\plot  7.912 15.769  7.861 15.737 /
\plot  7.861 15.737  7.813 15.706 /
\plot  7.813 15.706  7.762 15.672 /
\plot  7.762 15.672  7.709 15.638 /
\plot  7.709 15.638  7.658 15.602 /
\plot  7.658 15.602  7.605 15.564 /
\plot  7.605 15.564  7.550 15.526 /
\plot  7.550 15.526  7.495 15.486 /
\plot  7.495 15.486  7.440 15.445 /
\plot  7.440 15.445  7.383 15.405 /
\plot  7.383 15.405  7.326 15.361 /
\plot  7.326 15.361  7.269 15.318 /
\plot  7.269 15.318  7.211 15.274 /
\plot  7.211 15.274  7.154 15.229 /
\plot  7.154 15.229  7.097 15.185 /
\plot  7.097 15.185  7.042 15.138 /
\plot  7.042 15.138  6.985 15.094 /
\plot  6.985 15.094  6.932 15.050 /
\plot  6.932 15.050  6.877 15.005 /
\plot  6.877 15.005  6.826 14.961 /
\plot  6.826 14.961  6.775 14.916 /
\plot  6.775 14.916  6.725 14.874 /
\plot  6.725 14.874  6.676 14.831 /
\plot  6.676 14.831  6.629 14.791 /
\plot  6.629 14.791  6.585 14.751 /
\plot  6.585 14.751  6.541 14.711 /
\plot  6.541 14.711  6.498 14.671 /
\plot  6.498 14.671  6.456 14.630 /
\plot  6.456 14.630  6.411 14.590 /
\plot  6.411 14.590  6.367 14.548 /
\plot  6.367 14.548  6.325 14.506 /
\plot  6.325 14.506  6.280 14.463 /
\plot  6.280 14.463  6.238 14.419 /
\plot  6.238 14.419  6.195 14.376 /
\plot  6.195 14.376  6.151 14.332 /
\plot  6.151 14.332  6.109 14.287 /
\plot  6.109 14.287  6.066 14.243 /
\plot  6.066 14.243  6.026 14.196 /
\plot  6.026 14.196  5.984 14.152 /
\plot  5.984 14.152  5.944 14.105 /
\plot  5.944 14.105  5.903 14.061 /
\plot  5.903 14.061  5.863 14.014 /
\plot  5.863 14.014  5.827 13.972 /
\plot  5.827 13.972  5.789 13.928 /
\plot  5.789 13.928  5.755 13.885 /
\plot  5.755 13.885  5.721 13.843 /
\plot  5.721 13.843  5.690 13.803 /
\plot  5.690 13.803  5.658 13.763 /
\plot  5.658 13.763  5.628 13.724 /
\plot  5.628 13.724  5.601 13.686 /
\plot  5.601 13.686  5.575 13.650 /
\plot  5.575 13.650  5.550 13.614 /
\plot  5.550 13.614  5.527 13.581 /
\plot  5.527 13.581  5.503 13.547 /
\plot  5.503 13.547  5.491 13.528 /
\plot  5.491 13.528  5.478 13.509 /
\plot  5.478 13.509  5.467 13.490 /
\plot  5.467 13.490  5.455 13.473 /
\plot  5.455 13.473  5.442 13.454 /
\plot  5.442 13.454  5.431 13.437 /
\plot  5.431 13.437  5.421 13.420 /
\plot  5.421 13.420  5.408 13.403 /
\plot  5.408 13.403  5.397 13.388 /
\plot  5.397 13.388  5.387 13.373 /
\plot  5.387 13.373  5.376 13.358 /
\plot  5.376 13.358  5.366 13.346 /
\plot  5.366 13.346  5.355 13.333 /
\plot  5.355 13.333  5.345 13.322 /
\plot  5.345 13.322  5.334 13.314 /
\plot  5.334 13.314  5.326 13.305 /
\plot  5.326 13.305  5.315 13.299 /
\plot  5.315 13.299  5.306 13.293 /
\plot  5.306 13.293  5.296 13.291 /
\plot  5.296 13.291  5.287 13.288 /
\putrule from  5.287 13.288 to  5.279 13.288
\plot  5.279 13.288  5.268 13.291 /
\plot  5.268 13.291  5.260 13.295 /
\plot  5.260 13.295  5.251 13.301 /
\plot  5.251 13.301  5.245 13.310 /
\plot  5.245 13.310  5.237 13.320 /
\plot  5.237 13.320  5.228 13.335 /
\plot  5.228 13.335  5.222 13.350 /
\plot  5.222 13.350  5.215 13.367 /
\plot  5.215 13.367  5.209 13.388 /
\plot  5.209 13.388  5.203 13.411 /
\plot  5.203 13.411  5.196 13.437 /
\plot  5.196 13.437  5.190 13.464 /
\plot  5.190 13.464  5.184 13.496 /
\plot  5.184 13.496  5.177 13.528 /
\plot  5.177 13.528  5.173 13.564 /
\plot  5.173 13.564  5.169 13.604 /
\plot  5.169 13.604  5.163 13.644 /
\plot  5.163 13.644  5.158 13.688 /
\plot  5.158 13.688  5.154 13.735 /
\plot  5.154 13.735  5.150 13.784 /
\plot  5.150 13.784  5.148 13.837 /
\plot  5.148 13.837  5.143 13.892 /
\plot  5.143 13.892  5.139 13.949 /
\plot  5.139 13.949  5.137 14.010 /
\plot  5.137 14.010  5.133 14.076 /
\plot  5.133 14.076  5.131 14.144 /
\plot  5.131 14.144  5.127 14.216 /
\plot  5.127 14.216  5.124 14.292 /
\plot  5.124 14.292  5.120 14.370 /
\plot  5.120 14.370  5.118 14.453 /
\plot  5.118 14.453  5.116 14.537 /
\plot  5.116 14.537  5.114 14.628 /
\plot  5.114 14.628  5.112 14.721 /
\plot  5.112 14.721  5.110 14.817 /
\plot  5.110 14.817  5.108 14.918 /
\plot  5.108 14.918  5.105 15.022 /
\plot  5.105 15.022  5.103 15.128 /
\plot  5.103 15.128  5.101 15.238 /
\plot  5.101 15.238  5.099 15.350 /
\plot  5.099 15.350  5.097 15.464 /
\putrule from  5.097 15.464 to  5.097 15.583
\plot  5.097 15.583  5.095 15.701 /
\plot  5.095 15.701  5.093 15.824 /
\putrule from  5.093 15.824 to  5.093 15.947
\plot  5.093 15.947  5.091 16.072 /
\plot  5.091 16.072  5.088 16.199 /
\putrule from  5.088 16.199 to  5.088 16.326
\plot  5.088 16.326  5.086 16.453 /
\putrule from  5.086 16.453 to  5.086 16.580
\putrule from  5.086 16.580 to  5.086 16.709
\plot  5.086 16.709  5.084 16.836 /
\putrule from  5.084 16.836 to  5.084 16.961
\putrule from  5.084 16.961 to  5.084 17.088
\plot  5.084 17.088  5.082 17.213 /
\putrule from  5.082 17.213 to  5.082 17.335
\putrule from  5.082 17.335 to  5.082 17.456
\putrule from  5.082 17.456 to  5.082 17.575
\putrule from  5.082 17.575 to  5.082 17.691
\plot  5.082 17.691  5.080 17.805 /
\putrule from  5.080 17.805 to  5.080 17.918
\putrule from  5.080 17.918 to  5.080 18.026
\putrule from  5.080 18.026 to  5.080 18.131
\putrule from  5.080 18.131 to  5.080 18.235
\putrule from  5.080 18.235 to  5.080 18.335
\putrule from  5.080 18.335 to  5.080 18.432
\putrule from  5.080 18.432 to  5.080 18.525
\putrule from  5.080 18.525 to  5.080 18.616
\putrule from  5.080 18.616 to  5.080 18.703
\putrule from  5.080 18.703 to  5.080 18.785
\putrule from  5.080 18.785 to  5.080 18.866
\putrule from  5.080 18.866 to  5.080 18.944
\putrule from  5.080 18.944 to  5.080 19.003
\putrule from  5.080 19.003 to  5.080 19.061
\putrule from  5.080 19.061 to  5.080 19.116
\putrule from  5.080 19.116 to  5.080 19.171
\putrule from  5.080 19.171 to  5.080 19.224
\putrule from  5.080 19.224 to  5.080 19.272
\putrule from  5.080 19.272 to  5.080 19.323
\putrule from  5.080 19.323 to  5.080 19.370
\putrule from  5.080 19.370 to  5.080 19.416
\putrule from  5.080 19.416 to  5.080 19.461
\putrule from  5.080 19.461 to  5.080 19.503
\putrule from  5.080 19.503 to  5.080 19.545
\putrule from  5.080 19.545 to  5.080 19.583
\putrule from  5.080 19.583 to  5.080 19.622
\putrule from  5.080 19.622 to  5.080 19.660
\plot  5.080 19.660  5.082 19.693 /
\putrule from  5.082 19.693 to  5.082 19.727
\putrule from  5.082 19.727 to  5.082 19.761
\putrule from  5.082 19.761 to  5.082 19.791
\putrule from  5.082 19.791 to  5.082 19.820
\putrule from  5.082 19.820 to  5.082 19.848
\plot  5.082 19.848  5.084 19.873 /
\putrule from  5.084 19.873 to  5.084 19.897
\putrule from  5.084 19.897 to  5.084 19.920
\putrule from  5.084 19.920 to  5.084 19.941
\putrule from  5.084 19.941 to  5.084 19.960
\plot  5.084 19.960  5.086 19.979 /
\putrule from  5.086 19.979 to  5.086 19.994
\putrule from  5.086 19.994 to  5.086 20.009
\plot  5.086 20.009  5.088 20.024 /
\putrule from  5.088 20.024 to  5.088 20.034
\plot  5.088 20.034  5.091 20.045 /
\putrule from  5.091 20.045 to  5.091 20.053
\plot  5.091 20.053  5.093 20.060 /
\putrule from  5.093 20.060 to  5.093 20.066
\plot  5.093 20.066  5.095 20.070 /
\putrule from  5.095 20.070 to  5.095 20.072
\plot  5.095 20.072  5.097 20.074 /
\plot  5.097 20.074  5.099 20.072 /
\plot  5.099 20.072  5.101 20.070 /
\putrule from  5.101 20.070 to  5.101 20.066
\plot  5.101 20.066  5.103 20.062 /
\plot  5.103 20.062  5.105 20.055 /
\putrule from  5.105 20.055 to  5.105 20.049
\plot  5.105 20.049  5.108 20.041 /
\plot  5.108 20.041  5.110 20.030 /
\plot  5.110 20.030  5.112 20.019 /
\plot  5.112 20.019  5.114 20.009 /
\plot  5.114 20.009  5.116 19.996 /
\plot  5.116 19.996  5.118 19.983 /
\plot  5.118 19.983  5.120 19.971 /
\plot  5.120 19.971  5.122 19.956 /
\plot  5.122 19.956  5.124 19.939 /
\plot  5.124 19.939  5.127 19.924 /
\plot  5.127 19.924  5.129 19.907 /
\plot  5.129 19.907  5.131 19.888 /
\plot  5.131 19.888  5.133 19.871 /
\putrule from  5.133 19.871 to  5.133 19.869
\plot  5.133 19.869  5.137 19.833 /
\plot  5.137 19.833  5.143 19.793 /
\plot  5.143 19.793  5.150 19.748 /
\plot  5.150 19.748  5.154 19.702 /
\plot  5.154 19.702  5.163 19.653 /
\plot  5.163 19.653  5.169 19.600 /
\plot  5.169 19.600  5.175 19.545 /
\plot  5.175 19.545  5.184 19.486 /
\plot  5.184 19.486  5.192 19.425 /
\plot  5.192 19.425  5.199 19.361 /
\plot  5.199 19.361  5.209 19.296 /
\plot  5.209 19.296  5.218 19.226 /
\plot  5.218 19.226  5.226 19.156 /
\plot  5.226 19.156  5.237 19.084 /
\plot  5.237 19.084  5.245 19.010 /
\plot  5.245 19.010  5.256 18.936 /
\plot  5.256 18.936  5.266 18.860 /
\plot  5.266 18.860  5.277 18.783 /
\plot  5.277 18.783  5.287 18.709 /
\plot  5.287 18.709  5.298 18.633 /
\plot  5.298 18.633  5.306 18.559 /
\plot  5.306 18.559  5.317 18.485 /
\plot  5.317 18.485  5.328 18.411 /
\plot  5.328 18.411  5.338 18.339 /
\plot  5.338 18.339  5.349 18.269 /
\plot  5.349 18.269  5.359 18.201 /
\plot  5.359 18.201  5.368 18.133 /
\plot  5.368 18.133  5.378 18.068 /
\plot  5.378 18.068  5.387 18.002 /
\plot  5.387 18.002  5.397 17.939 /
}%
%
%
\put{\SetFigFont{6}{7.2}{\rmdefault}{\mddefault}{\updefault}{\color[rgb]{0,0,0}$\tilde{S}_-$}%
} [lB] at  4.0 12.700
%
%
\put{\SetFigFont{6}{7.2}{\rmdefault}{\mddefault}{\itdefault}{\color[rgb]{0,0,0}$\tilde{S}_+$}%
} [lB] at 15.558 15.240
%
%
\put{\SetFigFont{6}{7.2}{\rmdefault}{\mddefault}{\itdefault}{\color[rgb]{0,0,0}$\bar{N}_-$}%
} [lB] at  4.0 20.320
%
%
\put{\SetFigFont{6}{7.2}{\rmdefault}{\mddefault}{\itdefault}{\color[rgb]{0,0,0}$\bar{N}_+$}%
} [lB] at 15.558 17.939
%
%
\put{\SetFigFont{6}{7.2}{\rmdefault}{\mddefault}{\itdefault}{\color[rgb]{0,0,0}$\tilde{N}_+$}%
} [lB] at 15.558  8.096
%
%
\put{\SetFigFont{6}{7.2}{\rmdefault}{\mddefault}{\itdefault}{\color[rgb]{0,0,0}$\tilde{N}_-$}%
} [lB] at  4.0  9.842
%
%
\put{\SetFigFont{6}{7.2}{\rmdefault}{\mddefault}{\itdefault}{\color[rgb]{0,0,0}$0$}%
} [lB] at  4.604 16.351
%
%
\put{\SetFigFont{6}{7.2}{\rmdefault}{\mddefault}{\itdefault}{\color[rgb]{0,0,0}$-\pi$}%
} [lB] at  4.0 11.271
%
%
\put{\SetFigFont{6}{7.2}{\rmdefault}{\mddefault}{\itdefault}{\color[rgb]{0,0,0}$\pi$}%
} [lB] at  4.445 21.431
%
%
\put{\SetFigFont{6}{7.2}{\rmdefault}{\mddefault}{\itdefault}{\color[rgb]{0,0,0}$\cW^-$}%
} [lB] at  7.144 14.922
%
%
\put{\SetFigFont{6}{7.2}{\rmdefault}{\mddefault}{\itdefault}{\color[rgb]{0,0,0}$\cW^+$}%
} [lB] at 12.383 17.304
%
%
\put{\SetFigFont{6}{7.2}{\rmdefault}{\mddefault}{\itdefault}{\color[rgb]{0,0,0}$\cK_1$}%
} [lB] at  5.715 15.399
\linethickness=0pt
\putrectangle corners at  4.413 21.768 and 15.589  6.280
\endpicture}
 \end{center}
\caption{\label{fig:CORRIDORone} The corridor $\cK_1$ in $\cC$ and in $\bar{\cC}$.}
\end{figure}

Furthermore, in the universal cover $\bar{\cC}$, using the well-orderedness of $\RR$ it is possible to speak of $\widetilde{\cW}^+$ as being 
situated ``above'' or ``below" $\widetilde{\cW}^-$.
 It is  evident that the boundary of the corridor $\cK_1$ is oriented clockwise if $\widetilde{\cW}^+$ is below $\widetilde{\cW}^-$, 
and counterclockwise if it is the other way around.  see Figure~\ref{fig:CORRIDORone}.  
\subsubsection{Parameter-dependent flows}
Suppose that the function $g$ in \refeq{eq:flow} depends smoothly on a parameter $\mu \in I$ where $I$ is an open interval in $\RR$ 
(It's enough for the $\mu$-dependence to be $C^1$).
  Thus the flow is now
\beq\label{eq:flowmu}
\left\{\begin{array}{rcl} \dot{x} & = & f(x)\\ \dot{y} & = & g_\mu(x,y).\end{array}\right.
\eeq
  By the implicit function theorem, the locations of the equilibria $S_\pm,N_\pm$ also depend --in a $C^1$ fashion--  on $\mu$, so 
long as the non-degeneracy conditions \refeq{cond:hyp} are satisfied. 

We need to make certain assumptions about the $\mu$-dependence of the flow regarding its monotonicity, and the topology of its nullclines:

\subsubsubsection{Monotonicity}
We will only consider parameter-dependent flows for which the function $g_\mu(x,y)$ is {\em monotone non-increasing} in $\mu$.

{\bf Assumption (M)}:  For all $(x,y) \in \bar{\cC}$, we have
\beq\label{ass:mon}  \frac{\p}{\p\mu} g_{\mu}(x,y) \leq 0.
\eeq

This assumption in particular implies a corresponding monotonicity for the distinguished orbits of the flow \refeq{eq:flowmu}:
\begin{lem}\label{lem:mon}
 Let $\tilde{\cW}^\pm_\mu = \{(x(\tau),y^\pm_\mu(\tau)\}_{\tau\in\RR}$ be the distinguished orbits of \refeq{eq:flowmu}.  Then $y^\pm_\mu$ are monotone in $\mu$, i.e.
\beq\label{orbmon}
\mu_1<\mu_2 \implies y^-_{\mu_1}(\tau) \geq y^-_{\mu_2}(\tau),\quad
y^+_{\mu_1}(\tau) \leq y^+_{\mu_2}(\tau)\mbox{for all }\tau\in\RR.
\eeq
\end{lem}
\begin{proof}
Let 
$$ z(\tau) := \frac{\p}{\p \mu} y^-_\mu(\tau).$$
Then $z$ satisfies the ODE
\beq\label{odez}
\dot{z} = \frac{\p g_\mu}{\p y}(x,y) z + \frac{\p g_\mu}{\p \mu}(x,y).
\eeq
By the Implicit Function Theorem and the hypotheses above,
$$
z(-\infty) = \frac{ds^-(\mu)}{d\mu} = - \left.\frac{ \p g/\p y}{\p g /\p \mu}\right|_{x=x_-,y = s_-} < 0
$$
and 
\beq\label{odiz}
\dot{z} - \frac{\p g_\mu}{\p y}(x,y) z \leq 0.
\eeq
Let
$$
U(\tau_0,\tau) := \exp\left(-\int_{\tau_0}^\tau \frac{\p g_\mu}{\p y}(x(\tau'),y(\tau')) d\tau' \right) \geq 0.
$$
Integrating \refeq{odiz} on $[\tau_0,\tau]$ we then obtain
$$
U(\tau_0,\tau) z(\tau) \leq z(\tau_0),
$$
and taking the limit $\tau_0 \to -\infty$ we conclude $z(\tau)\leq 0$ for all $\tau \in \RR$. Therefore $y^-_\mu$ is monotone non-increasing in $\mu$.
The proof of monotonicity of $y^+_\mu$ is completely analogous. 
\end{proof}

\subsubsubsection{Topology of nullclines}
Consider the subset of the cylinder $\cC$ defined by
\beq
\Ga := \{(x,y)\in \cC\ | \ g_\mu(x,y) = 0\}.
\eeq
Thus $\Ga$ is the zero level-set of $g_\mu$.
  Since $g_\mu$ is smooth, $\Ga$ is a curve (or collection of curves) in $\cC$.
 These curves are referred to as the {\em $y$-nullclines} of the flow \refeq{eq:flow}.
  They have the property that any orbit of the flow that crosses them must have a horizontal tangent at the crossing point.
   Moreover, $\Ga$ divides $\cC$ into two regions, thus $\cC = \Ga\cup\cN\cup\cP$, where
\beq
\cN := \{(x,y)\in\cC\ |\ g_\mu(x,y) < 0\},\qquad \cP := \{(x,y)\ |\ g_\mu(x,y)>0\}.
\eeq
Thus the $y$ coordinate of (the lift to the universal cover of) any orbit must decrease in $\cN$ and must increase in $\cP$.  

 Evidently, $\Ga$ must include all the singular points of the flow: $S_\pm\in\Ga$, $N_\pm \in \Ga$.  

It may happen that as $\mu$ crosses a critical value $\mu^{}_c\in I$, the topology of the nullclines undergoes a dramatic change.
  This is indeed the case for the flows that we are studying in this paper.
  We introduce an assumption that amounts to having some control on this change in nullcline topology:
  
{\bf Assumption (A)}:  There exists a $\mu^{}_c \in I$ such that $\Ga$, the zero level-set of $g_{\mu^{}_c}$ has a {\em saddle point} 
at some interior point of $\cC$.
  In particular, for $\mu<\mu^{}_c$, the sets $\cN$ and $\cP$ are both connected, and $\Ga$ is the union of two disjoint curves 
$\Ga = \Ga_u \cup \Ga_d$, with $N_-,S_+ \in \overline{\Ga_d}$ and $S_-,N_+\in \overline{\Ga_u}$.
  On the other hand  for $\mu>\mu^{}_c$, $\cN$ is connected, while $\cP$ has two connected components $\cP_l$ and $\cP_r$, each 
being a convex subset of $\cC$.
  Moreover, $\Ga$ is the union of two disjoint curves $\Ga = \Ga_l \cup \Ga_r$, with $N_-,S_-\in\overline{\Ga_l}$
and $ S_+,N_+\in\overline{\Ga_r}$; see Figure~\ref{fig:topnul}.
\begin{figure}[ht]
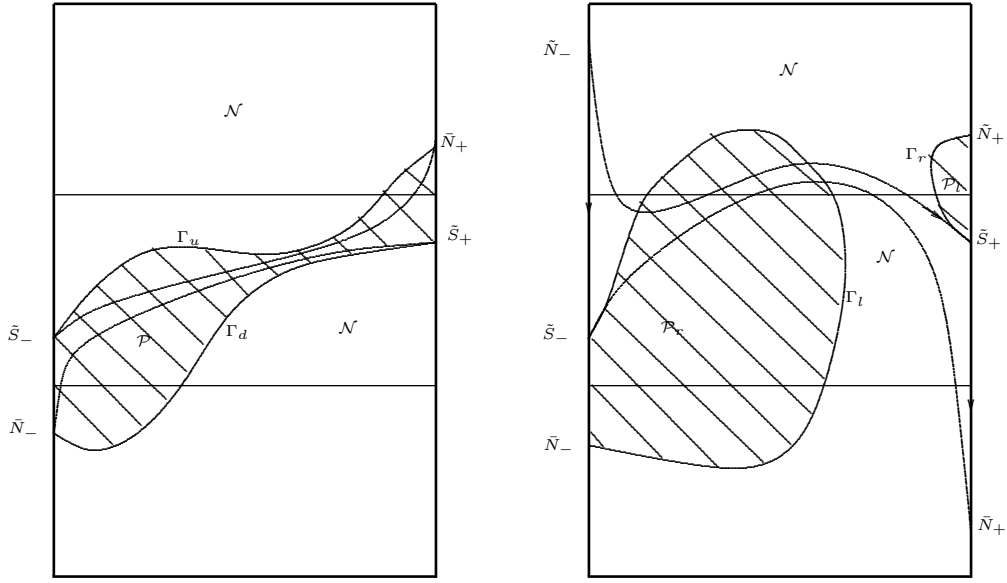

\begin{center}
%
%
\font\thinlinefont=cmr5
\begingroup\makeatletter\ifx\SetFigFont\undefined%
\gdef\SetFigFont#1#2#3#4#5{%
  \reset@font\fontsize{#1}{#2pt}%
  \fontfamily{#3}\fontseries{#4}\fontshape{#5}%
  \selectfont}%
\fi\endgroup%
\mbox{\beginpicture
\setcoordinatesystem units <0.50000cm,0.50000cm>
\unitlength=0.50000cm
\linethickness=1pt
\setplotsymbol ({\makebox(0,0)[l]{\tencirc\symbol{'160}}})
\setshadesymbol ({\thinlinefont .})
\setlinear
%
%
\linethickness=1pt
\setplotsymbol ({\makebox(0,0)[l]{\tencirc\symbol{'160}}})
{\color[rgb]{0,0,0}\putrule from  5.080 21.590 to 15.240 21.590
}%
%
%
\linethickness= 0.500pt
\setplotsymbol ({\thinlinefont .})
{\color[rgb]{0,0,0}\putrule from  5.080 16.510 to 15.240 16.510
}%
%
%
\linethickness= 0.500pt
\setplotsymbol ({\thinlinefont .})
{\color[rgb]{0,0,0}\putrule from  5.080 11.430 to 15.240 11.430
}%
%
%
\linethickness=1pt
\setplotsymbol ({\makebox(0,0)[l]{\tencirc\symbol{'160}}})
{\color[rgb]{0,0,0}\putrule from 15.240 21.590 to 15.240  6.315
\putrule from 15.240  6.350 to  5.045  6.350
\putrule from  5.080  6.350 to  5.080 21.590
}%
%
%
\linethickness= 0.500pt
\setplotsymbol ({\thinlinefont .})
\setdashes < 0.1905cm>
{\color[rgb]{0,0,0}\plot 15.240 15.240 15.236 15.240 /
\plot 15.236 15.240 15.223 15.238 /
\plot 15.223 15.238 15.202 15.234 /
\plot 15.202 15.234 15.170 15.229 /
\plot 15.170 15.229 15.124 15.223 /
\plot 15.124 15.223 15.064 15.215 /
\plot 15.064 15.215 14.994 15.204 /
\plot 14.994 15.204 14.910 15.193 /
\plot 14.910 15.193 14.817 15.179 /
\plot 14.817 15.179 14.717 15.164 /
\plot 14.717 15.164 14.611 15.149 /
\plot 14.611 15.149 14.503 15.134 /
\plot 14.503 15.134 14.393 15.117 /
\plot 14.393 15.117 14.285 15.102 /
\plot 14.285 15.102 14.180 15.085 /
\plot 14.180 15.085 14.076 15.071 /
\plot 14.076 15.071 13.978 15.056 /
\plot 13.978 15.056 13.885 15.043 /
\plot 13.885 15.043 13.796 15.030 /
\plot 13.796 15.030 13.712 15.018 /
\plot 13.712 15.018 13.631 15.005 /
\plot 13.631 15.005 13.555 14.994 /
\plot 13.555 14.994 13.483 14.982 /
\plot 13.483 14.982 13.413 14.971 /
\plot 13.413 14.971 13.348 14.963 /
\plot 13.348 14.963 13.282 14.952 /
\plot 13.282 14.952 13.219 14.942 /
\plot 13.219 14.942 13.157 14.933 /
\plot 13.157 14.933 13.098 14.922 /
\plot 13.098 14.922 13.032 14.912 /
\plot 13.032 14.912 12.967 14.901 /
\plot 12.967 14.901 12.903 14.891 /
\plot 12.903 14.891 12.838 14.880 /
\plot 12.838 14.880 12.772 14.870 /
\plot 12.772 14.870 12.706 14.857 /
\plot 12.706 14.857 12.639 14.844 /
\plot 12.639 14.844 12.571 14.831 /
\plot 12.571 14.831 12.503 14.819 /
\plot 12.503 14.819 12.435 14.806 /
\plot 12.435 14.806 12.366 14.791 /
\plot 12.366 14.791 12.298 14.776 /
\plot 12.298 14.776 12.228 14.762 /
\plot 12.228 14.762 12.158 14.745 /
\plot 12.158 14.745 12.090 14.730 /
\plot 12.090 14.730 12.023 14.713 /
\plot 12.023 14.713 11.955 14.694 /
\plot 11.955 14.694 11.889 14.677 /
\plot 11.889 14.677 11.826 14.658 /
\plot 11.826 14.658 11.762 14.639 /
\plot 11.762 14.639 11.699 14.620 /
\plot 11.699 14.620 11.637 14.601 /
\plot 11.637 14.601 11.578 14.582 /
\plot 11.578 14.582 11.521 14.561 /
\plot 11.521 14.561 11.464 14.539 /
\plot 11.464 14.539 11.407 14.518 /
\plot 11.407 14.518 11.352 14.495 /
\plot 11.352 14.495 11.299 14.472 /
\plot 11.299 14.472 11.246 14.450 /
\plot 11.246 14.450 11.195 14.427 /
\plot 11.195 14.427 11.144 14.404 /
\plot 11.144 14.404 11.093 14.379 /
\plot 11.093 14.379 11.041 14.351 /
\plot 11.041 14.351 10.988 14.323 /
\plot 10.988 14.323 10.935 14.294 /
\plot 10.935 14.294 10.880 14.262 /
\plot 10.880 14.262 10.827 14.230 /
\plot 10.827 14.230 10.772 14.196 /
\plot 10.772 14.196 10.715 14.163 /
\plot 10.715 14.163 10.660 14.127 /
\plot 10.660 14.127 10.605 14.089 /
\plot 10.605 14.089 10.547 14.048 /
\plot 10.547 14.048 10.492 14.008 /
\plot 10.492 14.008 10.437 13.968 /
\plot 10.437 13.968 10.382 13.926 /
\plot 10.382 13.926 10.327 13.883 /
\plot 10.327 13.883 10.274 13.839 /
\plot 10.274 13.839 10.224 13.796 /
\plot 10.224 13.796 10.171 13.752 /
\plot 10.171 13.752 10.122 13.705 /
\plot 10.122 13.705 10.073 13.661 /
\plot 10.073 13.661 10.025 13.617 /
\plot 10.025 13.617  9.980 13.570 /
\plot  9.980 13.570  9.934 13.523 /
\plot  9.934 13.523  9.889 13.477 /
\plot  9.889 13.477  9.847 13.430 /
\plot  9.847 13.430  9.804 13.384 /
\plot  9.804 13.384  9.764 13.335 /
\plot  9.764 13.335  9.724 13.288 /
\plot  9.724 13.288  9.686 13.242 /
\plot  9.686 13.242  9.648 13.193 /
\plot  9.648 13.193  9.610 13.144 /
\plot  9.610 13.144  9.572 13.094 /
\plot  9.572 13.094  9.533 13.043 /
\plot  9.533 13.043  9.495 12.988 /
\plot  9.495 12.988  9.455 12.933 /
\plot  9.455 12.933  9.415 12.878 /
\plot  9.415 12.878  9.377 12.821 /
\plot  9.377 12.821  9.337 12.761 /
\plot  9.337 12.761  9.296 12.702 /
\plot  9.296 12.702  9.256 12.641 /
\plot  9.256 12.641  9.216 12.579 /
\plot  9.216 12.579  9.176 12.518 /
\plot  9.176 12.518  9.136 12.454 /
\plot  9.136 12.454  9.095 12.391 /
\plot  9.095 12.391  9.055 12.330 /
\plot  9.055 12.330  9.017 12.268 /
\plot  9.017 12.268  8.977 12.205 /
\plot  8.977 12.205  8.939 12.143 /
\plot  8.939 12.143  8.901 12.084 /
\plot  8.901 12.084  8.862 12.025 /
\plot  8.862 12.025  8.827 11.966 /
\plot  8.827 11.966  8.788 11.908 /
\plot  8.788 11.908  8.752 11.851 /
\plot  8.752 11.851  8.719 11.796 /
\plot  8.719 11.796  8.683 11.741 /
\plot  8.683 11.741  8.649 11.688 /
\plot  8.649 11.688  8.615 11.637 /
\plot  8.615 11.637  8.581 11.587 /
\plot  8.581 11.587  8.547 11.536 /
\plot  8.547 11.536  8.507 11.479 /
\plot  8.507 11.479  8.467 11.422 /
\plot  8.467 11.422  8.426 11.366 /
\plot  8.426 11.366  8.386 11.309 /
\plot  8.386 11.309  8.344 11.254 /
\plot  8.344 11.254  8.302 11.197 /
\plot  8.302 11.197  8.259 11.142 /
\plot  8.259 11.142  8.215 11.085 /
\plot  8.215 11.085  8.170 11.028 /
\plot  8.170 11.028  8.124 10.973 /
\plot  8.124 10.973  8.077 10.916 /
\plot  8.077 10.916  8.029 10.861 /
\plot  8.029 10.861  7.982 10.806 /
\plot  7.982 10.806  7.933 10.753 /
\plot  7.933 10.753  7.887 10.700 /
\plot  7.887 10.700  7.838 10.649 /
\plot  7.838 10.649  7.789 10.598 /
\plot  7.789 10.598  7.743 10.549 /
\plot  7.743 10.549  7.696 10.503 /
\plot  7.696 10.503  7.650 10.458 /
\plot  7.650 10.458  7.603 10.414 /
\plot  7.603 10.414  7.559 10.374 /
\plot  7.559 10.374  7.514 10.334 /
\plot  7.514 10.334  7.472 10.295 /
\plot  7.472 10.295  7.427 10.259 /
\plot  7.427 10.259  7.385 10.226 /
\plot  7.385 10.226  7.345 10.192 /
\plot  7.345 10.192  7.303 10.160 /
\plot  7.303 10.160  7.254 10.124 /
\plot  7.254 10.124  7.205 10.090 /
\plot  7.205 10.090  7.156 10.056 /
\plot  7.156 10.056  7.106 10.025 /
\plot  7.106 10.025  7.055  9.993 /
\plot  7.055  9.993  7.002  9.963 /
\plot  7.002  9.963  6.949  9.934 /
\plot  6.949  9.934  6.894  9.906 /
\plot  6.894  9.906  6.841  9.881 /
\plot  6.841  9.881  6.786  9.855 /
\plot  6.786  9.855  6.731  9.832 /
\plot  6.731  9.832  6.674  9.811 /
\plot  6.674  9.811  6.621  9.792 /
\plot  6.621  9.792  6.566  9.775 /
\plot  6.566  9.775  6.513  9.758 /
\plot  6.513  9.758  6.460  9.745 /
\plot  6.460  9.745  6.407  9.735 /
\plot  6.407  9.735  6.358  9.724 /
\plot  6.358  9.724  6.310  9.718 /
\plot  6.310  9.718  6.263  9.711 /
\plot  6.263  9.711  6.217  9.709 /
\plot  6.217  9.709  6.172  9.707 /
\plot  6.172  9.707  6.128  9.707 /
\plot  6.128  9.707  6.085  9.709 /
\plot  6.085  9.709  6.039  9.716 /
\plot  6.039  9.716  5.994  9.722 /
\plot  5.994  9.722  5.948  9.730 /
\plot  5.948  9.730  5.901  9.743 /
\plot  5.901  9.743  5.853  9.758 /
\plot  5.853  9.758  5.802  9.775 /
\plot  5.802  9.775  5.749  9.796 /
\plot  5.749  9.796  5.692  9.821 /
\plot  5.692  9.821  5.632  9.849 /
\plot  5.632  9.849  5.569  9.881 /
\plot  5.569  9.881  5.503  9.914 /
\plot  5.503  9.914  5.438  9.950 /
\plot  5.438  9.950  5.370  9.989 /
\plot  5.370  9.989  5.306 10.025 /
\plot  5.306 10.025  5.247 10.061 /
\plot  5.247 10.061  5.194 10.090 /
\plot  5.194 10.090  5.152 10.116 /
\plot  5.152 10.116  5.120 10.137 /
\plot  5.120 10.137  5.097 10.149 /
\plot  5.097 10.149  5.086 10.156 /
\plot  5.086 10.156  5.080 10.160 /
}%
%
%
\linethickness= 0.500pt
\setplotsymbol ({\thinlinefont .})
{\color[rgb]{0,0,0}\plot 15.240 17.780 15.236 17.776 /
\plot 15.236 17.776 15.225 17.767 /
\plot 15.225 17.767 15.204 17.755 /
\plot 15.204 17.755 15.174 17.731 /
\plot 15.174 17.731 15.134 17.702 /
\plot 15.134 17.702 15.085 17.666 /
\plot 15.085 17.666 15.028 17.621 /
\plot 15.028 17.621 14.965 17.573 /
\plot 14.965 17.573 14.897 17.522 /
\plot 14.897 17.522 14.829 17.469 /
\plot 14.829 17.469 14.760 17.416 /
\plot 14.760 17.416 14.694 17.363 /
\plot 14.694 17.363 14.630 17.312 /
\plot 14.630 17.312 14.569 17.264 /
\plot 14.569 17.264 14.514 17.217 /
\plot 14.514 17.217 14.459 17.173 /
\plot 14.459 17.173 14.410 17.128 /
\plot 14.410 17.128 14.364 17.086 /
\plot 14.364 17.086 14.319 17.046 /
\plot 14.319 17.046 14.277 17.005 /
\plot 14.277 17.005 14.235 16.963 /
\plot 14.235 16.963 14.194 16.923 /
\plot 14.194 16.923 14.156 16.880 /
\plot 14.156 16.880 14.122 16.844 /
\plot 14.122 16.844 14.089 16.806 /
\plot 14.089 16.806 14.053 16.768 /
\plot 14.053 16.768 14.019 16.728 /
\plot 14.019 16.728 13.985 16.686 /
\plot 13.985 16.686 13.949 16.643 /
\plot 13.949 16.643 13.911 16.601 /
\plot 13.911 16.601 13.875 16.557 /
\plot 13.875 16.557 13.835 16.510 /
\plot 13.835 16.510 13.796 16.463 /
\plot 13.796 16.463 13.756 16.417 /
\plot 13.756 16.417 13.714 16.368 /
\plot 13.714 16.368 13.674 16.320 /
\plot 13.674 16.320 13.631 16.271 /
\plot 13.631 16.271 13.587 16.224 /
\plot 13.587 16.224 13.545 16.176 /
\plot 13.545 16.176 13.500 16.127 /
\plot 13.500 16.127 13.456 16.080 /
\plot 13.456 16.080 13.409 16.034 /
\plot 13.409 16.034 13.365 15.987 /
\plot 13.365 15.987 13.320 15.943 /
\plot 13.320 15.943 13.274 15.900 /
\plot 13.274 15.900 13.227 15.858 /
\plot 13.227 15.858 13.180 15.816 /
\plot 13.180 15.816 13.134 15.776 /
\plot 13.134 15.776 13.087 15.737 /
\plot 13.087 15.737 13.041 15.699 /
\plot 13.041 15.699 12.992 15.663 /
\plot 12.992 15.663 12.943 15.629 /
\plot 12.943 15.629 12.897 15.596 /
\plot 12.897 15.596 12.846 15.564 /
\plot 12.846 15.564 12.795 15.532 /
\plot 12.795 15.532 12.742 15.498 /
\plot 12.742 15.498 12.689 15.466 /
\plot 12.689 15.466 12.632 15.435 /
\plot 12.632 15.435 12.573 15.405 /
\plot 12.573 15.405 12.514 15.373 /
\plot 12.514 15.373 12.452 15.344 /
\plot 12.452 15.344 12.389 15.314 /
\plot 12.389 15.314 12.325 15.284 /
\plot 12.325 15.284 12.260 15.255 /
\plot 12.260 15.255 12.194 15.227 /
\plot 12.194 15.227 12.126 15.202 /
\plot 12.126 15.202 12.059 15.174 /
\plot 12.059 15.174 11.991 15.151 /
\plot 11.991 15.151 11.923 15.128 /
\plot 11.923 15.128 11.855 15.105 /
\plot 11.855 15.105 11.788 15.083 /
\plot 11.788 15.083 11.722 15.064 /
\plot 11.722 15.064 11.656 15.047 /
\plot 11.656 15.047 11.591 15.030 /
\plot 11.591 15.030 11.527 15.014 /
\plot 11.527 15.014 11.464 15.001 /
\plot 11.464 15.001 11.402 14.988 /
\plot 11.402 14.988 11.343 14.975 /
\plot 11.343 14.975 11.282 14.965 /
\plot 11.282 14.965 11.225 14.956 /
\plot 11.225 14.956 11.165 14.948 /
\plot 11.165 14.948 11.104 14.942 /
\plot 11.104 14.942 11.041 14.935 /
\plot 11.041 14.935 10.979 14.931 /
\plot 10.979 14.931 10.916 14.927 /
\plot 10.916 14.927 10.852 14.925 /
\plot 10.852 14.925 10.787 14.922 /
\plot 10.787 14.922 10.721 14.922 /
\plot 10.721 14.922 10.655 14.922 /
\plot 10.655 14.922 10.588 14.925 /
\plot 10.588 14.925 10.520 14.927 /
\plot 10.520 14.927 10.450 14.931 /
\plot 10.450 14.931 10.382 14.933 /
\plot 10.382 14.933 10.312 14.939 /
\plot 10.312 14.939 10.243 14.944 /
\plot 10.243 14.944 10.173 14.950 /
\plot 10.173 14.950 10.105 14.958 /
\plot 10.105 14.958 10.035 14.965 /
\plot 10.035 14.965  9.967 14.973 /
\plot  9.967 14.973  9.902 14.980 /
\plot  9.902 14.980  9.834 14.988 /
\plot  9.834 14.988  9.771 14.997 /
\plot  9.771 14.997  9.705 15.005 /
\plot  9.705 15.005  9.644 15.014 /
\plot  9.644 15.014  9.580 15.022 /
\plot  9.580 15.022  9.521 15.030 /
\plot  9.521 15.030  9.459 15.039 /
\plot  9.459 15.039  9.400 15.047 /
\plot  9.400 15.047  9.341 15.054 /
\plot  9.341 15.054  9.279 15.062 /
\plot  9.279 15.062  9.220 15.071 /
\plot  9.220 15.071  9.159 15.077 /
\plot  9.159 15.077  9.095 15.083 /
\plot  9.095 15.083  9.032 15.090 /
\plot  9.032 15.090  8.968 15.096 /
\plot  8.968 15.096  8.903 15.102 /
\plot  8.903 15.102  8.835 15.107 /
\plot  8.835 15.107  8.767 15.113 /
\plot  8.767 15.113  8.700 15.115 /
\plot  8.700 15.115  8.630 15.117 /
\plot  8.630 15.117  8.560 15.119 /
\plot  8.560 15.119  8.488 15.121 /
\plot  8.488 15.121  8.418 15.119 /
\plot  8.418 15.119  8.346 15.119 /
\plot  8.346 15.119  8.276 15.115 /
\plot  8.276 15.115  8.206 15.111 /
\plot  8.206 15.111  8.136 15.107 /
\plot  8.136 15.107  8.069 15.100 /
\plot  8.069 15.100  8.003 15.092 /
\plot  8.003 15.092  7.938 15.081 /
\plot  7.938 15.081  7.874 15.071 /
\plot  7.874 15.071  7.811 15.060 /
\plot  7.811 15.060  7.749 15.045 /
\plot  7.749 15.045  7.688 15.030 /
\plot  7.688 15.030  7.631 15.014 /
\plot  7.631 15.014  7.571 14.994 /
\plot  7.571 14.994  7.514 14.975 /
\plot  7.514 14.975  7.457 14.954 /
\plot  7.457 14.954  7.400 14.931 /
\plot  7.400 14.931  7.343 14.906 /
\plot  7.343 14.906  7.286 14.876 /
\plot  7.286 14.876  7.228 14.846 /
\plot  7.228 14.846  7.169 14.817 /
\plot  7.169 14.817  7.110 14.783 /
\plot  7.110 14.783  7.051 14.747 /
\plot  7.051 14.747  6.991 14.709 /
\plot  6.991 14.709  6.932 14.669 /
\plot  6.932 14.669  6.871 14.628 /
\plot  6.871 14.628  6.811 14.586 /
\plot  6.811 14.586  6.752 14.541 /
\plot  6.752 14.541  6.693 14.497 /
\plot  6.693 14.497  6.636 14.450 /
\plot  6.636 14.450  6.579 14.404 /
\plot  6.579 14.404  6.521 14.357 /
\plot  6.521 14.357  6.466 14.311 /
\plot  6.466 14.311  6.414 14.262 /
\plot  6.414 14.262  6.361 14.216 /
\plot  6.361 14.216  6.312 14.171 /
\plot  6.312 14.171  6.263 14.125 /
\plot  6.263 14.125  6.217 14.080 /
\plot  6.217 14.080  6.172 14.036 /
\plot  6.172 14.036  6.128 13.991 /
\plot  6.128 13.991  6.085 13.949 /
\plot  6.085 13.949  6.045 13.906 /
\plot  6.045 13.906  6.007 13.864 /
\plot  6.007 13.864  5.965 13.820 /
\plot  5.965 13.820  5.925 13.775 /
\plot  5.925 13.775  5.884 13.731 /
\plot  5.884 13.731  5.844 13.686 /
\plot  5.844 13.686  5.804 13.640 /
\plot  5.804 13.640  5.764 13.591 /
\plot  5.764 13.591  5.723 13.542 /
\plot  5.723 13.542  5.681 13.492 /
\plot  5.681 13.492  5.639 13.437 /
\plot  5.639 13.437  5.592 13.379 /
\plot  5.592 13.379  5.546 13.320 /
\plot  5.546 13.320  5.497 13.257 /
\plot  5.497 13.257  5.448 13.193 /
\plot  5.448 13.193  5.400 13.128 /
\plot  5.400 13.128  5.349 13.062 /
\plot  5.349 13.062  5.300 12.998 /
\plot  5.300 12.998  5.256 12.937 /
\plot  5.256 12.937  5.213 12.880 /
\plot  5.213 12.880  5.175 12.829 /
\plot  5.175 12.829  5.143 12.787 /
\plot  5.143 12.787  5.118 12.753 /
\plot  5.118 12.753  5.101 12.730 /
\plot  5.101 12.730  5.088 12.713 /
\plot  5.088 12.713  5.082 12.704 /
\plot  5.082 12.704  5.080 12.700 /
}%
%
%
\linethickness= 0.500pt
\setplotsymbol ({\thinlinefont .})
\setsolid
{\color[rgb]{0,0,0}\plot  5.080 12.700  5.082 12.702 /
\plot  5.082 12.702  5.091 12.708 /
\plot  5.091 12.708  5.103 12.721 /
\plot  5.103 12.721  5.122 12.738 /
\plot  5.122 12.738  5.150 12.761 /
\plot  5.150 12.761  5.186 12.791 /
\plot  5.186 12.791  5.228 12.827 /
\plot  5.228 12.827  5.277 12.867 /
\plot  5.277 12.867  5.330 12.912 /
\plot  5.330 12.912  5.387 12.958 /
\plot  5.387 12.958  5.446 13.005 /
\plot  5.446 13.005  5.508 13.051 /
\plot  5.508 13.051  5.569 13.096 /
\plot  5.569 13.096  5.630 13.140 /
\plot  5.630 13.140  5.692 13.183 /
\plot  5.692 13.183  5.753 13.223 /
\plot  5.753 13.223  5.814 13.261 /
\plot  5.814 13.261  5.876 13.297 /
\plot  5.876 13.297  5.937 13.331 /
\plot  5.937 13.331  6.001 13.365 /
\plot  6.001 13.365  6.064 13.396 /
\plot  6.064 13.396  6.132 13.428 /
\plot  6.132 13.428  6.200 13.460 /
\plot  6.200 13.460  6.274 13.490 /
\plot  6.274 13.490  6.350 13.519 /
\plot  6.350 13.519  6.401 13.540 /
\plot  6.401 13.540  6.456 13.561 /
\plot  6.456 13.561  6.511 13.581 /
\plot  6.511 13.581  6.568 13.602 /
\plot  6.568 13.602  6.627 13.623 /
\plot  6.627 13.623  6.689 13.644 /
\plot  6.689 13.644  6.752 13.665 /
\plot  6.752 13.665  6.820 13.686 /
\plot  6.820 13.686  6.888 13.710 /
\plot  6.888 13.710  6.957 13.733 /
\plot  6.957 13.733  7.032 13.754 /
\plot  7.032 13.754  7.106 13.777 /
\plot  7.106 13.777  7.184 13.803 /
\plot  7.184 13.803  7.264 13.826 /
\plot  7.264 13.826  7.345 13.849 /
\plot  7.345 13.849  7.430 13.875 /
\plot  7.430 13.875  7.514 13.898 /
\plot  7.514 13.898  7.601 13.923 /
\plot  7.601 13.923  7.690 13.949 /
\plot  7.690 13.949  7.781 13.974 /
\plot  7.781 13.974  7.872 14.000 /
\plot  7.872 14.000  7.963 14.023 /
\plot  7.963 14.023  8.056 14.048 /
\plot  8.056 14.048  8.149 14.074 /
\plot  8.149 14.074  8.242 14.099 /
\plot  8.242 14.099  8.338 14.125 /
\plot  8.338 14.125  8.431 14.148 /
\plot  8.431 14.148  8.526 14.173 /
\plot  8.526 14.173  8.619 14.196 /
\plot  8.619 14.196  8.714 14.220 /
\plot  8.714 14.220  8.807 14.243 /
\plot  8.807 14.243  8.901 14.268 /
\plot  8.901 14.268  8.992 14.290 /
\plot  8.992 14.290  9.085 14.313 /
\plot  9.085 14.313  9.176 14.336 /
\plot  9.176 14.336  9.267 14.357 /
\plot  9.267 14.357  9.358 14.381 /
\plot  9.358 14.381  9.449 14.402 /
\plot  9.449 14.402  9.540 14.425 /
\plot  9.540 14.425  9.631 14.446 /
\plot  9.631 14.446  9.718 14.467 /
\plot  9.718 14.467  9.804 14.489 /
\plot  9.804 14.489  9.893 14.510 /
\plot  9.893 14.510  9.982 14.531 /
\plot  9.982 14.531 10.071 14.552 /
\plot 10.071 14.552 10.162 14.573 /
\plot 10.162 14.573 10.253 14.594 /
\plot 10.253 14.594 10.346 14.618 /
\plot 10.346 14.618 10.439 14.639 /
\plot 10.439 14.639 10.535 14.662 /
\plot 10.535 14.662 10.630 14.685 /
\plot 10.630 14.685 10.727 14.709 /
\plot 10.727 14.709 10.825 14.734 /
\plot 10.825 14.734 10.922 14.757 /
\plot 10.922 14.757 11.021 14.783 /
\plot 11.021 14.783 11.119 14.808 /
\plot 11.119 14.808 11.218 14.834 /
\plot 11.218 14.834 11.318 14.859 /
\plot 11.318 14.859 11.417 14.884 /
\plot 11.417 14.884 11.515 14.910 /
\plot 11.515 14.910 11.614 14.937 /
\plot 11.614 14.937 11.712 14.963 /
\plot 11.712 14.963 11.807 14.990 /
\plot 11.807 14.990 11.902 15.018 /
\plot 11.902 15.018 11.997 15.045 /
\plot 11.997 15.045 12.088 15.071 /
\plot 12.088 15.071 12.181 15.098 /
\plot 12.181 15.098 12.270 15.126 /
\plot 12.270 15.126 12.357 15.153 /
\plot 12.357 15.153 12.444 15.181 /
\plot 12.444 15.181 12.526 15.206 /
\plot 12.526 15.206 12.609 15.234 /
\plot 12.609 15.234 12.687 15.261 /
\plot 12.687 15.261 12.766 15.289 /
\plot 12.766 15.289 12.840 15.314 /
\plot 12.840 15.314 12.914 15.342 /
\plot 12.914 15.342 12.984 15.369 /
\plot 12.984 15.369 13.051 15.395 /
\plot 13.051 15.395 13.119 15.422 /
\plot 13.119 15.422 13.185 15.450 /
\plot 13.185 15.450 13.246 15.477 /
\plot 13.246 15.477 13.310 15.505 /
\plot 13.310 15.505 13.392 15.543 /
\plot 13.392 15.543 13.473 15.583 /
\plot 13.473 15.583 13.551 15.625 /
\plot 13.551 15.625 13.629 15.665 /
\plot 13.629 15.665 13.703 15.710 /
\plot 13.703 15.710 13.775 15.754 /
\plot 13.775 15.754 13.847 15.799 /
\plot 13.847 15.799 13.917 15.845 /
\plot 13.917 15.845 13.985 15.892 /
\plot 13.985 15.892 14.050 15.941 /
\plot 14.050 15.941 14.116 15.989 /
\plot 14.116 15.989 14.177 16.040 /
\plot 14.177 16.040 14.237 16.091 /
\plot 14.237 16.091 14.296 16.142 /
\plot 14.296 16.142 14.351 16.192 /
\plot 14.351 16.192 14.404 16.243 /
\plot 14.404 16.243 14.455 16.294 /
\plot 14.455 16.294 14.503 16.345 /
\plot 14.503 16.345 14.548 16.396 /
\plot 14.548 16.396 14.590 16.444 /
\plot 14.590 16.444 14.633 16.493 /
\plot 14.633 16.493 14.669 16.540 /
\plot 14.669 16.540 14.704 16.586 /
\plot 14.704 16.586 14.738 16.631 /
\plot 14.738 16.631 14.768 16.675 /
\plot 14.768 16.675 14.798 16.720 /
\plot 14.798 16.720 14.825 16.760 /
\plot 14.825 16.760 14.851 16.802 /
\plot 14.851 16.802 14.874 16.842 /
\plot 14.874 16.842 14.897 16.880 /
\plot 14.897 16.880 14.925 16.933 /
\plot 14.925 16.933 14.952 16.984 /
\plot 14.952 16.984 14.975 17.035 /
\plot 14.975 17.035 14.999 17.088 /
\plot 14.999 17.088 15.022 17.141 /
\plot 15.022 17.141 15.043 17.194 /
\plot 15.043 17.194 15.062 17.251 /
\plot 15.062 17.251 15.083 17.310 /
\plot 15.083 17.310 15.102 17.374 /
\plot 15.102 17.374 15.121 17.437 /
\plot 15.121 17.437 15.138 17.505 /
\plot 15.138 17.505 15.157 17.575 /
\plot 15.157 17.575 15.174 17.642 /
\plot 15.174 17.642 15.189 17.708 /
\plot 15.189 17.708 15.204 17.769 /
\plot 15.204 17.769 15.215 17.822 /
\plot 15.215 17.822 15.225 17.867 /
\plot 15.225 17.867 15.232 17.899 /
\plot 15.232 17.899 15.236 17.922 /
\plot 15.236 17.922 15.238 17.932 /
\plot 15.238 17.932 15.240 17.939 /
}%
%
%
\linethickness= 0.500pt
\setplotsymbol ({\thinlinefont .})
{\color[rgb]{0,0,0}\putrule from 15.240 15.240 to 15.236 15.240
\putrule from 15.236 15.240 to 15.227 15.240
\plot 15.227 15.240 15.208 15.238 /
\plot 15.208 15.238 15.183 15.236 /
\plot 15.183 15.236 15.145 15.234 /
\plot 15.145 15.234 15.094 15.229 /
\plot 15.094 15.229 15.033 15.225 /
\plot 15.033 15.225 14.958 15.219 /
\plot 14.958 15.219 14.874 15.212 /
\plot 14.874 15.212 14.781 15.206 /
\plot 14.781 15.206 14.679 15.198 /
\plot 14.679 15.198 14.571 15.189 /
\plot 14.571 15.189 14.461 15.181 /
\plot 14.461 15.181 14.347 15.172 /
\plot 14.347 15.172 14.232 15.162 /
\plot 14.232 15.162 14.120 15.153 /
\plot 14.120 15.153 14.008 15.143 /
\plot 14.008 15.143 13.898 15.134 /
\plot 13.898 15.134 13.792 15.124 /
\plot 13.792 15.124 13.691 15.113 /
\plot 13.691 15.113 13.593 15.105 /
\plot 13.593 15.105 13.498 15.094 /
\plot 13.498 15.094 13.407 15.085 /
\plot 13.407 15.085 13.320 15.075 /
\plot 13.320 15.075 13.236 15.064 /
\plot 13.236 15.064 13.153 15.054 /
\plot 13.153 15.054 13.075 15.045 /
\plot 13.075 15.045 12.996 15.035 /
\plot 12.996 15.035 12.920 15.022 /
\plot 12.920 15.022 12.844 15.011 /
\plot 12.844 15.011 12.770 15.001 /
\plot 12.770 15.001 12.696 14.988 /
\plot 12.696 14.988 12.622 14.975 /
\plot 12.622 14.975 12.554 14.963 /
\plot 12.554 14.963 12.486 14.952 /
\plot 12.486 14.952 12.416 14.939 /
\plot 12.416 14.939 12.349 14.927 /
\plot 12.349 14.927 12.279 14.912 /
\plot 12.279 14.912 12.207 14.897 /
\plot 12.207 14.897 12.135 14.882 /
\plot 12.135 14.882 12.061 14.867 /
\plot 12.061 14.867 11.985 14.853 /
\plot 11.985 14.853 11.908 14.836 /
\plot 11.908 14.836 11.832 14.819 /
\plot 11.832 14.819 11.754 14.800 /
\plot 11.754 14.800 11.673 14.781 /
\plot 11.673 14.781 11.593 14.762 /
\plot 11.593 14.762 11.510 14.743 /
\plot 11.510 14.743 11.428 14.721 /
\plot 11.428 14.721 11.345 14.702 /
\plot 11.345 14.702 11.261 14.681 /
\plot 11.261 14.681 11.178 14.658 /
\plot 11.178 14.658 11.093 14.637 /
\plot 11.093 14.637 11.009 14.613 /
\plot 11.009 14.613 10.924 14.592 /
\plot 10.924 14.592 10.839 14.569 /
\plot 10.839 14.569 10.757 14.546 /
\plot 10.757 14.546 10.672 14.522 /
\plot 10.672 14.522 10.590 14.497 /
\plot 10.590 14.497 10.509 14.474 /
\plot 10.509 14.474 10.427 14.450 /
\plot 10.427 14.450 10.346 14.427 /
\plot 10.346 14.427 10.268 14.402 /
\plot 10.268 14.402 10.190 14.379 /
\plot 10.190 14.379 10.111 14.355 /
\plot 10.111 14.355 10.033 14.330 /
\plot 10.033 14.330  9.957 14.307 /
\plot  9.957 14.307  9.883 14.281 /
\plot  9.883 14.281  9.807 14.258 /
\plot  9.807 14.258  9.732 14.232 /
\plot  9.732 14.232  9.658 14.207 /
\plot  9.658 14.207  9.578 14.182 /
\plot  9.578 14.182  9.500 14.154 /
\plot  9.500 14.154  9.419 14.127 /
\plot  9.419 14.127  9.339 14.099 /
\plot  9.339 14.099  9.258 14.072 /
\plot  9.258 14.072  9.176 14.042 /
\plot  9.176 14.042  9.093 14.012 /
\plot  9.093 14.012  9.009 13.983 /
\plot  9.009 13.983  8.924 13.951 /
\plot  8.924 13.951  8.837 13.919 /
\plot  8.837 13.919  8.752 13.887 /
\plot  8.752 13.887  8.664 13.856 /
\plot  8.664 13.856  8.577 13.822 /
\plot  8.577 13.822  8.490 13.788 /
\plot  8.490 13.788  8.401 13.754 /
\plot  8.401 13.754  8.314 13.720 /
\plot  8.314 13.720  8.227 13.686 /
\plot  8.227 13.686  8.141 13.653 /
\plot  8.141 13.653  8.054 13.619 /
\plot  8.054 13.619  7.969 13.585 /
\plot  7.969 13.585  7.887 13.551 /
\plot  7.887 13.551  7.804 13.517 /
\plot  7.804 13.517  7.722 13.483 /
\plot  7.722 13.483  7.643 13.449 /
\plot  7.643 13.449  7.567 13.418 /
\plot  7.567 13.418  7.491 13.386 /
\plot  7.491 13.386  7.419 13.354 /
\plot  7.419 13.354  7.347 13.322 /
\plot  7.347 13.322  7.277 13.293 /
\plot  7.277 13.293  7.211 13.263 /
\plot  7.211 13.263  7.146 13.233 /
\plot  7.146 13.233  7.084 13.206 /
\plot  7.084 13.206  7.023 13.178 /
\plot  7.023 13.178  6.966 13.151 /
\plot  6.966 13.151  6.909 13.123 /
\plot  6.909 13.123  6.854 13.096 /
\plot  6.854 13.096  6.778 13.060 /
\plot  6.778 13.060  6.706 13.026 /
\plot  6.706 13.026  6.636 12.990 /
\plot  6.636 12.990  6.568 12.954 /
\plot  6.568 12.954  6.500 12.918 /
\plot  6.500 12.918  6.435 12.882 /
\plot  6.435 12.882  6.371 12.846 /
\plot  6.371 12.846  6.308 12.810 /
\plot  6.308 12.810  6.248 12.774 /
\plot  6.248 12.774  6.189 12.738 /
\plot  6.189 12.738  6.132 12.700 /
\plot  6.132 12.700  6.075 12.664 /
\plot  6.075 12.664  6.022 12.628 /
\plot  6.022 12.628  5.971 12.592 /
\plot  5.971 12.592  5.925 12.556 /
\plot  5.925 12.556  5.878 12.520 /
\plot  5.878 12.520  5.836 12.486 /
\plot  5.836 12.486  5.795 12.452 /
\plot  5.795 12.452  5.757 12.418 /
\plot  5.757 12.418  5.721 12.385 /
\plot  5.721 12.385  5.690 12.353 /
\plot  5.690 12.353  5.658 12.321 /
\plot  5.658 12.321  5.630 12.289 /
\plot  5.630 12.289  5.605 12.260 /
\plot  5.605 12.260  5.580 12.228 /
\plot  5.580 12.228  5.556 12.196 /
\plot  5.556 12.196  5.529 12.156 /
\plot  5.529 12.156  5.501 12.116 /
\plot  5.501 12.116  5.478 12.073 /
\plot  5.478 12.073  5.455 12.031 /
\plot  5.455 12.031  5.431 11.985 /
\plot  5.431 11.985  5.410 11.938 /
\plot  5.410 11.938  5.391 11.891 /
\plot  5.391 11.891  5.372 11.841 /
\plot  5.372 11.841  5.355 11.792 /
\plot  5.355 11.792  5.340 11.741 /
\plot  5.340 11.741  5.326 11.688 /
\plot  5.326 11.688  5.313 11.637 /
\plot  5.313 11.637  5.300 11.587 /
\plot  5.300 11.587  5.290 11.536 /
\plot  5.290 11.536  5.279 11.485 /
\plot  5.279 11.485  5.271 11.436 /
\plot  5.271 11.436  5.262 11.388 /
\plot  5.262 11.388  5.254 11.339 /
\plot  5.254 11.339  5.245 11.292 /
\plot  5.245 11.292  5.239 11.244 /
\plot  5.239 11.244  5.232 11.201 /
\plot  5.232 11.201  5.226 11.157 /
\plot  5.226 11.157  5.218 11.113 /
\plot  5.218 11.113  5.211 11.066 /
\plot  5.211 11.066  5.205 11.015 /
\plot  5.205 11.015  5.196 10.962 /
\plot  5.196 10.962  5.188 10.907 /
\plot  5.188 10.907  5.179 10.846 /
\plot  5.179 10.846  5.169 10.780 /
\plot  5.169 10.780  5.160 10.712 /
\plot  5.160 10.712  5.150 10.640 /
\plot  5.150 10.640  5.139 10.564 /
\plot  5.139 10.564  5.127 10.490 /
\plot  5.127 10.490  5.118 10.418 /
\plot  5.118 10.418  5.108 10.351 /
\plot  5.108 10.351  5.099 10.291 /
\plot  5.099 10.291  5.093 10.243 /
\plot  5.093 10.243  5.086 10.204 /
\plot  5.086 10.204  5.082 10.181 /
\plot  5.082 10.181  5.080 10.166 /
\putrule from  5.080 10.166 to  5.080 10.160
}%
%
%
\linethickness= 0.500pt
\setplotsymbol ({\thinlinefont .})
{\color[rgb]{0,0,0}\plot  5.112 11.430  6.636  9.874 /
}%
%
%
\linethickness= 0.500pt
\setplotsymbol ({\thinlinefont .})
{\color[rgb]{0,0,0}\plot  5.112 12.732  7.398 10.382 /
}%
%
%
\linethickness= 0.500pt
\setplotsymbol ({\thinlinefont .})
{\color[rgb]{0,0,0}\plot  5.620 13.399  8.064 11.081 /
}%
%
%
\linethickness= 0.500pt
\setplotsymbol ({\thinlinefont .})
{\color[rgb]{0,0,0}\plot  6.223 14.097  8.604 11.811 /
}%
%
%
\linethickness= 0.500pt
\setplotsymbol ({\thinlinefont .})
{\color[rgb]{0,0,0}\plot  6.890 14.637  9.049 12.541 /
}%
%
%
\linethickness= 0.500pt
\setplotsymbol ({\thinlinefont .})
{\color[rgb]{0,0,0}\plot  7.715 15.081  9.652 13.272 /
}%
%
%
\linethickness= 0.500pt
\setplotsymbol ({\thinlinefont .})
{\color[rgb]{0,0,0}\plot 12.522 15.369 12.975 14.956 /
}%
%
%
\linethickness= 0.500pt
\setplotsymbol ({\thinlinefont .})
{\color[rgb]{0,0,0}\plot 13.193 15.903 13.942 15.153 /
}%
%
%
\linethickness= 0.500pt
\setplotsymbol ({\thinlinefont .})
{\color[rgb]{0,0,0}\plot 13.843 16.612 15.145 15.291 /
}%
%
%
\linethickness= 0.500pt
\setplotsymbol ({\thinlinefont .})
{\color[rgb]{0,0,0}\plot 14.417 17.185 15.204 16.514 /
}%
%
%
\linethickness= 0.500pt
\setplotsymbol ({\thinlinefont .})
{\color[rgb]{0,0,0}\plot  9.131 15.035 10.234 13.851 /
}%
%
%
\linethickness= 0.500pt
\setplotsymbol ({\thinlinefont .})
{\color[rgb]{0,0,0}\plot 10.433 14.916 11.024 14.385 /
}%
%
%
\linethickness= 0.500pt
\setplotsymbol ({\thinlinefont .})
{\color[rgb]{0,0,0}\plot 11.557 15.035 11.891 14.719 /
}%
%
%
\put{\SetFigFont{6}{7.2}{\rmdefault}{\mddefault}{\updefault}{\color[rgb]{0,0,0}$\Ga_u$}%
} [lB] at  8.350 15.240
%
%
\put{\SetFigFont{6}{7.2}{\rmdefault}{\mddefault}{\updefault}{\color[rgb]{0,0,0}$\Ga_d$}%
} [lB] at  9.652 12.764
%
%
\put{\SetFigFont{6}{7.2}{\rmdefault}{\mddefault}{\updefault}{\color[rgb]{0,0,0}$\tilde{S}_-$}%
} [lB] at  3.842 12.637
%
%
\put{\SetFigFont{6}{7.2}{\rmdefault}{\mddefault}{\updefault}{\color[rgb]{0,0,0}$\bar{N}_-$}%
} [lB] at  3.873 10.160
%
%
\put{\SetFigFont{6}{7.2}{\rmdefault}{\mddefault}{\updefault}{\color[rgb]{0,0,0}$\bar{N}_+$}%
} [lB] at 15.335 17.780
%
%
\put{\SetFigFont{6}{7.2}{\rmdefault}{\mddefault}{\updefault}{\color[rgb]{0,0,0}$\tilde{S}_+$}%
} [lB] at 15.526 15.304
%
%
\put{\SetFigFont{6}{7.2}{\rmdefault}{\mddefault}{\updefault}{\color[rgb]{0,0,0}$\cN$}%
} [lB] at  9.620 18.605
%
%
\put{\SetFigFont{6}{7.2}{\rmdefault}{\mddefault}{\updefault}{\color[rgb]{0,0,0}$\cP$}%
} [lB] at  7.271 12.478
%
%
\put{\SetFigFont{6}{7.2}{\rmdefault}{\mddefault}{\updefault}{\color[rgb]{0,0,0}$\cN$}%
} [lB] at 12.668 12.827
\linethickness=0pt
\putrectangle corners at  3.810 21.660 and 15.558  6.280
\endpicture}
\qquad
%
%
\font\thinlinefont=cmr5
\begingroup\makeatletter\ifx\SetFigFont\undefined%
\gdef\SetFigFont#1#2#3#4#5{%
  \reset@font\fontsize{#1}{#2pt}%
  \fontfamily{#3}\fontseries{#4}\fontshape{#5}%
  \selectfont}%
\fi\endgroup%
\mbox{\beginpicture
\setcoordinatesystem units <0.50000cm,0.50000cm>
\unitlength=0.50000cm
\linethickness=1pt
\setplotsymbol ({\makebox(0,0)[l]{\tencirc\symbol{'160}}})
\setshadesymbol ({\thinlinefont .})
\setlinear
%
%
\linethickness= 0.500pt
\setplotsymbol ({\thinlinefont .})
{\color[rgb]{0,0,0}\putrule from  5.080 16.510 to 15.240 16.510
}%
%
%
\linethickness= 0.500pt
\setplotsymbol ({\thinlinefont .})
{\color[rgb]{0,0,0}\putrule from  5.080 11.430 to 15.240 11.430
}%
%
%
\linethickness=1pt
\setplotsymbol ({\makebox(0,0)[l]{\tencirc\symbol{'160}}})
{\color[rgb]{0,0,0}\putrule from  5.080 21.590 to 15.240 21.590
}%
%
%
\linethickness=1pt
\setplotsymbol ({\makebox(0,0)[l]{\tencirc\symbol{'160}}})
{\color[rgb]{0,0,0}\putrule from 15.240 21.590 to 15.240  6.315
\putrule from 15.240  6.350 to  5.045  6.350
\putrule from  5.080  6.350 to  5.080 21.590
}%
%
%
\linethickness= 0.500pt
\setplotsymbol ({\thinlinefont .})
{\color[rgb]{0,0,0}\plot 13.970 16.192 14.446 15.875 /
%
%
\plot 14.200 15.963 14.446 15.875 14.270 16.069 /
}%
%
%
\linethickness= 0.500pt
\setplotsymbol ({\thinlinefont .})
{\color[rgb]{0,0,0}\putrule from  5.080 17.462 to  5.080 16.034
%
%
\plot  5.017 16.288  5.080 16.034  5.143 16.288 /
}%
%
%
\linethickness= 0.500pt
\setplotsymbol ({\thinlinefont .})
{\color[rgb]{0,0,0}\putrule from 15.240 12.383 to 15.240 10.795
%
%
\plot 15.176 11.049 15.240 10.795 15.304 11.049 /
}%
%
%
\linethickness= 0.500pt
\setplotsymbol ({\thinlinefont .})
\setdashes < 0.1905cm>
{\color[rgb]{0,0,0}\plot  5.080 12.700  5.084 12.706 /
\plot  5.084 12.706  5.091 12.719 /
\plot  5.091 12.719  5.105 12.742 /
\plot  5.105 12.742  5.124 12.776 /
\plot  5.124 12.776  5.150 12.821 /
\plot  5.150 12.821  5.179 12.874 /
\plot  5.179 12.874  5.213 12.933 /
\plot  5.213 12.933  5.249 12.994 /
\plot  5.249 12.994  5.283 13.056 /
\plot  5.283 13.056  5.317 13.115 /
\plot  5.317 13.115  5.349 13.172 /
\plot  5.349 13.172  5.376 13.227 /
\plot  5.376 13.227  5.404 13.278 /
\plot  5.404 13.278  5.429 13.327 /
\plot  5.429 13.327  5.450 13.373 /
\plot  5.450 13.373  5.472 13.418 /
\plot  5.472 13.418  5.493 13.460 /
\plot  5.493 13.460  5.512 13.504 /
\plot  5.512 13.504  5.531 13.547 /
\plot  5.531 13.547  5.546 13.587 /
\plot  5.546 13.587  5.563 13.627 /
\plot  5.563 13.627  5.580 13.669 /
\plot  5.580 13.669  5.594 13.714 /
\plot  5.594 13.714  5.611 13.758 /
\plot  5.611 13.758  5.628 13.807 /
\plot  5.628 13.807  5.645 13.856 /
\plot  5.645 13.856  5.662 13.904 /
\plot  5.662 13.904  5.679 13.957 /
\plot  5.679 13.957  5.698 14.010 /
\plot  5.698 14.010  5.715 14.063 /
\plot  5.715 14.063  5.732 14.118 /
\plot  5.732 14.118  5.751 14.173 /
\plot  5.751 14.173  5.768 14.228 /
\plot  5.768 14.228  5.785 14.283 /
\plot  5.785 14.283  5.802 14.338 /
\plot  5.802 14.338  5.819 14.393 /
\plot  5.819 14.393  5.836 14.446 /
\plot  5.836 14.446  5.850 14.499 /
\plot  5.850 14.499  5.867 14.552 /
\plot  5.867 14.552  5.884 14.605 /
\plot  5.884 14.605  5.901 14.658 /
\plot  5.901 14.658  5.916 14.707 /
\plot  5.916 14.707  5.931 14.755 /
\plot  5.931 14.755  5.948 14.806 /
\plot  5.948 14.806  5.965 14.859 /
\plot  5.965 14.859  5.982 14.912 /
\plot  5.982 14.912  5.999 14.967 /
\plot  5.999 14.967  6.016 15.022 /
\plot  6.016 15.022  6.035 15.077 /
\plot  6.035 15.077  6.054 15.134 /
\plot  6.054 15.134  6.073 15.193 /
\plot  6.073 15.193  6.092 15.251 /
\plot  6.092 15.251  6.113 15.308 /
\plot  6.113 15.308  6.132 15.365 /
\plot  6.132 15.365  6.151 15.422 /
\plot  6.151 15.422  6.170 15.477 /
\plot  6.170 15.477  6.189 15.532 /
\plot  6.189 15.532  6.208 15.585 /
\plot  6.208 15.585  6.225 15.636 /
\plot  6.225 15.636  6.242 15.685 /
\plot  6.242 15.685  6.259 15.731 /
\plot  6.259 15.731  6.276 15.776 /
\plot  6.276 15.776  6.293 15.820 /
\plot  6.293 15.820  6.308 15.862 /
\plot  6.308 15.862  6.325 15.900 /
\plot  6.325 15.900  6.342 15.947 /
\plot  6.342 15.947  6.361 15.994 /
\plot  6.361 15.994  6.378 16.038 /
\plot  6.378 16.038  6.397 16.082 /
\plot  6.397 16.082  6.416 16.125 /
\plot  6.416 16.125  6.435 16.169 /
\plot  6.435 16.169  6.456 16.212 /
\plot  6.456 16.212  6.477 16.256 /
\plot  6.477 16.256  6.500 16.298 /
\plot  6.500 16.298  6.524 16.341 /
\plot  6.524 16.341  6.547 16.381 /
\plot  6.547 16.381  6.572 16.423 /
\plot  6.572 16.423  6.598 16.463 /
\plot  6.598 16.463  6.625 16.502 /
\plot  6.625 16.502  6.651 16.540 /
\plot  6.651 16.540  6.680 16.578 /
\plot  6.680 16.578  6.708 16.614 /
\plot  6.708 16.614  6.737 16.650 /
\plot  6.737 16.650  6.767 16.686 /
\plot  6.767 16.686  6.801 16.722 /
\plot  6.801 16.722  6.828 16.751 /
\plot  6.828 16.751  6.858 16.783 /
\plot  6.858 16.783  6.888 16.815 /
\plot  6.888 16.815  6.919 16.849 /
\plot  6.919 16.849  6.955 16.883 /
\plot  6.955 16.883  6.991 16.919 /
\plot  6.991 16.919  7.027 16.954 /
\plot  7.027 16.954  7.068 16.990 /
\plot  7.068 16.990  7.108 17.029 /
\plot  7.108 17.029  7.148 17.067 /
\plot  7.148 17.067  7.190 17.107 /
\plot  7.190 17.107  7.235 17.145 /
\plot  7.235 17.145  7.277 17.183 /
\plot  7.277 17.183  7.322 17.223 /
\plot  7.322 17.223  7.366 17.261 /
\plot  7.366 17.261  7.408 17.300 /
\plot  7.408 17.300  7.451 17.336 /
\plot  7.451 17.336  7.493 17.371 /
\plot  7.493 17.371  7.533 17.407 /
\plot  7.533 17.407  7.573 17.441 /
\plot  7.573 17.441  7.614 17.475 /
\plot  7.614 17.475  7.652 17.507 /
\plot  7.652 17.507  7.690 17.539 /
\plot  7.690 17.539  7.726 17.568 /
\plot  7.726 17.568  7.766 17.602 /
\plot  7.766 17.602  7.804 17.634 /
\plot  7.804 17.634  7.844 17.666 /
\plot  7.844 17.666  7.885 17.697 /
\plot  7.885 17.697  7.925 17.731 /
\plot  7.925 17.731  7.965 17.763 /
\plot  7.965 17.763  8.007 17.795 /
\plot  8.007 17.795  8.050 17.827 /
\plot  8.050 17.827  8.094 17.856 /
\plot  8.094 17.856  8.139 17.888 /
\plot  8.139 17.888  8.183 17.918 /
\plot  8.183 17.918  8.227 17.945 /
\plot  8.227 17.945  8.274 17.973 /
\plot  8.274 17.973  8.319 17.998 /
\plot  8.319 17.998  8.363 18.023 /
\plot  8.363 18.023  8.410 18.047 /
\plot  8.410 18.047  8.454 18.068 /
\plot  8.454 18.068  8.498 18.087 /
\plot  8.498 18.087  8.543 18.104 /
\plot  8.543 18.104  8.587 18.121 /
\plot  8.587 18.121  8.632 18.136 /
\plot  8.632 18.136  8.678 18.150 /
\plot  8.678 18.150  8.725 18.163 /
\plot  8.725 18.163  8.774 18.174 /
\plot  8.774 18.174  8.824 18.184 /
\plot  8.824 18.184  8.875 18.195 /
\plot  8.875 18.195  8.930 18.203 /
\plot  8.930 18.203  8.985 18.210 /
\plot  8.985 18.210  9.042 18.216 /
\plot  9.042 18.216  9.100 18.222 /
\plot  9.100 18.222  9.159 18.227 /
\plot  9.159 18.227  9.220 18.229 /
\plot  9.220 18.229  9.279 18.231 /
\plot  9.279 18.231  9.339 18.231 /
\plot  9.339 18.231  9.398 18.231 /
\plot  9.398 18.231  9.455 18.229 /
\plot  9.455 18.229  9.512 18.224 /
\plot  9.512 18.224  9.565 18.220 /
\plot  9.565 18.220  9.618 18.216 /
\plot  9.618 18.216  9.667 18.210 /
\plot  9.667 18.210  9.713 18.203 /
\plot  9.713 18.203  9.760 18.195 /
\plot  9.760 18.195  9.802 18.186 /
\plot  9.802 18.186  9.842 18.176 /
\plot  9.842 18.176  9.889 18.163 /
\plot  9.889 18.163  9.936 18.150 /
\plot  9.936 18.150  9.980 18.133 /
\plot  9.980 18.133 10.022 18.117 /
\plot 10.022 18.117 10.065 18.095 /
\plot 10.065 18.095 10.105 18.074 /
\plot 10.105 18.074 10.147 18.051 /
\plot 10.147 18.051 10.185 18.026 /
\plot 10.185 18.026 10.226 18.000 /
\plot 10.226 18.000 10.262 17.973 /
\plot 10.262 17.973 10.300 17.943 /
\plot 10.300 17.943 10.334 17.913 /
\plot 10.334 17.913 10.370 17.884 /
\plot 10.370 17.884 10.401 17.852 /
\plot 10.401 17.852 10.435 17.822 /
\plot 10.435 17.822 10.467 17.791 /
\plot 10.467 17.791 10.499 17.759 /
\plot 10.499 17.759 10.530 17.727 /
\plot 10.530 17.727 10.560 17.697 /
\plot 10.560 17.697 10.590 17.668 /
\plot 10.590 17.668 10.621 17.636 /
\plot 10.621 17.636 10.655 17.602 /
\plot 10.655 17.602 10.689 17.568 /
\plot 10.689 17.568 10.725 17.532 /
\plot 10.725 17.532 10.761 17.496 /
\plot 10.761 17.496 10.799 17.458 /
\plot 10.799 17.458 10.835 17.418 /
\plot 10.835 17.418 10.875 17.380 /
\plot 10.875 17.380 10.914 17.340 /
\plot 10.914 17.340 10.950 17.302 /
\plot 10.950 17.302 10.988 17.261 /
\plot 10.988 17.261 11.024 17.223 /
\plot 11.024 17.223 11.060 17.185 /
\plot 11.060 17.185 11.093 17.149 /
\plot 11.093 17.149 11.127 17.113 /
\plot 11.127 17.113 11.159 17.079 /
\plot 11.159 17.079 11.189 17.046 /
\plot 11.189 17.046 11.218 17.012 /
\plot 11.218 17.012 11.244 16.982 /
\plot 11.244 16.982 11.271 16.952 /
\plot 11.271 16.952 11.297 16.923 /
\plot 11.297 16.923 11.322 16.891 /
\plot 11.322 16.891 11.347 16.859 /
\plot 11.347 16.859 11.373 16.828 /
\plot 11.373 16.828 11.398 16.792 /
\plot 11.398 16.792 11.424 16.756 /
\plot 11.424 16.756 11.449 16.720 /
\plot 11.449 16.720 11.474 16.679 /
\plot 11.474 16.679 11.498 16.639 /
\plot 11.498 16.639 11.523 16.597 /
\plot 11.523 16.597 11.546 16.554 /
\plot 11.546 16.554 11.568 16.510 /
\plot 11.568 16.510 11.591 16.466 /
\plot 11.591 16.466 11.612 16.419 /
\plot 11.612 16.419 11.631 16.370 /
\plot 11.631 16.370 11.650 16.322 /
\plot 11.650 16.322 11.669 16.273 /
\plot 11.669 16.273 11.686 16.220 /
\plot 11.686 16.220 11.705 16.167 /
\plot 11.705 16.167 11.722 16.112 /
\plot 11.722 16.112 11.733 16.072 /
\plot 11.733 16.072 11.745 16.027 /
\plot 11.745 16.027 11.756 15.983 /
\plot 11.756 15.983 11.769 15.934 /
\plot 11.769 15.934 11.779 15.886 /
\plot 11.779 15.886 11.792 15.833 /
\plot 11.792 15.833 11.803 15.778 /
\plot 11.803 15.778 11.813 15.720 /
\plot 11.813 15.720 11.824 15.661 /
\plot 11.824 15.661 11.834 15.598 /
\plot 11.834 15.598 11.843 15.534 /
\plot 11.843 15.534 11.853 15.466 /
\plot 11.853 15.466 11.862 15.399 /
\plot 11.862 15.399 11.870 15.327 /
\plot 11.870 15.327 11.877 15.253 /
\plot 11.877 15.253 11.883 15.179 /
\plot 11.883 15.179 11.889 15.102 /
\plot 11.889 15.102 11.894 15.024 /
\plot 11.894 15.024 11.898 14.946 /
\plot 11.898 14.946 11.902 14.867 /
\plot 11.902 14.867 11.904 14.787 /
\plot 11.904 14.787 11.904 14.707 /
\plot 11.904 14.707 11.906 14.626 /
\plot 11.906 14.626 11.904 14.544 /
\plot 11.904 14.544 11.904 14.463 /
\plot 11.904 14.463 11.900 14.381 /
\plot 11.900 14.381 11.898 14.300 /
\plot 11.898 14.300 11.891 14.218 /
\plot 11.891 14.218 11.887 14.133 /
\plot 11.887 14.133 11.881 14.048 /
\plot 11.881 14.048 11.874 13.983 /
\plot 11.874 13.983 11.866 13.913 /
\plot 11.866 13.913 11.860 13.843 /
\plot 11.860 13.843 11.851 13.773 /
\plot 11.851 13.773 11.843 13.699 /
\plot 11.843 13.699 11.832 13.625 /
\plot 11.832 13.625 11.822 13.551 /
\plot 11.822 13.551 11.811 13.473 /
\plot 11.811 13.473 11.798 13.394 /
\plot 11.798 13.394 11.786 13.314 /
\plot 11.786 13.314 11.771 13.231 /
\plot 11.771 13.231 11.756 13.149 /
\plot 11.756 13.149 11.741 13.064 /
\plot 11.741 13.064 11.724 12.979 /
\plot 11.724 12.979 11.707 12.893 /
\plot 11.707 12.893 11.688 12.806 /
\plot 11.688 12.806 11.671 12.717 /
\plot 11.671 12.717 11.652 12.630 /
\plot 11.652 12.630 11.631 12.541 /
\plot 11.631 12.541 11.612 12.452 /
\plot 11.612 12.452 11.591 12.366 /
\plot 11.591 12.366 11.568 12.277 /
\plot 11.568 12.277 11.546 12.190 /
\plot 11.546 12.190 11.525 12.103 /
\plot 11.525 12.103 11.502 12.018 /
\plot 11.502 12.018 11.479 11.934 /
\plot 11.479 11.934 11.455 11.851 /
\plot 11.455 11.851 11.432 11.769 /
\plot 11.432 11.769 11.409 11.688 /
\plot 11.409 11.688 11.386 11.610 /
\plot 11.386 11.610 11.362 11.532 /
\plot 11.362 11.532 11.337 11.458 /
\plot 11.337 11.458 11.314 11.383 /
\plot 11.314 11.383 11.290 11.309 /
\plot 11.290 11.309 11.265 11.240 /
\plot 11.265 11.240 11.242 11.170 /
\plot 11.242 11.170 11.216 11.100 /
\plot 11.216 11.100 11.193 11.032 /
\plot 11.193 11.032 11.165 10.962 /
\plot 11.165 10.962 11.138 10.892 /
\plot 11.138 10.892 11.110 10.825 /
\plot 11.110 10.825 11.083 10.757 /
\plot 11.083 10.757 11.053 10.689 /
\plot 11.053 10.689 11.024 10.621 /
\plot 11.024 10.621 10.992 10.554 /
\plot 10.992 10.554 10.958 10.488 /
\plot 10.958 10.488 10.924 10.422 /
\plot 10.924 10.422 10.888 10.357 /
\plot 10.888 10.357 10.850 10.291 /
\plot 10.850 10.291 10.812 10.228 /
\plot 10.812 10.228 10.770 10.164 /
\plot 10.770 10.164 10.727 10.103 /
\plot 10.727 10.103 10.683 10.041 /
\plot 10.683 10.041 10.638  9.982 /
\plot 10.638  9.982 10.590  9.925 /
\plot 10.590  9.925 10.541  9.870 /
\plot 10.541  9.870 10.490  9.815 /
\plot 10.490  9.815 10.437  9.764 /
\plot 10.437  9.764 10.382  9.713 /
\plot 10.382  9.713 10.327  9.667 /
\plot 10.327  9.667 10.270  9.620 /
\plot 10.270  9.620 10.211  9.578 /
\plot 10.211  9.578 10.152  9.538 /
\plot 10.152  9.538 10.090  9.500 /
\plot 10.090  9.500 10.027  9.466 /
\plot 10.027  9.466  9.961  9.432 /
\plot  9.961  9.432  9.895  9.402 /
\plot  9.895  9.402  9.830  9.375 /
\plot  9.830  9.375  9.760  9.349 /
\plot  9.760  9.349  9.690  9.326 /
\plot  9.690  9.326  9.618  9.307 /
\plot  9.618  9.307  9.546  9.290 /
\plot  9.546  9.290  9.470  9.273 /
\plot  9.470  9.273  9.394  9.260 /
\plot  9.394  9.260  9.328  9.252 /
\plot  9.328  9.252  9.260  9.243 /
\plot  9.260  9.243  9.191  9.237 /
\plot  9.191  9.237  9.119  9.233 /
\plot  9.119  9.233  9.045  9.231 /
\plot  9.045  9.231  8.966  9.229 /
\plot  8.966  9.229  8.886  9.231 /
\plot  8.886  9.231  8.801  9.233 /
\plot  8.801  9.233  8.712  9.237 /
\plot  8.712  9.237  8.621  9.241 /
\plot  8.621  9.241  8.524  9.250 /
\plot  8.524  9.250  8.422  9.258 /
\plot  8.422  9.258  8.314  9.271 /
\plot  8.314  9.271  8.204  9.284 /
\plot  8.204  9.284  8.086  9.299 /
\plot  8.086  9.299  7.963  9.315 /
\plot  7.963  9.315  7.836  9.332 /
\plot  7.836  9.332  7.703  9.354 /
\plot  7.703  9.354  7.563  9.375 /
\plot  7.563  9.375  7.419  9.398 /
\plot  7.419  9.398  7.271  9.423 /
\plot  7.271  9.423  7.120  9.449 /
\plot  7.120  9.449  6.966  9.476 /
\plot  6.966  9.476  6.807  9.506 /
\plot  6.807  9.506  6.651  9.533 /
\plot  6.651  9.533  6.494  9.563 /
\plot  6.494  9.563  6.337  9.593 /
\plot  6.337  9.593  6.185  9.622 /
\plot  6.185  9.622  6.039  9.650 /
\plot  6.039  9.650  5.897  9.677 /
\plot  5.897  9.677  5.766  9.703 /
\plot  5.766  9.703  5.643  9.728 /
\plot  5.643  9.728  5.533  9.749 /
\plot  5.533  9.749  5.433  9.771 /
\plot  5.433  9.771  5.347  9.787 /
\plot  5.347  9.787  5.275  9.802 /
\plot  5.275  9.802  5.215  9.815 /
\plot  5.215  9.815  5.167  9.823 /
\plot  5.167  9.823  5.133  9.832 /
\plot  5.133  9.832  5.108  9.836 /
\plot  5.108  9.836  5.093  9.840 /
\plot  5.093  9.840  5.084  9.842 /
\plot  5.084  9.842  5.080  9.842 /
}%
%
%
\linethickness= 0.500pt
\setplotsymbol ({\thinlinefont .})
{\color[rgb]{0,0,0}\plot 15.240 15.240 15.234 15.244 /
\plot 15.234 15.244 15.221 15.255 /
\plot 15.221 15.255 15.196 15.274 /
\plot 15.196 15.274 15.164 15.301 /
\plot 15.164 15.301 15.121 15.335 /
\plot 15.121 15.335 15.073 15.373 /
\plot 15.073 15.373 15.022 15.416 /
\plot 15.022 15.416 14.973 15.456 /
\plot 14.973 15.456 14.925 15.498 /
\plot 14.925 15.498 14.878 15.536 /
\plot 14.878 15.536 14.838 15.574 /
\plot 14.838 15.574 14.800 15.608 /
\plot 14.800 15.608 14.768 15.642 /
\plot 14.768 15.642 14.736 15.674 /
\plot 14.736 15.674 14.709 15.706 /
\plot 14.709 15.706 14.683 15.737 /
\plot 14.683 15.737 14.658 15.769 /
\plot 14.658 15.769 14.635 15.803 /
\plot 14.635 15.803 14.611 15.837 /
\plot 14.611 15.837 14.588 15.873 /
\plot 14.588 15.873 14.565 15.909 /
\plot 14.565 15.909 14.544 15.949 /
\plot 14.544 15.949 14.520 15.989 /
\plot 14.520 15.989 14.499 16.032 /
\plot 14.499 16.032 14.478 16.076 /
\plot 14.478 16.076 14.457 16.121 /
\plot 14.457 16.121 14.438 16.165 /
\plot 14.438 16.165 14.419 16.212 /
\plot 14.419 16.212 14.400 16.256 /
\plot 14.400 16.256 14.383 16.300 /
\plot 14.383 16.300 14.368 16.345 /
\plot 14.368 16.345 14.353 16.387 /
\plot 14.353 16.387 14.340 16.430 /
\plot 14.340 16.430 14.326 16.470 /
\plot 14.326 16.470 14.315 16.510 /
\plot 14.315 16.510 14.302 16.550 /
\plot 14.302 16.550 14.290 16.593 /
\plot 14.290 16.593 14.279 16.633 /
\plot 14.279 16.633 14.268 16.677 /
\plot 14.268 16.677 14.258 16.720 /
\plot 14.258 16.720 14.247 16.766 /
\plot 14.247 16.766 14.237 16.811 /
\plot 14.237 16.811 14.228 16.857 /
\plot 14.228 16.857 14.220 16.904 /
\plot 14.220 16.904 14.211 16.950 /
\plot 14.211 16.950 14.205 16.995 /
\plot 14.205 16.995 14.199 17.039 /
\plot 14.199 17.039 14.192 17.081 /
\plot 14.192 17.081 14.188 17.124 /
\plot 14.188 17.124 14.186 17.164 /
\plot 14.186 17.164 14.184 17.202 /
\plot 14.184 17.202 14.182 17.240 /
\plot 14.182 17.240 14.182 17.276 /
\plot 14.182 17.276 14.182 17.314 /
\plot 14.182 17.314 14.184 17.350 /
\plot 14.184 17.350 14.186 17.386 /
\plot 14.186 17.386 14.190 17.422 /
\plot 14.190 17.422 14.196 17.460 /
\plot 14.196 17.460 14.203 17.496 /
\plot 14.203 17.496 14.211 17.534 /
\plot 14.211 17.534 14.222 17.570 /
\plot 14.222 17.570 14.235 17.606 /
\plot 14.235 17.606 14.247 17.642 /
\plot 14.247 17.642 14.264 17.674 /
\plot 14.264 17.674 14.281 17.708 /
\plot 14.281 17.708 14.300 17.738 /
\plot 14.300 17.738 14.321 17.765 /
\plot 14.321 17.765 14.343 17.793 /
\plot 14.343 17.793 14.368 17.816 /
\plot 14.368 17.816 14.393 17.839 /
\plot 14.393 17.839 14.421 17.858 /
\plot 14.421 17.858 14.446 17.877 /
\plot 14.446 17.877 14.476 17.894 /
\plot 14.476 17.894 14.506 17.909 /
\plot 14.506 17.909 14.541 17.924 /
\plot 14.541 17.924 14.580 17.939 /
\plot 14.580 17.939 14.624 17.954 /
\plot 14.624 17.954 14.673 17.968 /
\plot 14.673 17.968 14.726 17.983 /
\plot 14.726 17.983 14.785 18.000 /
\plot 14.785 18.000 14.848 18.015 /
\plot 14.848 18.015 14.914 18.030 /
\plot 14.914 18.030 14.980 18.045 /
\plot 14.980 18.045 15.045 18.057 /
\plot 15.045 18.057 15.102 18.070 /
\plot 15.102 18.070 15.153 18.081 /
\plot 15.153 18.081 15.191 18.089 /
\plot 15.191 18.089 15.219 18.093 /
\plot 15.219 18.093 15.234 18.095 /
\plot 15.234 18.095 15.240 18.098 /
}%
%
%
\linethickness= 0.500pt
\setplotsymbol ({\thinlinefont .})
\setsolid
{\color[rgb]{0,0,0}\plot  5.080 12.700  5.082 12.706 /
\plot  5.082 12.706  5.088 12.717 /
\plot  5.088 12.717  5.101 12.740 /
\plot  5.101 12.740  5.118 12.774 /
\plot  5.118 12.774  5.141 12.819 /
\plot  5.141 12.819  5.171 12.876 /
\plot  5.171 12.876  5.207 12.941 /
\plot  5.207 12.941  5.245 13.013 /
\plot  5.245 13.013  5.285 13.089 /
\plot  5.285 13.089  5.328 13.168 /
\plot  5.328 13.168  5.370 13.244 /
\plot  5.370 13.244  5.412 13.320 /
\plot  5.412 13.320  5.453 13.392 /
\plot  5.453 13.392  5.491 13.460 /
\plot  5.491 13.460  5.529 13.523 /
\plot  5.529 13.523  5.565 13.583 /
\plot  5.565 13.583  5.599 13.638 /
\plot  5.599 13.638  5.632 13.691 /
\plot  5.632 13.691  5.664 13.739 /
\plot  5.664 13.739  5.696 13.786 /
\plot  5.696 13.786  5.730 13.830 /
\plot  5.730 13.830  5.762 13.875 /
\plot  5.762 13.875  5.795 13.917 /
\plot  5.795 13.917  5.825 13.955 /
\plot  5.825 13.955  5.857 13.995 /
\plot  5.857 13.995  5.891 14.034 /
\plot  5.891 14.034  5.925 14.074 /
\plot  5.925 14.074  5.961 14.112 /
\plot  5.961 14.112  5.997 14.152 /
\plot  5.997 14.152  6.035 14.192 /
\plot  6.035 14.192  6.073 14.232 /
\plot  6.073 14.232  6.113 14.275 /
\plot  6.113 14.275  6.155 14.315 /
\plot  6.155 14.315  6.198 14.357 /
\plot  6.198 14.357  6.242 14.398 /
\plot  6.242 14.398  6.284 14.440 /
\plot  6.284 14.440  6.329 14.480 /
\plot  6.329 14.480  6.375 14.520 /
\plot  6.375 14.520  6.420 14.561 /
\plot  6.420 14.561  6.464 14.601 /
\plot  6.464 14.601  6.509 14.639 /
\plot  6.509 14.639  6.553 14.677 /
\plot  6.553 14.677  6.598 14.713 /
\plot  6.598 14.713  6.642 14.751 /
\plot  6.642 14.751  6.684 14.785 /
\plot  6.684 14.785  6.727 14.821 /
\plot  6.727 14.821  6.769 14.855 /
\plot  6.769 14.855  6.811 14.889 /
\plot  6.811 14.889  6.854 14.922 /
\plot  6.854 14.922  6.892 14.954 /
\plot  6.892 14.954  6.932 14.986 /
\plot  6.932 14.986  6.972 15.018 /
\plot  6.972 15.018  7.013 15.050 /
\plot  7.013 15.050  7.055 15.081 /
\plot  7.055 15.081  7.097 15.115 /
\plot  7.097 15.115  7.142 15.149 /
\plot  7.142 15.149  7.188 15.183 /
\plot  7.188 15.183  7.235 15.219 /
\plot  7.235 15.219  7.283 15.255 /
\plot  7.283 15.255  7.332 15.291 /
\plot  7.332 15.291  7.383 15.329 /
\plot  7.383 15.329  7.434 15.365 /
\plot  7.434 15.365  7.487 15.403 /
\plot  7.487 15.403  7.540 15.439 /
\plot  7.540 15.439  7.595 15.477 /
\plot  7.595 15.477  7.648 15.515 /
\plot  7.648 15.515  7.703 15.551 /
\plot  7.703 15.551  7.758 15.589 /
\plot  7.758 15.589  7.813 15.625 /
\plot  7.813 15.625  7.868 15.661 /
\plot  7.868 15.661  7.923 15.697 /
\plot  7.923 15.697  7.978 15.731 /
\plot  7.978 15.731  8.033 15.765 /
\plot  8.033 15.765  8.088 15.801 /
\plot  8.088 15.801  8.143 15.835 /
\plot  8.143 15.835  8.198 15.869 /
\plot  8.198 15.869  8.255 15.900 /
\plot  8.255 15.900  8.306 15.932 /
\plot  8.306 15.932  8.357 15.962 /
\plot  8.357 15.962  8.410 15.991 /
\plot  8.410 15.991  8.465 16.021 /
\plot  8.465 16.021  8.520 16.053 /
\plot  8.520 16.053  8.577 16.085 /
\plot  8.577 16.085  8.636 16.116 /
\plot  8.636 16.116  8.697 16.148 /
\plot  8.697 16.148  8.759 16.180 /
\plot  8.759 16.180  8.822 16.212 /
\plot  8.822 16.212  8.888 16.245 /
\plot  8.888 16.245  8.956 16.277 /
\plot  8.956 16.277  9.023 16.309 /
\plot  9.023 16.309  9.091 16.341 /
\plot  9.091 16.341  9.163 16.375 /
\plot  9.163 16.375  9.233 16.404 /
\plot  9.233 16.404  9.305 16.436 /
\plot  9.305 16.436  9.377 16.466 /
\plot  9.377 16.466  9.449 16.495 /
\plot  9.449 16.495  9.521 16.525 /
\plot  9.521 16.525  9.593 16.552 /
\plot  9.593 16.552  9.665 16.578 /
\plot  9.665 16.578  9.735 16.603 /
\plot  9.735 16.603  9.807 16.626 /
\plot  9.807 16.626  9.874 16.650 /
\plot  9.874 16.650  9.944 16.671 /
\plot  9.944 16.671 10.012 16.692 /
\plot 10.012 16.692 10.080 16.711 /
\plot 10.080 16.711 10.147 16.728 /
\plot 10.147 16.728 10.213 16.745 /
\plot 10.213 16.745 10.279 16.760 /
\plot 10.279 16.760 10.346 16.775 /
\plot 10.346 16.775 10.412 16.787 /
\plot 10.412 16.787 10.478 16.800 /
\plot 10.478 16.800 10.545 16.811 /
\plot 10.545 16.811 10.613 16.821 /
\plot 10.613 16.821 10.681 16.830 /
\plot 10.681 16.830 10.751 16.836 /
\plot 10.751 16.836 10.823 16.842 /
\plot 10.823 16.842 10.894 16.849 /
\plot 10.894 16.849 10.966 16.851 /
\plot 10.966 16.851 11.041 16.855 /
\putrule from 11.041 16.855 to 11.115 16.855
\putrule from 11.115 16.855 to 11.191 16.855
\plot 11.191 16.855 11.267 16.853 /
\plot 11.267 16.853 11.343 16.849 /
\plot 11.343 16.849 11.419 16.844 /
\plot 11.419 16.844 11.496 16.838 /
\plot 11.496 16.838 11.570 16.830 /
\plot 11.570 16.830 11.646 16.819 /
\plot 11.646 16.819 11.720 16.808 /
\plot 11.720 16.808 11.794 16.796 /
\plot 11.794 16.796 11.866 16.781 /
\plot 11.866 16.781 11.938 16.766 /
\plot 11.938 16.766 12.008 16.749 /
\plot 12.008 16.749 12.076 16.730 /
\plot 12.076 16.730 12.141 16.711 /
\plot 12.141 16.711 12.207 16.690 /
\plot 12.207 16.690 12.270 16.667 /
\plot 12.270 16.667 12.332 16.643 /
\plot 12.332 16.643 12.393 16.618 /
\plot 12.393 16.618 12.452 16.593 /
\plot 12.452 16.593 12.509 16.565 /
\plot 12.509 16.565 12.569 16.535 /
\plot 12.569 16.535 12.624 16.506 /
\plot 12.624 16.506 12.681 16.474 /
\plot 12.681 16.474 12.736 16.440 /
\plot 12.736 16.440 12.791 16.404 /
\plot 12.791 16.404 12.848 16.366 /
\plot 12.848 16.366 12.903 16.328 /
\plot 12.903 16.328 12.958 16.286 /
\plot 12.958 16.286 13.013 16.241 /
\plot 13.013 16.241 13.070 16.195 /
\plot 13.070 16.195 13.125 16.148 /
\plot 13.125 16.148 13.180 16.097 /
\plot 13.180 16.097 13.236 16.044 /
\plot 13.236 16.044 13.291 15.991 /
\plot 13.291 15.991 13.343 15.936 /
\plot 13.343 15.936 13.396 15.879 /
\plot 13.396 15.879 13.449 15.820 /
\plot 13.449 15.820 13.502 15.761 /
\plot 13.502 15.761 13.551 15.699 /
\plot 13.551 15.699 13.602 15.638 /
\plot 13.602 15.638 13.648 15.574 /
\plot 13.648 15.574 13.695 15.513 /
\plot 13.695 15.513 13.739 15.450 /
\plot 13.739 15.450 13.784 15.386 /
\plot 13.784 15.386 13.824 15.323 /
\plot 13.824 15.323 13.864 15.259 /
\plot 13.864 15.259 13.902 15.196 /
\plot 13.902 15.196 13.940 15.132 /
\plot 13.940 15.132 13.974 15.069 /
\plot 13.974 15.069 14.008 15.007 /
\plot 14.008 15.007 14.042 14.944 /
\plot 14.042 14.944 14.072 14.880 /
\plot 14.072 14.880 14.103 14.817 /
\plot 14.103 14.817 14.129 14.760 /
\plot 14.129 14.760 14.154 14.702 /
\plot 14.154 14.702 14.177 14.643 /
\plot 14.177 14.643 14.203 14.584 /
\plot 14.203 14.584 14.226 14.522 /
\plot 14.226 14.522 14.249 14.461 /
\plot 14.249 14.461 14.273 14.398 /
\plot 14.273 14.398 14.294 14.332 /
\plot 14.294 14.332 14.317 14.264 /
\plot 14.317 14.264 14.338 14.196 /
\plot 14.338 14.196 14.359 14.127 /
\plot 14.359 14.127 14.381 14.055 /
\plot 14.381 14.055 14.402 13.981 /
\plot 14.402 13.981 14.423 13.904 /
\plot 14.423 13.904 14.442 13.826 /
\plot 14.442 13.826 14.461 13.748 /
\plot 14.461 13.748 14.482 13.667 /
\plot 14.482 13.667 14.499 13.587 /
\plot 14.499 13.587 14.518 13.504 /
\plot 14.518 13.504 14.535 13.420 /
\plot 14.535 13.420 14.554 13.335 /
\plot 14.554 13.335 14.571 13.250 /
\plot 14.571 13.250 14.586 13.164 /
\plot 14.586 13.164 14.603 13.079 /
\plot 14.603 13.079 14.618 12.992 /
\plot 14.618 12.992 14.633 12.903 /
\plot 14.633 12.903 14.647 12.816 /
\plot 14.647 12.816 14.662 12.730 /
\plot 14.662 12.730 14.675 12.641 /
\plot 14.675 12.641 14.690 12.552 /
\plot 14.690 12.552 14.702 12.463 /
\plot 14.702 12.463 14.715 12.374 /
\plot 14.715 12.374 14.728 12.285 /
\plot 14.728 12.285 14.738 12.196 /
\plot 14.738 12.196 14.751 12.105 /
\plot 14.751 12.105 14.764 12.012 /
\plot 14.764 12.012 14.774 11.936 /
\plot 14.774 11.936 14.783 11.858 /
\plot 14.783 11.858 14.793 11.779 /
\plot 14.793 11.779 14.804 11.697 /
\plot 14.804 11.697 14.815 11.614 /
\plot 14.815 11.614 14.825 11.529 /
\plot 14.825 11.529 14.834 11.441 /
\plot 14.834 11.441 14.846 11.350 /
\plot 14.846 11.350 14.857 11.256 /
\plot 14.857 11.256 14.867 11.157 /
\plot 14.867 11.157 14.878 11.055 /
\plot 14.878 11.055 14.891 10.950 /
\plot 14.891 10.950 14.903 10.839 /
\plot 14.903 10.839 14.916 10.723 /
\plot 14.916 10.723 14.929 10.605 /
\plot 14.929 10.605 14.942 10.480 /
\plot 14.942 10.480 14.956 10.348 /
\plot 14.956 10.348 14.971 10.215 /
\plot 14.971 10.215 14.986 10.075 /
\plot 14.986 10.075 15.001  9.931 /
\plot 15.001  9.931 15.016  9.783 /
\plot 15.016  9.783 15.033  9.631 /
\plot 15.033  9.631 15.050  9.478 /
\plot 15.050  9.478 15.064  9.322 /
\plot 15.064  9.322 15.081  9.165 /
\plot 15.081  9.165 15.098  9.011 /
\plot 15.098  9.011 15.113  8.856 /
\plot 15.113  8.856 15.128  8.706 /
\plot 15.128  8.706 15.143  8.560 /
\plot 15.143  8.560 15.157  8.422 /
\plot 15.157  8.422 15.172  8.293 /
\plot 15.172  8.293 15.183  8.172 /
\plot 15.183  8.172 15.196  8.064 /
\plot 15.196  8.064 15.204  7.967 /
\plot 15.204  7.967 15.212  7.882 /
\plot 15.212  7.882 15.221  7.811 /
\plot 15.221  7.811 15.227  7.751 /
\plot 15.227  7.751 15.232  7.705 /
\plot 15.232  7.705 15.236  7.671 /
\plot 15.236  7.671 15.238  7.648 /
\putrule from 15.238  7.648 to 15.238  7.631
\plot 15.238  7.631 15.240  7.624 /
\putrule from 15.240  7.624 to 15.240  7.620
}%
%
%
\linethickness= 0.500pt
\setplotsymbol ({\thinlinefont .})
{\color[rgb]{0,0,0}\plot 15.240 15.240 15.236 15.242 /
\plot 15.236 15.242 15.229 15.248 /
\plot 15.229 15.248 15.215 15.261 /
\plot 15.215 15.261 15.193 15.278 /
\plot 15.193 15.278 15.164 15.304 /
\plot 15.164 15.304 15.124 15.335 /
\plot 15.124 15.335 15.077 15.373 /
\plot 15.077 15.373 15.024 15.418 /
\plot 15.024 15.418 14.963 15.466 /
\plot 14.963 15.466 14.897 15.519 /
\plot 14.897 15.519 14.829 15.574 /
\plot 14.829 15.574 14.760 15.632 /
\plot 14.760 15.632 14.688 15.687 /
\plot 14.688 15.687 14.616 15.744 /
\plot 14.616 15.744 14.546 15.797 /
\plot 14.546 15.797 14.476 15.852 /
\plot 14.476 15.852 14.408 15.903 /
\plot 14.408 15.903 14.343 15.953 /
\plot 14.343 15.953 14.275 16.002 /
\plot 14.275 16.002 14.211 16.049 /
\plot 14.211 16.049 14.146 16.095 /
\plot 14.146 16.095 14.080 16.142 /
\plot 14.080 16.142 14.014 16.188 /
\plot 14.014 16.188 13.947 16.235 /
\plot 13.947 16.235 13.877 16.281 /
\plot 13.877 16.281 13.805 16.328 /
\plot 13.805 16.328 13.733 16.377 /
\plot 13.733 16.377 13.678 16.413 /
\plot 13.678 16.413 13.623 16.449 /
\plot 13.623 16.449 13.568 16.485 /
\plot 13.568 16.485 13.509 16.521 /
\plot 13.509 16.521 13.447 16.559 /
\plot 13.447 16.559 13.386 16.597 /
\plot 13.386 16.597 13.322 16.635 /
\plot 13.322 16.635 13.255 16.673 /
\plot 13.255 16.673 13.187 16.711 /
\plot 13.187 16.711 13.115 16.749 /
\plot 13.115 16.749 13.041 16.789 /
\plot 13.041 16.789 12.967 16.828 /
\plot 12.967 16.828 12.888 16.866 /
\plot 12.888 16.866 12.808 16.904 /
\plot 12.808 16.904 12.725 16.942 /
\plot 12.725 16.942 12.641 16.978 /
\plot 12.641 16.978 12.556 17.014 /
\plot 12.556 17.014 12.467 17.050 /
\plot 12.467 17.050 12.378 17.081 /
\plot 12.378 17.081 12.287 17.115 /
\plot 12.287 17.115 12.194 17.145 /
\plot 12.194 17.145 12.101 17.175 /
\plot 12.101 17.175 12.006 17.200 /
\plot 12.006 17.200 11.910 17.225 /
\plot 11.910 17.225 11.815 17.249 /
\plot 11.815 17.249 11.718 17.268 /
\plot 11.718 17.268 11.621 17.287 /
\plot 11.621 17.287 11.523 17.304 /
\plot 11.523 17.304 11.426 17.316 /
\plot 11.426 17.316 11.328 17.327 /
\plot 11.328 17.327 11.229 17.333 /
\plot 11.229 17.333 11.132 17.340 /
\plot 11.132 17.340 11.034 17.342 /
\putrule from 11.034 17.342 to 10.937 17.342
\plot 10.937 17.342 10.837 17.338 /
\plot 10.837 17.338 10.740 17.331 /
\plot 10.740 17.331 10.643 17.323 /
\plot 10.643 17.323 10.543 17.310 /
\plot 10.543 17.310 10.446 17.295 /
\plot 10.446 17.295 10.346 17.276 /
\plot 10.346 17.276 10.253 17.257 /
\plot 10.253 17.257 10.162 17.236 /
\plot 10.162 17.236 10.069 17.211 /
\plot 10.069 17.211  9.974 17.183 /
\plot  9.974 17.183  9.878 17.156 /
\plot  9.878 17.156  9.781 17.124 /
\plot  9.781 17.124  9.682 17.090 /
\plot  9.682 17.090  9.582 17.054 /
\plot  9.582 17.054  9.481 17.016 /
\plot  9.481 17.016  9.379 16.978 /
\plot  9.379 16.978  9.273 16.938 /
\plot  9.273 16.938  9.169 16.895 /
\plot  9.169 16.895  9.064 16.851 /
\plot  9.064 16.851  8.956 16.806 /
\plot  8.956 16.806  8.850 16.762 /
\plot  8.850 16.762  8.740 16.717 /
\plot  8.740 16.717  8.632 16.671 /
\plot  8.632 16.671  8.524 16.626 /
\plot  8.524 16.626  8.414 16.580 /
\plot  8.414 16.580  8.306 16.535 /
\plot  8.306 16.535  8.198 16.491 /
\plot  8.198 16.491  8.090 16.447 /
\plot  8.090 16.447  7.984 16.404 /
\plot  7.984 16.404  7.878 16.362 /
\plot  7.878 16.362  7.775 16.324 /
\plot  7.775 16.324  7.671 16.286 /
\plot  7.671 16.286  7.569 16.250 /
\plot  7.569 16.250  7.470 16.216 /
\plot  7.470 16.216  7.374 16.184 /
\plot  7.374 16.184  7.279 16.157 /
\plot  7.279 16.157  7.186 16.131 /
\plot  7.186 16.131  7.095 16.108 /
\plot  7.095 16.108  7.008 16.089 /
\plot  7.008 16.089  6.924 16.072 /
\plot  6.924 16.072  6.841 16.059 /
\plot  6.841 16.059  6.761 16.051 /
\plot  6.761 16.051  6.682 16.044 /
\plot  6.682 16.044  6.608 16.042 /
\plot  6.608 16.042  6.538 16.044 /
\plot  6.538 16.044  6.469 16.049 /
\plot  6.469 16.049  6.403 16.059 /
\plot  6.403 16.059  6.337 16.072 /
\plot  6.337 16.072  6.276 16.091 /
\plot  6.276 16.091  6.219 16.112 /
\plot  6.219 16.112  6.164 16.140 /
\plot  6.164 16.140  6.111 16.169 /
\plot  6.111 16.169  6.060 16.205 /
\plot  6.060 16.205  6.013 16.245 /
\plot  6.013 16.245  5.967 16.290 /
\plot  5.967 16.290  5.922 16.341 /
\plot  5.922 16.341  5.880 16.398 /
\plot  5.880 16.398  5.840 16.461 /
\plot  5.840 16.461  5.800 16.529 /
\plot  5.800 16.529  5.762 16.605 /
\plot  5.762 16.605  5.726 16.688 /
\plot  5.726 16.688  5.690 16.777 /
\plot  5.690 16.777  5.656 16.872 /
\plot  5.656 16.872  5.622 16.976 /
\plot  5.622 16.976  5.590 17.086 /
\plot  5.590 17.086  5.558 17.204 /
\plot  5.558 17.204  5.529 17.329 /
\plot  5.529 17.329  5.499 17.462 /
\plot  5.499 17.462  5.469 17.602 /
\plot  5.469 17.602  5.442 17.748 /
\plot  5.442 17.748  5.412 17.903 /
\plot  5.412 17.903  5.387 18.062 /
\plot  5.387 18.062  5.359 18.227 /
\plot  5.359 18.227  5.334 18.396 /
\plot  5.334 18.396  5.309 18.570 /
\plot  5.309 18.570  5.285 18.743 /
\plot  5.285 18.743  5.262 18.919 /
\plot  5.262 18.919  5.241 19.094 /
\plot  5.241 19.094  5.220 19.266 /
\plot  5.220 19.266  5.201 19.435 /
\plot  5.201 19.435  5.182 19.596 /
\plot  5.182 19.596  5.165 19.751 /
\plot  5.165 19.751  5.150 19.895 /
\plot  5.150 19.895  5.137 20.028 /
\plot  5.137 20.028  5.124 20.149 /
\plot  5.124 20.149  5.114 20.256 /
\plot  5.114 20.256  5.105 20.350 /
\plot  5.105 20.350  5.099 20.428 /
\plot  5.099 20.428  5.093 20.494 /
\plot  5.093 20.494  5.088 20.544 /
\plot  5.088 20.544  5.084 20.582 /
\plot  5.084 20.582  5.082 20.608 /
\plot  5.082 20.608  5.080 20.625 /
\putrule from  5.080 20.625 to  5.080 20.633
\putrule from  5.080 20.633 to  5.080 20.637
}%
%
%
\linethickness= 0.500pt
\setplotsymbol ({\thinlinefont .})
{\color[rgb]{0,0,0}\plot  5.175 11.430  7.017  9.525 /
}%
%
%
\linethickness= 0.500pt
\setplotsymbol ({\thinlinefont .})
{\color[rgb]{0,0,0}\plot  5.175 12.668  8.509  9.271 /
}%
%
%
\linethickness= 0.500pt
\setplotsymbol ({\thinlinefont .})
{\color[rgb]{0,0,0}\plot  5.524 13.462  9.716  9.398 /
}%
%
%
\linethickness= 0.500pt
\setplotsymbol ({\thinlinefont .})
{\color[rgb]{0,0,0}\plot  5.874 14.415 10.446  9.874 /
}%
%
%
\linethickness= 0.500pt
\setplotsymbol ({\thinlinefont .})
{\color[rgb]{0,0,0}\plot  6.159 15.367 10.986 10.700 /
}%
%
%
\linethickness= 0.500pt
\setplotsymbol ({\thinlinefont .})
{\color[rgb]{0,0,0}\plot  6.445 16.320 11.335 11.621 /
}%
%
%
\linethickness= 0.500pt
\setplotsymbol ({\thinlinefont .})
{\color[rgb]{0,0,0}\plot  7.017 16.986 11.557 12.573 /
}%
%
%
\linethickness= 0.500pt
\setplotsymbol ({\thinlinefont .})
{\color[rgb]{0,0,0}\plot  7.747 17.526 11.716 13.589 /
}%
%
%
\linethickness= 0.500pt
\setplotsymbol ({\thinlinefont .})
{\color[rgb]{0,0,0}\plot  8.319 18.129 11.811 14.796 /
}%
%
%
\linethickness= 0.500pt
\setplotsymbol ({\thinlinefont .})
{\color[rgb]{0,0,0}\plot  9.493 18.193 11.462 16.510 /
}%
%
%
\linethickness= 0.500pt
\setplotsymbol ({\thinlinefont .})
{\color[rgb]{0,0,0}\plot  5.143 10.192  5.461  9.811 /
}%
%
%
\linethickness= 0.500pt
\setplotsymbol ({\thinlinefont .})
{\color[rgb]{0,0,0}\plot 14.383 16.224 15.145 15.589 /
}%
%
%
\linethickness= 0.500pt
\setplotsymbol ({\thinlinefont .})
{\color[rgb]{0,0,0}\plot 14.129 17.558 15.113 16.637 /
}%
%
%
\linethickness= 0.500pt
\setplotsymbol ({\thinlinefont .})
{\color[rgb]{0,0,0}\plot 14.859 18.034 15.145 17.780 /
}%
%
%
\put{\SetFigFont{6}{7.2}{\rmdefault}{\mddefault}{\updefault}{\color[rgb]{0,0,0}$\Ga_l$}%
} [lB] at 11.906 13.653
%
%
\put{\SetFigFont{6}{7.2}{\rmdefault}{\mddefault}{\updefault}{\color[rgb]{0,0,0}$\Ga_r$}%
} [lB] at 13.494 17.462
%
%
\put{\SetFigFont{6}{7.2}{\rmdefault}{\mddefault}{\updefault}{\color[rgb]{0,0,0}$\cP_l$}%
} [lB] at 14.446 16.669
%
%
\put{\SetFigFont{6}{7.2}{\rmdefault}{\mddefault}{\updefault}{\color[rgb]{0,0,0}$\cP_r$}%
} [lB] at  6.985 12.859
%
%
\put{\SetFigFont{6}{7.2}{\rmdefault}{\mddefault}{\updefault}{\color[rgb]{0,0,0}$\cN$}%
} [lB] at 10.160 19.685
%
%
\put{\SetFigFont{6}{7.2}{\rmdefault}{\mddefault}{\updefault}{\color[rgb]{0,0,0}$\tilde{S}_-$}%
} [lB] at  3.842 12.668
%
%
\put{\SetFigFont{6}{7.2}{\rmdefault}{\mddefault}{\updefault}{\color[rgb]{0,0,0}
$\tilde{S}_+$}%
} [lB] at 15.367 15.208
%
%
\put{\SetFigFont{6}{7.2}{\rmdefault}{\mddefault}{\updefault}{\color[rgb]{0,0,0}$\tilde{N}_+$}%
} [lB] at 15.367 18.002
%
%
\put{\SetFigFont{6}{7.2}{\rmdefault}{\mddefault}{\updefault}{\color[rgb]{0,0,0}$\bar{N}_+$}%
} [lB] at 15.399  7.588
%
%
\put{\SetFigFont{6}{7.2}{\rmdefault}{\mddefault}{\updefault}{\color[rgb]{0,0,0}$\bar{N}_-$}%
} [lB] at  3.873  9.716
%
%
\put{\SetFigFont{6}{7.2}{\rmdefault}{\mddefault}{\updefault}{\color[rgb]{0,0,0}$\tilde{N}_-$}%
} [lB] at  3.842 20.225
%
%
\put{\SetFigFont{6}{7.2}{\rmdefault}{\mddefault}{\updefault}{\color[rgb]{0,0,0}$\cN$}%
} [lB] at 12.732 14.732
\linethickness=0pt
\putrectangle corners at  3.810 21.660 and 15.431  6.280
\endpicture}

\end{center}
\caption{Change in the topology of nullclines: $\mu<\mu_c$ (left) and $\mu>\mu_c$ (right).\label{fig:topnul}}
\end{figure}

\subsubsection{Winding number of orbits and corridors}
Assuming \refeq{assump}, let
\beq w_0 := \frac{s_+^{}- n_-^{}}{2\pi}= \frac{n_+^{}-s_-^{}}{2\pi}.\eeq
Given any orbit $\cO = \{(x_o(\tau),y_o(\tau))\ | \ \tau\in\RR\}$ of the flow \refeq{eq:flow}, the following quantity is well-defined:
\beq\label{def:wn}
w(\cO) = w_0 - \frac{1}{2\pi} \left(y_o(\infty) - y_o(-\infty)\right) 
= w_0-\frac{1}{2\pi}\int_{-\infty}^\infty g_\mu(x_o(\tau),y_o(\tau)) d\tau
\eeq
 In particular for the two distinguished orbits $\widetilde{\cW}^\pm$ this is easily calculated to be
\beq
w(\widetilde{\cW}^-) = w_0 - \frac{1}{2\pi}(n_+^{} - 2\pi k_+ - s_-^{})=  k_+ \qquad
w(\widetilde{\cW}^+) = w_0-\frac{1}{2\pi}\left(s_+^{} - ( n_-^{} + 2\pi k_-)\right)=  k_-
\eeq
and is therefore an integer. 
 We call these the {\em winding numbers}, of $\cW^-$ and $\cW^+$, respectively.

On the other hand it is easy to see that these two must in fact be equal.
  The reason is that, since $\widetilde{\cW}^-$ goes from $\widetilde{S}^- = (x_-^{},s_-^{})$ to $\bar{N}^+=(x_+^{},n_+^{}-2\pi k_+)$, 
there are two copies of it in the universal cover that go from $(x_-^{},s_-^{} + 2\pi k_+)$ to $(x_+^{},n_+^{})$ and from 
$(x_-^{},s_-^{}+2\pi(k_+-1))$ to $(x_+^{},n_+^{}-2\pi)$.  Since $n_+^{}-2\pi<s_+^{}<n_+^{}$, these two copies of 
$\widetilde{\cW}^-$ must sandwich any orbit that goes into $\widetilde{S}_+$, in particular $\widetilde{\cW}^+$.  
Thus the $\al$-limit point of $\widetilde{\cW}^+$ has to be $(x_-^{},n_-^{}+2\pi k_+)$, and therefore $k_+=k_-$.
\begin{defn}
 The {\em winding number of the corridor} $\cK_1$ is  the common value of the winding numbers of $\widetilde{\cW}^\pm$.
\end{defn}

Figure~\ref{fig:topnul} shows a corridor of winding number zero when $\mu<\mu_c$ (left) and one of winding number equal to one for
 $\mu>\mu_c$ (right).
\subsubsection{Continuity argument for existence of saddles connectors}
As the parameter $\mu$ varies, the two
distinguished orbits $\widetilde{\cW}^\pm$, and hence the corridor $\cK_1$ that they form will also vary.
  Let $\cK_1(\mu)$ denote the corresponding corridor for parameter value $\mu$ (if it exists).
Since the winding number $w(\cK_1(\mu))$ is integer-valued, if  $\cK_1$ varies continuously with respect to $\mu$, its winding number would have to remain constant. 
 It is however possible that for some value of $\mu$ the two orbits $\widetilde{\cW}^\pm$ coincide and the corridor
 $\cK_1(\mu)$ disappears, leaving a saddles connector behind (for which the winding number will not be an integer). 
 We would like to show that this is the only way for the winding number of $\cK_1(\mu)$ to be different for two 
different values of the parameter $\mu$.
  In other words:

\begin{prop} \label{prop:ssexists}
Suppose there exist two values $\mu^{}_0<\mu^{}_1$ in the interval $I$ for each of which a non-empty corridor $\cK_1$ of finite winding number
exists, and such that
\beq
w(\cK_1(\mu^{}_0))\leq 0\qquad\mbox{and}\qquad w(\cK_1(\mu^{}_1))\geq 1.
\eeq
Then there exists $\mu^{}_s \in (\mu^{}_0,\mu^{}_1)$ such that the flow \refeq{eq:flowmu} with $\mu=\mu^{}_s$ has a 
saddles connector $\cS(\mu^{}_s)$, whose lift to the universal cover $\bar{\cC}$ 
 connects the saddle-node $\widetilde{S}_-$ to the saddle-node $\widetilde{S}_+$.
\end{prop}

\begin{proof}

Let $\ba(\mu)$ denote the {\em signed area} of $\cK_1(\mu)$, defined via Green's theorem:
\beq\label{def:area}
\ba(\mu) = \oint_{\p \cK_1(\mu)} (-y)dx = \int_{x_-}^{x_+} \left(y_\mu^+ - y_\mu^-\right) dx.
\eeq
Here $y_\mu^\pm$ are the $y$-components of the orbits $\widetilde{\cW}^\pm$, thought of as functions of $x$ (which is always possible 
since $x(\tau)$ is a monotone increasing function of flow parameter $\tau$.)
The following facts about $\ba(\mu)$ are easily verified:
\begin{enumerate}
\item $\ba(\mu)>0$ if and only if $w(\cK_1(\mu))\geq 1$.

{\em Proof:} Since $\widetilde{\cW}^+$ and $\widetilde{\cW}^-$ cannot intersect without coinciding, it is clear from \refeq{def:area} 
that $\ba(\mu)>0$ if and only if $\widetilde{\cW}^+$ is above $\widetilde{\cW}^-$, which is equivalent to $\widetilde{S}_+ = (x_+,s_+)$ 
being above $\bar{N}_+ = (x_+,n_+-2\pi k)$ where $k=w(\cK_1(\mu))$.  Since $n_+>s_+$ we must have $k\geq 1$.
\item $\ba(\mu)<0$ if and only if $w(\cK_1(\mu)) \leq 0$.

{\em Proof:} Similar to above.
\item $\ba(\mu) = 0$ if and only if $\cK_1(\mu) = \emptyset$.

{\em Proof:} From the definition \refeq{def:area} it is clear that the only way for $\ba(\mu)$ to be zero is for the two orbits 
$\widetilde{\cW}^+$ and $\widetilde{\cW}^-$ to coincide.
\item $\ba(\mu)$ is a continuous function of $\mu$ for all $\mu\in [\mu_0,\mu_1]$.

{\em Proof:} Let $\ep>0$ be given.  For $\mu,\mu'\in [\mu_0,\mu_1]$ and $\xi>0$  small enough, we have
\bea
|\ba(\mu') - \ba(\mu)| & = &\left| \int_{x_-}^{x_+} \left(y_{\mu'}^+ - y_{\mu}^+\right) - \left( y_{\mu'}^- - y_{\mu}^-\right) dx \right| \\
& \leq & \int_{x_-}^{x_-+\xi} |y_{\mu'}^+| + |y_{\mu}^+| dx + \int_{x_-+\xi}^{x_+-\xi} \left|y_{\mu'}^+ - y_{\mu}^+\right|  dx 
+ \int_{x_+-\xi}^{x_+} |y_{\mu'}^+| + |y_{\mu}^+| dx \\
& & \mbox{} + \int_{x_-}^{x_-+\xi} |y_{\mu'}^-| + |y_{\mu}^-| dx + \int_{x_-+\xi}^{x_+ - \xi} \left|y_{\mu'}^- - y_{\mu}^-\right| dx 
+  \int_{x_+-\xi}^{x_+} |y_{\mu'}^-| + |y_{\mu}^-| dx \\
& = &\mbox{I}+\mbox{II}+\mbox{III}+\mbox{IV}+\mbox{V}+\mbox{VI}.
\eea
By Lemma~\ref{lem:mon} the functions $y_\mu^\pm$ are monotone in $\mu$, thus for all $\tau \in \RR$ we have
$$
|y_\mu^\pm(\tau)| \leq \max\{ |y_{\mu_0}^\pm(\tau)|, |y_{\mu_1}^\pm(\tau)|\}
$$
Thus is particular, using the continuity of orbits of \refeq{eq:flowmu} in the flow parameter $\tau$, 
for $x$ near the boundary points $x_\pm$ and all $\mu\in[\mu_0,\mu_1]$ we have $|y_\mu^\pm(x)| \leq C$,
where $C>0$ is a constant depending only on $y^\pm_{\mu_0}(x_\pm)$ and $y^\pm_{\mu_1}(x_\pm)$, 
or in other words, on the finite winding numbers of the corridors $\cK_1(\mu_0)$ and $\cK_1(\mu_1)$.  
Thus, given $\ep$, we may choose $\xi$ small enough (depending on $\ep$) such that $\mbox{I},\mbox{III},\mbox{IV}$ and $\mbox{VI}$ are all less than $\ep/6$.

Fixing $\xi$ in this way, 
by continuous dependence of orbits of the flow \refeq{eq:flowmu} on the parameter $\mu$ (and since we are on a compact interval in the flow parameter $\tau$), 
it is possible for  $\delta>0$ to be chosen small enough (depending on $\ep$ and $\xi$,) 
such that $|\mu' - \mu|<\de$ implies $|y_{\mu'}^+(x) - y_{\mu}^+(x)| < (x_+ - x_-)^{-1}\ep /6$ for all $x \in [x_-+\xi,x_+-\xi]$.  
Therefore $\mbox{II} < \ep/6$.  
Similarly, $\mbox{V}<\ep/6$ as well, and we are done.
\end{enumerate} 
 
Having established the above properties for the signed area $\ba(\mu)$, the standard continuity argument can now be applied:  By assumption we have $\ba(\mu_0)<0$ and $\ba(\mu_1)>0$.  Thus by the Intermediate Value Theorem there exists $\mu_s\in(\mu_0,\mu_1)$ such that $\ba(\mu_s)=0$, which is equivalent to the existence of a connector. 
 Since the corridors $\cK_1(\mu)$ by definition always have $\widetilde{S}_\pm$ on their boundary, the saddles connector also has to go 
from $\widetilde{S}_-$ to $\widetilde{S}_+$.
\end{proof}
We will use this proposition to establish the existence of  saddles connectors for both the $\Theta$ \refeq{dynsysTh} and the $\Omega$ 
\refeq{dynsysOm} equations.
 In each case we need to show that the flow satisfies the assumptions we have made in this subsection about the flow on a cylinder 
\refeq{eq:flowmu}.

\subsubsection{Existence of saddles connectors for the $\Theta$ equation}
For simplicity we are only going to consider the case $\ka = \half$.
To see that \refeq{dynsysTh} is a flow on a cylinder of the type we have considered in the above, we make the following 
identifications: $x=\theta$, $y=\Theta$, $x_-^{} = 0$, $x_+^{} = \pi$.
  We have
$$
f(\theta) = \sin\theta,
\qquad 
g(\theta,\Theta) = -2a\sin\theta\cos\theta\cos\Theta+\left(2aE\sin^2\theta - 1\right) \sin\Theta + 2\la\sin\theta.
$$
Therefore, the equilibria are at
$$
S_- = (0,0),\quad N_- = (0,\pi)\qquad S_+ = (\pi,-\pi),\quad N_+ = (\pi,0);
$$
all of these are hyperbolic.
Condition \refeq{assump} is clearly satisfied: $s_-^{} - n_-^{} = -\pi$ and $n_+^{} - s_+^{} = \pi$.

Here we make a note of the fact, easily verified, that the $\Theta$ flow \refeq{dynsysTh} possesses a discrete symmetry:  
Let $\cO = (\theta(\tau),\Theta(\tau))$ be any orbit of the flow \refeq{dynsysTh}.
  Then $\cO' = (\pi - \theta(-t), \pi-\Theta(-t))$ is also an orbit of \refeq{dynsysTh}.
  This is simply due to the fact that 
$$
f(\theta) = f(\pi-\theta),\qquad g(\theta,\Theta) = g(\pi-\theta,\pi-\Theta).
$$
This symmetry will prove useful in constructing corridors of given winding number for the flow.
 Next we check that the assumptions we made in the previous subsection about the topology of nullclines 
hold in the case of \refeq{dynsysTh}.
\subsubsection{Topology of the nullclines}

Let $T := \tan(\Theta/2)$.  Then,
$$
g_{E,\la}(\theta,\Theta) = \frac{1}{1+T^2}\left( 2\sin\theta(\la - a\cos\theta)T^2 +2(2aE\sin^2\theta - 1)T + 
2\sin\theta(\la+ a\cos\theta)\right) =: \frac{q(T)}{1+T^2}
$$
where $q$ is a quadratic polynomial in $T$, whose discriminant we calculate to be
$$
\Delta_q := (2aE\sin^2\theta - 1)^2 - 4\la^2\sin^2\theta + 4 a^2 \sin^2\theta\cos^2\theta.
$$
Setting $\tau :=- \cot \theta$ we thus obtain
$$
\Delta_q = \frac{\delta_q(\tau)}{(1+\tau^2)^2},\qquad \delta_q(\tau) = \tau^4 + 2(1-2aE + 2a^2 - 2\la^2)\tau^2+(1-2aE)^2 - 4\la^2
$$
Therefore $\de_q$ is a quartic polynomial which is quadratic in $s=\tau^2$.
  The discriminant of $\de_q(s)$ is 
$$
\Delta_{\delta_q} = \la^4 + 2a(E-a)\la^2 + a^2(a^2-2aE+1) = (\la^2+a(E-a))^2 + a^2(1-E^2)
$$
which is always positive.  
Thus $\de_q(s) = 0$ will always have two roots, which will be of opposite signs if $|\la|\geq\half - aE$.
  In that case there exists $s_1 = \tau_1^2>0$ such that $\Delta_q(\pm \tau_1^{})=0$, and $\Delta_q(\tau)<0$ for $\tau\in(-\tau_1^{},\tau_1^{})$.
  It follows that the quadratic equation $q(T)=0$ will have two roots so long as $|t|>\tau_1^{}$, it will have repeated roots 
when $t=\pm \tau_1^{}$, and no roots when $|t|<\tau_1^{}$.  

If on the other hand $|\la|<\half-aE\leq \half$ then $s_1$ and $s_2$, the roots of $\de_q(s)=0$, will be of the same sign.
  Since the sum of these roots will be $4(\la^2 - a^2 - (\half-aE))$ it is easy to see that $s_1$ and $s_2$ will both be 
negative, thus they will not correspond to real values of $\tau = \pm\sqrt{s}$.
  Therefore $\Delta_q>0$ and hence $q(T)=0$ will always have two roots.
  Thus the critical value of the parameter $\la$ (thinking of the other parameters $a$ and $E$ as given and fixed) is
$$
\la_c := -\half + aE.
$$

Now any zero of $q$ will be a zero of $g$, and thus will give us a point on the nullcline $\Ga$.
  In addition, $g$ may also have a zero at $\Theta = 0$ or $\pm\pi$, where $T=\pm\infty$.
  For $g$ to be zero there we need the coefficient of $T^2$ in $q$ to vanish, i.e. either $\sin\theta = 0$, 
which will give us the equilibrium points, or $\cos\theta = \la/a$ which is impossible so long as 
$$
\la < -a.
$$
Under this condition therefore, the nullclines have the topology we assumed in the previous subsection, with 
$\la$ playing the role of the parameter $\mu$, i.e., given $a\in(0,\half)$, $E\in[0,1]$, $\la\in(-\infty,-a)$ 
and $\la_c$ as in the above, the topology of $\Theta$-nullclines for the flow \refeq{dynsysTh} changes across 
$\la=\la_c$ in the manner described in assumption ({\bf A}).  
Figure~\ref{fig:thetanull} shows Maple plots of the $\Theta$-nullclines for values of $\la$ below and above the critical value.
\begin{figure}[ht]
\begin{center}
\includegraphics[scale=0.3]{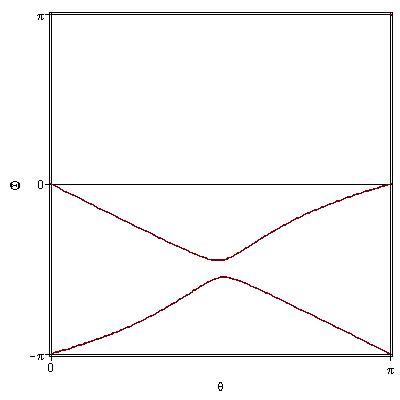}
\includegraphics[scale=0.3]{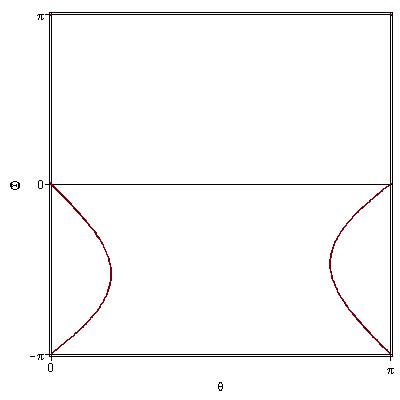}
\end{center}
\caption{$\Theta$ nullclines for $\la=-0.4$ (left) and $\la=-0.9$ (right), with $a=0.1$, $E=0.95$.\label{fig:thetanull}}
\end{figure}

\subsubsection{Explicit solutions of the $\Theta$ equation}\label{explicit}
One easily verifies that given $a\in[0,\half)$, there is an explicit solution of \refeq{eq:Theta} (with $\ka = \half$), for $E=1$ and $\la = -1+a$, namely
$$
\Theta_0(\theta) = - \theta.
$$
This furthermore generates a saddles connector for \refeq{dynsysTh}:  $\cS_0 := (\theta_0(\tau),\Theta_0(\theta_0(\tau))$ where $\theta_0(\tau)$ is 
the unique solution to the ODE $\dot{\theta} = \sin\theta$ with $\theta(-\infty) = 0$ and $\theta(\infty) = \pi$.

This solution will help us get the iteration started.
\subsubsection{Existence of corridors with unequal winding number}
Throughout this section, $a$ will be a fixed number in $(0,\half)$.
The following two propositions will start things off:

\begin{prop}\label{prop:lacorlow}
Given $E\in[0,1]$ and $\la \leq \la_l := -1-a$  the flow \refeq{dynsysTh} has a corridor $\cK_1(E,\la)$ with $w(\cK_1)\geq 1$.
\end{prop}
\begin{proof}
The linearization of the flow at $S_-$ gives us eigenvalues $\la_1= 1$, $\la_2=-1$ and a corresponding set of eigenvectors is 
$v_1 = \left(\begin{array}{c} 1\\ \la-a\end{array}\right)$ and $v_2 = \left(\begin{array}{c} 0\\ 1\end{array}\right)$.
  The orbit $\cW^-$, being the unstable manifold of $S_-$ is tangent to the unstable direction $v_1$, therefore the slope 
of $\cW^-$ at $S_-$ is $\frac{d\Theta}{d\theta}|_{\cW^-}(S_-) = \la - a < 0$.
  The slope of $\cW^+$ at $S_+$ is similarly calculated, and it turns out to be the same 
$\frac{d\Theta}{d\theta}|_{\cW^+}(S_+) = \la - a$.
\begin{figure}[ht]
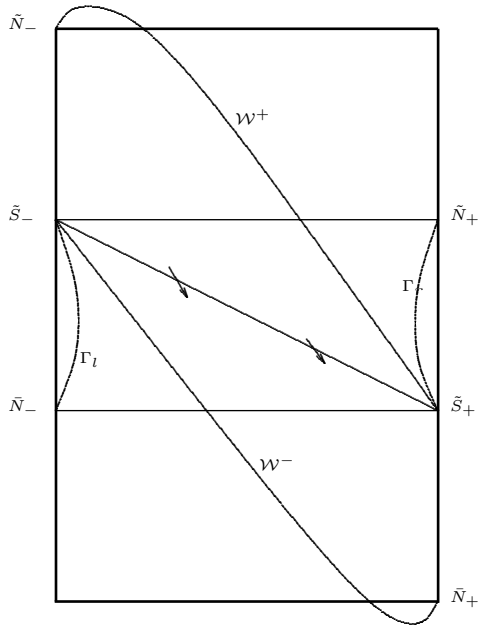

 \begin{center}
\font\thinlinefont=cmr5
\begingroup\makeatletter\ifx\SetFigFont\undefined%
\gdef\SetFigFont#1#2#3#4#5{%
  \reset@font\fontsize{#1}{#2pt}%
  \fontfamily{#3}\fontseries{#4}\fontshape{#5}%
  \selectfont}%
\fi\endgroup%
\mbox{\beginpicture
\setcoordinatesystem units <0.50000cm,0.50000cm>
\unitlength=0.50000cm
\linethickness=1pt
\setplotsymbol ({\makebox(0,0)[l]{\tencirc\symbol{'160}}})
\setshadesymbol ({\thinlinefont .})
\setlinear
%
%
\linethickness= 0.500pt
\setplotsymbol ({\thinlinefont .})
{\color[rgb]{0,0,0}\putrule from  5.080 16.510 to 15.240 16.510
}%
%
%
\linethickness=1pt
\setplotsymbol ({\makebox(0,0)[l]{\tencirc\symbol{'160}}})
{\color[rgb]{0,0,0}\putrule from  5.080 21.590 to 15.240 21.590
}%
%
%
\linethickness=1pt
\setplotsymbol ({\makebox(0,0)[l]{\tencirc\symbol{'160}}})
{\color[rgb]{0,0,0}\putrule from 15.240 21.590 to 15.240  6.315
\putrule from 15.240  6.350 to  5.045  6.350
\putrule from  5.080  6.350 to  5.080 21.590
}%
%
%
\linethickness= 0.500pt
\setplotsymbol ({\thinlinefont .})
{\color[rgb]{0,0,0}\putrule from  5.080 11.430 to 15.240 11.430
}%
%
%
\linethickness= 0.500pt
\setplotsymbol ({\thinlinefont .})
{\color[rgb]{0,0,0}\plot  5.080 16.510 15.240 11.430 /
}%
%
%
\linethickness= 0.500pt
\setplotsymbol ({\thinlinefont .})
{\color[rgb]{0,0,0}\plot  8.096 15.240  8.572 14.446 /
%
%
\plot  8.387 14.631  8.572 14.446  8.496 14.697 /
}%
%
%
\linethickness= 0.500pt
\setplotsymbol ({\thinlinefont .})
{\color[rgb]{0,0,0}\plot 11.748 13.335 12.224 12.700 /
%
%
\plot 12.021 12.865 12.224 12.700 12.122 12.941 /
}%
%
%
\linethickness= 0.500pt
\setplotsymbol ({\thinlinefont .})
{\color[rgb]{0,0,0}\plot  5.080 16.510  5.082 16.506 /
\plot  5.082 16.506  5.086 16.493 /
\plot  5.086 16.493  5.093 16.474 /
\plot  5.093 16.474  5.103 16.442 /
\plot  5.103 16.442  5.118 16.398 /
\plot  5.118 16.398  5.137 16.343 /
\plot  5.137 16.343  5.160 16.277 /
\plot  5.160 16.277  5.186 16.201 /
\plot  5.186 16.201  5.215 16.116 /
\plot  5.215 16.116  5.245 16.025 /
\plot  5.245 16.025  5.277 15.932 /
\plot  5.277 15.932  5.309 15.839 /
\plot  5.309 15.839  5.338 15.744 /
\plot  5.338 15.744  5.368 15.653 /
\plot  5.368 15.653  5.395 15.564 /
\plot  5.395 15.564  5.423 15.477 /
\plot  5.423 15.477  5.448 15.397 /
\plot  5.448 15.397  5.472 15.318 /
\plot  5.472 15.318  5.493 15.244 /
\plot  5.493 15.244  5.512 15.174 /
\plot  5.512 15.174  5.529 15.107 /
\plot  5.529 15.107  5.546 15.041 /
\plot  5.546 15.041  5.560 14.978 /
\plot  5.560 14.978  5.573 14.916 /
\plot  5.573 14.916  5.586 14.857 /
\plot  5.586 14.857  5.599 14.798 /
\plot  5.599 14.798  5.609 14.736 /
\plot  5.609 14.736  5.620 14.677 /
\plot  5.620 14.677  5.628 14.618 /
\plot  5.628 14.618  5.637 14.556 /
\plot  5.637 14.556  5.645 14.493 /
\plot  5.645 14.493  5.654 14.431 /
\plot  5.654 14.431  5.660 14.366 /
\plot  5.660 14.366  5.666 14.300 /
\plot  5.666 14.300  5.671 14.235 /
\plot  5.671 14.235  5.675 14.167 /
\plot  5.675 14.167  5.679 14.097 /
\plot  5.679 14.097  5.681 14.029 /
\plot  5.681 14.029  5.683 13.959 /
\plot  5.683 13.959  5.685 13.890 /
\putrule from  5.685 13.890 to  5.685 13.818
\putrule from  5.685 13.818 to  5.685 13.748
\plot  5.685 13.748  5.683 13.680 /
\plot  5.683 13.680  5.681 13.610 /
\plot  5.681 13.610  5.679 13.542 /
\plot  5.679 13.542  5.675 13.477 /
\plot  5.675 13.477  5.671 13.411 /
\plot  5.671 13.411  5.666 13.346 /
\plot  5.666 13.346  5.660 13.284 /
\plot  5.660 13.284  5.654 13.223 /
\plot  5.654 13.223  5.645 13.164 /
\plot  5.645 13.164  5.637 13.106 /
\plot  5.637 13.106  5.628 13.049 /
\plot  5.628 13.049  5.620 12.992 /
\plot  5.620 12.992  5.609 12.937 /
\plot  5.609 12.937  5.599 12.880 /
\plot  5.599 12.880  5.586 12.821 /
\plot  5.586 12.821  5.571 12.764 /
\plot  5.571 12.764  5.556 12.704 /
\plot  5.556 12.704  5.539 12.645 /
\plot  5.539 12.645  5.520 12.581 /
\plot  5.520 12.581  5.501 12.518 /
\plot  5.501 12.518  5.480 12.450 /
\plot  5.480 12.450  5.455 12.380 /
\plot  5.455 12.380  5.429 12.306 /
\plot  5.429 12.306  5.400 12.230 /
\plot  5.400 12.230  5.370 12.150 /
\plot  5.370 12.150  5.338 12.067 /
\plot  5.338 12.067  5.304 11.982 /
\plot  5.304 11.982  5.273 11.898 /
\plot  5.273 11.898  5.239 11.813 /
\plot  5.239 11.813  5.207 11.735 /
\plot  5.207 11.735  5.175 11.661 /
\plot  5.175 11.661  5.150 11.597 /
\plot  5.150 11.597  5.127 11.542 /
\plot  5.127 11.542  5.110 11.498 /
\plot  5.110 11.498  5.095 11.468 /
\plot  5.095 11.468  5.086 11.447 /
\plot  5.086 11.447  5.082 11.434 /
\plot  5.082 11.434  5.080 11.430 /
}%
%
%
\linethickness= 0.500pt
\setplotsymbol ({\thinlinefont .})
{\color[rgb]{0,0,0}\plot 15.240 16.510 15.238 16.506 /
\plot 15.238 16.506 15.234 16.493 /
\plot 15.234 16.493 15.227 16.474 /
\plot 15.227 16.474 15.217 16.442 /
\plot 15.217 16.442 15.202 16.398 /
\plot 15.202 16.398 15.183 16.343 /
\plot 15.183 16.343 15.160 16.275 /
\plot 15.160 16.275 15.134 16.199 /
\plot 15.134 16.199 15.105 16.114 /
\plot 15.105 16.114 15.075 16.023 /
\plot 15.075 16.023 15.043 15.930 /
\plot 15.043 15.930 15.011 15.835 /
\plot 15.011 15.835 14.982 15.740 /
\plot 14.982 15.740 14.952 15.646 /
\plot 14.952 15.646 14.925 15.558 /
\plot 14.925 15.558 14.897 15.471 /
\plot 14.897 15.471 14.872 15.388 /
\plot 14.872 15.388 14.848 15.308 /
\plot 14.848 15.308 14.827 15.234 /
\plot 14.827 15.234 14.808 15.162 /
\plot 14.808 15.162 14.791 15.092 /
\plot 14.791 15.092 14.774 15.026 /
\plot 14.774 15.026 14.760 14.961 /
\plot 14.760 14.961 14.747 14.897 /
\plot 14.747 14.897 14.734 14.836 /
\plot 14.734 14.836 14.721 14.772 /
\plot 14.721 14.772 14.711 14.711 /
\plot 14.711 14.711 14.700 14.647 /
\plot 14.700 14.647 14.692 14.586 /
\plot 14.692 14.586 14.683 14.520 /
\plot 14.683 14.520 14.675 14.455 /
\plot 14.675 14.455 14.666 14.389 /
\plot 14.666 14.389 14.660 14.321 /
\plot 14.660 14.321 14.654 14.252 /
\plot 14.654 14.252 14.649 14.182 /
\plot 14.649 14.182 14.645 14.110 /
\plot 14.645 14.110 14.641 14.038 /
\plot 14.641 14.038 14.639 13.964 /
\plot 14.639 13.964 14.637 13.892 /
\plot 14.637 13.892 14.635 13.818 /
\putrule from 14.635 13.818 to 14.635 13.744
\putrule from 14.635 13.744 to 14.635 13.669
\plot 14.635 13.669 14.637 13.595 /
\plot 14.637 13.595 14.639 13.523 /
\plot 14.639 13.523 14.641 13.454 /
\plot 14.641 13.454 14.645 13.384 /
\plot 14.645 13.384 14.649 13.316 /
\plot 14.649 13.316 14.654 13.248 /
\plot 14.654 13.248 14.660 13.185 /
\plot 14.660 13.185 14.666 13.121 /
\plot 14.666 13.121 14.675 13.060 /
\plot 14.675 13.060 14.683 13.001 /
\plot 14.683 13.001 14.692 12.943 /
\plot 14.692 12.943 14.700 12.888 /
\plot 14.700 12.888 14.711 12.831 /
\plot 14.711 12.831 14.721 12.774 /
\plot 14.721 12.774 14.734 12.717 /
\plot 14.734 12.717 14.749 12.660 /
\plot 14.749 12.660 14.764 12.603 /
\plot 14.764 12.603 14.781 12.543 /
\plot 14.781 12.543 14.800 12.486 /
\plot 14.800 12.486 14.819 12.425 /
\plot 14.819 12.425 14.840 12.361 /
\plot 14.840 12.361 14.865 12.296 /
\plot 14.865 12.296 14.891 12.228 /
\plot 14.891 12.228 14.920 12.156 /
\plot 14.920 12.156 14.950 12.082 /
\plot 14.950 12.082 14.982 12.006 /
\plot 14.982 12.006 15.016 11.930 /
\plot 15.016 11.930 15.047 11.851 /
\plot 15.047 11.851 15.081 11.777 /
\plot 15.081 11.777 15.113 11.705 /
\plot 15.113 11.705 15.145 11.637 /
\plot 15.145 11.637 15.170 11.580 /
\plot 15.170 11.580 15.193 11.532 /
\plot 15.193 11.532 15.210 11.491 /
\plot 15.210 11.491 15.225 11.464 /
\plot 15.225 11.464 15.234 11.445 /
\plot 15.234 11.445 15.238 11.434 /
\plot 15.238 11.434 15.240 11.430 /
}%
%
%
\linethickness= 0.500pt
\setplotsymbol ({\thinlinefont .})
{\color[rgb]{0,0,0}\plot  5.080 16.510  5.082 16.508 /
\plot  5.082 16.508  5.084 16.504 /
\plot  5.084 16.504  5.091 16.497 /
\plot  5.091 16.497  5.101 16.485 /
\plot  5.101 16.485  5.116 16.466 /
\plot  5.116 16.466  5.135 16.440 /
\plot  5.135 16.440  5.160 16.408 /
\plot  5.160 16.408  5.192 16.368 /
\plot  5.192 16.368  5.230 16.317 /
\plot  5.230 16.317  5.277 16.260 /
\plot  5.277 16.260  5.330 16.190 /
\plot  5.330 16.190  5.391 16.114 /
\plot  5.391 16.114  5.461 16.025 /
\plot  5.461 16.025  5.537 15.928 /
\plot  5.537 15.928  5.624 15.820 /
\plot  5.624 15.820  5.715 15.701 /
\plot  5.715 15.701  5.817 15.574 /
\plot  5.817 15.574  5.922 15.439 /
\plot  5.922 15.439  6.037 15.293 /
\plot  6.037 15.293  6.157 15.141 /
\plot  6.157 15.141  6.282 14.982 /
\plot  6.282 14.982  6.413 14.815 /
\plot  6.413 14.815  6.549 14.643 /
\plot  6.549 14.643  6.687 14.467 /
\plot  6.687 14.467  6.830 14.285 /
\plot  6.830 14.285  6.974 14.101 /
\plot  6.974 14.101  7.123 13.915 /
\plot  7.123 13.915  7.271 13.727 /
\plot  7.271 13.727  7.421 13.536 /
\plot  7.421 13.536  7.569 13.348 /
\plot  7.569 13.348  7.719 13.157 /
\plot  7.719 13.157  7.870 12.969 /
\plot  7.870 12.969  8.016 12.780 /
\plot  8.016 12.780  8.164 12.596 /
\plot  8.164 12.596  8.308 12.412 /
\plot  8.308 12.412  8.450 12.232 /
\plot  8.450 12.232  8.589 12.057 /
\plot  8.589 12.057  8.727 11.883 /
\plot  8.727 11.883  8.862 11.714 /
\plot  8.862 11.714  8.994 11.546 /
\plot  8.994 11.546  9.121 11.386 /
\plot  9.121 11.386  9.248 11.227 /
\plot  9.248 11.227  9.368 11.074 /
\plot  9.368 11.074  9.487 10.924 /
\plot  9.487 10.924  9.603 10.780 /
\plot  9.603 10.780  9.716 10.638 /
\plot  9.716 10.638  9.823 10.503 /
\plot  9.823 10.503  9.929 10.370 /
\plot  9.929 10.370 10.031 10.240 /
\plot 10.031 10.240 10.130 10.118 /
\plot 10.130 10.118 10.228  9.997 /
\plot 10.228  9.997 10.321  9.881 /
\plot 10.321  9.881 10.412  9.768 /
\plot 10.412  9.768 10.499  9.658 /
\plot 10.499  9.658 10.583  9.553 /
\plot 10.583  9.553 10.668  9.451 /
\plot 10.668  9.451 10.746  9.351 /
\plot 10.746  9.351 10.825  9.256 /
\plot 10.825  9.256 10.901  9.163 /
\plot 10.901  9.163 10.975  9.072 /
\plot 10.975  9.072 11.047  8.983 /
\plot 11.047  8.983 11.117  8.898 /
\plot 11.117  8.898 11.184  8.816 /
\plot 11.184  8.816 11.250  8.735 /
\plot 11.250  8.735 11.316  8.657 /
\plot 11.316  8.657 11.379  8.581 /
\plot 11.379  8.581 11.441  8.505 /
\plot 11.441  8.505 11.502  8.433 /
\plot 11.502  8.433 11.563  8.361 /
\plot 11.563  8.361 11.665  8.240 /
\plot 11.665  8.240 11.764  8.122 /
\plot 11.764  8.122 11.860  8.009 /
\plot 11.860  8.009 11.955  7.899 /
\plot 11.955  7.899 12.046  7.794 /
\plot 12.046  7.794 12.135  7.690 /
\plot 12.135  7.690 12.224  7.590 /
\plot 12.224  7.590 12.308  7.493 /
\plot 12.308  7.493 12.391  7.400 /
\plot 12.391  7.400 12.474  7.309 /
\plot 12.474  7.309 12.554  7.222 /
\plot 12.554  7.222 12.630  7.137 /
\plot 12.630  7.137 12.706  7.055 /
\plot 12.706  7.055 12.780  6.977 /
\plot 12.780  6.977 12.852  6.900 /
\plot 12.852  6.900 12.922  6.828 /
\plot 12.922  6.828 12.990  6.759 /
\plot 12.990  6.759 13.056  6.693 /
\plot 13.056  6.693 13.119  6.629 /
\plot 13.119  6.629 13.180  6.570 /
\plot 13.180  6.570 13.240  6.513 /
\plot 13.240  6.513 13.297  6.460 /
\plot 13.297  6.460 13.352  6.409 /
\plot 13.352  6.409 13.405  6.361 /
\plot 13.405  6.361 13.456  6.316 /
\plot 13.456  6.316 13.502  6.276 /
\plot 13.502  6.276 13.549  6.236 /
\plot 13.549  6.236 13.593  6.200 /
\plot 13.593  6.200 13.636  6.166 /
\plot 13.636  6.166 13.676  6.136 /
\plot 13.676  6.136 13.714  6.107 /
\plot 13.714  6.107 13.750  6.079 /
\plot 13.750  6.079 13.786  6.056 /
\plot 13.786  6.056 13.820  6.032 /
\plot 13.820  6.032 13.851  6.011 /
\plot 13.851  6.011 13.881  5.992 /
\plot 13.881  5.992 13.913  5.973 /
\plot 13.913  5.973 13.940  5.958 /
\plot 13.940  5.958 13.968  5.941 /
\plot 13.968  5.941 13.998  5.927 /
\plot 13.998  5.927 14.046  5.901 /
\plot 14.046  5.901 14.093  5.878 /
\plot 14.093  5.878 14.141  5.857 /
\plot 14.141  5.857 14.188  5.838 /
\plot 14.188  5.838 14.232  5.821 /
\plot 14.232  5.821 14.279  5.806 /
\plot 14.279  5.806 14.323  5.791 /
\plot 14.323  5.791 14.368  5.781 /
\plot 14.368  5.781 14.412  5.772 /
\plot 14.412  5.772 14.455  5.764 /
\plot 14.455  5.764 14.497  5.759 /
\plot 14.497  5.759 14.537  5.755 /
\putrule from 14.537  5.755 to 14.575  5.755
\putrule from 14.575  5.755 to 14.611  5.755
\plot 14.611  5.755 14.647  5.759 /
\plot 14.647  5.759 14.679  5.764 /
\plot 14.679  5.764 14.711  5.770 /
\plot 14.711  5.770 14.740  5.776 /
\plot 14.740  5.776 14.768  5.787 /
\plot 14.768  5.787 14.796  5.798 /
\plot 14.796  5.798 14.819  5.808 /
\plot 14.819  5.808 14.844  5.821 /
\plot 14.844  5.821 14.874  5.840 /
\plot 14.874  5.840 14.903  5.863 /
\plot 14.903  5.863 14.933  5.889 /
\plot 14.933  5.889 14.963  5.918 /
\plot 14.963  5.918 14.992  5.954 /
\plot 14.992  5.954 15.022  5.994 /
\plot 15.022  5.994 15.054  6.041 /
\plot 15.054  6.041 15.088  6.092 /
\plot 15.088  6.092 15.121  6.145 /
\plot 15.121  6.145 15.155  6.200 /
\plot 15.155  6.200 15.183  6.251 /
\plot 15.183  6.251 15.208  6.293 /
\plot 15.208  6.293 15.225  6.322 /
\plot 15.225  6.322 15.236  6.342 /
\plot 15.236  6.342 15.240  6.348 /
\putrule from 15.240  6.348 to 15.240  6.350
}%
%
%
\linethickness= 0.500pt
\setplotsymbol ({\thinlinefont .})
{\color[rgb]{0,0,0}\plot 15.240 11.430 15.238 11.432 /
\plot 15.238 11.432 15.236 11.436 /
\plot 15.236 11.436 15.229 11.445 /
\plot 15.229 11.445 15.219 11.458 /
\plot 15.219 11.458 15.206 11.479 /
\plot 15.206 11.479 15.187 11.504 /
\plot 15.187 11.504 15.162 11.540 /
\plot 15.162 11.540 15.130 11.582 /
\plot 15.130 11.582 15.094 11.635 /
\plot 15.094 11.635 15.050 11.697 /
\plot 15.050 11.697 14.997 11.769 /
\plot 14.997 11.769 14.937 11.853 /
\plot 14.937 11.853 14.872 11.946 /
\plot 14.872 11.946 14.796 12.050 /
\plot 14.796 12.050 14.715 12.164 /
\plot 14.715 12.164 14.626 12.287 /
\plot 14.626 12.287 14.531 12.421 /
\plot 14.531 12.421 14.427 12.565 /
\plot 14.427 12.565 14.319 12.715 /
\plot 14.319 12.715 14.205 12.874 /
\plot 14.205 12.874 14.086 13.039 /
\plot 14.086 13.039 13.964 13.212 /
\plot 13.964 13.212 13.835 13.390 /
\plot 13.835 13.390 13.703 13.572 /
\plot 13.703 13.572 13.570 13.758 /
\plot 13.570 13.758 13.434 13.947 /
\plot 13.434 13.947 13.297 14.137 /
\plot 13.297 14.137 13.159 14.330 /
\plot 13.159 14.330 13.020 14.522 /
\plot 13.020 14.522 12.882 14.715 /
\plot 12.882 14.715 12.744 14.906 /
\plot 12.744 14.906 12.607 15.096 /
\plot 12.607 15.096 12.471 15.284 /
\plot 12.471 15.284 12.338 15.469 /
\plot 12.338 15.469 12.205 15.653 /
\plot 12.205 15.653 12.076 15.831 /
\plot 12.076 15.831 11.949 16.006 /
\plot 11.949 16.006 11.824 16.180 /
\plot 11.824 16.180 11.701 16.347 /
\plot 11.701 16.347 11.582 16.510 /
\plot 11.582 16.510 11.466 16.669 /
\plot 11.466 16.669 11.354 16.825 /
\plot 11.354 16.825 11.244 16.976 /
\plot 11.244 16.976 11.138 17.122 /
\plot 11.138 17.122 11.034 17.264 /
\plot 11.034 17.264 10.933 17.399 /
\plot 10.933 17.399 10.835 17.532 /
\plot 10.835 17.532 10.740 17.661 /
\plot 10.740 17.661 10.649 17.786 /
\plot 10.649 17.786 10.560 17.907 /
\plot 10.560 17.907 10.473 18.023 /
\plot 10.473 18.023 10.389 18.136 /
\plot 10.389 18.136 10.308 18.246 /
\plot 10.308 18.246 10.230 18.351 /
\plot 10.230 18.351 10.152 18.453 /
\plot 10.152 18.453 10.077 18.553 /
\plot 10.077 18.553 10.005 18.648 /
\plot 10.005 18.648  9.934 18.741 /
\plot  9.934 18.741  9.866 18.832 /
\plot  9.866 18.832  9.798 18.919 /
\plot  9.798 18.919  9.732 19.006 /
\plot  9.732 19.006  9.669 19.088 /
\plot  9.669 19.088  9.605 19.169 /
\plot  9.605 19.169  9.544 19.247 /
\plot  9.544 19.247  9.485 19.325 /
\plot  9.485 19.325  9.426 19.399 /
\plot  9.426 19.399  9.366 19.473 /
\plot  9.366 19.473  9.269 19.594 /
\plot  9.269 19.594  9.174 19.713 /
\plot  9.174 19.713  9.083 19.827 /
\plot  9.083 19.827  8.992 19.937 /
\plot  8.992 19.937  8.903 20.043 /
\plot  8.903 20.043  8.816 20.146 /
\plot  8.816 20.146  8.729 20.248 /
\plot  8.729 20.248  8.644 20.345 /
\plot  8.644 20.345  8.562 20.441 /
\plot  8.562 20.441  8.481 20.534 /
\plot  8.481 20.534  8.401 20.623 /
\plot  8.401 20.623  8.323 20.709 /
\plot  8.323 20.709  8.244 20.792 /
\plot  8.244 20.792  8.168 20.872 /
\plot  8.168 20.872  8.094 20.951 /
\plot  8.094 20.951  8.020 21.025 /
\plot  8.020 21.025  7.948 21.097 /
\plot  7.948 21.097  7.878 21.167 /
\plot  7.878 21.167  7.811 21.232 /
\plot  7.811 21.232  7.745 21.294 /
\plot  7.745 21.294  7.679 21.353 /
\plot  7.679 21.353  7.618 21.410 /
\plot  7.618 21.410  7.557 21.463 /
\plot  7.557 21.463  7.497 21.514 /
\plot  7.497 21.514  7.440 21.560 /
\plot  7.440 21.560  7.385 21.605 /
\plot  7.385 21.605  7.330 21.647 /
\plot  7.330 21.647  7.279 21.685 /
\plot  7.279 21.685  7.228 21.721 /
\plot  7.228 21.721  7.182 21.755 /
\plot  7.182 21.755  7.135 21.787 /
\plot  7.135 21.787  7.089 21.816 /
\plot  7.089 21.816  7.046 21.844 /
\plot  7.046 21.844  7.004 21.867 /
\plot  7.004 21.867  6.964 21.891 /
\plot  6.964 21.891  6.924 21.914 /
\plot  6.924 21.914  6.886 21.933 /
\plot  6.886 21.933  6.847 21.952 /
\plot  6.847 21.952  6.809 21.971 /
\plot  6.809 21.971  6.773 21.986 /
\plot  6.773 21.986  6.714 22.013 /
\plot  6.714 22.013  6.653 22.037 /
\plot  6.653 22.037  6.593 22.060 /
\plot  6.593 22.060  6.534 22.081 /
\plot  6.534 22.081  6.475 22.098 /
\plot  6.475 22.098  6.416 22.115 /
\plot  6.416 22.115  6.356 22.130 /
\plot  6.356 22.130  6.299 22.145 /
\plot  6.299 22.145  6.240 22.155 /
\plot  6.240 22.155  6.183 22.164 /
\plot  6.183 22.164  6.126 22.172 /
\plot  6.126 22.172  6.068 22.176 /
\plot  6.068 22.176  6.016 22.181 /
\plot  6.016 22.181  5.963 22.183 /
\putrule from  5.963 22.183 to  5.912 22.183
\plot  5.912 22.183  5.863 22.181 /
\plot  5.863 22.181  5.817 22.178 /
\plot  5.817 22.178  5.772 22.174 /
\plot  5.772 22.174  5.732 22.168 /
\plot  5.732 22.168  5.692 22.159 /
\plot  5.692 22.159  5.656 22.151 /
\plot  5.656 22.151  5.620 22.142 /
\plot  5.620 22.142  5.588 22.130 /
\plot  5.588 22.130  5.556 22.119 /
\plot  5.556 22.119  5.518 22.102 /
\plot  5.518 22.102  5.480 22.083 /
\plot  5.480 22.083  5.446 22.060 /
\plot  5.446 22.060  5.412 22.035 /
\plot  5.412 22.035  5.378 22.007 /
\plot  5.378 22.007  5.347 21.973 /
\plot  5.347 21.973  5.313 21.935 /
\plot  5.313 21.935  5.279 21.895 /
\plot  5.279 21.895  5.245 21.848 /
\plot  5.245 21.848  5.211 21.800 /
\plot  5.211 21.800  5.179 21.753 /
\plot  5.179 21.753  5.150 21.706 /
\plot  5.150 21.706  5.124 21.666 /
\plot  5.124 21.666  5.105 21.632 /
\plot  5.105 21.632  5.091 21.609 /
\plot  5.091 21.609  5.084 21.596 /
\plot  5.084 21.596  5.080 21.590 /
}%
%
%
\put{\SetFigFont{6}{7.2}{\rmdefault}{\mddefault}{\updefault}{\color[rgb]{0,0,0}$\tilde{S}_-$}%
} [lB] at  3.810 16.510
%
%
\put{\SetFigFont{6}{7.2}{\rmdefault}{\mddefault}{\updefault}{\color[rgb]{0,0,0}$\bar{N}_-$}%
} [lB] at  3.810 11.430
%
%
\put{\SetFigFont{6}{7.2}{\rmdefault}{\mddefault}{\updefault}{\color[rgb]{0,0,0}$\tilde{N}_-$}%
} [lB] at  3.810 21.590
%
%
\put{\SetFigFont{6}{7.2}{\rmdefault}{\mddefault}{\updefault}{\color[rgb]{0,0,0}$\tilde{S}_+$}%
} [lB] at 15.558 11.430
%
%
\put{\SetFigFont{6}{7.2}{\rmdefault}{\mddefault}{\updefault}{\color[rgb]{0,0,0}$\tilde{N}_+$}%
} [lB] at 15.558 16.510
%
%
\put{\SetFigFont{6}{7.2}{\rmdefault}{\mddefault}{\updefault}{\color[rgb]{0,0,0}$\bar{N}_+$}%
} [lB] at 15.558  6.350
%
%
\put{\SetFigFont{6}{7.2}{\rmdefault}{\mddefault}{\updefault}{\color[rgb]{0,0,0}$\cW^-$}%
} [lB] at 10.478  9.842
%
%
\put{\SetFigFont{6}{7.2}{\rmdefault}{\mddefault}{\updefault}{\color[rgb]{0,0,0}$\cW^+$}%
} [lB] at  9.842 19.050
%
%
\put{\SetFigFont{6}{7.2}{\rmdefault}{\mddefault}{\updefault}{\color[rgb]{0,0,0}$\Ga_l$}%
} [lB] at  5.715 12.700
%
%
\put{\SetFigFont{6}{7.2}{\rmdefault}{\mddefault}{\updefault}{\color[rgb]{0,0,0}$\Ga_r$}%
} [lB] at 14.287 14.605
\linethickness=0pt
\putrectangle corners at  3.778 22.229 and 15.589  5.709
\endpicture}

 \end{center}
\caption{Existence of corridor with $w(\cK_1)\geq 1$}
\end{figure}
On the other hand, since $\la_l<\la_c$, we have $\Ga = \Ga_l \cup \Ga_r$.
  $S_-$ is a terminal point of the curve $\Ga_l$, on which $g_{E,\la}(\theta,\Theta)=0$, hence the slope of 
$\Ga_l$ at $S_-$ can be calculated from implicit function theorem to be 
$$
\left.\frac{d\Theta}{d\theta}\right|_{\Ga_l}(S_-) = -\frac{\p_\theta g_{E,\la}}{\p_\Theta g_{E,\la}} (S_-) = 2(\la - a) < \la -a.
$$
Same is true for the slope of $\Ga_r$ at $S_+$, as can be easily verified.
 Thus $\cW^-$ starts off inside $\cN$, which is connected, and its initial slope is $\la - a < -1$ for $\la<-1+a$.
  Consider the diagonal line segment $S_-S_+$, on which $\Theta = -\theta$.
  Let us compute the slope of orbits that cross this line, and compare it to the slope of the line:
$$
\frac{g_{E,\la}(\theta,-\theta)}{f(\theta)}  - (-1) =
 -2a\sin\theta( \cos\theta+E\sin\theta) + 2\la +2 \leq 0 \quad\mbox{ for } \la \leq -1-a.
$$
Thus $S_-S_+$ acts as a ``barrier", not allowing $\cW^-$ to cross it from below to above.
  Hence the $\om$-limit of $\cW^-$ cannot be $(\pi,0)$.
  The first possible terminal point is then $(\pi,-2\pi)$, hence $w(W^-) \geq 1$.  
\end{proof}
\begin{prop}\label{prop:thlowercor}
Suppose that for some $E\in(0,1]$ and some $\la\leq \la_0 := -1+a$, the flow \refeq{dynsysTh} has a saddles connector $\sS^\Theta(E,\la)$ 
whose lift to the universal cover of the cylinder goes from $\widetilde{S}_- = (0,0)$ to $\widetilde{S}_+ = (\pi,-\pi)$.
 Then, for all $E'\in [0,E)$, there exists a corridor $\cK_1(E',\la)$ of winding number $w(\cK_1) = 0$ for \refeq{dynsysTh}.
\end{prop}
\begin{proof}
Let $\cS(E,\la)= (\theta(\tau),\Theta_{E,\la}(\tau))$. Since $\cS^\Theta = \cW^- = \cW^+$, by the calculation done in the proof the 
previous Proposition, the graph of $\Theta_{E,\la}$ is entirely contained in $\cN$, thus it has to be monotone decreasing, and 
thus $\sin\Theta_{E,\la}(\tau)\leq 0$ for all $\tau$.  
\begin{figure}[ht]
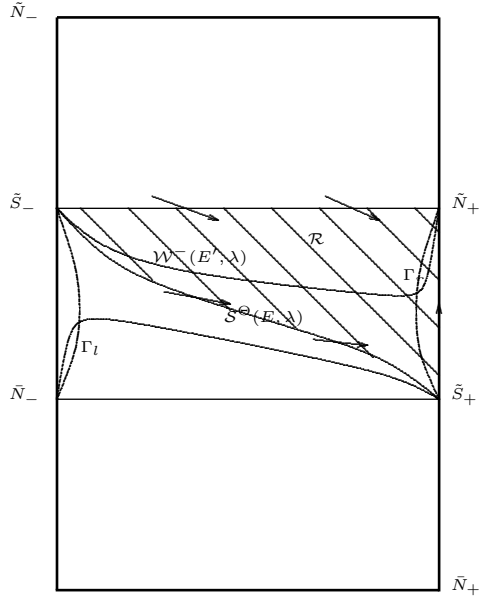

 \begin{center}
\font\thinlinefont=cmr5
\begingroup\makeatletter\ifx\SetFigFont\undefined%
\gdef\SetFigFont#1#2#3#4#5{%
  \reset@font\fontsize{#1}{#2pt}%
  \fontfamily{#3}\fontseries{#4}\fontshape{#5}%
  \selectfont}%
\fi\endgroup%
\mbox{\beginpicture
\setcoordinatesystem units <0.50000cm,0.50000cm>
\unitlength=0.50000cm
\linethickness=1pt
\setplotsymbol ({\makebox(0,0)[l]{\tencirc\symbol{'160}}})
\setshadesymbol ({\thinlinefont .})
\setlinear
%
%
\linethickness= 0.500pt
\setplotsymbol ({\thinlinefont .})
{\color[rgb]{0,0,0}\putrule from  5.080 16.510 to 15.240 16.510
}%
%
%
\linethickness=1pt
\setplotsymbol ({\makebox(0,0)[l]{\tencirc\symbol{'160}}})
{\color[rgb]{0,0,0}\putrule from  5.080 21.590 to 15.240 21.590
}%
%
%
\linethickness=1pt
\setplotsymbol ({\makebox(0,0)[l]{\tencirc\symbol{'160}}})
{\color[rgb]{0,0,0}\putrule from 15.240 21.590 to 15.240  6.315
\putrule from 15.240  6.350 to  5.045  6.350
\putrule from  5.080  6.350 to  5.080 21.590
}%
%
%
\linethickness= 0.500pt
\setplotsymbol ({\thinlinefont .})
{\color[rgb]{0,0,0}\putrule from  5.080 11.430 to 15.240 11.430
}%
%
%
\linethickness= 0.500pt
\setplotsymbol ({\thinlinefont .})
{\color[rgb]{0,0,0}\plot  7.938 14.287  9.684 13.970 /
%
%
\plot  9.422 13.953  9.684 13.970  9.445 14.078 /
}%
%
%
\linethickness= 0.500pt
\setplotsymbol ({\thinlinefont .})
{\color[rgb]{0,0,0}\plot 11.906 13.018 13.335 12.859 /
%
%
\plot 13.076 12.824 13.335 12.859 13.090 12.950 /
}%
%
%
\linethickness= 0.500pt
\setplotsymbol ({\thinlinefont .})
{\color[rgb]{0,0,0}\plot  7.620 16.828  9.366 16.192 /
%
%
\plot  9.106 16.220  9.366 16.192  9.149 16.339 /
}%
%
%
\linethickness= 0.500pt
\setplotsymbol ({\thinlinefont .})
{\color[rgb]{0,0,0}\plot 12.224 16.828 13.652 16.192 /
%
%
\plot 13.395 16.238 13.652 16.192 13.446 16.354 /
}%
%
%
\linethickness= 0.500pt
\setplotsymbol ({\thinlinefont .})
{\color[rgb]{0,0,0}\putrule from 15.240 12.700 to 15.240 13.970
%
%
\plot 15.304 13.716 15.240 13.970 15.176 13.716 /
}%
%
%
\linethickness= 0.500pt
\setplotsymbol ({\thinlinefont .})
{\color[rgb]{0,0,0}\plot  5.715 16.510  7.620 14.605 /
}%
%
%
\linethickness= 0.500pt
\setplotsymbol ({\thinlinefont .})
{\color[rgb]{0,0,0}\plot  6.985 16.510  9.684 13.811 /
}%
%
%
\linethickness= 0.500pt
\setplotsymbol ({\thinlinefont .})
{\color[rgb]{0,0,0}\plot  8.255 16.510 11.430 13.335 /
}%
%
%
\linethickness= 0.500pt
\setplotsymbol ({\thinlinefont .})
{\color[rgb]{0,0,0}\plot  9.525 16.510 13.494 12.541 /
}%
%
%
\linethickness= 0.500pt
\setplotsymbol ({\thinlinefont .})
{\color[rgb]{0,0,0}\plot 10.795 16.510 15.240 12.065 /
}%
%
%
\linethickness= 0.500pt
\setplotsymbol ({\thinlinefont .})
{\color[rgb]{0,0,0}\plot 12.065 16.510 15.240 13.335 /
}%
%
%
\linethickness= 0.500pt
\setplotsymbol ({\thinlinefont .})
{\color[rgb]{0,0,0}\plot 13.335 16.510 15.240 14.605 /
}%
%
%
\linethickness= 0.500pt
\setplotsymbol ({\thinlinefont .})
{\color[rgb]{0,0,0}\plot 14.605 16.510 15.240 15.875 /
}%
%
%
\linethickness= 0.500pt
\setplotsymbol ({\thinlinefont .})
{\color[rgb]{0,0,0}\plot  5.080 16.510  5.082 16.506 /
\plot  5.082 16.506  5.086 16.493 /
\plot  5.086 16.493  5.093 16.474 /
\plot  5.093 16.474  5.103 16.442 /
\plot  5.103 16.442  5.118 16.398 /
\plot  5.118 16.398  5.137 16.343 /
\plot  5.137 16.343  5.160 16.277 /
\plot  5.160 16.277  5.186 16.201 /
\plot  5.186 16.201  5.215 16.116 /
\plot  5.215 16.116  5.245 16.025 /
\plot  5.245 16.025  5.277 15.932 /
\plot  5.277 15.932  5.309 15.839 /
\plot  5.309 15.839  5.338 15.744 /
\plot  5.338 15.744  5.368 15.653 /
\plot  5.368 15.653  5.395 15.564 /
\plot  5.395 15.564  5.423 15.477 /
\plot  5.423 15.477  5.448 15.397 /
\plot  5.448 15.397  5.472 15.318 /
\plot  5.472 15.318  5.493 15.244 /
\plot  5.493 15.244  5.512 15.174 /
\plot  5.512 15.174  5.529 15.107 /
\plot  5.529 15.107  5.546 15.041 /
\plot  5.546 15.041  5.560 14.978 /
\plot  5.560 14.978  5.573 14.916 /
\plot  5.573 14.916  5.586 14.857 /
\plot  5.586 14.857  5.599 14.798 /
\plot  5.599 14.798  5.609 14.736 /
\plot  5.609 14.736  5.620 14.677 /
\plot  5.620 14.677  5.628 14.618 /
\plot  5.628 14.618  5.637 14.556 /
\plot  5.637 14.556  5.645 14.493 /
\plot  5.645 14.493  5.654 14.431 /
\plot  5.654 14.431  5.660 14.366 /
\plot  5.660 14.366  5.666 14.300 /
\plot  5.666 14.300  5.671 14.235 /
\plot  5.671 14.235  5.675 14.167 /
\plot  5.675 14.167  5.679 14.097 /
\plot  5.679 14.097  5.681 14.029 /
\plot  5.681 14.029  5.683 13.959 /
\plot  5.683 13.959  5.685 13.890 /
\putrule from  5.685 13.890 to  5.685 13.818
\putrule from  5.685 13.818 to  5.685 13.748
\plot  5.685 13.748  5.683 13.680 /
\plot  5.683 13.680  5.681 13.610 /
\plot  5.681 13.610  5.679 13.542 /
\plot  5.679 13.542  5.675 13.477 /
\plot  5.675 13.477  5.671 13.411 /
\plot  5.671 13.411  5.666 13.346 /
\plot  5.666 13.346  5.660 13.284 /
\plot  5.660 13.284  5.654 13.223 /
\plot  5.654 13.223  5.645 13.164 /
\plot  5.645 13.164  5.637 13.106 /
\plot  5.637 13.106  5.628 13.049 /
\plot  5.628 13.049  5.620 12.992 /
\plot  5.620 12.992  5.609 12.937 /
\plot  5.609 12.937  5.599 12.880 /
\plot  5.599 12.880  5.586 12.821 /
\plot  5.586 12.821  5.571 12.764 /
\plot  5.571 12.764  5.556 12.704 /
\plot  5.556 12.704  5.539 12.645 /
\plot  5.539 12.645  5.520 12.581 /
\plot  5.520 12.581  5.501 12.518 /
\plot  5.501 12.518  5.480 12.450 /
\plot  5.480 12.450  5.455 12.380 /
\plot  5.455 12.380  5.429 12.306 /
\plot  5.429 12.306  5.400 12.230 /
\plot  5.400 12.230  5.370 12.150 /
\plot  5.370 12.150  5.338 12.067 /
\plot  5.338 12.067  5.304 11.982 /
\plot  5.304 11.982  5.273 11.898 /
\plot  5.273 11.898  5.239 11.813 /
\plot  5.239 11.813  5.207 11.735 /
\plot  5.207 11.735  5.175 11.661 /
\plot  5.175 11.661  5.150 11.597 /
\plot  5.150 11.597  5.127 11.542 /
\plot  5.127 11.542  5.110 11.498 /
\plot  5.110 11.498  5.095 11.468 /
\plot  5.095 11.468  5.086 11.447 /
\plot  5.086 11.447  5.082 11.434 /
\plot  5.082 11.434  5.080 11.430 /
}%
%
%
\linethickness= 0.500pt
\setplotsymbol ({\thinlinefont .})
{\color[rgb]{0,0,0}\plot 15.240 16.510 15.238 16.506 /
\plot 15.238 16.506 15.234 16.493 /
\plot 15.234 16.493 15.227 16.474 /
\plot 15.227 16.474 15.217 16.442 /
\plot 15.217 16.442 15.202 16.398 /
\plot 15.202 16.398 15.183 16.343 /
\plot 15.183 16.343 15.160 16.275 /
\plot 15.160 16.275 15.134 16.199 /
\plot 15.134 16.199 15.105 16.114 /
\plot 15.105 16.114 15.075 16.023 /
\plot 15.075 16.023 15.043 15.930 /
\plot 15.043 15.930 15.011 15.835 /
\plot 15.011 15.835 14.982 15.740 /
\plot 14.982 15.740 14.952 15.646 /
\plot 14.952 15.646 14.925 15.558 /
\plot 14.925 15.558 14.897 15.471 /
\plot 14.897 15.471 14.872 15.388 /
\plot 14.872 15.388 14.848 15.308 /
\plot 14.848 15.308 14.827 15.234 /
\plot 14.827 15.234 14.808 15.162 /
\plot 14.808 15.162 14.791 15.092 /
\plot 14.791 15.092 14.774 15.026 /
\plot 14.774 15.026 14.760 14.961 /
\plot 14.760 14.961 14.747 14.897 /
\plot 14.747 14.897 14.734 14.836 /
\plot 14.734 14.836 14.721 14.772 /
\plot 14.721 14.772 14.711 14.711 /
\plot 14.711 14.711 14.700 14.647 /
\plot 14.700 14.647 14.692 14.586 /
\plot 14.692 14.586 14.683 14.520 /
\plot 14.683 14.520 14.675 14.455 /
\plot 14.675 14.455 14.666 14.389 /
\plot 14.666 14.389 14.660 14.321 /
\plot 14.660 14.321 14.654 14.252 /
\plot 14.654 14.252 14.649 14.182 /
\plot 14.649 14.182 14.645 14.110 /
\plot 14.645 14.110 14.641 14.038 /
\plot 14.641 14.038 14.639 13.964 /
\plot 14.639 13.964 14.637 13.892 /
\plot 14.637 13.892 14.635 13.818 /
\putrule from 14.635 13.818 to 14.635 13.744
\putrule from 14.635 13.744 to 14.635 13.669
\plot 14.635 13.669 14.637 13.595 /
\plot 14.637 13.595 14.639 13.523 /
\plot 14.639 13.523 14.641 13.454 /
\plot 14.641 13.454 14.645 13.384 /
\plot 14.645 13.384 14.649 13.316 /
\plot 14.649 13.316 14.654 13.248 /
\plot 14.654 13.248 14.660 13.185 /
\plot 14.660 13.185 14.666 13.121 /
\plot 14.666 13.121 14.675 13.060 /
\plot 14.675 13.060 14.683 13.001 /
\plot 14.683 13.001 14.692 12.943 /
\plot 14.692 12.943 14.700 12.888 /
\plot 14.700 12.888 14.711 12.831 /
\plot 14.711 12.831 14.721 12.774 /
\plot 14.721 12.774 14.734 12.717 /
\plot 14.734 12.717 14.749 12.660 /
\plot 14.749 12.660 14.764 12.603 /
\plot 14.764 12.603 14.781 12.543 /
\plot 14.781 12.543 14.800 12.486 /
\plot 14.800 12.486 14.819 12.425 /
\plot 14.819 12.425 14.840 12.361 /
\plot 14.840 12.361 14.865 12.296 /
\plot 14.865 12.296 14.891 12.228 /
\plot 14.891 12.228 14.920 12.156 /
\plot 14.920 12.156 14.950 12.082 /
\plot 14.950 12.082 14.982 12.006 /
\plot 14.982 12.006 15.016 11.930 /
\plot 15.016 11.930 15.047 11.851 /
\plot 15.047 11.851 15.081 11.777 /
\plot 15.081 11.777 15.113 11.705 /
\plot 15.113 11.705 15.145 11.637 /
\plot 15.145 11.637 15.170 11.580 /
\plot 15.170 11.580 15.193 11.532 /
\plot 15.193 11.532 15.210 11.491 /
\plot 15.210 11.491 15.225 11.464 /
\plot 15.225 11.464 15.234 11.445 /
\plot 15.234 11.445 15.238 11.434 /
\plot 15.238 11.434 15.240 11.430 /
}%
%
%
\linethickness= 0.500pt
\setplotsymbol ({\thinlinefont .})
{\color[rgb]{0,0,0}\plot  5.080 16.510  5.082 16.508 /
\plot  5.082 16.508  5.086 16.504 /
\plot  5.086 16.504  5.095 16.493 /
\plot  5.095 16.493  5.108 16.478 /
\plot  5.108 16.478  5.127 16.457 /
\plot  5.127 16.457  5.150 16.430 /
\plot  5.150 16.430  5.182 16.396 /
\plot  5.182 16.396  5.220 16.353 /
\plot  5.220 16.353  5.264 16.305 /
\plot  5.264 16.305  5.315 16.248 /
\plot  5.315 16.248  5.372 16.188 /
\plot  5.372 16.188  5.433 16.121 /
\plot  5.433 16.121  5.499 16.053 /
\plot  5.499 16.053  5.571 15.979 /
\plot  5.571 15.979  5.643 15.905 /
\plot  5.643 15.905  5.719 15.831 /
\plot  5.719 15.831  5.798 15.754 /
\plot  5.798 15.754  5.876 15.680 /
\plot  5.876 15.680  5.956 15.606 /
\plot  5.956 15.606  6.035 15.534 /
\plot  6.035 15.534  6.115 15.464 /
\plot  6.115 15.464  6.195 15.397 /
\plot  6.195 15.397  6.274 15.331 /
\plot  6.274 15.331  6.354 15.268 /
\plot  6.354 15.268  6.435 15.208 /
\plot  6.435 15.208  6.515 15.149 /
\plot  6.515 15.149  6.596 15.094 /
\plot  6.596 15.094  6.678 15.039 /
\plot  6.678 15.039  6.763 14.986 /
\plot  6.763 14.986  6.847 14.935 /
\plot  6.847 14.935  6.934 14.887 /
\plot  6.934 14.887  7.023 14.838 /
\plot  7.023 14.838  7.114 14.791 /
\plot  7.114 14.791  7.209 14.743 /
\plot  7.209 14.743  7.307 14.696 /
\plot  7.307 14.696  7.408 14.652 /
\plot  7.408 14.652  7.514 14.605 /
\plot  7.514 14.605  7.597 14.569 /
\plot  7.597 14.569  7.681 14.535 /
\plot  7.681 14.535  7.768 14.499 /
\plot  7.768 14.499  7.859 14.465 /
\plot  7.859 14.465  7.950 14.429 /
\plot  7.950 14.429  8.045 14.393 /
\plot  8.045 14.393  8.141 14.357 /
\plot  8.141 14.357  8.240 14.321 /
\plot  8.240 14.321  8.342 14.285 /
\plot  8.342 14.285  8.448 14.247 /
\plot  8.448 14.247  8.553 14.211 /
\plot  8.553 14.211  8.664 14.173 /
\plot  8.664 14.173  8.774 14.135 /
\plot  8.774 14.135  8.888 14.097 /
\plot  8.888 14.097  9.002 14.059 /
\plot  9.002 14.059  9.121 14.021 /
\plot  9.121 14.021  9.239 13.981 /
\plot  9.239 13.981  9.360 13.942 /
\plot  9.360 13.942  9.483 13.904 /
\plot  9.483 13.904  9.605 13.864 /
\plot  9.605 13.864  9.730 13.824 /
\plot  9.730 13.824  9.857 13.786 /
\plot  9.857 13.786  9.982 13.746 /
\plot  9.982 13.746 10.109 13.705 /
\plot 10.109 13.705 10.236 13.667 /
\plot 10.236 13.667 10.363 13.627 /
\plot 10.363 13.627 10.488 13.589 /
\plot 10.488 13.589 10.615 13.551 /
\plot 10.615 13.551 10.740 13.513 /
\plot 10.740 13.513 10.865 13.473 /
\plot 10.865 13.473 10.988 13.437 /
\plot 10.988 13.437 11.108 13.399 /
\plot 11.108 13.399 11.229 13.360 /
\plot 11.229 13.360 11.347 13.324 /
\plot 11.347 13.324 11.466 13.288 /
\plot 11.466 13.288 11.580 13.252 /
\plot 11.580 13.252 11.692 13.219 /
\plot 11.692 13.219 11.803 13.183 /
\plot 11.803 13.183 11.910 13.149 /
\plot 11.910 13.149 12.016 13.115 /
\plot 12.016 13.115 12.120 13.083 /
\plot 12.120 13.083 12.222 13.049 /
\plot 12.222 13.049 12.319 13.018 /
\plot 12.319 13.018 12.416 12.986 /
\plot 12.416 12.986 12.509 12.956 /
\plot 12.509 12.956 12.601 12.924 /
\plot 12.601 12.924 12.687 12.895 /
\plot 12.687 12.895 12.774 12.865 /
\plot 12.774 12.865 12.857 12.835 /
\plot 12.857 12.835 12.937 12.806 /
\putrule from 12.937 12.806 to 12.939 12.806
\plot 12.939 12.806 13.053 12.764 /
\plot 13.053 12.764 13.164 12.721 /
\plot 13.164 12.721 13.269 12.679 /
\plot 13.269 12.679 13.371 12.636 /
\plot 13.371 12.636 13.470 12.594 /
\plot 13.470 12.594 13.566 12.552 /
\plot 13.566 12.552 13.657 12.509 /
\plot 13.657 12.509 13.746 12.465 /
\plot 13.746 12.465 13.835 12.421 /
\plot 13.835 12.421 13.921 12.374 /
\plot 13.921 12.374 14.006 12.325 /
\plot 14.006 12.325 14.091 12.275 /
\plot 14.091 12.275 14.175 12.224 /
\plot 14.175 12.224 14.260 12.171 /
\plot 14.260 12.171 14.343 12.116 /
\plot 14.343 12.116 14.425 12.059 /
\plot 14.425 12.059 14.510 12.002 /
\plot 14.510 12.002 14.590 11.944 /
\plot 14.590 11.944 14.671 11.885 /
\plot 14.671 11.885 14.747 11.826 /
\plot 14.747 11.826 14.823 11.769 /
\plot 14.823 11.769 14.893 11.716 /
\plot 14.893 11.716 14.958 11.663 /
\plot 14.958 11.663 15.018 11.616 /
\plot 15.018 11.616 15.071 11.572 /
\plot 15.071 11.572 15.115 11.536 /
\plot 15.115 11.536 15.153 11.504 /
\plot 15.153 11.504 15.183 11.479 /
\plot 15.183 11.479 15.206 11.460 /
\plot 15.206 11.460 15.223 11.445 /
\plot 15.223 11.445 15.232 11.436 /
\plot 15.232 11.436 15.238 11.432 /
\plot 15.238 11.432 15.240 11.430 /
}%
%
%
\linethickness= 0.500pt
\setplotsymbol ({\thinlinefont .})
{\color[rgb]{0,0,0}\plot  5.080 16.510  5.084 16.506 /
\plot  5.084 16.506  5.091 16.499 /
\plot  5.091 16.499  5.105 16.487 /
\plot  5.105 16.487  5.127 16.466 /
\plot  5.127 16.466  5.156 16.438 /
\plot  5.156 16.438  5.194 16.402 /
\plot  5.194 16.402  5.241 16.360 /
\plot  5.241 16.360  5.292 16.313 /
\plot  5.292 16.313  5.347 16.260 /
\plot  5.347 16.260  5.406 16.207 /
\plot  5.406 16.207  5.465 16.154 /
\plot  5.465 16.154  5.527 16.099 /
\plot  5.527 16.099  5.588 16.046 /
\plot  5.588 16.046  5.645 15.996 /
\plot  5.645 15.996  5.702 15.949 /
\plot  5.702 15.949  5.759 15.903 /
\plot  5.759 15.903  5.812 15.860 /
\plot  5.812 15.860  5.863 15.820 /
\plot  5.863 15.820  5.914 15.782 /
\plot  5.914 15.782  5.965 15.746 /
\plot  5.965 15.746  6.013 15.712 /
\plot  6.013 15.712  6.064 15.678 /
\plot  6.064 15.678  6.115 15.646 /
\plot  6.115 15.646  6.166 15.615 /
\plot  6.166 15.615  6.219 15.583 /
\plot  6.219 15.583  6.259 15.560 /
\plot  6.259 15.560  6.301 15.536 /
\plot  6.301 15.536  6.344 15.513 /
\plot  6.344 15.513  6.390 15.490 /
\plot  6.390 15.490  6.437 15.466 /
\plot  6.437 15.466  6.485 15.443 /
\plot  6.485 15.443  6.536 15.418 /
\plot  6.536 15.418  6.591 15.395 /
\plot  6.591 15.395  6.646 15.369 /
\plot  6.646 15.369  6.703 15.344 /
\plot  6.703 15.344  6.765 15.318 /
\plot  6.765 15.318  6.828 15.295 /
\plot  6.828 15.295  6.894 15.270 /
\plot  6.894 15.270  6.964 15.244 /
\plot  6.964 15.244  7.036 15.219 /
\plot  7.036 15.219  7.110 15.193 /
\plot  7.110 15.193  7.186 15.168 /
\plot  7.186 15.168  7.264 15.143 /
\plot  7.264 15.143  7.345 15.117 /
\plot  7.345 15.117  7.429 15.094 /
\plot  7.429 15.094  7.514 15.069 /
\plot  7.514 15.069  7.601 15.045 /
\plot  7.601 15.045  7.690 15.022 /
\plot  7.690 15.022  7.781 14.999 /
\plot  7.781 14.999  7.874 14.975 /
\plot  7.874 14.975  7.967 14.954 /
\plot  7.967 14.954  8.062 14.931 /
\plot  8.062 14.931  8.160 14.910 /
\plot  8.160 14.910  8.259 14.889 /
\plot  8.259 14.889  8.361 14.870 /
\plot  8.361 14.870  8.462 14.848 /
\plot  8.462 14.848  8.568 14.829 /
\plot  8.568 14.829  8.674 14.810 /
\plot  8.674 14.810  8.784 14.789 /
\plot  8.784 14.789  8.867 14.776 /
\plot  8.867 14.776  8.951 14.762 /
\plot  8.951 14.762  9.038 14.747 /
\plot  9.038 14.747  9.129 14.734 /
\plot  9.129 14.734  9.220 14.719 /
\plot  9.220 14.719  9.313 14.704 /
\plot  9.313 14.704  9.409 14.692 /
\plot  9.409 14.692  9.506 14.677 /
\plot  9.506 14.677  9.605 14.662 /
\plot  9.605 14.662  9.707 14.647 /
\plot  9.707 14.647  9.811 14.635 /
\plot  9.811 14.635  9.914 14.620 /
\plot  9.914 14.620 10.022 14.605 /
\plot 10.022 14.605 10.132 14.590 /
\plot 10.132 14.590 10.243 14.577 /
\plot 10.243 14.577 10.355 14.563 /
\plot 10.355 14.563 10.467 14.548 /
\plot 10.467 14.548 10.581 14.533 /
\plot 10.581 14.533 10.696 14.520 /
\plot 10.696 14.520 10.812 14.506 /
\plot 10.812 14.506 10.926 14.493 /
\plot 10.926 14.493 11.043 14.478 /
\plot 11.043 14.478 11.159 14.465 /
\plot 11.159 14.465 11.273 14.450 /
\plot 11.273 14.450 11.390 14.438 /
\plot 11.390 14.438 11.504 14.425 /
\plot 11.504 14.425 11.616 14.412 /
\plot 11.616 14.412 11.728 14.400 /
\plot 11.728 14.400 11.839 14.389 /
\plot 11.839 14.389 11.946 14.376 /
\plot 11.946 14.376 12.052 14.366 /
\plot 12.052 14.366 12.158 14.353 /
\plot 12.158 14.353 12.260 14.343 /
\plot 12.260 14.343 12.361 14.332 /
\plot 12.361 14.332 12.459 14.323 /
\plot 12.459 14.323 12.554 14.313 /
\plot 12.554 14.313 12.647 14.304 /
\plot 12.647 14.304 12.736 14.294 /
\plot 12.736 14.294 12.823 14.285 /
\plot 12.823 14.285 12.907 14.277 /
\plot 12.907 14.277 12.990 14.271 /
\plot 12.990 14.271 13.070 14.262 /
\plot 13.070 14.262 13.147 14.256 /
\plot 13.147 14.256 13.221 14.247 /
\plot 13.221 14.247 13.293 14.241 /
\plot 13.293 14.241 13.363 14.235 /
\plot 13.363 14.235 13.464 14.226 /
\plot 13.464 14.226 13.559 14.216 /
\plot 13.559 14.216 13.652 14.209 /
\plot 13.652 14.209 13.739 14.201 /
\plot 13.739 14.201 13.824 14.194 /
\plot 13.824 14.194 13.904 14.190 /
\plot 13.904 14.190 13.981 14.186 /
\plot 13.981 14.186 14.055 14.182 /
\plot 14.055 14.182 14.125 14.180 /
\putrule from 14.125 14.180 to 14.190 14.180
\putrule from 14.190 14.180 to 14.254 14.180
\plot 14.254 14.180 14.313 14.182 /
\plot 14.313 14.182 14.368 14.186 /
\plot 14.368 14.186 14.421 14.192 /
\plot 14.421 14.192 14.470 14.199 /
\plot 14.470 14.199 14.514 14.207 /
\plot 14.514 14.207 14.556 14.218 /
\plot 14.556 14.218 14.594 14.230 /
\plot 14.594 14.230 14.630 14.245 /
\plot 14.630 14.245 14.662 14.262 /
\plot 14.662 14.262 14.692 14.279 /
\plot 14.692 14.279 14.719 14.300 /
\plot 14.719 14.300 14.745 14.321 /
\plot 14.745 14.321 14.766 14.345 /
\plot 14.766 14.345 14.787 14.370 /
\plot 14.787 14.370 14.806 14.398 /
\plot 14.806 14.398 14.823 14.427 /
\plot 14.823 14.427 14.840 14.457 /
\plot 14.840 14.457 14.855 14.491 /
\plot 14.855 14.491 14.870 14.525 /
\plot 14.870 14.525 14.884 14.563 /
\plot 14.884 14.563 14.897 14.603 /
\plot 14.897 14.603 14.912 14.647 /
\plot 14.912 14.647 14.927 14.694 /
\plot 14.927 14.694 14.939 14.745 /
\plot 14.939 14.745 14.954 14.800 /
\plot 14.954 14.800 14.967 14.861 /
\plot 14.967 14.861 14.982 14.925 /
\plot 14.982 14.925 14.997 14.997 /
\plot 14.997 14.997 15.011 15.073 /
\plot 15.011 15.073 15.026 15.155 /
\plot 15.026 15.155 15.043 15.242 /
\plot 15.043 15.242 15.058 15.337 /
\plot 15.058 15.337 15.075 15.437 /
\plot 15.075 15.437 15.092 15.541 /
\plot 15.092 15.541 15.109 15.649 /
\plot 15.109 15.649 15.128 15.756 /
\plot 15.128 15.756 15.145 15.867 /
\plot 15.145 15.867 15.160 15.972 /
\plot 15.160 15.972 15.176 16.076 /
\plot 15.176 16.076 15.189 16.169 /
\plot 15.189 16.169 15.202 16.256 /
\plot 15.202 16.256 15.215 16.328 /
\plot 15.215 16.328 15.223 16.389 /
\plot 15.223 16.389 15.229 16.436 /
\plot 15.229 16.436 15.234 16.470 /
\plot 15.234 16.470 15.238 16.493 /
\plot 15.238 16.493 15.240 16.504 /
\putrule from 15.240 16.504 to 15.240 16.510
}%
%
%
\linethickness= 0.500pt
\setplotsymbol ({\thinlinefont .})
{\color[rgb]{0,0,0}\plot 15.240 11.430 15.238 11.432 /
\plot 15.238 11.432 15.232 11.434 /
\plot 15.232 11.434 15.221 11.441 /
\plot 15.221 11.441 15.206 11.451 /
\plot 15.206 11.451 15.183 11.466 /
\plot 15.183 11.466 15.153 11.483 /
\plot 15.153 11.483 15.115 11.506 /
\plot 15.115 11.506 15.069 11.532 /
\plot 15.069 11.532 15.018 11.561 /
\plot 15.018 11.561 14.961 11.593 /
\plot 14.961 11.593 14.897 11.627 /
\plot 14.897 11.627 14.829 11.663 /
\plot 14.829 11.663 14.757 11.699 /
\plot 14.757 11.699 14.683 11.735 /
\plot 14.683 11.735 14.609 11.771 /
\plot 14.609 11.771 14.531 11.805 /
\plot 14.531 11.805 14.450 11.841 /
\plot 14.450 11.841 14.368 11.874 /
\plot 14.368 11.874 14.285 11.906 /
\plot 14.285 11.906 14.199 11.938 /
\plot 14.199 11.938 14.110 11.970 /
\plot 14.110 11.970 14.019 12.002 /
\plot 14.019 12.002 13.921 12.033 /
\plot 13.921 12.033 13.822 12.063 /
\plot 13.822 12.063 13.716 12.095 /
\plot 13.716 12.095 13.604 12.126 /
\plot 13.604 12.126 13.485 12.158 /
\plot 13.485 12.158 13.360 12.190 /
\plot 13.360 12.190 13.229 12.224 /
\plot 13.229 12.224 13.149 12.243 /
\plot 13.149 12.243 13.068 12.264 /
\plot 13.068 12.264 12.984 12.283 /
\plot 12.984 12.283 12.897 12.304 /
\plot 12.897 12.304 12.806 12.325 /
\plot 12.806 12.325 12.713 12.349 /
\plot 12.713 12.349 12.617 12.370 /
\plot 12.617 12.370 12.520 12.391 /
\plot 12.520 12.391 12.421 12.414 /
\plot 12.421 12.414 12.317 12.438 /
\plot 12.317 12.438 12.211 12.461 /
\plot 12.211 12.461 12.103 12.486 /
\plot 12.103 12.486 11.993 12.509 /
\plot 11.993 12.509 11.881 12.535 /
\plot 11.881 12.535 11.764 12.558 /
\plot 11.764 12.558 11.648 12.584 /
\plot 11.648 12.584 11.527 12.609 /
\plot 11.527 12.609 11.407 12.634 /
\plot 11.407 12.634 11.284 12.662 /
\plot 11.284 12.662 11.161 12.687 /
\plot 11.161 12.687 11.034 12.713 /
\plot 11.034 12.713 10.907 12.740 /
\plot 10.907 12.740 10.780 12.766 /
\plot 10.780 12.766 10.653 12.791 /
\plot 10.653 12.791 10.524 12.819 /
\plot 10.524 12.819 10.395 12.844 /
\plot 10.395 12.844 10.266 12.869 /
\plot 10.266 12.869 10.137 12.897 /
\plot 10.137 12.897 10.008 12.922 /
\plot 10.008 12.922  9.881 12.948 /
\plot  9.881 12.948  9.754 12.973 /
\plot  9.754 12.973  9.627 12.996 /
\plot  9.627 12.996  9.504 13.022 /
\plot  9.504 13.022  9.379 13.045 /
\plot  9.379 13.045  9.258 13.068 /
\plot  9.258 13.068  9.138 13.092 /
\plot  9.138 13.092  9.021 13.115 /
\plot  9.021 13.115  8.905 13.136 /
\plot  8.905 13.136  8.791 13.159 /
\plot  8.791 13.159  8.680 13.180 /
\plot  8.680 13.180  8.570 13.200 /
\plot  8.570 13.200  8.465 13.221 /
\plot  8.465 13.221  8.361 13.240 /
\plot  8.361 13.240  8.259 13.259 /
\plot  8.259 13.259  8.160 13.276 /
\plot  8.160 13.276  8.064 13.295 /
\plot  8.064 13.295  7.969 13.312 /
\plot  7.969 13.312  7.878 13.329 /
\plot  7.878 13.329  7.791 13.343 /
\plot  7.791 13.343  7.705 13.358 /
\plot  7.705 13.358  7.622 13.373 /
\plot  7.622 13.373  7.542 13.388 /
\plot  7.542 13.388  7.415 13.409 /
\plot  7.415 13.409  7.294 13.430 /
\plot  7.294 13.430  7.178 13.449 /
\plot  7.178 13.449  7.065 13.468 /
\plot  7.065 13.468  6.960 13.485 /
\plot  6.960 13.485  6.856 13.500 /
\plot  6.856 13.500  6.759 13.515 /
\plot  6.759 13.515  6.663 13.528 /
\plot  6.663 13.528  6.574 13.538 /
\plot  6.574 13.538  6.488 13.547 /
\plot  6.488 13.547  6.405 13.555 /
\plot  6.405 13.555  6.327 13.559 /
\plot  6.327 13.559  6.253 13.564 /
\plot  6.253 13.564  6.183 13.566 /
\putrule from  6.183 13.566 to  6.115 13.566
\plot  6.115 13.566  6.054 13.564 /
\plot  6.054 13.564  5.994 13.561 /
\plot  5.994 13.561  5.941 13.555 /
\plot  5.941 13.555  5.891 13.547 /
\plot  5.891 13.547  5.844 13.538 /
\plot  5.844 13.538  5.802 13.526 /
\plot  5.802 13.526  5.762 13.513 /
\plot  5.762 13.513  5.726 13.498 /
\plot  5.726 13.498  5.692 13.481 /
\plot  5.692 13.481  5.660 13.462 /
\plot  5.660 13.462  5.632 13.441 /
\plot  5.632 13.441  5.607 13.420 /
\plot  5.607 13.420  5.584 13.394 /
\plot  5.584 13.394  5.563 13.369 /
\plot  5.563 13.369  5.541 13.341 /
\plot  5.541 13.341  5.522 13.314 /
\plot  5.522 13.314  5.503 13.282 /
\plot  5.503 13.282  5.482 13.244 /
\plot  5.482 13.244  5.463 13.204 /
\plot  5.463 13.204  5.444 13.159 /
\plot  5.444 13.159  5.425 13.113 /
\plot  5.425 13.113  5.406 13.060 /
\plot  5.406 13.060  5.389 13.003 /
\plot  5.389 13.003  5.370 12.939 /
\plot  5.370 12.939  5.353 12.871 /
\plot  5.353 12.871  5.334 12.797 /
\plot  5.334 12.797  5.315 12.717 /
\plot  5.315 12.717  5.296 12.630 /
\plot  5.296 12.630  5.277 12.537 /
\plot  5.277 12.537  5.256 12.438 /
\plot  5.256 12.438  5.237 12.334 /
\plot  5.237 12.334  5.215 12.228 /
\plot  5.215 12.228  5.196 12.118 /
\plot  5.196 12.118  5.177 12.010 /
\plot  5.177 12.010  5.158 11.904 /
\plot  5.158 11.904  5.141 11.803 /
\plot  5.141 11.803  5.127 11.712 /
\plot  5.127 11.712  5.112 11.631 /
\plot  5.112 11.631  5.101 11.565 /
\plot  5.101 11.565  5.093 11.513 /
\plot  5.093 11.513  5.086 11.474 /
\plot  5.086 11.474  5.082 11.449 /
\plot  5.082 11.449  5.080 11.436 /
\putrule from  5.080 11.436 to  5.080 11.430
}%
%
%
\put{\SetFigFont{6}{7.2}{\rmdefault}{\mddefault}{\updefault}{\color[rgb]{0,0,0}$\tilde{S}_-$}%
} [lB] at  3.810 16.510
%
%
\put{\SetFigFont{6}{7.2}{\rmdefault}{\mddefault}{\updefault}{\color[rgb]{0,0,0}$\bar{N}_-$}%
} [lB] at  3.810 11.430
%
%
\put{\SetFigFont{6}{7.2}{\rmdefault}{\mddefault}{\updefault}{\color[rgb]{0,0,0}$\tilde{N}_-$}%
} [lB] at  3.810 21.590
%
%
\put{\SetFigFont{6}{7.2}{\rmdefault}{\mddefault}{\updefault}{\color[rgb]{0,0,0}$\tilde{S}_+$}%
} [lB] at 15.558 11.430
%
%
\put{\SetFigFont{6}{7.2}{\rmdefault}{\mddefault}{\updefault}{\color[rgb]{0,0,0}$\tilde{N}_+$}%
} [lB] at 15.558 16.510
%
%
\put{\SetFigFont{6}{7.2}{\rmdefault}{\mddefault}{\updefault}{\color[rgb]{0,0,0}$\bar{N}_+$}%
} [lB] at 15.558  6.350
%
%
\put{\SetFigFont{6}{7.2}{\rmdefault}{\mddefault}{\updefault}{\color[rgb]{0,0,0}$\Ga_l$}%
} [lB] at  5.715 12.700
%
%
\put{\SetFigFont{6}{7.2}{\rmdefault}{\mddefault}{\updefault}{\color[rgb]{0,0,0}$\Ga_r$}%
} [lB] at 14.287 14.605
%
%
\put{\SetFigFont{6}{7.2}{\rmdefault}{\mddefault}{\updefault}{\color[rgb]{0,0,0}$\cS^\Theta(E,\la)$}%
} [lB] at  9.525 13.494
%
%
\put{\SetFigFont{6}{7.2}{\rmdefault}{\mddefault}{\updefault}{\color[rgb]{0,0,0}$\cR$}%
} [lB] at 11.748 15.558
%
%
\put{\SetFigFont{6}{7.2}{\rmdefault}{\mddefault}{\updefault}{\color[rgb]{0,0,0}$\cW^-(E',\la)$}%
} [lB] at  7.620 15.081
\linethickness=0pt
\putrectangle corners at  3.778 21.901 and 15.589  6.215
\endpicture}

 \end{center}
\caption{Trapped region $\cR$ and the corridor with zero winding number.}
\end{figure}

 Consider the region $\cR$ in $\bar{\cC}$ whose boundary consists of the following three curves: (i) From $S_-$ to $S_+$ along 
$\Theta_{E,\la}$.
  (ii) From $S_+$ to $N_+$ along the right boundary, and (iii) from $N_+$ to $S_-$ along a horizontal line segment.
  We claim that  $\cR$ is a trapped region for \refeq{dynsysTh} at parameter values $(E',\la)$ provided $E'<E$.
  Since (ii) is always an orbit, and (iii) is entirely contained in $\cN$, this only needs to be checked for (i).
  We have
$$
g_{E',\la}(\theta,\Theta_{E,\la}) = g_{E,\la}(\theta,\Theta_{E,\la}) + 2a(E'-E)\sin^2\theta \sin\Theta_{E,\la} 
\geq \dot{\Theta}_{E,\la}
$$
which shows that the new flow crosses the old solution from left to right.
  Thus $\cR$ is trapped and since $\cW^-$ starts in $\cR$, it must terminate in $\widetilde{N}_+ = (\pi,0)$, so that $w(\cW^-) = 0$.
\end{proof}
Setting $E_0 =1$ and $\la_0  = -1+a$, let $\cS_0 := \cS(E_0,\la_0)$ denote the explicit solution found in \ref{explicit}.
  For all $E_1\in (0,1)$, the above two propositions, together with the following immediate corollary of 
Proposition~\ref{prop:ssexists}, establish the existence of a saddles connector $\cS_1 := \cS(E_1,\la_1)$ for \refeq{dynsysTh}, 
for some $\la_1\in (\la_l,\la_0)$:
\begin{cor}\label{cor:thconex}
Let $E\in[0,1]$ be fixed. Suppose that there exists $\la_1< \la_2<0$ such that the flow \refeq{dynsysTh} has corridors 
$\cK_1(E,\la_1)$ and $\cK_1(E,\la_2)$ with $w(\cK_1(E,\la_2))=0$ and $w(\cK_1(E,\la_1))\geq 1$.
  Then there is a $\la\in(\la_1,\la_2)$ such that \refeq{dynsysTh} has a saddles connector $\cS^\Theta(E,\la)$ going from $(0,0)$ 
to $(\pi,-\pi)$.
\end{cor}
\begin{proof}
Proposition~\ref{prop:ssexists} applies, with $-\la$ playing the role of the parameter $\mu$. 
\end{proof}
Proceeding iteratively, suppose that given $E_n\in[0,1]$ a saddles connector $\cS^\Theta_n=\cS^\Theta(E_n,\la_n)$ has been found 
for \refeq{dynsysTh}, for some $\la_n <\la_0$.
 In the next subsection we shall see how the newly-found $\la_n$ can be used to prove the existence of a saddles connector for the 
$\Omega$ flow \refeq{dynsysOm}, namely $\cS^\Om_n:= \cS^\Om_n(E_{n+1},\la_n)$ for some $E_{n+1} \in (0,1)$.
  Coming back to the $\Theta$ flow then, a new saddles connector $\cS_{n+1}^\Theta$ needs to be found with 
the updated energy $E_{n+1}$, {\em given that a saddles connector $\cS^\Theta_n(E_n,\la_n)$ already exists.}
 Since $E_{n+1}$ can be on either side of $E_n$, in addition to Prop.~\ref{prop:thlowercor} we also need the following:
\begin{thm}\label{thm:La}
Given any $a\in(0,\half)$ and $E \in [0,1]$, there exists a unique 
$$
\la = \Lambda(E) \in [-1-a,-1+a]
$$ 
such that \refeq{dynsysTh} has a saddles connector $\cS(E,\la)$ going from $(0,0)$ to $(\pi,-\pi)$.
 Moreover, $\La$ is an increasing $C^1$ function, and $\frac{\p \La}{\p E} < a$. 
\end{thm}
\begin{proof}
If $E=1$ then $\la = -1+a$ works and $\cS = \cS_0 = (\theta,-\theta)$.
  For $E<1$ existence is guaranteed by Prop.~\ref{prop:thlowercor}, Prop.~\ref{prop:lacorlow}, and Corollary~\ref{cor:thconex}.  
To prove uniqueness, suppose that for a given $E$, there are two saddles connectors $\cS(E,\la)$ and $\cS'(E,\la')$, with $\la'>\la$.  
Let $\Theta_{E,\la}$ and $\Theta_{E,\la'}$ denote the corresponding $\Theta$-components of $\cS$ and $\cS'$, respectively.  
For $\theta\in(0,\pi)$,
$$
g_{E,\la'}(\theta,\Theta_{E,\la}) = g_{E,\la}(\theta,\Theta_{E,\la}) + 2(\la'-\la)\sin\theta > 
\dot{\Theta}_{E,\la}.
$$
Thus orbits of the $(E,\la')$ flow can only cross $\cS$ from below to above.
On the other hand, since $\cS'$ is a saddles connector, near $S_-$ it coincides with $\cW^-(E,\la')$, and near $S_+$ it coincides 
with $\cW^+$.  
Thus from the linearizartion of the flow at $S_\pm$,
$$
\left.\frac{d\Theta_{E,\la'}}{d\theta}\right|_{S_\pm} = \la'-a>\la-a.
$$
Therefore $\cS'$ must be above $\cS$ near $S_-$ and below it near $S_+$, so $\cS'$ would have to cross $\cS$ from above to 
below, which is a contradiction, unless they coincide.

Given $E$ then, let $\La(E)$ denote the unique value of $\la$ for which a saddles connector $\cS(E,\la)$ exists. 
The fact that $\La$ is continuously differentiable (in fact analytic), and the bound on the derivative, have already been 
shown in \cite{WinklmeierPHD} and \cite{WINKLMEIERa} using analytic perturbation theory. 
 Here we give a simple proof of monotonicity of $\La$ which also establishes the bound on the derivative:

Given $E\in[0,1]$ let $\la=\La(E)$ and let $(\theta(\tau),\Theta_E(\tau))$ denote the unique (modulo translations in $\tau$) saddles 
connector for \refeq{dynsysTh} whose existence we have established.
Let $u: = \frac{\p \Theta_E}{\p E}$. 
By differentiating the $\Theta$ equation in \refeq{dynsysTh} with respect to $E$ we obtain an equation for $u$:
\beq\label{eq:uu}
\frac{du}{d\tau} = P(\tau) u + Q(\tau),\qquad \left\{\begin{array}{rcl} P & := & 2a\sin\theta\cos\theta \sin\Theta_E + 
(2aE\sin^2\theta - 1)\cos\Theta_E \\
Q & := & 2a\sin^2\theta\sin\Theta_E + 2 \frac{d\La}{dE} \sin\theta.\end{array}\right.
\eeq
Let 
$$
U(\tau_1^{},\tau_2^{}) := e^{-\int_{\tau_1^{}}^{\tau_2^{}} P(\tau) d\tau}.
$$
Thus we have 
$$
U(\tau_2^{},\tau_1^{}) = \frac{1}{U(\tau_1^{},\tau_2^{})},\qquad U(\tau_1^{},\tau_2^{})U(\tau_2^{},\tau_3^{}) = U(\tau_1^{},\tau_3^{}).
$$
Solving the first-order linear ODE \refeq{eq:uu} for $u$ we obtain
\beq\label{UQ}
u(\tau) = U(\tau,\tau_1^{}) u(\tau_1^{}) + \int_{\tau_1^{}}^\tau U(\tau,\tau') Q(\tau') d\tau'.
\eeq
Note that $P$ and $Q$ are bounded functions of $\tau$ and 
$$
\lim_{\tau\to -\infty} P(\tau) = -1,\qquad \lim_{\tau\to\infty} P(\tau) = 1.
$$
Therefore, for any fixed $\tau$,
$$
U(\tau,\tau_1^{}) \to 0\mbox{ as }\tau_1^{}\to -\infty,\qquad U(\tau,\tau_1^{}) \to 0\mbox{ as } \tau_1^{}\to \infty.
$$
 Moreover $u(\pm\infty) = 0$.
 For any finite $\tau$ from \refeq{UQ} we thus obtain two equivalent expressions for $u(\tau)$.
  For example, setting $\tau=0$,
$$
u(0) = \int_{-\infty}^0 U(0,\tau') Q(\tau') d\tau' =-  \int_0^\infty U(0,\tau') d\tau'.
$$
Thus in particular
$$
0 
= 
\int_{-\infty}^\infty U(0,\tau') Q(\tau') d\tau' = 2a\int_\RR U(0,\tau')\sin^2\theta(\tau')\sin\Theta_E(\tau') d\tau' +
 2 \frac{d\La}{d E} \int_\RR U(0,\tau') \sin\theta(\tau') d\tau'
$$
Therefore
$$
\frac{d\La}{dE} 
=a \frac{\int_\RR U(0,\tau')\sin^2\theta(\tau')\left(-\sin\Theta_E(\tau')\right) d\tau'}{\int_\RR U(0,\tau') \sin\theta(\tau') d\tau'} \geq 0.
$$
We have already shown that $\Theta_E(\tau)\in(-\pi,0)$ and $\theta(\tau)\in(0,\pi)$ for all $\tau$.
  Thus the numerator in the above fraction is strictly less than the denominator, hence
\beq\label{est:dlade}
0\leq \frac{\p\La}{\p E} < a
\eeq
\end{proof}
\subsubsection{Existence of saddles connectors for the $\Om$ equation}
We now show that the flow \refeq{dynsysOm} also satisfies all the hypotheses we had made about flows on a cylinder.
 Once again, for simplicity we are only going to consider the case $\ka =\half$.
  The situation is somewhat more complicated than what we have done in the above for the $\Theta$ equation, due to 
the presence of an extra parameter, namely $\gamma$, which breaks the symmetry that was present for the $\Theta$ flow, 
as well as the fact that the equilibria of the $\Omega$ flow are degenerate (non-hyperbolic). 

 Let us make the identifications $x = \xi$ and $y=\Om$.  
Thus $x_-^{} = -\pi/2$, $x_+^{} = \pi/2$, and we now have 
$$
f(\xi) = \cos^2 \xi,
\qquad 
g_{E,\la}(\xi,\Om) = 2a\sin\xi\cos\Om+2\la\cos\xi\sin\Om + 2\ga\sin\xi\cos\xi+\cos^2\xi - 2aE
$$
Therefore, for $E\in[0,1)$,
$$
s_-^{} = -\pi+\cos^{-1}(E),\quad n_-^{} = \pi - \cos^{-1}(E),\qquad s_+^{} = -\cos^{-1}(E),\quad n_+^{} = \cos^{-1}(E),
$$
where by $\cos^{-1}$ we mean the principal branch of the arccosine, $0\leq \cos^{-1}x\leq \pi$, and
$$
S_\pm = (\pm{\textstyle\frac{\pi}{2}},s_\pm),\qquad N_\pm = (\pm{\textstyle\frac{\pi}{2}},n_\pm)
$$
as before.
  We note that this time, all the equilibria are non-hyperbolic, since $f'(\pm\pi/2) = 0$, and that for $\ga= 0$,
 there is a discrete symmetry: $f(-\xi) = f(\xi)$ and $g_{|_{\ga=0}}(-\xi,\pi-\Om) = g_{|_{\ga=0}}(\xi,\Om)$, which 
is broken when $\ga$ is turned on.

Also in the case $E=1$ there is a further degeneracy: the two equilibria on each side coalesce into one singular point 
with both eigenvalues equal to zero.
  For these type of singular points center manifolds can be non-unique, so that the distinguished orbits $\cW^\pm$ and 
the index theory we have developed for the corridor they form, are not directly relevent to this case.  

We now check the hypotheses about the topology of the nullclines.
\subsubsection{Topology of the null-clines}
Let $T := \tan(\Om/2)$.  Then
$$
g_{E,\la}(\xi,\Om) = \frac{2q(T)}{1+T^2},
$$
where
\bea
q(T) & := & \left(\ga\sin\xi\cos\xi +{\textstyle\half}\cos^2\xi - aE -a\sin\xi\right)T^2 + 2\la\cos\xi T\\
&&\mbox{} + \left(\ga\sin\xi\cos\xi +{\textstyle\half}\cos^2\xi - aE +a\sin\xi\right).
\eea
Thus $q$ is a quadratic polynomial in $T$ with coefficients that are functions of $\xi$.  The discriminant of $q$ is
$$
\Delta_q(\xi) := \la^2\cos^2\xi - (\ga\sin\xi\cos\xi + {\textstyle\half}\cos^2\xi - aE)^2 + 4a^2\sin^2\xi.
$$
Let $\tau := \tan\xi$.  
Thus $-\infty<\tau<\infty$, and
$$
\Delta_q(\tau) = \frac{p(\tau)}{(1+\tau^2)^2}
$$
with 
$$
p(\tau) 
:= 
a^2(1-E^2) \tau^4 + 2\ga a E \tau^3 + \left( \la^2 -\ga^2 + a^2 + 2a({\textstyle\half} - aE)\right) \tau^2 - 2 \ga({\textstyle\half} - aE) \tau 
+ \la^2 - ({\textstyle\half} - aE)^2.
$$
Since $ |E|<1$,  $p$ is an irreducible quartic in $\tau$.
 Let us write it as $p(\tau) = \sum_{j=0}^4 c_j \tau^j$.
 Suppose that $E\in(0,1)$ and
\beq\label{paramrange}
a \in (0,{\textstyle\half}),\qquad \ga \in (- \sqrt{2a(1-2a)},0),\qquad \la \in [-1-a,-1+a]
\eeq
It then follows that
$$
c_4>0,\qquad c_3<0,\qquad c_2>0,\qquad c_1>0,\qquad c_0>0.
$$
By Descartes' Rule of Signs, then, $p(\tau)$ has either
 two or no real positive roots, and either two or no real negative roots, counting multiplicity.
For a more accurate root count, one needs to use the discriminant of the quartic.
  Since the discriminant theory for general quartics is somewhat complicated, here we opt for a 
simpler analysis by estimating $p$ from above and below with two {\em reducible} quartics.
To this end, first we note that
$
p(\tau) = Q(\tau) - (q_1(\tau))^2
$
with
$$
Q(\tau) = \la^2 + (\la^2+a^2) \tau^2 + a^2 \tau^4,\qquad q_1(\tau) = {\textstyle\half} - aE + \ga \tau - aE \tau^2.
$$
Thus on the one hand, by completing the square,
$$
Q(\tau) = \left(a\tau^2+\frac{\la^2+a^2}{2a}\right)^2 - \frac{(\la^2 - a^2)^2}{4a^2} \leq (q_2^+(\tau))^2,\qquad q_2^+(\tau) 
:= a\tau^2 + \frac{\la^2+a^2}{2a}
$$
and on the other hand,
$$
Q(\tau) \geq \la^2 + 2 |\la| a  \tau^2 + a^2 \tau^4  = (q_2^-(\tau))^2,\qquad q_2^-(\tau) = a\tau^2 + |\la|
$$
Thus we have upper and lower bounds for $p$ in terms of factorizable quartics $Q^\pm(\tau)$:
$$
Q^-(\tau) := \left(q_2^-(\tau)\right)^2 - \left(q_1(\tau)\right)^2 \leq p(\tau) \leq \left(q_2^+(\tau)\right)^2 - \left(q_1(\tau)\right)^2 =: Q^+(\tau)
$$
\begin{figure}[ht]
 \begin{center}
\includegraphics[scale =0.3]{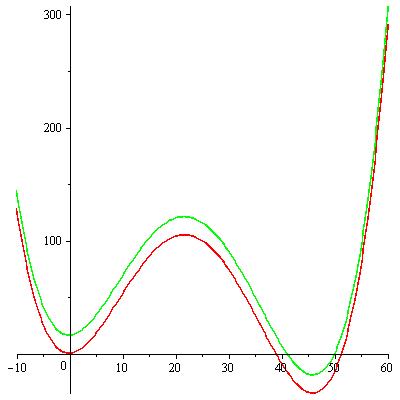}
\includegraphics[scale=0.3]{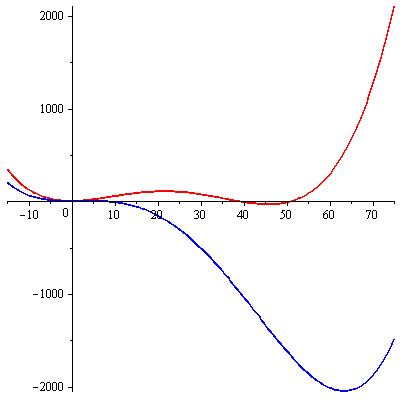}
 \end{center}
\caption{The bounding quartics: $Q^+$ and $p$ (left), $p$ and $Q^-$ (right) for parameter values $a=0.1$, $\ga=-0.2$, $\la = -0.9$, $E = 0.978$.}
\end{figure}

Consider first the upper quartic $Q^+$.
  We have $Q^+ = (q_2^+ + q_1)(q_2^+ - q_1) = q^+ q^-$ where $q^\pm$ are two quadratic polynomials.
  It is clear that if $q^+$ has any real roots, they will be positive, and if $q^-$ has any roots, 
they will be negative (recall that we are assuming $\ga<0$).
  The discriminants of $q^\pm$ are computed to be
$$
\Delta_\pm := \ga^2 - 2 (1 \mp E)\left( \la^2 + a^2 \pm a(1 -2aE)\right).
$$

Similarly, $Q^- = (q_2^-+q_1)(q_2^--q_1) = \tilde{q}^+\tilde{q}^-$ where $\tilde{q}^\pm$ are two quadratics, and once again,
 any real roots of $\tilde{q}^+$ must be positive and any real root of $\tilde{q}^-$, negative.
  The discriminants of $\tilde{q}^\pm$ are
$$
\tilde{\Delta}_\pm = \ga^2 - 2(1\mp E) (-2a\la \pm a(1 - 2aE) )\geq \Delta_\pm
$$
It thus follows that there are two subsets of the $(a,\ga,\la)$ parameter space \refeq{paramrange}  that are of interest:
 (R1) where both $\tilde{\Delta}_+$ and $\tilde{\Delta}_-$ are negative; and (R2) where $\Delta_+>0$ and $\tilde{\Delta}_-<0$.
  For parameter values in the region (R1) the quartic $Q^-$ will have no real zeros, and will be always positive, while 
for those in (R2) the quartic $Q^+$ will have exactly two positive roots and no negative root.  

Let us fix $a,\ga,\la$ as in \refeq{paramrange}.  We find that the range (R1) corresponds to $0\leq  E < E_l$, where
$$
E_l(\la) := \frac{1}{2a}\left[ -\la+a+{\textstyle\half} - \sqrt{(\la+a-{\textstyle\half})^2+\ga^2}\right],
$$
while range (R2) corresponds to $E_h< E\leq 1$, with
$$
E_h(\la) = \frac{1}{4a^2}\left[\la^2+3a^2+a - \sqrt{(\la^2 - a^2 + a)^2 + 4 a^2 \ga^2}\right]
$$
Note that  $0< E_l<E_h<1$ for all values of $a,\ga,\la$ as in \refeq{paramrange}.

 For $E\in(E_h,1]$, therefore, since $Q^+$ has two positive roots, the quartic $p(\tau)$ must also have at least two roots, one of 
which will definitely be positive.
  Thus by the Rule of Signs, $p(\tau)$ has exactly two positive roots.
  We call them $\tau_1^{}$ and $\tau_2^{}$, and $p(\tau)<0$ for $\tau_1^{}<\tau<\tau_2^{}$.
  It follows that the quadratic $q(T)$ will have two roots for $\tau\notin[\tau_1^{},\tau_2^{}]$, 
double roots at $\tau_1^{}$ and at $\tau_2^{}$, and no real roots for $\tau\in(\tau_1^{},\tau_2^{})$.  

For $E\in[0,E_l)$ since $Q^-$ has no real roots, $p$ cannot have any either.
 Thus $p$ is always positive and $q(T)=0$ will have two roots for all $\tau\in \RR$.  

Combining these two, one concludes that a critical value for the energy $E=E_c(a,\ga,\la)$ exists, $E_c\in[E_l,E_h]$, such 
that assumption ({\bf A}) is satisfied, with the role of parameter $\mu$  played by $E$.  
\begin{figure}[ht]
 \begin{center}
  \includegraphics[scale=0.3]{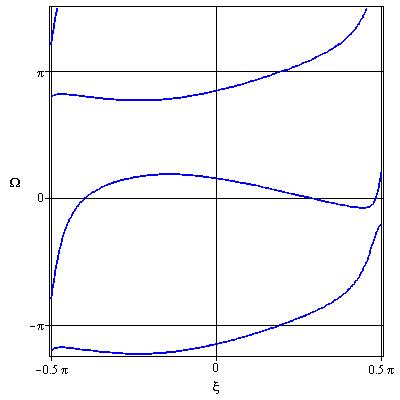}
\includegraphics[scale=0.3]{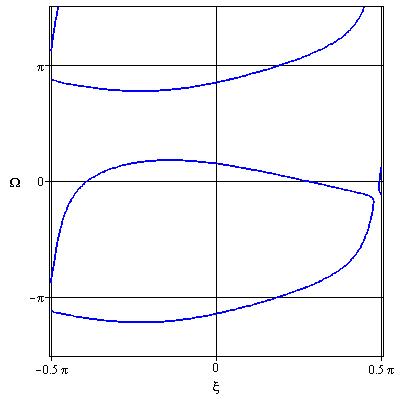}
 \end{center}
\caption{$\Omega$-nullclines for parameter values $E=0.8$ (left) and $E=0.93$ (right), with $a=0.1$, $\ga = -0.4$, 
and $\lambda = -0.9$.}
\end{figure}

\subsubsection{Existence of corridors with unequal winding number}
Throughout this section, $a$ will be a fixed number in $(0,\half)$ and $\ga$ a fixed number in $(- \sqrt{2a(1-2a)},0)$. 
 The following two propositions  help us get started:
\begin{prop}\label{prop:om1}
Given $\la \leq -1 +a$ there exists $\bar{E}\in(E_h(\la),1)$ such that for all $E\in[\bar{E},1)$ the corridor
 $\cK_1(E,\la)$ of the flow \refeq{dynsysOm} has winding number greater than or equal to one.
\end{prop}
\begin{proof}
We compute the slope of solution orbits that cross the following line in $\bar{\cC}$
$$L := \{(\xi,\Om)\in \bar{\cC}\ |\ -{\textstyle\frac{\pi}{2}}\leq\xi\leq{\textstyle\frac{\pi}{2}},\quad \Om = {\textstyle\frac{\pi}{2}} - \cos^{-1}E - \xi\}$$
and compare it to the slope of $L$.
  Note that $L$ passes through $\widetilde{N}_-$ and $\widetilde{S}_+$.
  We have
\beq\label{slopes}
\frac{g_{E,\la}(\xi,\Om)}{f(\xi)} - (-1) = 2\left( 1 - (a-\la)E + [ \sqrt{1-E^2} (a-\la) + \ga]\tan\xi \right)
\eeq
Consider first the case $\la = -1+a$.  Let
$$
\xi_0(E) := \tan^{-1} \frac{1 - E}{-\ga-\sqrt{1-E^2}} \searrow 0 \mbox{ as } E \nearrow 1.
$$
For $\xi\geq \xi_0(E)$, the slope of any orbit of the flow crossing $L$ is less than (i.e. more negative than) the slope of $L$.
  Hence on the portion of $L$ where $\xi\geq \xi_0$ orbits can only cross $L$ from above to below.

On the other hand, suppose $\la<-1+a$.
  Let
$$
E_m := \max\left\{\frac{1}{a-\la},\sqrt{1-\frac{\ga^2}{(a-\la)^2}},E_h(\la)\right\}\in(E_h,1)
$$
For $E\in (E_m,1)$ we have $\xi_0(E)\in(-\frac{\pi}{2},0)$.
  Let 
$$
\eta_0 := \left\{\begin{array}{lr}
                     \cos^{-1}E & \la<-1+a \\ 
\cos^{-1}E +\xi_0(E) & \la =-1+a
                    \end{array}\right.
$$
We note that $\eta_0 \to 0$ as $E\to 1$.
  Let us consider the horizontal line
$$
L' := \{(\xi,\Om)\in\bar{\cC}\ |\ -{\textstyle\frac{\pi}{2}}\leq \xi\leq {\textstyle\frac{\pi}{2}},\quad \Om = {\textstyle\frac{\pi}{2}} - \eta_0\}
$$
We compute the slope of orbits crossing $L'$:
$$
h(\xi) := g_{E,\la}(\xi,{\textstyle\frac{\pi}{2}}-\eta_0) = 2a\sin\xi \sin\eta_0 + 2\la\cos\xi\cos\eta_0 + 2\ga \sin\xi\cos\xi + \cos^2\xi -2aE.
$$
Clearly
\bea
h(\xi) & \leq & 2a(\sin\eta_0 - E)+|\ga| +2\la\cos\xi\cos\eta_0+ \cos^2\xi \\
&  \leq &  2a(\sin\eta_0 - E)+|\ga| + 2\la\cos\eta_0 +1 \\
& \to & -2a + |\ga| + 2\la + 1\quad\mbox{as } E\to 1\\
& \leq & |\ga| - 1 < 0
\eea
Thus there exists $E$ sufficiently close to 1 such that $h(\xi)<0$ for $\xi\in(-{\textstyle\frac{\pi}{2}},{\textstyle\frac{\pi}{2}})$.  
\begin{figure}[ht]
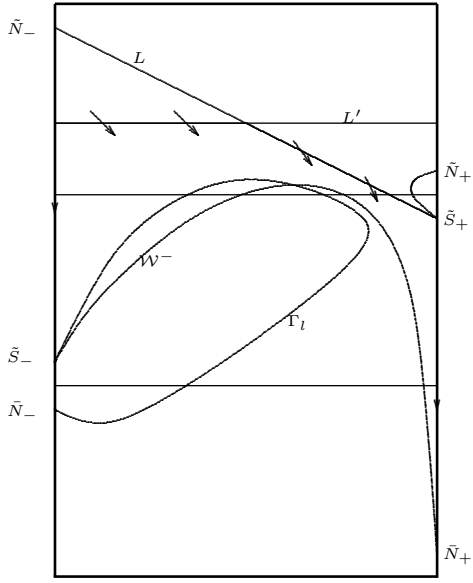

 \begin{center}
%
%
\font\thinlinefont=cmr5
\begingroup\makeatletter\ifx\SetFigFont\undefined%
\gdef\SetFigFont#1#2#3#4#5{%
  \reset@font\fontsize{#1}{#2pt}%
  \fontfamily{#3}\fontseries{#4}\fontshape{#5}%
  \selectfont}%
\fi\endgroup%
\mbox{\beginpicture
\setcoordinatesystem units <0.50000cm,0.50000cm>
\unitlength=0.50000cm
\linethickness=1pt
\setplotsymbol ({\makebox(0,0)[l]{\tencirc\symbol{'160}}})
\setshadesymbol ({\thinlinefont .})
\setlinear
%
%
\linethickness= 0.500pt
\setplotsymbol ({\thinlinefont .})
{\color[rgb]{0,0,0}\putrule from  5.080 11.430 to 15.240 11.430
}%
%
%
\linethickness=1pt
\setplotsymbol ({\makebox(0,0)[l]{\tencirc\symbol{'160}}})
{\color[rgb]{0,0,0}\putrule from  5.080 21.590 to 15.240 21.590
}%
%
%
\linethickness= 0.500pt
\setplotsymbol ({\thinlinefont .})
{\color[rgb]{0,0,0}\putrule from  5.080 17.462 to  5.080 16.034
%
%
\plot  5.017 16.288  5.080 16.034  5.143 16.288 /
}%
%
%
\linethickness= 0.500pt
\setplotsymbol ({\thinlinefont .})
{\color[rgb]{0,0,0}\putrule from 15.240 12.383 to 15.240 10.795
%
%
\plot 15.176 11.049 15.240 10.795 15.304 11.049 /
}%
%
%
\linethickness= 0.500pt
\setplotsymbol ({\thinlinefont .})
{\color[rgb]{0,0,0}\putrule from  5.080 16.510 to 15.240 16.510
}%
%
%
\linethickness=1pt
\setplotsymbol ({\makebox(0,0)[l]{\tencirc\symbol{'160}}})
{\color[rgb]{0,0,0}\putrule from 15.240 21.590 to 15.240  6.315
\putrule from 15.240  6.350 to  5.045  6.350
\putrule from  5.080  6.350 to  5.080 21.590
}%
%
%
\linethickness= 0.500pt
\setplotsymbol ({\thinlinefont .})
{\color[rgb]{0,0,0}\plot 15.240 15.875  5.080 20.955 /
}%
%
%
\linethickness= 0.500pt
\setplotsymbol ({\thinlinefont .})
{\color[rgb]{0,0,0}\putrule from 15.240 18.415 to  5.080 18.415
}%
%
%
\linethickness= 0.500pt
\setplotsymbol ({\thinlinefont .})
{\color[rgb]{0,0,0}\putrule from  5.080 18.415 to 10.160 18.415
\plot 10.160 18.415 15.240 15.875 /
}%
%
%
\linethickness= 0.500pt
\setplotsymbol ({\thinlinefont .})
{\color[rgb]{0,0,0}\plot  6.032 18.733  6.668 18.098 /
%
%
\plot  6.443 18.232  6.668 18.098  6.533 18.322 /
}%
%
%
\linethickness= 0.500pt
\setplotsymbol ({\thinlinefont .})
{\color[rgb]{0,0,0}\plot  8.255 18.733  8.890 18.098 /
%
%
\plot  8.665 18.232  8.890 18.098  8.755 18.322 /
}%
%
%
\linethickness= 0.500pt
\setplotsymbol ({\thinlinefont .})
{\color[rgb]{0,0,0}\plot 11.430 17.939 11.906 17.304 /
%
%
\plot 11.703 17.469 11.906 17.304 11.805 17.545 /
}%
%
%
\linethickness= 0.500pt
\setplotsymbol ({\thinlinefont .})
{\color[rgb]{0,0,0}\plot 13.335 16.986 13.652 16.351 /
%
%
\plot 13.482 16.550 13.652 16.351 13.596 16.607 /
}%
%
%
\linethickness= 0.500pt
\setplotsymbol ({\thinlinefont .})
{\color[rgb]{0,0,0}\plot 15.240 15.875 15.234 15.879 /
\plot 15.234 15.879 15.221 15.890 /
\plot 15.221 15.890 15.200 15.909 /
\plot 15.200 15.909 15.168 15.934 /
\plot 15.168 15.934 15.128 15.968 /
\plot 15.128 15.968 15.079 16.010 /
\plot 15.079 16.010 15.028 16.055 /
\plot 15.028 16.055 14.973 16.101 /
\plot 14.973 16.101 14.920 16.148 /
\plot 14.920 16.148 14.872 16.192 /
\plot 14.872 16.192 14.825 16.235 /
\plot 14.825 16.235 14.783 16.273 /
\plot 14.783 16.273 14.747 16.311 /
\plot 14.747 16.311 14.713 16.345 /
\plot 14.713 16.345 14.685 16.375 /
\plot 14.685 16.375 14.660 16.404 /
\plot 14.660 16.404 14.639 16.432 /
\plot 14.639 16.432 14.620 16.457 /
\plot 14.620 16.457 14.605 16.482 /
\plot 14.605 16.482 14.590 16.512 /
\plot 14.590 16.512 14.577 16.540 /
\plot 14.577 16.540 14.567 16.567 /
\plot 14.567 16.567 14.561 16.597 /
\plot 14.561 16.597 14.554 16.624 /
\plot 14.554 16.624 14.552 16.654 /
\putrule from 14.552 16.654 to 14.552 16.681
\plot 14.552 16.681 14.554 16.709 /
\plot 14.554 16.709 14.558 16.736 /
\plot 14.558 16.736 14.565 16.764 /
\plot 14.565 16.764 14.575 16.789 /
\plot 14.575 16.789 14.586 16.815 /
\plot 14.586 16.815 14.599 16.838 /
\plot 14.599 16.838 14.613 16.859 /
\plot 14.613 16.859 14.628 16.880 /
\plot 14.628 16.880 14.647 16.899 /
\plot 14.647 16.899 14.664 16.916 /
\plot 14.664 16.916 14.685 16.933 /
\plot 14.685 16.933 14.709 16.950 /
\plot 14.709 16.950 14.734 16.967 /
\plot 14.734 16.967 14.764 16.984 /
\plot 14.764 16.984 14.796 16.999 /
\plot 14.796 16.999 14.834 17.016 /
\plot 14.834 17.016 14.876 17.033 /
\plot 14.876 17.033 14.925 17.050 /
\plot 14.925 17.050 14.978 17.067 /
\plot 14.978 17.067 15.033 17.084 /
\plot 15.033 17.084 15.088 17.101 /
\plot 15.088 17.101 15.138 17.115 /
\plot 15.138 17.115 15.181 17.128 /
\plot 15.181 17.128 15.212 17.137 /
\plot 15.212 17.137 15.232 17.143 /
\plot 15.232 17.143 15.238 17.145 /
\putrule from 15.238 17.145 to 15.240 17.145
}%
%
%
\linethickness= 0.500pt
\setplotsymbol ({\thinlinefont .})
{\color[rgb]{0,0,0}\plot  5.080 12.065  5.082 12.067 /
\plot  5.082 12.067  5.084 12.076 /
\plot  5.084 12.076  5.091 12.088 /
\plot  5.091 12.088  5.101 12.107 /
\plot  5.101 12.107  5.114 12.137 /
\plot  5.114 12.137  5.133 12.175 /
\plot  5.133 12.175  5.156 12.224 /
\plot  5.156 12.224  5.186 12.283 /
\plot  5.186 12.283  5.220 12.353 /
\plot  5.220 12.353  5.258 12.431 /
\plot  5.258 12.431  5.302 12.520 /
\plot  5.302 12.520  5.351 12.617 /
\plot  5.351 12.617  5.402 12.721 /
\plot  5.402 12.721  5.457 12.831 /
\plot  5.457 12.831  5.514 12.946 /
\plot  5.514 12.946  5.573 13.064 /
\plot  5.573 13.064  5.635 13.183 /
\plot  5.635 13.183  5.698 13.303 /
\plot  5.698 13.303  5.759 13.424 /
\plot  5.759 13.424  5.823 13.542 /
\plot  5.823 13.542  5.884 13.659 /
\plot  5.884 13.659  5.946 13.771 /
\plot  5.946 13.771  6.005 13.881 /
\plot  6.005 13.881  6.064 13.989 /
\plot  6.064 13.989  6.121 14.091 /
\plot  6.121 14.091  6.179 14.190 /
\plot  6.179 14.190  6.234 14.285 /
\plot  6.234 14.285  6.289 14.376 /
\plot  6.289 14.376  6.342 14.463 /
\plot  6.342 14.463  6.394 14.548 /
\plot  6.394 14.548  6.447 14.628 /
\plot  6.447 14.628  6.498 14.704 /
\plot  6.498 14.704  6.549 14.779 /
\plot  6.549 14.779  6.598 14.848 /
\plot  6.598 14.848  6.648 14.918 /
\plot  6.648 14.918  6.699 14.984 /
\plot  6.699 14.984  6.750 15.050 /
\plot  6.750 15.050  6.801 15.111 /
\plot  6.801 15.111  6.854 15.172 /
\plot  6.854 15.172  6.905 15.234 /
\plot  6.905 15.234  6.960 15.293 /
\plot  6.960 15.293  7.015 15.354 /
\plot  7.015 15.354  7.074 15.416 /
\plot  7.074 15.416  7.133 15.475 /
\plot  7.133 15.475  7.195 15.534 /
\plot  7.195 15.534  7.256 15.593 /
\plot  7.256 15.593  7.319 15.653 /
\plot  7.319 15.653  7.385 15.710 /
\plot  7.385 15.710  7.453 15.767 /
\plot  7.453 15.767  7.521 15.824 /
\plot  7.521 15.824  7.592 15.881 /
\plot  7.592 15.881  7.664 15.936 /
\plot  7.664 15.936  7.739 15.991 /
\plot  7.739 15.991  7.815 16.046 /
\plot  7.815 16.046  7.891 16.099 /
\plot  7.891 16.099  7.969 16.152 /
\plot  7.969 16.152  8.050 16.203 /
\plot  8.050 16.203  8.132 16.254 /
\plot  8.132 16.254  8.215 16.303 /
\plot  8.215 16.303  8.297 16.351 /
\plot  8.297 16.351  8.382 16.396 /
\plot  8.382 16.396  8.467 16.440 /
\plot  8.467 16.440  8.553 16.482 /
\plot  8.553 16.482  8.638 16.525 /
\plot  8.638 16.525  8.725 16.563 /
\plot  8.725 16.563  8.812 16.601 /
\plot  8.812 16.601  8.896 16.635 /
\plot  8.896 16.635  8.983 16.669 /
\plot  8.983 16.669  9.068 16.698 /
\plot  9.068 16.698  9.155 16.728 /
\plot  9.155 16.728  9.239 16.753 /
\plot  9.239 16.753  9.322 16.779 /
\plot  9.322 16.779  9.406 16.802 /
\plot  9.406 16.802  9.489 16.821 /
\plot  9.489 16.821  9.572 16.840 /
\plot  9.572 16.840  9.652 16.855 /
\plot  9.652 16.855  9.732 16.870 /
\plot  9.732 16.870  9.815 16.883 /
\plot  9.815 16.883  9.893 16.893 /
\plot  9.893 16.893  9.974 16.899 /
\plot  9.974 16.899 10.054 16.906 /
\plot 10.054 16.906 10.139 16.912 /
\plot 10.139 16.912 10.224 16.914 /
\putrule from 10.224 16.914 to 10.308 16.914
\putrule from 10.308 16.914 to 10.395 16.914
\plot 10.395 16.914 10.482 16.910 /
\plot 10.482 16.910 10.569 16.906 /
\plot 10.569 16.906 10.657 16.899 /
\plot 10.657 16.899 10.746 16.889 /
\plot 10.746 16.889 10.837 16.878 /
\plot 10.837 16.878 10.928 16.866 /
\plot 10.928 16.866 11.019 16.851 /
\plot 11.019 16.851 11.110 16.836 /
\plot 11.110 16.836 11.204 16.817 /
\plot 11.204 16.817 11.295 16.798 /
\plot 11.295 16.798 11.388 16.777 /
\plot 11.388 16.777 11.479 16.753 /
\plot 11.479 16.753 11.570 16.730 /
\plot 11.570 16.730 11.661 16.705 /
\plot 11.661 16.705 11.750 16.679 /
\plot 11.750 16.679 11.839 16.652 /
\plot 11.839 16.652 11.925 16.622 /
\plot 11.925 16.622 12.010 16.593 /
\plot 12.010 16.593 12.093 16.563 /
\plot 12.093 16.563 12.175 16.531 /
\plot 12.175 16.531 12.253 16.499 /
\plot 12.253 16.499 12.330 16.468 /
\plot 12.330 16.468 12.404 16.436 /
\plot 12.404 16.436 12.476 16.404 /
\plot 12.476 16.404 12.545 16.372 /
\plot 12.545 16.372 12.611 16.339 /
\plot 12.611 16.339 12.675 16.307 /
\plot 12.675 16.307 12.736 16.275 /
\plot 12.736 16.275 12.793 16.243 /
\plot 12.793 16.243 12.848 16.212 /
\plot 12.848 16.212 12.901 16.180 /
\plot 12.901 16.180 12.952 16.148 /
\plot 12.952 16.148 12.998 16.118 /
\plot 12.998 16.118 13.045 16.087 /
\plot 13.045 16.087 13.083 16.059 /
\plot 13.083 16.059 13.119 16.032 /
\plot 13.119 16.032 13.155 16.004 /
\plot 13.155 16.004 13.189 15.977 /
\plot 13.189 15.977 13.221 15.947 /
\plot 13.221 15.947 13.250 15.919 /
\plot 13.250 15.919 13.278 15.890 /
\plot 13.278 15.890 13.303 15.860 /
\plot 13.303 15.860 13.327 15.828 /
\plot 13.327 15.828 13.348 15.799 /
\plot 13.348 15.799 13.367 15.767 /
\plot 13.367 15.767 13.382 15.733 /
\plot 13.382 15.733 13.396 15.699 /
\plot 13.396 15.699 13.407 15.665 /
\plot 13.407 15.665 13.415 15.629 /
\plot 13.415 15.629 13.420 15.591 /
\plot 13.420 15.591 13.422 15.553 /
\putrule from 13.422 15.553 to 13.422 15.515
\plot 13.422 15.515 13.418 15.475 /
\plot 13.418 15.475 13.409 15.433 /
\plot 13.409 15.433 13.399 15.390 /
\plot 13.399 15.390 13.386 15.346 /
\plot 13.386 15.346 13.369 15.301 /
\plot 13.369 15.301 13.348 15.255 /
\plot 13.348 15.255 13.324 15.208 /
\plot 13.324 15.208 13.299 15.160 /
\plot 13.299 15.160 13.267 15.111 /
\plot 13.267 15.111 13.233 15.060 /
\plot 13.233 15.060 13.197 15.009 /
\plot 13.197 15.009 13.157 14.956 /
\plot 13.157 14.956 13.115 14.903 /
\plot 13.115 14.903 13.068 14.848 /
\plot 13.068 14.848 13.020 14.793 /
\plot 13.020 14.793 12.967 14.736 /
\plot 12.967 14.736 12.912 14.679 /
\plot 12.912 14.679 12.852 14.620 /
\plot 12.852 14.620 12.791 14.558 /
\plot 12.791 14.558 12.728 14.497 /
\plot 12.728 14.497 12.660 14.434 /
\plot 12.660 14.434 12.588 14.370 /
\plot 12.588 14.370 12.514 14.302 /
\plot 12.514 14.302 12.435 14.235 /
\plot 12.435 14.235 12.370 14.177 /
\plot 12.370 14.177 12.302 14.120 /
\plot 12.302 14.120 12.232 14.061 /
\plot 12.232 14.061 12.158 14.000 /
\plot 12.158 14.000 12.082 13.936 /
\plot 12.082 13.936 12.004 13.873 /
\plot 12.004 13.873 11.923 13.807 /
\plot 11.923 13.807 11.841 13.741 /
\plot 11.841 13.741 11.754 13.672 /
\plot 11.754 13.672 11.665 13.602 /
\plot 11.665 13.602 11.574 13.532 /
\plot 11.574 13.532 11.481 13.458 /
\plot 11.481 13.458 11.386 13.386 /
\plot 11.386 13.386 11.286 13.310 /
\plot 11.286 13.310 11.187 13.233 /
\plot 11.187 13.233 11.083 13.157 /
\plot 11.083 13.157 10.979 13.079 /
\plot 10.979 13.079 10.871 13.001 /
\plot 10.871 13.001 10.763 12.920 /
\plot 10.763 12.920 10.653 12.840 /
\plot 10.653 12.840 10.543 12.759 /
\plot 10.543 12.759 10.431 12.679 /
\plot 10.431 12.679 10.317 12.596 /
\plot 10.317 12.596 10.202 12.516 /
\plot 10.202 12.516 10.088 12.435 /
\plot 10.088 12.435  9.972 12.355 /
\plot  9.972 12.355  9.855 12.275 /
\plot  9.855 12.275  9.741 12.194 /
\plot  9.741 12.194  9.624 12.116 /
\plot  9.624 12.116  9.510 12.037 /
\plot  9.510 12.037  9.396 11.961 /
\plot  9.396 11.961  9.282 11.885 /
\plot  9.282 11.885  9.167 11.811 /
\plot  9.167 11.811  9.057 11.737 /
\plot  9.057 11.737  8.945 11.665 /
\plot  8.945 11.665  8.837 11.595 /
\plot  8.837 11.595  8.729 11.527 /
\plot  8.729 11.527  8.623 11.460 /
\plot  8.623 11.460  8.520 11.394 /
\plot  8.520 11.394  8.416 11.333 /
\plot  8.416 11.333  8.316 11.271 /
\plot  8.316 11.271  8.217 11.212 /
\plot  8.217 11.212  8.122 11.157 /
\plot  8.122 11.157  8.029 11.102 /
\plot  8.029 11.102  7.935 11.049 /
\plot  7.935 11.049  7.846 10.998 /
\plot  7.846 10.998  7.760 10.950 /
\plot  7.760 10.950  7.675 10.905 /
\plot  7.675 10.905  7.590 10.861 /
\plot  7.590 10.861  7.510 10.818 /
\plot  7.510 10.818  7.432 10.780 /
\plot  7.432 10.780  7.355 10.742 /
\plot  7.355 10.742  7.243 10.689 /
\plot  7.243 10.689  7.133 10.640 /
\plot  7.133 10.640  7.029 10.598 /
\plot  7.029 10.598  6.930 10.560 /
\plot  6.930 10.560  6.833 10.528 /
\plot  6.833 10.528  6.739 10.501 /
\plot  6.739 10.501  6.648 10.478 /
\plot  6.648 10.478  6.560 10.458 /
\plot  6.560 10.458  6.475 10.446 /
\plot  6.475 10.446  6.388 10.437 /
\plot  6.388 10.437  6.303 10.433 /
\putrule from  6.303 10.433 to  6.221 10.433
\plot  6.221 10.433  6.136 10.437 /
\plot  6.136 10.437  6.054 10.446 /
\plot  6.054 10.446  5.971 10.458 /
\plot  5.971 10.458  5.889 10.473 /
\plot  5.889 10.473  5.806 10.492 /
\plot  5.806 10.492  5.726 10.516 /
\plot  5.726 10.516  5.647 10.541 /
\plot  5.647 10.541  5.569 10.566 /
\plot  5.569 10.566  5.495 10.594 /
\plot  5.495 10.594  5.425 10.624 /
\plot  5.425 10.624  5.362 10.651 /
\plot  5.362 10.651  5.302 10.679 /
\plot  5.302 10.679  5.249 10.704 /
\plot  5.249 10.704  5.205 10.727 /
\plot  5.205 10.727  5.167 10.746 /
\plot  5.167 10.746  5.137 10.763 /
\plot  5.137 10.763  5.114 10.776 /
\plot  5.114 10.776  5.097 10.784 /
\plot  5.097 10.784  5.088 10.791 /
\plot  5.088 10.791  5.082 10.793 /
\plot  5.082 10.793  5.080 10.795 /
}%
%
%
\linethickness= 0.500pt
\setplotsymbol ({\thinlinefont .})
{\color[rgb]{0,0,0}\plot  5.080 12.065  5.082 12.069 /
\plot  5.082 12.069  5.086 12.076 /
\plot  5.086 12.076  5.095 12.090 /
\plot  5.095 12.090  5.108 12.112 /
\plot  5.108 12.112  5.127 12.143 /
\plot  5.127 12.143  5.150 12.181 /
\plot  5.150 12.181  5.177 12.230 /
\plot  5.177 12.230  5.211 12.285 /
\plot  5.211 12.285  5.249 12.349 /
\plot  5.249 12.349  5.292 12.416 /
\plot  5.292 12.416  5.334 12.488 /
\plot  5.334 12.488  5.381 12.565 /
\plot  5.381 12.565  5.429 12.639 /
\plot  5.429 12.639  5.476 12.715 /
\plot  5.476 12.715  5.524 12.789 /
\plot  5.524 12.789  5.571 12.863 /
\plot  5.571 12.863  5.618 12.933 /
\plot  5.618 12.933  5.664 13.003 /
\plot  5.664 13.003  5.709 13.068 /
\plot  5.709 13.068  5.753 13.132 /
\plot  5.753 13.132  5.798 13.193 /
\plot  5.798 13.193  5.840 13.252 /
\plot  5.840 13.252  5.884 13.312 /
\plot  5.884 13.312  5.929 13.369 /
\plot  5.929 13.369  5.973 13.424 /
\plot  5.973 13.424  6.020 13.481 /
\plot  6.020 13.481  6.066 13.538 /
\plot  6.066 13.538  6.115 13.595 /
\plot  6.115 13.595  6.166 13.652 /
\plot  6.166 13.652  6.208 13.701 /
\plot  6.208 13.701  6.253 13.752 /
\plot  6.253 13.752  6.297 13.801 /
\plot  6.297 13.801  6.346 13.854 /
\plot  6.346 13.854  6.394 13.906 /
\plot  6.394 13.906  6.445 13.959 /
\plot  6.445 13.959  6.498 14.014 /
\plot  6.498 14.014  6.551 14.069 /
\plot  6.551 14.069  6.608 14.127 /
\plot  6.608 14.127  6.665 14.186 /
\plot  6.665 14.186  6.725 14.243 /
\plot  6.725 14.243  6.786 14.304 /
\plot  6.786 14.304  6.847 14.364 /
\plot  6.847 14.364  6.911 14.425 /
\plot  6.911 14.425  6.977 14.486 /
\plot  6.977 14.486  7.042 14.550 /
\plot  7.042 14.550  7.110 14.611 /
\plot  7.110 14.611  7.178 14.673 /
\plot  7.178 14.673  7.245 14.736 /
\plot  7.245 14.736  7.313 14.798 /
\plot  7.313 14.798  7.383 14.857 /
\plot  7.383 14.857  7.451 14.918 /
\plot  7.451 14.918  7.521 14.978 /
\plot  7.521 14.978  7.588 15.037 /
\plot  7.588 15.037  7.656 15.094 /
\plot  7.656 15.094  7.724 15.149 /
\plot  7.724 15.149  7.791 15.204 /
\plot  7.791 15.204  7.857 15.257 /
\plot  7.857 15.257  7.923 15.310 /
\plot  7.923 15.310  7.988 15.361 /
\plot  7.988 15.361  8.052 15.409 /
\plot  8.052 15.409  8.115 15.458 /
\plot  8.115 15.458  8.177 15.505 /
\plot  8.177 15.505  8.238 15.549 /
\plot  8.238 15.549  8.299 15.593 /
\plot  8.299 15.593  8.361 15.636 /
\plot  8.361 15.636  8.424 15.680 /
\plot  8.424 15.680  8.490 15.725 /
\plot  8.490 15.725  8.553 15.767 /
\plot  8.553 15.767  8.619 15.809 /
\plot  8.619 15.809  8.683 15.852 /
\plot  8.683 15.852  8.750 15.892 /
\plot  8.750 15.892  8.816 15.932 /
\plot  8.816 15.932  8.884 15.970 /
\plot  8.884 15.970  8.951 16.010 /
\plot  8.951 16.010  9.019 16.049 /
\plot  9.019 16.049  9.089 16.087 /
\plot  9.089 16.087  9.159 16.123 /
\plot  9.159 16.123  9.231 16.161 /
\plot  9.231 16.161  9.301 16.195 /
\plot  9.301 16.195  9.373 16.231 /
\plot  9.373 16.231  9.442 16.264 /
\plot  9.442 16.264  9.514 16.296 /
\plot  9.514 16.296  9.584 16.328 /
\plot  9.584 16.328  9.656 16.358 /
\plot  9.656 16.358  9.726 16.387 /
\plot  9.726 16.387  9.794 16.415 /
\plot  9.794 16.415  9.864 16.442 /
\plot  9.864 16.442  9.931 16.468 /
\plot  9.931 16.468  9.997 16.491 /
\plot  9.997 16.491 10.063 16.514 /
\plot 10.063 16.514 10.128 16.535 /
\plot 10.128 16.535 10.192 16.557 /
\plot 10.192 16.557 10.253 16.576 /
\plot 10.253 16.576 10.315 16.595 /
\plot 10.315 16.595 10.376 16.612 /
\plot 10.376 16.612 10.435 16.626 /
\plot 10.435 16.626 10.494 16.641 /
\plot 10.494 16.641 10.552 16.656 /
\plot 10.552 16.656 10.611 16.669 /
\plot 10.611 16.669 10.674 16.684 /
\plot 10.674 16.684 10.740 16.696 /
\plot 10.740 16.696 10.808 16.709 /
\plot 10.808 16.709 10.873 16.720 /
\plot 10.873 16.720 10.941 16.730 /
\plot 10.941 16.730 11.009 16.739 /
\plot 11.009 16.739 11.079 16.747 /
\plot 11.079 16.747 11.151 16.753 /
\plot 11.151 16.753 11.220 16.760 /
\plot 11.220 16.760 11.295 16.764 /
\plot 11.295 16.764 11.366 16.768 /
\putrule from 11.366 16.768 to 11.441 16.768
\plot 11.441 16.768 11.515 16.770 /
\plot 11.515 16.770 11.589 16.768 /
\plot 11.589 16.768 11.665 16.766 /
\plot 11.665 16.766 11.739 16.760 /
\plot 11.739 16.760 11.813 16.756 /
\plot 11.813 16.756 11.887 16.747 /
\plot 11.887 16.747 11.961 16.739 /
\plot 11.961 16.739 12.033 16.726 /
\plot 12.033 16.726 12.105 16.715 /
\plot 12.105 16.715 12.175 16.701 /
\plot 12.175 16.701 12.245 16.686 /
\plot 12.245 16.686 12.313 16.669 /
\plot 12.313 16.669 12.378 16.650 /
\plot 12.378 16.650 12.444 16.631 /
\plot 12.444 16.631 12.509 16.609 /
\plot 12.509 16.609 12.573 16.586 /
\plot 12.573 16.586 12.636 16.563 /
\plot 12.636 16.563 12.700 16.535 /
\plot 12.700 16.535 12.747 16.516 /
\plot 12.747 16.516 12.795 16.495 /
\plot 12.795 16.495 12.842 16.472 /
\plot 12.842 16.472 12.891 16.447 /
\plot 12.891 16.447 12.939 16.421 /
\plot 12.939 16.421 12.988 16.394 /
\plot 12.988 16.394 13.037 16.364 /
\plot 13.037 16.364 13.085 16.332 /
\plot 13.085 16.332 13.136 16.298 /
\plot 13.136 16.298 13.185 16.264 /
\plot 13.185 16.264 13.236 16.226 /
\plot 13.236 16.226 13.286 16.184 /
\plot 13.286 16.184 13.339 16.142 /
\plot 13.339 16.142 13.390 16.097 /
\plot 13.390 16.097 13.441 16.049 /
\plot 13.441 16.049 13.492 15.998 /
\plot 13.492 15.998 13.542 15.945 /
\plot 13.542 15.945 13.593 15.888 /
\plot 13.593 15.888 13.644 15.828 /
\plot 13.644 15.828 13.695 15.767 /
\plot 13.695 15.767 13.746 15.704 /
\plot 13.746 15.704 13.794 15.636 /
\plot 13.794 15.636 13.843 15.566 /
\plot 13.843 15.566 13.890 15.494 /
\plot 13.890 15.494 13.936 15.420 /
\plot 13.936 15.420 13.983 15.344 /
\plot 13.983 15.344 14.027 15.263 /
\plot 14.027 15.263 14.069 15.183 /
\plot 14.069 15.183 14.112 15.098 /
\plot 14.112 15.098 14.154 15.014 /
\plot 14.154 15.014 14.194 14.925 /
\plot 14.194 14.925 14.232 14.834 /
\plot 14.232 14.834 14.271 14.740 /
\plot 14.271 14.740 14.307 14.645 /
\plot 14.307 14.645 14.340 14.548 /
\plot 14.340 14.548 14.374 14.448 /
\plot 14.374 14.448 14.408 14.347 /
\plot 14.408 14.347 14.440 14.241 /
\plot 14.440 14.241 14.470 14.133 /
\plot 14.470 14.133 14.499 14.023 /
\plot 14.499 14.023 14.520 13.942 /
\plot 14.520 13.942 14.539 13.858 /
\plot 14.539 13.858 14.561 13.773 /
\plot 14.561 13.773 14.580 13.684 /
\plot 14.580 13.684 14.599 13.593 /
\plot 14.599 13.593 14.616 13.500 /
\plot 14.616 13.500 14.635 13.405 /
\plot 14.635 13.405 14.652 13.303 /
\plot 14.652 13.303 14.669 13.202 /
\plot 14.669 13.202 14.688 13.094 /
\plot 14.688 13.094 14.704 12.984 /
\plot 14.704 12.984 14.721 12.867 /
\plot 14.721 12.867 14.738 12.749 /
\plot 14.738 12.749 14.753 12.624 /
\plot 14.753 12.624 14.770 12.495 /
\plot 14.770 12.495 14.787 12.359 /
\plot 14.787 12.359 14.804 12.220 /
\plot 14.804 12.220 14.821 12.076 /
\plot 14.821 12.076 14.836 11.925 /
\plot 14.836 11.925 14.853 11.769 /
\plot 14.853 11.769 14.870 11.608 /
\plot 14.870 11.608 14.887 11.441 /
\plot 14.887 11.441 14.903 11.267 /
\plot 14.903 11.267 14.920 11.089 /
\plot 14.920 11.089 14.937 10.907 /
\plot 14.937 10.907 14.954 10.719 /
\plot 14.954 10.719 14.969 10.526 /
\plot 14.969 10.526 14.986 10.331 /
\plot 14.986 10.331 15.003 10.132 /
\plot 15.003 10.132 15.020  9.931 /
\plot 15.020  9.931 15.037  9.728 /
\plot 15.037  9.728 15.054  9.525 /
\plot 15.054  9.525 15.069  9.322 /
\plot 15.069  9.322 15.085  9.121 /
\plot 15.085  9.121 15.100  8.924 /
\plot 15.100  8.924 15.115  8.729 /
\plot 15.115  8.729 15.128  8.541 /
\plot 15.128  8.541 15.143  8.361 /
\plot 15.143  8.361 15.155  8.187 /
\plot 15.155  8.187 15.166  8.022 /
\plot 15.166  8.022 15.179  7.870 /
\plot 15.179  7.870 15.189  7.728 /
\plot 15.189  7.728 15.198  7.599 /
\plot 15.198  7.599 15.206  7.482 /
\plot 15.206  7.482 15.212  7.379 /
\plot 15.212  7.379 15.219  7.290 /
\plot 15.219  7.290 15.225  7.211 /
\plot 15.225  7.211 15.229  7.148 /
\plot 15.229  7.148 15.232  7.097 /
\plot 15.232  7.097 15.236  7.057 /
\plot 15.236  7.057 15.238  7.027 /
\putrule from 15.238  7.027 to 15.238  7.008
\plot 15.238  7.008 15.240  6.996 /
\putrule from 15.240  6.996 to 15.240  6.987
\putrule from 15.240  6.987 to 15.240  6.985
}%
%
%
\put{\SetFigFont{6}{7.2}{\rmdefault}{\mddefault}{\updefault}{\color[rgb]{0,0,0}$\bar{N}_-$}%
} [lB] at  3.810 10.636
%
%
\put{\SetFigFont{6}{7.2}{\rmdefault}{\mddefault}{\updefault}{\color[rgb]{0,0,0}$\tilde{S}_-$}%
} [lB] at  3.810 12.065
%
%
\put{\SetFigFont{6}{7.2}{\rmdefault}{\mddefault}{\updefault}{\color[rgb]{0,0,0}$\tilde{N}_-$}%
} [lB] at  3.810 20.796
%
%
\put{\SetFigFont{6}{7.2}{\rmdefault}{\mddefault}{\updefault}{\color[rgb]{0,0,0}$\tilde{N}_+$}%
} [lB] at 15.399 16.986
%
%
\put{\SetFigFont{6}{7.2}{\rmdefault}{\mddefault}{\updefault}{\color[rgb]{0,0,0}$\tilde{S}_+$}%
} [lB] at 15.399 15.716
%
%
\put{\SetFigFont{6}{7.2}{\rmdefault}{\mddefault}{\updefault}{\color[rgb]{0,0,0}$\bar{N}_+$}%
} [lB] at 15.399  6.826
%
%
\put{\SetFigFont{6}{7.2}{\rmdefault}{\mddefault}{\updefault}{\color[rgb]{0,0,0}$L$}%
} [lB] at  7.144 20.003
%
%
\put{\SetFigFont{6}{7.2}{\rmdefault}{\mddefault}{\updefault}{\color[rgb]{0,0,0}$L'$}%
} [lB] at 12.700 18.415
%
%
\put{\SetFigFont{6}{7.2}{\rmdefault}{\mddefault}{\updefault}{\color[rgb]{0,0,0}$\cW^-$}%
} [lB] at  7.303 14.605
%
%
\put{\SetFigFont{6}{7.2}{\rmdefault}{\mddefault}{\updefault}{\color[rgb]{0,0,0}$\Ga_l$}%
} [lB] at 11.271 13.018
\linethickness=0pt
\putrectangle corners at  3.778 21.660 and 15.431  6.280
\endpicture}

 \end{center}
\caption{Construction of a barrier}
\end{figure}

Let $\Upsilon$ be the curve in $\cC$  consisting of two line segments: From $(-\frac{\pi}{2},\frac{\pi}{2} - \eta_0)$ along the horizontal 
line $L'$, upto the intersection point of $L$ and $L'$, and then from there along $L$ to $\widetilde{S}_+$.
 The curve $\Upsilon$ provides a barrier for the flow: no orbit can cross it from the region below $\Upsilon$ into the region above $\Upsilon$.
  Let $\widetilde{\cW}^-$ be the unstable manifold af $\widetilde{S}_-$.
  Thus $\widetilde{\cW}^-$ must stay below $\Upsilon$, and as a result the $\om$-limit of $\cW^-$ cannot be $\widetilde{N}_+$ so therefore its 
winding number is not zero or negative.
\end{proof}
\begin{prop}\label{prop:om0}
Given $\la \leq -1+a$ and $E\in[0,E_l(\la))$  the corridor $\cK_1(E,\la)$ of the flow \refeq{dynsysOm} has winding number equal to zero.
\end{prop}
\begin{proof}
Once again we find a barrier that prevents $\cW^-$ from going down:  Let us compute the slope of orbits crossing the line  
$$
L = \{(\xi,\Om)\ |\ \Om = \xi-{\textstyle\frac{\pi}{2}}\}
$$ 
and compare it to the slope of this line. 
$$
j(\xi):= g_{E,\la}(\xi,\xi-{\textstyle\frac{\pi}{2}}) - \cos^2\xi = 2a(1-E) -2(\la+a)\cos^2\xi + 2\ga\sin\xi\cos\xi.
$$
Thus $j(\pm\pi/2) = 2a(1-E)>0$.
  Any interior minimum of $j$ must be achieved at a critical point:
$$
j'(\xi_0) = 0 \implies \xi_0 = {\textstyle\half} \tan^{-1}\frac{\ga}{\la+a} >0.
$$
However we have 
$$
j(\xi_0) = 2a(1-E) -(\la+a)- (\la+a)\cos2\xi_0 +\ga\sin2\xi_0 = 2a(1-E)- \frac{2(\la+a)^3}{(\la+a)^2 +\ga^2}> 2a(1-E).
$$
\begin{figure}[ht]
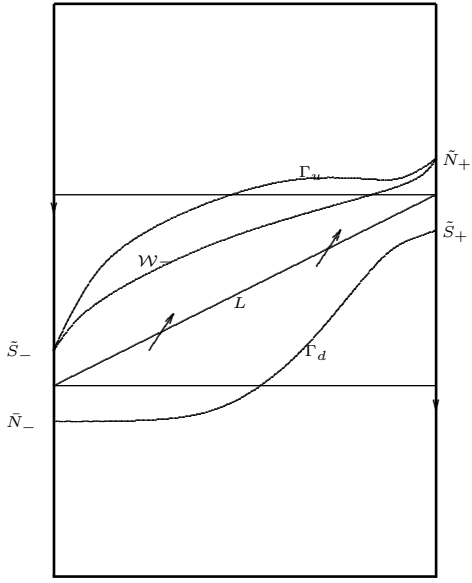

 \begin{center}
%
%
\font\thinlinefont=cmr5
\begingroup\makeatletter\ifx\SetFigFont\undefined%
\gdef\SetFigFont#1#2#3#4#5{%
  \reset@font\fontsize{#1}{#2pt}%
  \fontfamily{#3}\fontseries{#4}\fontshape{#5}%
  \selectfont}%
\fi\endgroup%
\mbox{\beginpicture
\setcoordinatesystem units <0.50000cm,0.50000cm>
\unitlength=0.50000cm
\linethickness=1pt
\setplotsymbol ({\makebox(0,0)[l]{\tencirc\symbol{'160}}})
\setshadesymbol ({\thinlinefont .})
\setlinear
%
%
\linethickness= 0.500pt
\setplotsymbol ({\thinlinefont .})
{\color[rgb]{0,0,0}\putrule from  5.080 11.430 to 15.240 11.430
}%
%
%
\linethickness=1pt
\setplotsymbol ({\makebox(0,0)[l]{\tencirc\symbol{'160}}})
{\color[rgb]{0,0,0}\putrule from  5.080 21.590 to 15.240 21.590
}%
%
%
\linethickness= 0.500pt
\setplotsymbol ({\thinlinefont .})
{\color[rgb]{0,0,0}\putrule from  5.080 17.462 to  5.080 16.034
%
%
\plot  5.017 16.288  5.080 16.034  5.143 16.288 /
}%
%
%
\linethickness= 0.500pt
\setplotsymbol ({\thinlinefont .})
{\color[rgb]{0,0,0}\putrule from 15.240 12.383 to 15.240 10.795
%
%
\plot 15.176 11.049 15.240 10.795 15.304 11.049 /
}%
%
%
\linethickness= 0.500pt
\setplotsymbol ({\thinlinefont .})
{\color[rgb]{0,0,0}\putrule from  5.080 16.510 to 15.240 16.510
}%
%
%
\linethickness=1pt
\setplotsymbol ({\makebox(0,0)[l]{\tencirc\symbol{'160}}})
{\color[rgb]{0,0,0}\putrule from 15.240 21.590 to 15.240  6.315
\putrule from 15.240  6.350 to  5.045  6.350
\putrule from  5.080  6.350 to  5.080 21.590
}%
%
%
\linethickness= 0.500pt
\setplotsymbol ({\thinlinefont .})
{\color[rgb]{0,0,0}\plot  5.080 11.430 15.240 16.510 /
}%
%
%
\linethickness= 0.500pt
\setplotsymbol ({\thinlinefont .})
{\color[rgb]{0,0,0}\plot  7.620 12.383  8.255 13.335 /
%
%
\plot  8.167 13.088  8.255 13.335  8.061 13.159 /
}%
%
%
\linethickness= 0.500pt
\setplotsymbol ({\thinlinefont .})
{\color[rgb]{0,0,0}\plot 12.065 14.605 12.700 15.558 /
%
%
\plot 12.612 15.311 12.700 15.558 12.506 15.381 /
}%
%
%
\linethickness= 0.500pt
\setplotsymbol ({\thinlinefont .})
{\color[rgb]{0,0,0}\plot 15.240 17.462 15.236 17.458 /
\plot 15.236 17.458 15.225 17.452 /
\plot 15.225 17.452 15.208 17.439 /
\plot 15.208 17.439 15.181 17.420 /
\plot 15.181 17.420 15.143 17.393 /
\plot 15.143 17.393 15.096 17.361 /
\plot 15.096 17.361 15.041 17.323 /
\plot 15.041 17.323 14.980 17.283 /
\plot 14.980 17.283 14.912 17.240 /
\plot 14.912 17.240 14.844 17.198 /
\plot 14.844 17.198 14.774 17.158 /
\plot 14.774 17.158 14.707 17.120 /
\plot 14.707 17.120 14.637 17.081 /
\plot 14.637 17.081 14.571 17.050 /
\plot 14.571 17.050 14.506 17.020 /
\plot 14.506 17.020 14.442 16.995 /
\plot 14.442 16.995 14.381 16.974 /
\plot 14.381 16.974 14.319 16.954 /
\plot 14.319 16.954 14.258 16.940 /
\plot 14.258 16.940 14.194 16.927 /
\plot 14.194 16.927 14.131 16.919 /
\plot 14.131 16.919 14.065 16.910 /
\plot 14.065 16.910 13.998 16.906 /
\putrule from 13.998 16.906 to 13.949 16.906
\plot 13.949 16.906 13.900 16.904 /
\plot 13.900 16.904 13.851 16.906 /
\putrule from 13.851 16.906 to 13.799 16.906
\plot 13.799 16.906 13.746 16.908 /
\plot 13.746 16.908 13.688 16.910 /
\plot 13.688 16.910 13.631 16.912 /
\plot 13.631 16.912 13.570 16.916 /
\plot 13.570 16.916 13.506 16.921 /
\plot 13.506 16.921 13.443 16.923 /
\plot 13.443 16.923 13.375 16.929 /
\plot 13.375 16.929 13.305 16.933 /
\plot 13.305 16.933 13.231 16.938 /
\plot 13.231 16.938 13.157 16.942 /
\plot 13.157 16.942 13.079 16.946 /
\plot 13.079 16.946 13.001 16.950 /
\plot 13.001 16.950 12.918 16.954 /
\plot 12.918 16.954 12.835 16.957 /
\plot 12.835 16.957 12.751 16.959 /
\plot 12.751 16.959 12.664 16.963 /
\putrule from 12.664 16.963 to 12.575 16.963
\plot 12.575 16.963 12.486 16.965 /
\putrule from 12.486 16.965 to 12.395 16.965
\plot 12.395 16.965 12.302 16.963 /
\plot 12.302 16.963 12.209 16.961 /
\plot 12.209 16.961 12.116 16.959 /
\plot 12.116 16.959 12.021 16.954 /
\plot 12.021 16.954 11.925 16.948 /
\plot 11.925 16.948 11.830 16.942 /
\plot 11.830 16.942 11.735 16.933 /
\plot 11.735 16.933 11.637 16.925 /
\plot 11.637 16.925 11.540 16.912 /
\plot 11.540 16.912 11.441 16.902 /
\plot 11.441 16.902 11.341 16.887 /
\plot 11.341 16.887 11.242 16.872 /
\plot 11.242 16.872 11.140 16.853 /
\plot 11.140 16.853 11.057 16.840 /
\plot 11.057 16.840 10.975 16.823 /
\plot 10.975 16.823 10.890 16.806 /
\plot 10.890 16.806 10.806 16.787 /
\plot 10.806 16.787 10.719 16.768 /
\plot 10.719 16.768 10.630 16.747 /
\plot 10.630 16.747 10.541 16.724 /
\plot 10.541 16.724 10.448 16.701 /
\plot 10.448 16.701 10.355 16.677 /
\plot 10.355 16.677 10.262 16.650 /
\plot 10.262 16.650 10.164 16.622 /
\plot 10.164 16.622 10.067 16.593 /
\plot 10.067 16.593  9.967 16.563 /
\plot  9.967 16.563  9.866 16.531 /
\plot  9.866 16.531  9.764 16.497 /
\plot  9.764 16.497  9.663 16.463 /
\plot  9.663 16.463  9.559 16.427 /
\plot  9.559 16.427  9.453 16.391 /
\plot  9.453 16.391  9.349 16.353 /
\plot  9.349 16.353  9.243 16.315 /
\plot  9.243 16.315  9.138 16.275 /
\plot  9.138 16.275  9.032 16.233 /
\plot  9.032 16.233  8.926 16.190 /
\plot  8.926 16.190  8.822 16.148 /
\plot  8.822 16.148  8.716 16.104 /
\plot  8.716 16.104  8.613 16.059 /
\plot  8.613 16.059  8.511 16.015 /
\plot  8.511 16.015  8.410 15.968 /
\plot  8.410 15.968  8.310 15.924 /
\plot  8.310 15.924  8.211 15.877 /
\plot  8.211 15.877  8.113 15.828 /
\plot  8.113 15.828  8.018 15.782 /
\plot  8.018 15.782  7.925 15.735 /
\plot  7.925 15.735  7.834 15.687 /
\plot  7.834 15.687  7.745 15.638 /
\plot  7.745 15.638  7.658 15.591 /
\plot  7.658 15.591  7.571 15.543 /
\plot  7.571 15.543  7.489 15.494 /
\plot  7.489 15.494  7.408 15.447 /
\plot  7.408 15.447  7.330 15.399 /
\plot  7.330 15.399  7.254 15.350 /
\plot  7.254 15.350  7.180 15.301 /
\plot  7.180 15.301  7.108 15.253 /
\plot  7.108 15.253  7.040 15.206 /
\plot  7.040 15.206  6.972 15.157 /
\plot  6.972 15.157  6.907 15.107 /
\plot  6.907 15.107  6.828 15.047 /
\plot  6.828 15.047  6.754 14.988 /
\plot  6.754 14.988  6.682 14.927 /
\plot  6.682 14.927  6.612 14.863 /
\plot  6.612 14.863  6.545 14.800 /
\plot  6.545 14.800  6.479 14.734 /
\plot  6.479 14.734  6.416 14.666 /
\plot  6.416 14.666  6.352 14.597 /
\plot  6.352 14.597  6.291 14.525 /
\plot  6.291 14.525  6.229 14.448 /
\plot  6.229 14.448  6.168 14.368 /
\plot  6.168 14.368  6.109 14.285 /
\plot  6.109 14.285  6.049 14.201 /
\plot  6.049 14.201  5.988 14.110 /
\plot  5.988 14.110  5.929 14.017 /
\plot  5.929 14.017  5.867 13.917 /
\plot  5.867 13.917  5.808 13.818 /
\plot  5.808 13.818  5.747 13.712 /
\plot  5.747 13.712  5.687 13.606 /
\plot  5.687 13.606  5.628 13.496 /
\plot  5.628 13.496  5.569 13.386 /
\plot  5.569 13.386  5.512 13.276 /
\plot  5.512 13.276  5.455 13.166 /
\plot  5.455 13.166  5.402 13.060 /
\plot  5.402 13.060  5.351 12.956 /
\plot  5.351 12.956  5.304 12.861 /
\plot  5.304 12.861  5.260 12.770 /
\plot  5.260 12.770  5.222 12.687 /
\plot  5.222 12.687  5.188 12.615 /
\plot  5.188 12.615  5.158 12.554 /
\plot  5.158 12.554  5.133 12.501 /
\plot  5.133 12.501  5.116 12.461 /
\plot  5.116 12.461  5.101 12.429 /
\plot  5.101 12.429  5.091 12.408 /
\plot  5.091 12.408  5.084 12.393 /
\plot  5.084 12.393  5.082 12.387 /
\plot  5.082 12.387  5.080 12.383 /
}%
%
%
\linethickness= 0.500pt
\setplotsymbol ({\thinlinefont .})
{\color[rgb]{0,0,0}\plot 15.240 15.558 15.236 15.555 /
\plot 15.236 15.555 15.225 15.553 /
\plot 15.225 15.553 15.206 15.547 /
\plot 15.206 15.547 15.179 15.538 /
\plot 15.179 15.538 15.141 15.526 /
\plot 15.141 15.526 15.092 15.511 /
\plot 15.092 15.511 15.035 15.492 /
\plot 15.035 15.492 14.971 15.471 /
\plot 14.971 15.471 14.901 15.447 /
\plot 14.901 15.447 14.827 15.420 /
\plot 14.827 15.420 14.753 15.392 /
\plot 14.753 15.392 14.677 15.365 /
\plot 14.677 15.365 14.605 15.335 /
\plot 14.605 15.335 14.533 15.306 /
\plot 14.533 15.306 14.465 15.276 /
\plot 14.465 15.276 14.402 15.246 /
\plot 14.402 15.246 14.338 15.215 /
\plot 14.338 15.215 14.281 15.183 /
\plot 14.281 15.183 14.224 15.151 /
\plot 14.224 15.151 14.171 15.117 /
\plot 14.171 15.117 14.118 15.083 /
\plot 14.118 15.083 14.067 15.045 /
\plot 14.067 15.045 14.017 15.007 /
\plot 14.017 15.007 13.968 14.967 /
\plot 13.968 14.967 13.917 14.922 /
\plot 13.917 14.922 13.879 14.887 /
\plot 13.879 14.887 13.839 14.851 /
\plot 13.839 14.851 13.801 14.812 /
\plot 13.801 14.812 13.758 14.772 /
\plot 13.758 14.772 13.718 14.730 /
\plot 13.718 14.730 13.676 14.685 /
\plot 13.676 14.685 13.633 14.641 /
\plot 13.633 14.641 13.589 14.592 /
\plot 13.589 14.592 13.542 14.541 /
\plot 13.542 14.541 13.496 14.491 /
\plot 13.496 14.491 13.449 14.436 /
\plot 13.449 14.436 13.401 14.379 /
\plot 13.401 14.379 13.352 14.321 /
\plot 13.352 14.321 13.301 14.262 /
\plot 13.301 14.262 13.248 14.201 /
\plot 13.248 14.201 13.197 14.137 /
\plot 13.197 14.137 13.142 14.074 /
\plot 13.142 14.074 13.089 14.008 /
\plot 13.089 14.008 13.034 13.942 /
\plot 13.034 13.942 12.979 13.877 /
\plot 12.979 13.877 12.924 13.809 /
\plot 12.924 13.809 12.869 13.741 /
\plot 12.869 13.741 12.814 13.674 /
\plot 12.814 13.674 12.759 13.606 /
\plot 12.759 13.606 12.702 13.538 /
\plot 12.702 13.538 12.647 13.470 /
\plot 12.647 13.470 12.592 13.403 /
\plot 12.592 13.403 12.535 13.335 /
\plot 12.535 13.335 12.480 13.269 /
\plot 12.480 13.269 12.425 13.202 /
\plot 12.425 13.202 12.368 13.136 /
\plot 12.368 13.136 12.311 13.070 /
\plot 12.311 13.070 12.256 13.005 /
\plot 12.256 13.005 12.198 12.937 /
\plot 12.198 12.937 12.148 12.882 /
\plot 12.148 12.882 12.097 12.825 /
\plot 12.097 12.825 12.046 12.768 /
\plot 12.046 12.768 11.995 12.708 /
\plot 11.995 12.708 11.940 12.651 /
\plot 11.940 12.651 11.885 12.592 /
\plot 11.885 12.592 11.830 12.533 /
\plot 11.830 12.533 11.773 12.471 /
\plot 11.773 12.471 11.714 12.410 /
\plot 11.714 12.410 11.652 12.349 /
\plot 11.652 12.349 11.591 12.287 /
\plot 11.591 12.287 11.527 12.224 /
\plot 11.527 12.224 11.462 12.160 /
\plot 11.462 12.160 11.394 12.099 /
\plot 11.394 12.099 11.326 12.035 /
\plot 11.326 12.035 11.256 11.972 /
\plot 11.256 11.972 11.184 11.908 /
\plot 11.184 11.908 11.113 11.847 /
\plot 11.113 11.847 11.038 11.786 /
\plot 11.038 11.786 10.964 11.722 /
\plot 10.964 11.722 10.888 11.663 /
\plot 10.888 11.663 10.812 11.604 /
\plot 10.812 11.604 10.734 11.544 /
\plot 10.734 11.544 10.655 11.487 /
\plot 10.655 11.487 10.577 11.430 /
\plot 10.577 11.430 10.497 11.375 /
\plot 10.497 11.375 10.418 11.322 /
\plot 10.418 11.322 10.338 11.271 /
\plot 10.338 11.271 10.257 11.223 /
\plot 10.257 11.223 10.179 11.174 /
\plot 10.179 11.174 10.099 11.127 /
\plot 10.099 11.127 10.018 11.083 /
\plot 10.018 11.083  9.938 11.041 /
\plot  9.938 11.041  9.857 11.000 /
\plot  9.857 11.000  9.775 10.962 /
\plot  9.775 10.962  9.694 10.924 /
\plot  9.694 10.924  9.614 10.890 /
\plot  9.614 10.890  9.531 10.856 /
\plot  9.531 10.856  9.449 10.825 /
\plot  9.449 10.825  9.366 10.795 /
\plot  9.366 10.795  9.290 10.770 /
\plot  9.290 10.770  9.212 10.744 /
\plot  9.212 10.744  9.133 10.721 /
\plot  9.133 10.721  9.053 10.700 /
\plot  9.053 10.700  8.970 10.679 /
\plot  8.970 10.679  8.886 10.660 /
\plot  8.886 10.660  8.799 10.643 /
\plot  8.799 10.643  8.710 10.626 /
\plot  8.710 10.626  8.617 10.611 /
\plot  8.617 10.611  8.520 10.596 /
\plot  8.520 10.596  8.420 10.581 /
\plot  8.420 10.581  8.316 10.569 /
\plot  8.316 10.569  8.208 10.558 /
\plot  8.208 10.558  8.096 10.547 /
\plot  8.096 10.547  7.978 10.537 /
\plot  7.978 10.537  7.857 10.528 /
\plot  7.857 10.528  7.730 10.520 /
\plot  7.730 10.520  7.599 10.513 /
\plot  7.599 10.513  7.463 10.507 /
\plot  7.463 10.507  7.324 10.501 /
\plot  7.324 10.501  7.180 10.494 /
\plot  7.180 10.494  7.032 10.490 /
\plot  7.032 10.490  6.881 10.486 /
\plot  6.881 10.486  6.731 10.484 /
\plot  6.731 10.484  6.579 10.480 /
\plot  6.579 10.480  6.428 10.478 /
\putrule from  6.428 10.478 to  6.278 10.478
\plot  6.278 10.478  6.132 10.475 /
\putrule from  6.132 10.475 to  5.992 10.475
\plot  5.992 10.475  5.857 10.473 /
\putrule from  5.857 10.473 to  5.732 10.473
\putrule from  5.732 10.473 to  5.616 10.473
\putrule from  5.616 10.473 to  5.510 10.473
\plot  5.510 10.473  5.417 10.475 /
\putrule from  5.417 10.475 to  5.334 10.475
\putrule from  5.334 10.475 to  5.264 10.475
\putrule from  5.264 10.475 to  5.207 10.475
\plot  5.207 10.475  5.163 10.478 /
\putrule from  5.163 10.478 to  5.129 10.478
\putrule from  5.129 10.478 to  5.105 10.478
\putrule from  5.105 10.478 to  5.091 10.478
\putrule from  5.091 10.478 to  5.084 10.478
\putrule from  5.084 10.478 to  5.080 10.478
}%
%
%
\linethickness= 0.500pt
\setplotsymbol ({\thinlinefont .})
{\color[rgb]{0,0,0}\plot  5.080 12.383  5.082 12.387 /
\plot  5.082 12.387  5.086 12.393 /
\plot  5.086 12.393  5.095 12.406 /
\plot  5.095 12.406  5.110 12.427 /
\plot  5.110 12.427  5.129 12.454 /
\plot  5.129 12.454  5.154 12.493 /
\plot  5.154 12.493  5.186 12.537 /
\plot  5.186 12.537  5.222 12.588 /
\plot  5.222 12.588  5.262 12.643 /
\plot  5.262 12.643  5.306 12.704 /
\plot  5.306 12.704  5.353 12.766 /
\plot  5.353 12.766  5.402 12.829 /
\plot  5.402 12.829  5.453 12.895 /
\plot  5.453 12.895  5.503 12.958 /
\plot  5.503 12.958  5.556 13.020 /
\plot  5.556 13.020  5.609 13.079 /
\plot  5.609 13.079  5.662 13.138 /
\plot  5.662 13.138  5.715 13.193 /
\plot  5.715 13.193  5.768 13.248 /
\plot  5.768 13.248  5.823 13.299 /
\plot  5.823 13.299  5.880 13.352 /
\plot  5.880 13.352  5.937 13.401 /
\plot  5.937 13.401  5.999 13.451 /
\plot  5.999 13.451  6.062 13.500 /
\plot  6.062 13.500  6.128 13.551 /
\plot  6.128 13.551  6.198 13.602 /
\plot  6.198 13.602  6.272 13.652 /
\plot  6.272 13.652  6.322 13.686 /
\plot  6.322 13.686  6.375 13.722 /
\plot  6.375 13.722  6.430 13.758 /
\plot  6.430 13.758  6.485 13.796 /
\plot  6.485 13.796  6.545 13.832 /
\plot  6.545 13.832  6.606 13.871 /
\plot  6.606 13.871  6.672 13.911 /
\plot  6.672 13.911  6.737 13.951 /
\plot  6.737 13.951  6.805 13.991 /
\plot  6.805 13.991  6.877 14.034 /
\plot  6.877 14.034  6.949 14.074 /
\plot  6.949 14.074  7.025 14.118 /
\plot  7.025 14.118  7.104 14.161 /
\plot  7.104 14.161  7.184 14.205 /
\plot  7.184 14.205  7.269 14.252 /
\plot  7.269 14.252  7.353 14.296 /
\plot  7.353 14.296  7.440 14.343 /
\plot  7.440 14.343  7.531 14.389 /
\plot  7.531 14.389  7.622 14.436 /
\plot  7.622 14.436  7.715 14.482 /
\plot  7.715 14.482  7.808 14.529 /
\plot  7.808 14.529  7.906 14.575 /
\plot  7.906 14.575  8.003 14.624 /
\plot  8.003 14.624  8.100 14.671 /
\plot  8.100 14.671  8.200 14.717 /
\plot  8.200 14.717  8.302 14.764 /
\plot  8.302 14.764  8.401 14.810 /
\plot  8.401 14.810  8.503 14.855 /
\plot  8.503 14.855  8.604 14.899 /
\plot  8.604 14.899  8.706 14.944 /
\plot  8.706 14.944  8.807 14.988 /
\plot  8.807 14.988  8.909 15.033 /
\plot  8.909 15.033  9.011 15.075 /
\plot  9.011 15.075  9.112 15.115 /
\plot  9.112 15.115  9.214 15.157 /
\plot  9.214 15.157  9.315 15.198 /
\plot  9.315 15.198  9.415 15.236 /
\plot  9.415 15.236  9.517 15.276 /
\plot  9.517 15.276  9.618 15.314 /
\plot  9.618 15.314  9.718 15.352 /
\plot  9.718 15.352  9.819 15.388 /
\plot  9.819 15.388  9.923 15.424 /
\plot  9.923 15.424 10.020 15.460 /
\plot 10.020 15.460 10.118 15.494 /
\plot 10.118 15.494 10.219 15.528 /
\plot 10.219 15.528 10.319 15.562 /
\plot 10.319 15.562 10.420 15.596 /
\plot 10.420 15.596 10.524 15.629 /
\plot 10.524 15.629 10.628 15.663 /
\plot 10.628 15.663 10.734 15.697 /
\plot 10.734 15.697 10.842 15.731 /
\plot 10.842 15.731 10.950 15.765 /
\plot 10.950 15.765 11.060 15.797 /
\plot 11.060 15.797 11.170 15.831 /
\plot 11.170 15.831 11.280 15.864 /
\plot 11.280 15.864 11.392 15.898 /
\plot 11.392 15.898 11.504 15.932 /
\plot 11.504 15.932 11.618 15.964 /
\plot 11.618 15.964 11.731 15.998 /
\plot 11.731 15.998 11.845 16.032 /
\plot 11.845 16.032 11.957 16.063 /
\plot 11.957 16.063 12.071 16.095 /
\plot 12.071 16.095 12.184 16.127 /
\plot 12.184 16.127 12.294 16.159 /
\plot 12.294 16.159 12.406 16.190 /
\plot 12.406 16.190 12.514 16.222 /
\plot 12.514 16.222 12.622 16.252 /
\plot 12.622 16.252 12.728 16.281 /
\plot 12.728 16.281 12.831 16.311 /
\plot 12.831 16.311 12.935 16.341 /
\plot 12.935 16.341 13.034 16.368 /
\plot 13.034 16.368 13.132 16.396 /
\plot 13.132 16.396 13.227 16.423 /
\plot 13.227 16.423 13.318 16.451 /
\plot 13.318 16.451 13.409 16.476 /
\plot 13.409 16.476 13.496 16.502 /
\plot 13.496 16.502 13.578 16.525 /
\plot 13.578 16.525 13.661 16.550 /
\plot 13.661 16.550 13.739 16.574 /
\plot 13.739 16.574 13.813 16.595 /
\plot 13.813 16.595 13.887 16.618 /
\plot 13.887 16.618 13.957 16.639 /
\plot 13.957 16.639 14.023 16.660 /
\plot 14.023 16.660 14.086 16.681 /
\plot 14.086 16.681 14.150 16.703 /
\plot 14.150 16.703 14.209 16.722 /
\plot 14.209 16.722 14.311 16.758 /
\plot 14.311 16.758 14.404 16.794 /
\plot 14.404 16.794 14.491 16.828 /
\plot 14.491 16.828 14.569 16.861 /
\plot 14.569 16.861 14.641 16.895 /
\plot 14.641 16.895 14.709 16.931 /
\plot 14.709 16.931 14.770 16.967 /
\plot 14.770 16.967 14.827 17.005 /
\plot 14.827 17.005 14.880 17.043 /
\plot 14.880 17.043 14.931 17.084 /
\plot 14.931 17.084 14.978 17.124 /
\plot 14.978 17.124 15.020 17.168 /
\plot 15.020 17.168 15.060 17.211 /
\plot 15.060 17.211 15.098 17.253 /
\plot 15.098 17.253 15.130 17.295 /
\plot 15.130 17.295 15.160 17.333 /
\plot 15.160 17.333 15.183 17.369 /
\plot 15.183 17.369 15.204 17.399 /
\plot 15.204 17.399 15.219 17.424 /
\plot 15.219 17.424 15.229 17.441 /
\plot 15.229 17.441 15.236 17.454 /
\plot 15.236 17.454 15.238 17.460 /
\plot 15.238 17.460 15.240 17.462 /
}%
%
%
\put{\SetFigFont{6}{7.2}{\rmdefault}{\mddefault}{\updefault}{\color[rgb]{0,0,0}$\cW-$}%
} [lB] at  7.303 14.605
%
%
\put{\SetFigFont{6}{7.2}{\rmdefault}{\mddefault}{\updefault}{\color[rgb]{0,0,0}$L$}%
} [lB] at  9.842 13.494
%
%
\put{\SetFigFont{6}{7.2}{\rmdefault}{\mddefault}{\updefault}{\color[rgb]{0,0,0}$\Ga_d$}%
} [lB] at 11.748 12.224
%
%
\put{\SetFigFont{6}{7.2}{\rmdefault}{\mddefault}{\updefault}{\color[rgb]{0,0,0}$\Ga_u$}%
} [lB] at 11.589 16.986
%
%
\put{\SetFigFont{6}{7.2}{\rmdefault}{\mddefault}{\updefault}{\color[rgb]{0,0,0}$\tilde{S}_-$}%
} [lB] at  3.810 12.224
%
%
\put{\SetFigFont{6}{7.2}{\rmdefault}{\mddefault}{\updefault}{\color[rgb]{0,0,0}$\bar{N}_-$}%
} [lB] at  3.810 10.319
%
%
\put{\SetFigFont{6}{7.2}{\rmdefault}{\mddefault}{\updefault}{\color[rgb]{0,0,0}$\tilde{N}_+$}%
} [lB] at 15.399 17.304
%
%
\put{\SetFigFont{6}{7.2}{\rmdefault}{\mddefault}{\updefault}{\color[rgb]{0,0,0}$\tilde{S}_+$}%
} [lB] at 15.399 15.399
\linethickness=0pt
\putrectangle corners at  3.778 21.660 and 15.431  6.280
\endpicture}

 \end{center}
\caption{Barrier for zero winding number}
\end{figure}

Thus $j(\xi)>0$ for all $\xi\in(-\frac{\pi}{2},\frac{\pi}{2})$.
  It follows that the slope of any orbit crossing the line is greater than the slope of the line.
  Thus orbits cannot cross this line from above to below.
  In particular, the orbit $\cW^-$ starts at $\widetilde{S}_-$, which is above this line.
  Hence $\cW^-$ cannot end at any copy of the node $N_+$ other than $\widetilde{N}_+$, so its winding number cannot be positive.
  Since the region $\cN$ is connected, once $\cW^-$ leaves $\cP$ and enters $\cN$, its $\Om$ must decrease, hence $\cW^-$ 
cannot end at any copy of the node $N_+$  that is  higher than $\widetilde{N}_+$ either, and therefore the winding number of 
$\cW^-$ cannot be negative, hence $w(\cW^-) = 0$.
\end{proof}
Let $\la_0 = -1+a$. 
 The above two propositions, in conjunction with the following immediate corollary of Proposition~\ref{prop:ssexists}, 
establish the existence a saddles connector $\cS_0^\Om(E_1,\la_0)$ forthe flow \refeq{dynsysOm}, for some $E_1\in(0,1)$:
\begin{cor}\label{cor:om}
Let $\la\leq  -1+a$ be fixed.
 Suppose that there exists $0\leq E_1< E_2< 1$ such that the flow \refeq{dynsysOm} has corridors $\cK_1(E_1,\la)$ 
and $\cK_1(E_2,\la)$ with $w(\cK_1(E_1,\la))=0$ and $w(\cK_1(E_2,\la))\geq 1$.  
  Then there is an $E\in(E_1,E_2)$ such that \refeq{dynsysOm} has a saddles connector $\cS(E,\la)$.
\end{cor}
\begin{proof}
Proposition~\ref{prop:ssexists} applies, with $E$ playing the role of the parameter $\mu$.
\end{proof}

 Let $n\geq 1$.
Suppose that given $\la_{n-1}\leq -1+a$ a saddles connector $\cS^\Om_{n-1}=\cS^\Om(E_n,\la_{n-1})$ has been found for \refeq{dynsysOm},
for some $E_n\in(0,1)$.
 In the previous subsection we saw how this newly-found $E_n$ can be used to prove the existence of a saddles connector 
for the $\Theta$ flow \refeq{dynsysTh}, namely $\cS^\Theta_n:= \cS^\Theta_n(E_n,\la_n)$ for some $\la_n\leq -1+a$.
  Coming back to the $\Om$ flow then, given the updated value $\la=\la_n$, a new saddles connector $\cS_n^\Om$ needs to be 
found with an updated energy $E_{n+1}$, {\em given that a saddles connector $\cS^\Om_{n-1}(E_{n},\la_{n-1})$ already exists.} 
%
Then, for all $\la'\in (\la,-1+a)$, there exists a corridor $\cK_1(E,\la')$ of winding number $w(\cK_1) = 1$ for \refeq{dynsysOm}.
%
More generally, we have
\begin{thm}\label{thm:E}
Fix $a\in(0,\half)$ and $\ga\in (- \sqrt{2a(1-2a)},0)$.
  Then given any $\la\in[-1-a,-1+a]$, there exists a unique $$E = \cE(\la)\in (0,1)$$ such that \refeq{dynsysOm} has a 
saddles connector $\cS^\Om(E,\la)$.
 Moreover, $\cE$ is a $C^1$ function, and $|\frac{\p \cE}{\p \la}| < \frac{1}{a}$. 
\end{thm}
\begin{proof}
  Existence of a saddles connector is guaranteed by Propositions~\ref{prop:om1} and ~\ref{prop:om0}, and Corollary~\ref{cor:om}.
  To see uniqueness, suppose that there exists two saddles connectors $\cS^\Om(E,\la)$ and $\cS^\Om(E',\la)$ for \refeq{dynsysOm}
 for $E$ and $E'$ in $(0,1)$, and suppose $E<E'$.
Let $\Om_{E,\la}$ be the $\Om$ component of $\cS^\Om(E,\la)$.
  We have
$$
g_{E',\la}(\xi,\Om_{E,\la}) = g_{E,\la}(\xi,\Om_{E,\la})-2a(E'-E) <\dot{\Om}_{E,\la}.
$$
It thus follows that orbits of the $(E',\la)$ flow can only cross $\cS^\Om(E,\la)$ from above to below.
  On the other hand, since $E'>E$ the equilibrium point $\widetilde{S}_-(E',\la)$ is situated below $\widetilde{S}_-(E,\la)$, 
while $\widetilde{S}_+(E',\la)$ is above $\widetilde{S}_+(E,\la)$.
  Since $\cS^\Om(E',\la)$ coincides with both $\cW^-(E',\la)$ and $\cW^+(E',\la)$ it begins below $\cS^\Om(E,\la)$ and 
it ends above it, which is a contradiction, hence $E'=E$.

Given $\la$, let $\cE(\la)$ denote the unique value of $E$ 
for which a saddles connector $\cS^\Om(\cE(\la),\la) = (\xi(\tau),\Om_{\cE(\la),\la}(\tau))$ exists.
  We now prove that $\cE$ is a $C^1$ function: Consider the two initial value problems for $\Om^\pm(\tau)$, 
the $\Om$ components of $\cW^\pm$:
$$
\dot{\Om}^\pm = g_{E,\la}(\xi(\tau),\Om^\pm),\qquad \Om^-(-\infty) = -\pi+\cos^{-1}E,\quad \Om^+(\infty) = -\cos^{-1}E
$$
By standard ODE theory these two problems have unique smooth solutions $\Om^\pm_{E,\la}(\tau)$ which also depend smoothly 
on the parameters $\la$ and $E$ (so long as $E<1$) for any finite $\tau$. 

Next recall that $\cW^-$ is a saddles connector if it coincides with $\cW^+$, which will be the case if these two orbits intersect 
at one point, e.g. if $\Om_{E,\la}^+(0) = \Om_{E,\la}^-(0)$.
  Let us define a smooth function
\beq\label{eq:IFT}
\cF(E,\la) := \Om^+_{E,\la}(0) -\Om^-_{E,\la}(0)
\eeq
Let $\la_0\in[-1-a,-1+a]$ be fixed, and set $E_0=\cE(\la_0)$.
  Then $\cF(E_0,\la_0) = 0$.
  By the Implicit Function Theorem, if 
\beq\label{cond:IFT}
\frac{\p \cF}{\p E} (E_0,\la_0) \ne 0,
\eeq
then there is a neighborhood $\cI$ of $\la_0$ and a $C^1$ function $\tilde{\cE}$ defined on $\cI$ such that 
$\cF(\tilde{\cE}(\la),\la) =0$ for all $\la\in\cI$. 
 By the uniqueness result we have already shown, we must have $\tilde{\cE} = \cE$.
  Thus we only need to verify the condition \refeq{cond:IFT}.

For $\la\in[-1-a,-1+a]$ and $E\in(0,1)$, let
 $$
u_\pm(\tau) := \frac{\p}{\p E} \Om^\pm_{E,\la}(\tau).
$$
Then $u_\pm$ satisfy the linear ODEs 
\beq\label{eq:v}
\frac{du_\pm}{d\tau} = P_\pm(\tau) u_\pm -2a,\qquad P_\pm(\tau):=-2a\sin\xi(\tau)\sin\Om^\pm_{E,\la}(\tau) + 2\la \cos\xi(\tau)\cos\Om^\pm_{E,\la}(\tau)
\eeq
together with the initial conditions
$$
u_-(-\infty) = \frac{-1}{\sqrt{1-E^2}},\qquad u_+(\infty) = \frac{1}{\sqrt{1-E^2}}.
$$
Moreover,
$$
\cF(E,\la) = u_+(0) - u_-(0).
$$
For $\tau_1^{},\tau_2^{}\in\RR$ let 
$$
U_\pm(\tau_1^{},\tau_2^{}) := e^{-\int_{\tau_1^{}}^{\tau_2^{}} P_\pm(\tau) d\tau}.
$$
Solving the ODEs \refeq{eq:v} for $u_\pm$ we obtain
\beq\label{Us}
U_\pm(\tau_1^{},\tau_2^{})u_\pm(\tau_2^{}) = u_\pm(\tau_1^{}) -2a \int_{\tau_1^{}}^{\tau_2^{}} U_\pm(\tau_1^{},\tau) d\tau.
\eeq
Note that 
$$
\lim_{\tau\to -\infty}P_-(\tau) = -2a\sqrt{1-E^2}< 0,\qquad 
\lim_{\tau\to \infty} P_+(\tau) = 2a\sqrt{1-E^2} > 0.
$$
Thus for any fixed $\tau$, 
$$
U_-(\tau,\tau_1^{})\to 0 \mbox{ as } \tau_1^{}\to -\infty,\qquad
U_+(\tau,\tau_2^{}) \to 0\mbox{ as } \tau_2^{} \to \infty.
$$
And so, from \refeq{Us} we obtain
\beq\label{solv}
u_-(\tau) =  -2a \int_{-\infty}^\tau U_-(\tau,\tau') d\tau',\qquad u_+(\tau) = 2a\int_\tau^\infty U_+(\tau,\tau') d\tau'.
\eeq
We know that $\Om^+_{E_0,\la_0}(\tau) = \Om^-_{E_0,\la_0}(\tau)$ for all $\tau$. 
 Hence $U_+(\tau_1^{},\tau_2^{}) = U_- (\tau_1^{},\tau_2^{})=: U(\tau_1^{},\tau_2^{})$ and thus 
$$
\frac{\p \cF}{\p E} (E_0,\la_0) = 2a\int_{-\infty}^\infty U(0,\tau') d\tau' >0
$$
so that \refeq{cond:IFT} is clearly satisfied.
  Since $\la_0$ was arbitrary we have shown that $\cE\in C^1((-1-a,-1+a))$.

We can furthermore compute the derivative of $\cE$ by implicit differentiation.
  Let 
$$
v_\pm(\tau) := \frac{\p}{\p\la} \Om^\pm_{E,\la}(\tau).
$$
Then $v_\pm$ satisfy
$$
\dot{v}_\pm = P_\pm(\tau) v_\pm + 2 \cos\xi(\tau)\sin\Om^\pm_{E,\la}(\tau),\qquad v_-(-\infty) =0,\quad v_+(\infty) = 0.
$$
Thus by a similar argument to above,
$$
v_-(\tau) = \int_{-\infty}^\tau U_-(\tau,\tau') \cos\xi(\tau')\sin\Om^-_{E,\la}(\tau') d\tau',\qquad
v_+(\tau) = \int_\tau^\infty U_+(\tau,\tau') \cos\xi(\tau')\sin\Om^+_{E,\la}(\tau') d\tau',
$$
so that
$$
\frac{\p \cF}{\p \la}(E_0,\la_0) = -2\int_{-\infty}^\infty U(0,\tau') \cos\xi(\tau')\sin\Om_{E_0,\la_0}(\tau') d\tau',
$$
and thus
$$
\frac{d\cE}{d\la} = -\frac{{\p \cF}/{\p \la}}{{\p \cF}/{\p E}} = 
 \frac{\int_{-\infty}^\infty U(0,\tau) \cos\xi(\tau)\sin\Om_{E,\la}(\tau) d\tau}{a\int_{-\infty}^\infty U(0,\tau)d\tau}.
$$
 Moreover, clearly
$$
\left|\int_{-\infty}^\infty U(0,\tau) \cos\xi(\tau)\sin\Om_{E,\la}(\tau) d\tau\right| <
\int_{-\infty}^\infty U(0,\tau)d\tau.
$$
so that
$$
\left|\frac{d\cE}{d\la}\right|  < \frac{1}{a}.
$$
\end{proof}

\subsubsection{The iteration argument}
\begin{thm}\label{thereAREsaddlesaddleCONNECTORS}
Let $a\in (0,\half)$ and $\ga \in (- \sqrt{2a(1-2a)},0)$ be fixed.  There exists a $\la\in [-1-a,-1+a] $ and $E \in (0,1)$ such that \refeq{dynsysTh} has a saddles connector $\cS^\Theta(E,\la)$ and
 \eqref{dynsysOm} has a saddles connector $\cS^\Om(E,\la)$.
\end{thm}

\begin{proof}
Set $\la_0 = -1+a$.  For $n\geq 1$ let
$$E_n := \cE(\la_{n-1})\in(0,1),\qquad \la_n := \La(E_n)\in [-1-a,-1+a].
$$
Thus 
$$
E_{n+1}= \cE(\La(E_{n})).
$$
By Theorems~\ref{thm:E} and \ref{thm:La} we have
$$
|E_{n+1} - E_n| \leq \de |E_n- E_{n-1}| 
$$
where
$$
\de := \max_{0\leq E\leq1}\left|\frac{d\La}{dE}\right| \max_{-1-a\leq \la\leq -1+a}\left| \frac{d\cE}{d\la}\right|  < 
a\cdot \frac{1}{a} = 1
$$
Thus by the contraction mapping theorem, the sequence $E_n$ converges, and thus so does the sequence $\la_n$.
  Let $\la = \lim_{n\to \infty} \la_n \in [-1-a,-1+a]$ and $E := \lim_{n\to \infty} E_n$.
  We must have $E = \cE(\la)$ and thus $0<E<1$.  
\end{proof}

\section{Summary and Outlook}

        We have studied the Dirac equation for a point electron in static, electromagnetic, flat 
spacetimes with Zipoy topology which include the zero-gravity limit of the electromagnetic Kerr--Newman spacetimes as 
special case, but which can feature a generalization of the Appell--Sommerfeld electromagnetic fields with any 
charge $Q$ and current $I$ one wants; the zero-$G$ Kerr--Newman spacetimes correspond to $Q=I\pi a$.
        In contrast to similar-spirited studies of the Dirac equation for a point electron on the Kerr--Newman spacetime,
which are plagued by the presence of a Cauchy horizon and regions of closed timelike loops, 
\cite{KalMil1992,BelMar99,FinsterETalDperDNEerr,FinsterETalDperDNE,FinsterETalDKNa,FinsterETalDKNb,WINKLMEIERa,WINKLMEIERb,WINKLMEIERc,BelCac2010},
our zero-$G$ spacetimes do not possess any such physically troublesome features.
       Moreover, by working with the topologically non-trivial maximal analytical extension of z$G$KN and its electromagnetic
fields, our treatment does not encounter physically troublesome problems like infinite charges, currents, and masses, which plague the topologically trivial Minkowski spacetime interpretations 
\cite{Isr70,Lyn04,GaiLyn07,Kai04}.\footnote{For $G>0$, the disk singularity of the ``single-leafed" truncation of the maximal analytically extended z$G$KN spacetime also features a negative mass density rotating at superluminal speed.}

        We proved that the  spectrum of any self-adjoint extension of the pertinent Dirac Hamiltonian 
is symmetric about zero; this result holds for any charge $Q$ and current $I$ of the generalized z$G$KN spacetimes.\footnote{This 
  result was to be expected, for on the z$G$KN spacetime the Dirac electron ``sees'' a charge $Q$ in one sheet and a charge $-Q$ 
in the other, and recalling that the Dirac operator for an electron in the Coulomb potential of a point proton in Minkowski 
spacetime has a symmetric continuous spectrum, plus a positive point spectrum which maps into a negative mirror image of it 
under the switch of the coupling constant $-e^2\to e^2$ (by charge conjugation).}
  	We have also shown that the formal Dirac Hamiltonian on a spacelike slice of the maximal analytically 
extended static z$G$KN spacetime is essentially self-adjoint.\footnote{We note that by Stone's theorem there exists a 
  unitary one-parameter group on the Hilbert space, generated by the unique self-adjoint extension of the Hamiltonian, 
  which yields the time-evolution of the Dirac bi-spinors on the spacelike static slice of the z$G$KN spacetime. 
   Thus the naked ring singularity does not cause any trouble.}
           We also showed that the self-adjoint Dirac operator on the z$G$KN spacetime has a continuous spectrum 
with a gap about zero which, under two smallness conditions, contains a pure point spectrum.

         Our results are far from exhaustive. 
         In the following we list a number of interesting open problems which we hope will be solved in some future work.

         We begin with the Dirac point electron in z$G$KN spacetimes:

\begin{itemize}

\item{\textbf{Problem 1}}:
Characterize the point spectrum of the Dirac Hamiltonian on z$G$KN in complete detail; to the extent possible, compute it
analytically, or at least numerically in some representative situations.

\begin{rem}\textit{
We suspect that there is a countably infinite set of energy eigenvalues
which correspond  with pairs of saddles connectors of arbitrary  winding numbers $(w_1,w_2)\in\Zset^2$;
the possibility of such saddles connectors we already established, see Theorem \ref{thereAREsaddlesaddleCONNECTORS}.
}
\end{rem}

\item {\textbf{Problem 2}}: 
Discuss the generalized scattering problem for Dirac spinor fields on the z$G$KN spacetimes.
In particular, investigate the evolution when (part of) the Dirac spinor field ``dives'' through the ring from one
sheet to the other.

\end{itemize}

We now come to the Dirac point electron in z$G$K spacetimes equipped with a generalization of the Appell--Sommerfeld electromagnetic 
fields to arbitrary charge $Q$ and current $I$:

\begin{itemize}
\item {\textbf{Problem 3}}: 
For the formal Dirac Hamiltonian on the $(Q,I)$-generalization of the z$G$KN spacetime (given $a$), show
that essential self-adjointness holds if the ``coupling constant'' $(Q-I\pi a)e$ is small in magnitude; 
perhaps using a so-called Hardy--Dirac type estimate.  

\item{\textbf{Problem 4}}:
Suppose essential self-adjointness fails if $|(Q-I\pi a)e|$ is too large. 
If so, what is the sharp constant for  $|(Q-I\pi a)e|$?
Determine the $(Q,I)$-parameter regimes (given $a$)
in which the Dirac Hamiltonian on a generalization of the z$G$KN spacetime with Sommerfeld  fields \refeq{def:AQI}
has several self-adjoint extensions, respectively has no self-adjoint extension.
Amongst the self-adjoint extensions, can one identify a distinguished one?

\item{\textbf{Problem 5}}:
In the cases of self-adjointness, characterize the spectrum of the Dirac Hamiltonian in complete detail; to the extent possible, 
compute the spectrum analytically, or at least numerically in some representative situations. 
In particular:

\item{\textbf{Problem 5a}}:
 The continuous spectrum of quantum physical operator families is usually very robust.  
Show that Theorem \ref{thm:essspec} holds for the $(Q,I)$ generalization of our Dirac operators, 
at least as long as $(Q -I\pi a)e$ is small in magnitude. 

\item{\textbf{Problem 5b}}:
Same consideration as above, for Theorem \ref{thm:ptspec}; thus:
  Can an eigenvalue of $\hat{H}$ on $\sf H$ for $Q =I\pi a$ be continuously deformed into an eigenvalue of $\hat{H}$ on $\sf H$ 
for $Q\ne I\pi a$ as long as $(Q -I\pi a)e$ is sufficiently small in magnitude? 
If so, does the size of the neighborhood of $Q=I\pi a$ into which an eigenvalue can be continued  depend on the eigenvalue, 
or can one have a uniform control on the spectrum w.r.t. the coupling constant $(Q -I\pi a)e$?
Is it possible that the point spectrum disappears completely if $(Q -I\pi a)e$ becomes too large in magnitude?

\item {\textbf{Problem 6}}: 
Same as Problem 2, now for Dirac spinor fields on z$G$K equipped with generalizations of the Appell--Sommerfeld 
fields to arbitrary $Q$ and $I$. 

\end{itemize}

So far our problems concern the Dirac equation on the zero-gravity limit case of the KN spacetimes, and its generalization
to arbitrary $Q$ and $I$. 
To make contact with the existing studies of Dirac's electron in Kerr--Newman spacetimes, the following bifurcation
problem suggests itself:

\begin{itemize}
\item {\textbf{Problem 7}}: 
Deform these zero-gravity spacetimes perturbatively by ``switching on'' $G$ and discuss the Dirac equation on them 
perturbatively as well.
In particular, is it possible to perturb into generalizations of the Kerr--Newman spacetime with gyromagnetic ratios 
amounting to $g$-factors $g\neq g_{\mbox{\tiny{KN}}}^{}=2$, or is such a perturbation feasible only if $g =g_{\mbox{\tiny{KN}}}^{}=2$,
viz. if $Q=I\pi a$?
\end{itemize}

In all the above problems,  the energy(-density)-momentum(-density)-stress tensor is of classical electromagnetic nature.
The Dirac spinor field does not influence the spacetime structure (in the zero-gravity limit the electromagnetic
fields do not influence the spacetime structure, either).
More to the point, Dirac's point electron is treated as a \emph{test particle} in this paper and in all the above problems,
which should be a good approximation if the z$G$KN ring singularity is much more massive and can approximately be 
treated as ``infinitely massive.''
 Yet, since test particles are only physical fiction, no matter how useful practically, an obvious task is to
investigate the Dirac electron \emph{not} as a test particle. 
Thus:

\begin{itemize}

\item {\textbf{Problem 8}}: 
Treat the quantum-mechanical interaction of the Dirac electron with the ring singularity of the zero-$G$ Kerr--Newman spacetime 
symmetrically as a two-body problem.

\begin{rem}\textit{
Of course, the same problem has not even been completely solved yet for the simpler setting of ``Dirac Hydrogen'' in Minkowski
spacetime; i.e., how to go beyond the traditional textbook problem where the relativistic Hydrogen problem is treated
by solving the Dirac equation of a point electron in flat Minkowski spacetime containing an \emph{infinitely massive}
positive point charge (representing the proton).
   Traditionally the problem of finite mass of the proton (or nucleus, more generally) 
is addressed by perturbation theory, starting from Pauli's two-body equation and adding
``relativistic corrections'' in powers of $1/c$; or by perturbative QED-type calculations.
 The constrained relativistic two-body approach of Bethe--Salpeter \cite{BeSaBOOK} is perhaps the
closest one has come to solving this problem.
  In a similar vein one may be able to take a finite ``ADM mass'' of the ring singularity into account.
 }
\end{rem}

\item {\textbf{Problem 9}}: 
Same as Problem 8, but now with the $(Q,I)$ generalizations of z$G$KN.

\item {\textbf{Problem 10}}: 
Same as Problems 8 and 9, but now perturbatively for $G>0$.

\item {\textbf{Problem 11}}: 
Compute the feedback of the Dirac spinor field onto the spacetime structure perturbatively when $G>0$.
This amounts to the perturbative discussion of the so-called Einstein--Maxwell--Dirac system for small $G$. 
In this problem the energy-momentum-stress tensor, in addition to classical electromagnetic fields, also
involves the Dirac spinor field which now influences the spacetime structure.

\end{itemize}

All these problems are difficult, and the amount of work needed to solve all of them can only be handled by the involvement of
many mathematical physicists. 
In this vein we hope that our paper inspires some readers to join us in our pursuit.
We ourselves have made some progress on problems 1, 3, 8, and 9, which we plan to report in forthcoming publications.

\bigskip

\noindent{\textbf{Acknowledgement:}} We thank Friedrich Hehl for his comments on an earlier version of
this paper and for bringing several references to our attention.
We are  also thankful to the anonymous referee for constructive comments.

\baselineskip=13pt
\bibliographystyle{plain}
\def\cprime{$'$}

\end{document}